\documentclass[12pt]{extarticle}%
\usepackage{eurosym}
\usepackage[top=2.5cm, bottom=2.5cm, left=2.2cm, right=2.2cm]{geometry}
\usepackage[utf8]{inputenc}
\usepackage{natbib}
\usepackage{amsmath,amsfonts,amssymb,amsthm,enumitem}
\usepackage{graphicx}
\usepackage{rotating}
\usepackage{subcaption}
\usepackage{multirow}
\usepackage{amsthm}
\usepackage{pgf}
\usepackage{cases}
\usepackage{setspace}
\usepackage{bibunits,booktabs}
\usepackage{hyperref}
\usepackage{bbm}
\usepackage{comment}
\usepackage{tikz,calc}
\usepackage{accents}
\usepackage{xcolor}
\usepackage{soul}
\usepackage{pgfplots}
\usepackage{graphicx}
\usepackage{caption}
\usepackage{pdflscape}
\usepackage{float}
\usepackage{array}
\usepackage{multirow}
\usepackage{needspace}
\usepackage{placeins}
\usepackage{amsmath}

\usepackage{amsfonts}
\usepackage{amssymb}%
\providecommand{\U}[1]{\protect\rule{.1in}{.1in}}
\newtheorem{theorem}{Theorem}
\newtheorem{lemma}{Lemma}

\newtheorem{corollary}{Corollary}
\newtheorem{assumption}{Assumption}

\newtheorem{proposition}{Proposition}
\newtheorem{remark}{Remark}

\foreach \x in {A,B,...,Z}{
\expandafter\xdef \csname c\x\endcsname{\noexpand\mathcal{\x}}
\expandafter\xdef \csname e\x\endcsname{\noexpand\mathbb{\x}}
\expandafter\xdef \csname b\x\endcsname{\noexpand\mathbf{\x}}
\expandafter\xdef \csname br\x\endcsname{\noexpand\overline{\x}}
\expandafter\xdef \csname hh\x\endcsname{\noexpand\widehat{\x}}
\expandafter\xdef \csname tt\x\endcsname{\noexpand\widetilde{\x}}
\expandafter\xdef \csname brb\x\endcsname{\noexpand\overline{ \noexpand \mathbf{\x}}}
\expandafter\xdef \csname hhb\x\endcsname{\noexpand\widehat{\noexpand\mathbf{\x}}}
\expandafter\xdef \csname ttb\x\endcsname{\noexpand\widetilde{\noexpand\mathbf{\x}}}
}
\foreach \x in {a,b,...,z}{
\expandafter\xdef \csname b\x\endcsname{\noexpand\mathbf{\x}}
\expandafter\xdef \csname br\x\endcsname{\noexpand\overline{\x}}
\expandafter\xdef \csname hh\x\endcsname{\noexpand\widehat{\x}}
\expandafter\xdef \csname tt\x\endcsname{\noexpand\widetilde{\x}}
\expandafter\xdef \csname brb\x\endcsname{\noexpand\overline{ \noexpand \mathbf{\x}}}
\expandafter\xdef \csname hhb\x\endcsname{\noexpand\widehat{\noexpand\mathbf{\x}}}
\expandafter\xdef \csname ttb\x\endcsname{\noexpand\widetilde{\noexpand\mathbf{\x}}}
}
\foreach \x in {beta, delta, mu, lambda, psi, phi, rho, theta, varepsilon, gamma, zeta}{
\expandafter\xdef\csname b\x\endcsname{\noexpand\boldsymbol{\csname\x\endcsname}}
\expandafter\xdef\csname ttb\x\endcsname{\noexpand\widetilde{\noexpand\boldsymbol{\csname\x\endcsname}}}
\expandafter\xdef\csname hhb\x\endcsname{\noexpand\widehat{\noexpand\boldsymbol{\csname\x\endcsname}}}
\expandafter\xdef\csname brb\x\endcsname{\noexpand\overline{\noexpand\boldsymbol{\csname\x\endcsname}}}
}
\foreach \x in {Lambda, Omega, Pi, Psi, Sigma, Theta}{
\expandafter\xdef\csname b\x\endcsname{\noexpand\boldsymbol{\csname\x\endcsname}}
\expandafter\xdef\csname ttb\x\endcsname{\noexpand\widetilde{\noexpand\boldsymbol{\csname\x\endcsname}}}
\expandafter\xdef\csname hhb\x\endcsname{\noexpand\widehat{\noexpand\boldsymbol{\csname\x\endcsname}}}
\expandafter\xdef\csname brb\x\endcsname{\noexpand\overline{\noexpand\boldsymbol{\csname\x\endcsname}}}
}

\newcommand{\Rom}[1]{\uppercase\expandafter{\romannumeral #1\relax}}
\newcommand{\rom}[1]{\lowercase\expandafter{\romannumeral #1\relax}}

\makeatletter
\let\oldabs\abs
\def\abs{\@ifstar{\oldabs}{\oldabs*}}
\let\oldnorm\norm
\def\norm{\@ifstar{\oldnorm}{\oldnorm*}}
\let\oldpwrap\pwrap
\def\pwrap{\@ifstar{\oldpwrap}{\oldpwrap*}}
\let\oldbwrap\bwrap
\def\bwrap{\@ifstar{\oldbwrap}{\oldbwrap*}}
\let\oldcbwrap\cbwrap
\def\cbwrap{\@ifstar{\oldcbwrap}{\oldcbwrap*}}
\makeatother
\DeclareMathOperator*\argmin{\arg\min}

\hypersetup{hypertexnames=false}
\newcommand{\bmtbibprefix}{main.}
\makeatletter
\def\hyper@natlinkstart#1{%
  \Hy@backout{#1}%
  \hyper@linkstart{cite}{cite.\bmtbibprefix#1}%
  \def\hyper@nat@current{#1}%
}
\def\hyper@natlinkbreak#1#2{%
  \hyper@linkend#1\hyper@linkstart{cite}{cite.\bmtbibprefix#2}%
}
\def\hyper@natanchorstart#1{%
  \Hy@raisedlink{\hyper@anchorstart{cite.\bmtbibprefix#1}}%
}
\makeatother

\begin{document}

\title{Model Selection in High Dimensional Linear Regression using Boosting with
Multiple Testing%
\thanks{We would like to thank Alex Chudik, Arturas Juodis, Frank Kleibergen, Hashem Pesaran, Martin Weidner and seminar participants and attendees at the Tinbergen Institute, UC Davis, CFE 2025, the 2025 Conference on Sustainable Finance, EcoSta 2025, and the 2025 Bari Panel and High Dimensional Data Conference for helpful comments and suggestions. Any remaining errors are our own.}}
\author{George Kapetanios\\King's College, London
\and Vasilis Sarafidis\\Brunel University London
\and Alexia Ventouri\\King's College, London}
\date{}
\maketitle

\begin{abstract}
This paper studies model selection in high dimensional regression settings. We propose a Boosting with Multiple Testing (BMT) approach designed to address proxy signal contamination in environments with strongly correlated regressors, including settings driven by latent common factors. At each stage, at most one regressor is selected conditional on those already included, using a multiple testing threshold for the remaining candidates.
Under stated dependence, detectability and design conditions, BMT recovers the exact true model with probability approaching one. A separate extension permits growing signal counts under additional restrictions. For fixed numbers of signals and proxies, we also establish recovery of an approximating model containing all true signals and no noise variables under a separate set of conditions. Following exact selection, fixed dimensional least squares has the oracle asymptotic distribution under the stated moment conditions. A theoretical comparison identifies designs and tuning sequences for which BMT is consistent while the Lasso is not. In the main Monte Carlo comparison, BMT has the highest median Matthews correlation coefficient and exact recovery rate and the lowest median coefficient root mean squared error among the compared specifications. Applications to U.S. emissions and inflation show that BMT combines sparse specifications with favourable forecast performance, including the lowest emissions forecasting errors in the reported comparisons and the best one quarter ahead forecasts in the baseline inflation comparison.
\end{abstract}



\section{Introduction}

The problem of model selection in high dimensional regression settings is well known and challenging. A central difficulty is how to achieve parsimony without sacrificing fit in environments with a large number of candidate regressors.
Widely used approaches include penalized regression methods (e.g.\ \citealp{tib1996,fan2001,efr2004,zou2005,can2007,fan2013}), and stepwise procedures such as boosting \citep{fri2001,buh2006,kap2010}. Penalised methods shrink coefficients via a tuning parameter, while boosting builds the model sequentially by adding regressors in a direction that improves in sample fit.

The choice of tuning and stopping rules matters for noise exclusion. In penalised regression, tuning parameters are typically chosen to optimise predictive performance, which can lead to insufficient penalisation and the inclusion of irrelevant variables (see, e.g., \citet{LengEtal2006} and Ch. 7 in \citet{HastieEtal2009}). Ordinary information criteria can also favour overly large models when the candidate model space is large \citep{ChenChen2008}.

A prominent approach to filtering out pure noise covariates is based on sure screening methods  (e.g.\ \citealp{fan2008,hua2008,fan2010}). These procedures perform an initial reduction of the covariate space based on marginal measures such as pairwise correlations or univariate regressions, with the aim of retaining all relevant regressors with high probability. Building on this approach, \citet{chu2018} propose the One Covariate at a Time, Multiple Testing (OCMT) procedure, which applies a multiple testing correction to assess the significance of individual regressors and filter out noise variables. The Generalised OCMT (GOCMT) procedure of \citet{Sharifvaghefi2025} accounts for strong covariate dependence through estimated common factors. Our analysis also allows correlated proxies that need not be generated by pervasive common factors.

A second difficulty is the presence of proxy signal covariates: variables with zero coefficients in the data generating process that can nevertheless have nonzero marginal effects because they are correlated with true signals. Screening methods and OCMT can retain such variables.

From an applied perspective, the presence of proxy signal covariates in the selected model can obscure the interpretation of coefficients and induce multicollinearity, inflating standard errors and making estimates unstable. As a result, empirical conclusions can become sensitive to small changes in the model or sample, undermining reliable inference. A structural interpretation requires further identification assumptions; recovering the nonzero coefficients of a predictive regression does not by itself identify policy effects.

Proxy signal covariates are pervasive in empirical work. For example, macroeconomic studies of how economic uncertainty affects output growth, employment, or asset returns often use several related uncertainty indices (e.g., newspaper based, forecast based, or market based measures, and different horizons such as one month ahead, quarterly, or annual). In monetary economics, analyses of how liquidity affects inflation, credit growth, or investment commonly use multiple monetary aggregates, such as $M0$, $M1$, $M2$, and $M3$, which represent progressively broader money concepts. In both cases, the regressor set contains groups of variables aimed at capturing the same underlying concept. Whether a particular variable is a proxy in our sense depends on its coefficient in the specified DGP.

More broadly, many macroeconomic and financial series are partially driven by a small number of latent common factors, so that a large collection of observables is correlated with the same underlying shocks. In such environments, allowing all available proxies to enter the model blurs the distinction between genuine drivers and variables that merely comove with common shocks, undermining the ability to single out regressors with genuine independent explanatory content.

We focus on a sparse data generating process whose covariates with nonzero coefficients form the true signal set. Although this sparsity assumption can be restrictive in applications where many regressors have small but nonzero effects, it remains the dominant modelling paradigm in a variety of disciplines, including economics and finance, where the objective typically extends beyond prediction to the estimation, inference, and interpretation of individual coefficients. Therefore, our aim is to design a selection procedure that mitigates the influence of proxy signal covariates and recovers the true signal set as accurately and parsimoniously as possible. 

To this end, we propose a model selection procedure, Boosting with Multiple Testing (BMT), which combines forward stepwise selection with a family wise multiple testing stopping rule. The researcher observes a large pool of candidate regressors but does not know which are true signals, which are proxies, and which are pure noise. Whereas OCMT admits all candidates that pass its threshold at a stage, BMT admits at most one. Specifically, the model is constructed iteratively by selecting the most significant regressor conditional on those already included, and retesting the remaining candidates given the updated specification, with a multiple testing based stopping rule determining when the procedure terminates. Under the conditions below, the probability of admitting any noise variable tends to zero.
Sequential conditioning can prevent a proxy from entering after the relevant signals have been selected, although ordinary BMT retains every variable already admitted.

The present paper establishes five sets of novel results. First, BMT recovers the exact true model under the conditions in Theorems~\ref{th3}--\ref{th4}. The single signal result permits a shrinking coefficient whose magnitude dominates $\sqrt{\log n/T}$, with the local covariance bounds keeping every proxy's correlation with the signal uniformly below one in absolute value. With several signals, at least one remaining signal must have a larger absolute population score than every proxy after conditioning on any incomplete subset of true signals, with the difference large relative to sampling error. A separate corollary permits growing signal counts under additional coefficient, projection and rate conditions.

Second, under separate assumptions with fixed numbers of signals and proxies, BMT includes all true signals and excludes noise with probability approaching one, while possibly retaining some proxies. These results also give rates for the fitted residual variance and coefficient error.

Third, after exact selection and under additional fixed dimensional moment conditions, post selection least squares has the same first order distribution as the infeasible oracle estimator. This inference result is pointwise and does not extend to growing signal dimension.

Fourth, we compare BMT with penalised regression methods. A positive definite proxy signal design with heterogeneous coefficients gives a region in which BMT recovers the true model while the Lasso fails under the stated penalty sequence. The relative strength of their design conditions depends on the coefficients and covariance structure. Appendix~\ref{app:scad_mcp} gives examples in both directions for SCAD and MCP and distinguishes their global solutions from a specified local linear approximation algorithm. 

Finally, in Section \ref{subsec:true_model} and Appendix~\ref{app:bmt_retesting}, we present two novel retesting variants of BMT. We establish exact recovery for them, under fixed signal and proxy counts and suitable joint regression and signal strength conditions. Retesting can remove an early proxy but can also delete a genuine weak signal in finite samples; ordinary BMT remains our main procedure due to its simplicity and good finite sample properties. 

The proofs use the probability inequalities of \citet{den2022} to control local sample moments uniformly under exponential mixing and sub Weibull tails. Together with the martingale difference restrictions on regression scores, these bounds justify the sequence of conditional tests used by BMT in dynamic regressions. Mixing extensions also appear in OCMT \citep{chu2018}.

Related ideas appear in the literature on sequential and ordered testing. On the multiple testing side, \citet{gse2015}, \citet{li2015}, \citet{tib2014}, and \citet{fit2014} develop stopping rules, exact or selective $p$ values, and inferential tools for ordered or sequential model paths. Related work on testing based forward selection includes \citet{kozbur2017,kozbur2020}.

In addition to the theoretical results outlined above, we provide extensive simulation evidence in support of the proposed approach. The Monte Carlo design varies sample size, sparsity, and the strength and source of multicollinearity induced by latent common and local factors. When the strength of the common factor is high, virtually all regressors in the candidate set behave as proxy signals, so these designs are particularly demanding for procedures that aim to recover the true signals.
We compare BMT with OCMT, GOCMT, Lasso with and without principal component prefiltering, Adaptive Lasso, MCP and SCAD. For Lasso and Adaptive Lasso, we consider both the minimum MSE and one standard error (1SE) tuning rules. We evaluate performance using the Matthews Correlation Coefficient (MCC) \citep{Matthews1975}, and standard selection metrics, including the true discovery rate (TDR, share of selected regressors that are true signals), true positive rate (TPR, share of true signals selected), false discovery rate (FDR, share of selected regressors that are not true signals) and false positive rate (FPR, share of noise or proxy variables selected). MCC can be interpreted as a correlation between the true and selected inclusion indicators, aggregating all four classification outcomes (true positives, true negatives, false positives and false negatives) into a single balanced index that remains informative in sparse, highly unbalanced settings. 
Overall, BMT delivers high MCC values, with high TPR and TDR, low FPR and FDR, and model sizes that closely track the true sparsity level. Adaptive Lasso under the 1SE rule, MCP and SCAD also perform well. BMT has the highest median MCC and exact recovery rate and the lowest median coefficient RMSE across the reported designs, but the ranking varies across designs and criteria.

Finally, we provide two empirical illustrations based on high dimensional information from the FRED-QD database of \citet{McCrackenNg2021}. In the first, aggregate annual emissions for the United States economy are modelled as a function of penalties imposed on emitting facilities together with a broad set of macro financial indicators. In the second, we use the link between reduced form Phillips curves and inflation forecasting to examine variable selection from a large set of macroeconomic indicators. In the emissions application, BMT selects one additional macroeconomic indicator and achieves the lowest forecast errors in the reported forecasting comparisons. In the baseline inflation forecasting exercise, BMT gives the most accurate one quarter ahead forecasts among the methods compared while selecting an average of two predictors. At two quarters, its forecasts of average inflation are close to those of the best performing method, MCP, but use substantially fewer predictors.

The remainder of the paper is organised as follows. Section \ref{setup} outlines the BMT algorithm, and Section \ref{sec:theory} develops its asymptotic properties. Section \ref{simulation} presents the Monte Carlo design and finite sample results, while Section \ref{empirics} reports the empirical illustrations. Section \ref{conclusions} concludes. The Online Supplement provides additional theoretical results and proofs, as well as detailed Monte Carlo results.

\textbf{Notation.} Generic positive finite constants $C,C_0,C_1,\ldots$ may differ between occurrences. For a positive sequence $g_n$, $f_n=O(g_n)$ denotes boundedness of $|f_n|/g_n$, and $f_n=o(g_n)$ denotes the limit $f_n/g_n\to0$. For positive sequences, $f_n\asymp g_n$, also written $f_n=\ominus(g_n)$, denotes a ratio eventually bounded above and away from zero; $f_n\sim g_n$ denotes convergence of the ratio to one. The stochastic orders $O_p(g_n)$ and $o_p(g_n)$ denote, respectively, boundedness in probability and convergence to zero in probability of the ratio to $g_n$. The arrows $\to_p$ and $\xrightarrow{p}$ denote convergence in probability, and $\xrightarrow{d}$ and $\Rightarrow$ denote convergence in distribution, under the asymptotic sequence in the relevant result.

For sets, $A\subseteq B$ allows equality, whereas $A\subsetneq B$ denotes proper inclusion of $A$ in $B$: every element of $A$ belongs to $B$, but at least one element of $B$ is missing. The empty set $\varnothing$ is a proper subset of any nonempty set. The symbols $|A|$ and $\#A$ denote the number of elements of $A$, and $A\setminus B$ contains the elements of $A$ that are not in $B$. Thus the minimum over $A\subsetneq\mathcal S_s$ includes the empty set and every incomplete subset of the true signal set $\mathcal S_s$.

For vectors, $\|v\|=\|v\|_2$ is the Euclidean norm, $\|v\|_1$ is the sum of absolute coordinates, and $\|v\|_\infty$ is the largest absolute coordinate. For a matrix $M$, $\|M\|_2$ is its spectral norm; $\lambda_{\min}(M)$ and $\lambda_{\max}(M)$ are its smallest and largest eigenvalues when $M$ is symmetric. A prime denotes transpose, $I_d$ is the $d\times d$ identity matrix, and $\mathbf1_d$ is the $d$ vector of ones. The operations $a\vee b=\max\{a,b\}$ and $a_+=\max\{a,0\}$ denote the maximum and positive part. The direct sum $M\oplus N$ is the block diagonal matrix with blocks $M$ and $N$. The notation $\sigma(\cdot)$ denotes the sigma field generated by the listed random variables. 

\section{Setup}

\label{setup}

We consider the following data generating process (DGP):
\begin{equation}
y_{t}=\mathbf{a}^{\prime}\boldsymbol{z}_{t}+%
{\textstyle\sum\nolimits_{i=1}^{k}}
\beta_{i}x_{it}+u_{t}\text{,} %
\label{dgp1}
\end{equation}
where $\boldsymbol{z}_{t}$ is a known vector of $\zeta$ preselected variables,
$x_{1t},x_{2t},...,x_{kt}$ are the $k$ unknown \textit{true signal} variables, $0<\left\vert \beta_{i}\right\vert \leq C<\infty$,
for $i=1,2,...,k$, and $u_{t}$ is an error term. It is assumed that
$\boldsymbol{z}_{t}$ and $x_{it}$, $i=1,2,...,k,$ are contemporaneously uncorrelated with $u_{t}$ at time $t$.  $\boldsymbol{z}_{t}$ may include deterministic components such as a constant, a linear trend, and seasonal dummy variables, as well as stochastic variables, such as lagged values of $y_{t}$ or observed common factors, that are typically chosen on \textit{a priori} theoretical grounds.

Further suppose that the $k$ signals are contained in the \textit{candidate set} $\mathcal{S}%
_{n}=\left\{  x_{it},i=1,2,...,n\right\}$, with $n$ being potentially larger than $T$.\footnote{The results assume that every true signal belongs to the candidate set. The algorithm can still be applied if this condition fails, but the stated recovery and noise exclusion guarantees do not cover that case.} In addition
to the $k$ signals, the candidate set is comprised of \textit{noise} or \textit{distant} variables, whose residuals are orthogonal to the signals after projection on $\boldsymbol z_t$, and a remaining group of variables which, conditional on $\boldsymbol{z}_{t}$, are correlated with the signals. We refer to the latter as \textit{pseudo signals} or \textit{proxy} variables, since they may be mistakenly viewed as signals. Denote the sets of true signals, pseudo signals, and noise/distant variables by $\mathcal{S}_{s}$, ${P}$, and $\mathcal{S}_{d}$, respectively. Noise variables are also assumed to be orthogonal to every proxy after the controls are removed. Here orthogonality denotes a zero time averaged expected cross product of the projection residuals, as defined in Section~3.1. With a constant among the controls this is zero covariance under the same averaged moments; without a constant the cross product restriction is imposed directly. In set operations, $\mathcal S_n=\{1,\ldots,n\}$ and $\mathcal S_s$, $P$ and $\mathcal S_d$ denote the corresponding candidate index sets.

We now present the full details of the Boosting with Multiple Testing (BMT) algorithm, which is a multistage procedure. In the first stage, we run the $n$ regressions of $y_{t}$ on $x_{it}$ and the always in controls $\boldsymbol{z}_{t}$, for $i = 1,\dots,n$. Let $t_{i}$ denote the $t$ ratio from the regression involving $x_{it}$, and write this as $t_{i,(1)}$, where the subscript $(1)$ indexes the stage. We \ identify the variable with the largest absolute $t$ ratio, that is, the $x_{it}$ for which $|t_{i,(1)}|$ is maximised over $i$, resolving a tie by the smallest candidate index under the fixed ordering. The stage one maximiser is admitted only if $\max_{1\leq i\leq n}|t_{i,(1)}|>c_{\alpha_0}(n,\delta)$; otherwise the procedure stops and returns the empty selected candidate set. Thus, stage~1 has an empty previously selected set but is not an automatic admission. Conditional on admission, denote the index of this selected variable by $\mathcal{S}_{(1)}^{c}$, and let $\boldsymbol{x}_{(1)}^{o}$ be its $T \times 1$ observation vector. Define $\boldsymbol{X}_{(1)}=\left(\boldsymbol{Z}, \boldsymbol{x}_{(1)}^{o}\right)=\left(\boldsymbol{x}%
_{(1),1},\dots,\boldsymbol{x}_{(1),T}\right)^{\prime}$, a $T \times (\zeta+1)$ matrix, where $\boldsymbol{Z} = \left(\boldsymbol{z}_{1},\dots,\boldsymbol{z}_{T}\right)'$, and set $\mathcal{S}_{(1)} = \mathcal{S}_{(1)}^{c}$ and $\mathfrak{A}_{(2)} = \{1,\dots,n\} \setminus \mathcal{S}_{(1)}$. \ Following that admission, in the second stage, we run regressions of $y_{t}$ on $\boldsymbol{x}_{(1),t}$ and $x_{it}$, for $i \in \mathfrak{A}_{(2)}$, and obtain the corresponding $t$ ratios for $x_{it}$, denoted by $t_{i,(2)}$.
We construct selection indicators, given by
\begin{equation}
\widehat{\mathcal{L}}_{i,(2)}=I[|t_{i,(2)}|>c_{\alpha_0}\left(  n,\delta\right)
]\text{, for }i\in\mathfrak{A}_{\left(  2\right)  }\text{, }%
\end{equation}
where $c_{\alpha_0}(n,\delta)$ is the critical value function, defined by%
\begin{equation}
c_{\alpha_0}\left(  n,\delta\right)  =\Phi^{-1}\left(  1-\frac{\alpha_0}{2f\left(
n,\delta\right)  }\right)  \text{,} \label{cvf}%
\end{equation}
$\Phi^{-1}\left(  .\right)  $ is the inverse of standard normal distribution
function, $f\left(  n,\delta\right)  =cn^{\delta}$ for some positive constants
$\delta$ and $c$, and $\alpha_0$ ($0<\alpha_0<1$) is the nominal calibration level
to be set by the investigator. We take $c\geq\alpha_0$, so that the threshold
is finite and nonnegative for every positive number of remaining candidates. We refer to $\delta$ as the critical
value exponent. If no remaining regressor passes at a stage, the procedure stops with the current selected set. Among regressors that pass the multiple testing filter, only the single regressor with the largest absolute t statistic is selected, and all others are deferred to subsequent stages. If at least one regressor passes, we denote the index of the
selected unit by $\mathcal{S}_{\left(  2\right)  }^{c}$ and the $T\times1$
observation vector of the selected variable by $\boldsymbol{x}_{(2)}^{o}$. Conditional on that admission, we let $\boldsymbol{X}_{(2)}=\left(\boldsymbol{X}_{(1)},\boldsymbol{x}_{(2)}^{o}  \right)=\left(\boldsymbol{x}_{(2),1}%
,...,\boldsymbol{x}_{(2),T}\right)^{\prime}$, a $T \times \left(\zeta+2 \right)$ matrix, and $\mathcal{S}_{\left(  2\right)
}=\mathcal{S}_{\left(  1\right)  }\cup\mathcal{S}_{\left(  2\right)  }^{c}$,
$\mathfrak{A}_{\left(  3\right)  }=\left\{  1,...,n\right\}
\setminus\mathcal{S}_{\left(  2\right)  }$.
\ Following an admission, we proceed to subsequent stages $j = 3,4,\ldots$ in the same manner, defining
$\boldsymbol{X}_{(j)} = \left(\boldsymbol{X}_{(j-1)}, \boldsymbol{x}_{(j)}^{o}\right)$, $\mathcal{S}_{(j)} = \mathcal{S}_{(j-1)} \cup \mathcal{S}_{(j)}^{c}$ and $\mathfrak{A}_{(j+1)} = \{1,\dots,n\} \setminus \mathcal{S}_{(j)}$. At each stage $j$, the value of $n$ in $c_{\alpha_0}(n,\delta)$ is set equal to the cardinality of $\mathfrak{A}_{(j)}$.

The stage specific construction above has the following formal recursion. Set
\[
\widehat S_{0}=\varnothing,\qquad R_{0}=\{1,\ldots,n\},\qquad
\boldsymbol X_{(0)}=\boldsymbol Z.
\]
At a stage $r\geq1$ with $R_{r-1}\neq\varnothing$, let
$t_{i\mid A}$ be the conventional OLS $t$ statistic on $x_{it}$ in the
regression of $y_t$ on $\boldsymbol z_t$, the candidates indexed by $A$,
and $x_{it}$.  For $i\in R_{r-1}$, define
\[
t_{i,(r)}:=t_{i\mid\widehat S_{r-1}},\qquad
\mathfrak A_{(r)}:=R_{r-1},\qquad m_r:=|R_{r-1}|,
\]
and use the same exponent $\delta$ at every stage to set
\[
\widehat{\mathcal L}_{i,(r)}
=I\!\left[|t_{i,(r)}|>c_{\alpha_0}(m_r,\delta)\right].
\]
The deterministic stagewise maximiser is
\[
\widehat j_r =\mbox{argmax}_{i\in R_{r-1}}|t_{i,(r)}|.
\]
If several candidates attain the maximum, select the one with the smallest index.
If $\widehat{\mathcal L}_{\widehat j_r,(r)}=0$, stop and report
$\widehat S_{r-1}$. If it equals one, admit $\widehat j_r$ and update
\[
\widehat S_r=\widehat S_{r-1}\cup\{\widehat j_r\},\qquad
R_r=R_{r-1}\setminus\{\widehat j_r\},\qquad
\boldsymbol X_{(r)}=(\boldsymbol X_{(r-1)},\boldsymbol x_{(r)}^o),
\]
where $\boldsymbol x_{(r)}^o$ is the observation vector for
$x_{\widehat j_r,t}$.  After a successful admission,
$\mathcal S_{(r)}=\widehat S_r$,
$\mathcal S_{(r)}^c=\widehat j_r$, and
$\mathfrak A_{(r+1)}=R_r$.  If $R_r=\varnothing$, stop without a further test;
otherwise proceed to stage $r+1$. Thus the first stage corresponds to
$\widehat S_0=\varnothing$, has multiplicity $m_1=n$, and applies the same
threshold before admission as every later stage.

\noindent
Let $\widehat{J}$ denote the last (random) stage at which a regressor is admitted, with $\widehat J=0$ if no candidate is admitted. The final observation matrix, including the preselected controls, is denoted as $\boldsymbol{X}_{(\widehat{J})}$ and is of dimension $T \times (\zeta + \widehat{J})$.

The final candidate index set is $\widehat S=\widehat S_{\widehat J}$;
in particular, $\widehat S=\varnothing$ and
$\boldsymbol X_{(\widehat J)}=\boldsymbol X_{(0)}=\boldsymbol Z$ when
$\widehat J=0$. To define the stagewise hypothesis family, let
$\theta_{i\mid A}$ be the population coefficient on $x_{it}$ in the linear
projection corresponding to the regression that defines $t_{i\mid A}$, and
write
\[
H_{0,i}(A):\ \theta_{i\mid A}=0,\qquad
\mathcal H_{0,r}=\{H_{0,i}(\widehat S_{r-1}):i\in R_{r-1}\}.
\]
The pathwise family wise true null rejection event is
\[
\mathcal E_{0}
=\{\text{at some tested stage }r,\text{ some }i\in R_{r-1}
\text{ has }\theta_{i\mid\widehat S_{r-1}}=0
\text{ and }\widehat{\mathcal L}_{i,(r)}=1\}.
\]
The family wise noise inclusion event is
\[
\mathcal E_{d}=\{\widehat S\cap\mathcal S_d\neq\varnothing\},
\]
the event that at least one noise variable is admitted at some stage of the realised path. Under the maintained orthogonality restrictions, the first such admission rejects a true stagewise null: all previously selected candidates are signals or proxies. Thus $\mathcal E_d\subseteq\mathcal E_0$. Exact model recovery additionally requires exclusion of proxies.


Each successful BMT stage adds exactly one variable to the model.\footnote{The procedure can be generalised to select a prespecified, possibly data dependent, number of variables per stage.} OCMT admits all candidates that pass its multiple testing threshold at a stage. 

Sequential conditioning changes the scores of the remaining candidates. To illustrate, consider $y_t=\beta x_{1t}+u_t$ and $x_{2t}=\rho x_{1t}+v_t$, where $\beta\rho\neq0$, $x_{1t}$ and $u_t$ are uncorrelated, and $v_t$ is uncorrelated with $(x_{1t},u_t)$. Suppose $x_{1t}$ and $v_t$ have positive variances. The proxy $x_{2t}$ has a nonzero marginal coefficient, but its population partial coefficient is zero after conditioning on $x_{1t}$. If the signal enters first, the stopping conditions below exclude the proxy with probability approaching one. A proxy admitted earlier is retained by ordinary BMT. Remark~\ref{rem:am_early_stopping} gives the corresponding early stopping result.

\section{Theoretical Results}\label{sec:theory}

\subsection{Assumptions}\label{subsec:assumptions}

We use $x_{it}$ and $y_t$ for scalar observations,
$\boldsymbol x_i=(x_{i1},\ldots,x_{iT})'$ and
$\boldsymbol y=(y_1,\ldots,y_T)'$ for observation columns, and
$\boldsymbol X_A=(\boldsymbol x_i:i\in A)$ for the candidate data
matrix, with columns in index order. The joint regression matrix is
$(\boldsymbol Z,\boldsymbol X_A)$. In the fixed design appendix
examples the explicitly defined columns $x_i$ and matrix $X$ denote
these same observation objects; a subscript $t$ denotes a scalar
observation. A conditioning set $A$ contains candidate indices and does
not include the always in controls. Complements of candidate sets are
taken in $\{1,\ldots,n\}$.

\begin{assumption}
\label{ass1}
Let $\mathcal F_{t-1}^{u}=\sigma(u_s:s\leq t-1)$ and
\(
\mathcal F_t=\sigma\{u_s,\boldsymbol z_s,x_{is}:s\leq t,\ 1\leq i\leq n\}.
\)
The error satisfies $E(u_t\mid\mathcal F_{t-1}^{u})=0$. Separately, the regression
scores satisfy the following restrictions conditional on the larger information
set $\mathcal F_{t-1}$:
\(
E(\boldsymbol z_tu_t\mid\mathcal F_{t-1})=\boldsymbol 0,
E(x_{it}u_t\mid\mathcal F_{t-1})=0
\)
almost surely for every $i$ and $t$. Moreover, $E(u_t^2)=\sigma^2$ for
every $t$, where $0<\sigma^2<C<\infty$.
The error process is strongly mixing in time, with coefficient
$\alpha_u(h)\leq C_u\xi_u^h$ for $h\geq1$, where
$C_u<\infty$ and $0<\xi_u<1$.
\end{assumption}

\begin{assumption}
\label{ass1A_mixingall}
For every $i,\ell\in\{1,\ldots,n\}$, the process
$(\boldsymbol z_t',x_{it},x_{\ell t},u_t)'$ is strongly mixing in time.  Let
$\alpha_{xu}(h)$ be the supremum, over $i,\ell,n$, and $T$, of its mixing
coefficients at temporal separation $h$.  There are constants
$C_{xu}<\infty$ and $0<\xi_{xu}<1$ such that
\(
\alpha_{xu}(h)\leq C_{xu}\xi_{xu}^h, h\geq1.
\)
\end{assumption}

\begin{assumption}
\label{ass2} There exist positive constants $C_{0}%
,C_{1},C_{2}$ and $C_{3}$ and $s_{x},s_{u}>0$ such that the covariates in the
candidate set $\mathcal{S}_{n}$ satisfy%
\begin{equation}
\sup\nolimits_{i,t}\Pr\left(  \left\vert x_{it}\right\vert >\alpha\right)
\leq C_{0}\exp\left(  -C_{1}\alpha^{s_{x}}\right)  ,\text{ for all }\alpha>0,
\label{expprob}%
\end{equation}
and the errors, $u_{t}$, in DGP (\ref{dgp1}) satisfy
\begin{equation}
\sup\nolimits_{t}\Pr\left(  \left\vert u_{t}\right\vert >\alpha\right)  \leq
C_{2}\exp\left(  -C_{3}\alpha^{s_{u}}\right)  ,\text{ for all }\alpha
>0\text{.} \label{expprob2}%
\end{equation}

The first bound also holds for every stochastic component of $\boldsymbol z_t$;
the normalised deterministic components are uniformly bounded.

\end{assumption}

For any scalar or vector array $h_t$, write
\[
\overline E_T(h_t)=T^{-1}\sum_{t=1}^T E(h_t).
\]
Let $A\subseteq\{1,\ldots,n\}$ be deterministic, and set
\[
q_{A,t}=(\boldsymbol z_t',\boldsymbol x_{A,t}')',\qquad
Q_{A,T}=\overline E_T(q_{A,t}q_{A,t}'),
\]
with $q_{\varnothing,t}=\boldsymbol z_t$.  Whenever $Q_{A,T}$ is nonsingular,
define, for $w_t=y_t$ and for $w_t=x_{jt}$ with $j\notin A$,
\begin{equation}
\label{ux}
g_{w,A,T}=\overline E_T(q_{A,t}w_t),\qquad
\gamma_{w,A,T}=Q_{A,T}^{-1}g_{w,A,T},\qquad
w_t^A=w_t-\gamma_{w,A,T}'q_{A,t}.
\end{equation}
Thus the superscript $A$ denotes the population residual after projecting $w_t$
on the compulsory controls and the candidate regressors indexed by $A$.
In particular, $w_t^{\varnothing}$ removes only the compulsory controls;
it does not denote the original, unadjusted variable unless there are no controls.
All conditions on $\boldsymbol z_t$ refer to a deterministic nonsingular
choice of basis for the always in controls. The change of basis may depend
on $T$ but is common to all observations and candidate regressions. It may
rescale deterministic components and subtract their span from stochastic
controls; the resulting components must satisfy \ref{ass2} and \ref{reg}. This leaves the
control span, fitted values, and candidate $t$ ratios unchanged. The
coefficient $\mathbf a$ is expressed in the same basis and need not be
bounded as $T$ grows. For example, a raw trend with fixed slope has a
coefficient of order $T$ when written as $t/T$.

For $v_t,w_t\in\{y_t\}\cup\{x_{jt}:j\notin A\}$, define
\[
s_{vw\mid A,T}
=\overline E_T(v_tw_t)-g_{v,A,T}'Q_{A,T}^{-1}g_{w,A,T}
=\overline E_T(v_t^Aw_t^A).
\]
In particular, let
\[
v_{j,A,T}=s_{x_jx_j\mid A,T},\qquad
v_{y,A,T}=s_{yy\mid A,T},\qquad
c_{j,A,T}=s_{x_jy\mid A,T}.
\]
For $j\notin A$ with $v_{j,A,T}v_{y,A,T}>0$, define
\[
\theta_{j\mid A,T}=\frac{c_{j,A,T}}{v_{j,A,T}},\qquad
\rho_{j,A,T}=\frac{c_{j,A,T}}
 {\{v_{j,A,T}v_{y,A,T}\}^{1/2}},\qquad
\omega_{j,A,T}^2=v_{y,A,T}-\frac{c_{j,A,T}^2}{v_{j,A,T}}.
\]
When, in addition, $|\rho_{j,A,T}|<1$, let
\[
\psi_{j,A,T}=\frac{\rho_{j,A,T}}
 {\sqrt{1-\rho_{j,A,T}^2}}.
\]
Thus $\theta_{j\mid A,T}$ is the population coefficient denoted by
$\theta_{j\mid A}$ in Section~2, and
$\omega_{j,A,T}^2=v_{y,A,T}(1-\rho_{j,A,T}^2)$ is the population outcome
residual variance after adding candidate $j$ to $q_{A,t}$.

The deterministic conditioning classes used below are
\[
\mathcal A_T^{\mathrm{ER}}=\{A:A\subseteq\mathcal S_s\},\qquad
\mathcal C_T^{\mathrm{ER}}
=\{(A,j):A\in\mathcal A_T^{\mathrm{ER}},\ j\in\mathcal S_n\setminus A\},
\]
for exact recovery, and
\[
\mathcal A_T^{\mathrm{AM}}
=\{A:A\subseteq\mathcal S_s\cup{P}\},\qquad
\mathcal C_T^{\mathrm{AM}}
=\{(A,j):A\in\mathcal A_T^{\mathrm{AM}},\ j\in\mathcal S_n\setminus A\},
\]
for the approximating model results with finite $k+k^*$.  Thus the
exact recovery class conditions only on subsets of signals, whereas the
approximating model class also allows proxies in the conditioning set.  For
any class $\mathcal C_T$, put
\[
K_T(\mathcal C_T)=\max_{(A,j)\in\mathcal C_T}|A|,\qquad
d_A=\zeta+|A|+1.
\]
Here $K_T(\mathcal C_T)$ is the largest conditioning set size in the class;
$\zeta$ is the separate, fixed dimension of the always in controls, and $d_A$
is the dimension of a regression that tests one further candidate.  In what
follows, $(\mathcal A_T,\mathcal C_T)$ denotes
$(\mathcal A_T^{\mathrm{ER}},\mathcal C_T^{\mathrm{ER}})$ for exact recovery
results and $(\mathcal A_T^{\mathrm{AM}},\mathcal C_T^{\mathrm{AM}})$ for
approximating model results; write $K_T=K_T(\mathcal C_T)$.

\begin{assumption}
\label{ass_proj}

For the relevant choice of $(\mathcal A_T,\mathcal C_T)$ defined
above, $Q_{A,T}$ is nonsingular for every $A\in\mathcal A_T$.  There
exist constants $0<c<C<\infty$, independent of $A,j$, and $T$, such that
\[
\sup_{(A,j)\in\mathcal C_T}
 \|\gamma_{x_j,A,T}\|_2\le C
\]
and, uniformly over $(A,j)\in\mathcal C_T$,
\[
c\le v_{j,A,T}\le C,\qquad
c\le v_{y,A,T}\le C,\qquad
c\le\omega_{j,A,T}^2\le C.
\]
\vspace{-0.08in}
\end{assumption}

The lower bound on $\omega_{j,A,T}^2$ keeps the population partial correlation
away from one in absolute value. Under the fixed $k$ DGP, the bounds in \ref{ass_proj}
follow from \ref{ass1}, \ref{ass2}, and \ref{reg}: projection cannot increase a second moment,
\ref{reg} provides uniform bounds on the norms of the relevant inverse matrices and on the candidate residual variances, and score orthogonality bounds the outcome residual variance below by
$\sigma^2$. Its upper bound follows by removing $\mathbf a'\boldsymbol z_t$
and using the fixed signal count and bounded signal coefficients.
There is no bound on the projection coefficient of the full outcome
on rescaled controls. The moment and Schur calculations below remove
the always in component first.

The uniform lower eigenvalue bound still excludes a candidate that becomes
arbitrarily close to the span of a covered conditioning set. For example,
$x+\eta_T e$ with $\eta_T\to0$ is excluded when conditioning on $x$.
For the approximating model class this restriction applies to every
signal--proxy subset, including subsets the algorithm might not visit.

\begin{assumption}
\label{ass10} The number of signals, $k$, in DGP \eqref{dgp1} is finite, and their
slope coefficients are bounded away from zero.\vspace{-0.08in}
\end{assumption}

\begin{assumption}
\label{reg}

For the same choice of $(\mathcal A_T,\mathcal C_T)$, define
\[
r_{A,j,t}=(q_{A,t}',x_{jt})',\qquad
G_{A,j,T}=\overline E_T(r_{A,j,t}r_{A,j,t}').
\]
There exist constants $0<c<C<\infty$, independent of $A,j$, and $T$, such
that
\[
c\leq\lambda_{\min}(G_{A,j,T})
\leq\lambda_{\max}(G_{A,j,T})\leq C
\quad\text{uniformly over }(A,j)\in\mathcal C_T.
\]
The same bounds hold for $Q_{A,T}$ uniformly over $A\in\mathcal A_T$.
Moreover, $\zeta+K_T+1<T$.\vspace{-0.08in}
\end{assumption}

These bounds concern prespecified population matrices, not a matrix selected
along the realised BMT path. If there are no controls and $A=\varnothing$,
$Q_{A,T}$ is empty; its eigenvalue restriction is vacuous and the corresponding
projection is zero.

\noindent\textit{Stronger projection condition.}
For growing conditioning classes, a stronger coordinatewise condition is
\[
\sup_{(A,j)\in\mathcal C_T}
\{\|\gamma_{x_j,A,T}\|_1+\|\gamma_{y,A,T}-(\mathbf a',\boldsymbol 0')'\|_1\}\leq C_{\mathrm{proj}},
\]
where the zero vector has dimension $|A|$ and $C_{\mathrm{proj}}<\infty$ is independent of $A,j$, and $T$. The second norm concerns the projection of $y_t-\mathbf a^{\prime}\boldsymbol z_t$. This condition is not
imposed in the local eigenvalue conditions above.

We comment on the above assumptions. They (i) impose sub Weibull tails on the candidate regressors and the error (Assumption~\textbf{\ref{ass2}}); (ii) impose uniform bounds on population second-moment 
matrices over the deterministic class relevant to each result, before the algorithm is run (Assumption~\textbf{\ref{reg}}); (iii) impose bounds on candidate projection norms and the relevant residual variances over the same deterministic class, without a finite support restriction (Assumption~\textbf{\ref{ass_proj}}); and (iv) restrict the number and minimal strength of signals so that their associated $t$ statistics diverge and nonnegligible effects can still be detected consistently as $T$ grows (Assumption~\textbf{\ref{ass10}}). Assumptions~\textbf{\ref{ass1}}--\textbf{\ref{ass1A_mixingall}} combine a martingale difference restriction on the regression scores with uniform exponential mixing of the local arrays used in the moment calculations. The score restriction gives exogeneity; mixing is used for uniform second moment bounds. Extensions of OCMT also allow exponential mixing; see Section~4.2 and Supplement~C of \citet{chu2018}. The selection results use the sub Weibull tail condition in Assumption~\textbf{\ref{ass2}}.


\subsection{True model recovery}\label{subsec:true_model}

We first consider exact model selection. Let $\widehat{\mathcal{L}}_{i}$ denote the final selection indicator for regressor $i$, equal to one if the variable is selected and zero otherwise. The true model is recovered when all signals are selected and no noise or proxy variable enters the final specification. Formally, we define
\begin{equation}\label{eq:true_model}
\mathcal{A}_{1}=\left\{
{\textstyle\sum\nolimits_{i\in\mathcal{S}_{s}}}
\widehat{\mathcal{L}}_{i}=k\right\}  \cap\left\{
{\textstyle\sum\nolimits_{i\in\mathcal{S}_{d}\cup{P}}}
\widehat{\mathcal{L}}_{i}=0\right\}.
\end{equation}
Thus $\mathcal{A}_{1}$ is the event that the procedure selects all $k$ true signals in $\mathcal{S}_{s}$ and excludes all variables in $\mathcal{S}_{d}$ and ${P}$.

Further, define the numbers of true and false positives, relative to the true model $\mathcal{A}_{1}$, by
\begin{equation}
TP_{n,T}=\sum_{i\in\mathcal{S}_{s}}\widehat{\mathcal{L}}_{i};
\qquad
FP_{n,T}^{(1)}=\sum_{i\in{P}\cup\mathcal{S}_{d}}\widehat{\mathcal{L}}_{i},
\end{equation}
so that $|\mathcal{S}_{s}|=k$ and $|{P}\cup\mathcal{S}_{d}|=n-k$. The True Positive Rate, $TPR_{n,T}$, is $TPR_{n,T}=TP_{n,T}/k$. We now define the False Positive Rate, False Discovery Rate and True Discovery Rate relative to $\mathcal{A}_{1}$ as
\begin{equation}\label{eq:TPR_etc_true}
FPR_{n,T}^{(1)}=\frac{FP_{n,T}^{(1)}}{n-k};\qquad
FDR_{n,T}^{(1)}=\frac{FP_{n,T}^{(1)}}{TP_{n,T}+FP_{n,T}^{(1)}};\qquad
TDR_{n,T}^{(1)}=\frac{TP_{n,T}}{TP_{n,T}+FP_{n,T}^{(1)}}.
\end{equation}
We set $FDR_{n,T}^{(1)}=TDR_{n,T}^{(1)}=0$ when no variable is selected.

For exact recovery, we keep $k$ fixed but allow the nonzero coefficients
to depend on $T$ and approach zero. The lower coefficient bound in
Assumption~\ref{ass10} is replaced by the conditions below. The number of proxies
may grow with $n$. All projection conditions refer to
$\mathcal C_T^{\mathrm{ER}}$, which conditions on subsets of signals and
tests one further candidate at a time.

\begin{theorem}
\label{th3}
Consider \eqref{dgp1} with $k=1$. Suppose that Assumptions~\ref{ass1}--\ref{ass_proj} and \ref{reg}
hold, with \ref{ass_proj} and \ref{reg} imposed on the exact recovery class, $n\to\infty$, and $n=O(T^\kappa)$
for some fixed $\kappa>0$. Suppose also that
\begin{equation}
\label{eq:one_signal_detectability}
\sqrt{\frac{T}{\log n}}\,|\beta_{1,T}|\longrightarrow\infty.
\end{equation}
Use the threshold $c_{\alpha_0}(m,\delta)$ in \eqref{cvf}, with
$f(m,\delta)=cm^\delta$, $0<\alpha_0<1$, $c\geq\alpha_0$, and the same exponent
at every stage, where $m$ is the number of remaining candidates.
There is a finite constant $\delta_0$ such that, for each fixed
$\delta>\delta_0$, BMT selects the true signal at its first stage and
stops after that admission with probability tending to one. In particular,
\begin{equation}
\label{eq:ER_one_signal}
\Pr(\mathcal A_1)\longrightarrow1,\qquad
\Pr(\widehat J=1)\longrightarrow1.
\end{equation}
Moreover,
\[
E(TPR_{n,T})\longrightarrow1,\qquad
E(TDR_{n,T}^{(1)})\longrightarrow1,\qquad
E(FPR_{n,T}^{(1)})\longrightarrow0,\qquad
E(FDR_{n,T}^{(1)})\longrightarrow0,
\]
and the four measures converge in probability to the same limits.
\end{theorem}

Condition~\eqref{eq:one_signal_detectability} permits a shrinking signal,
but requires it to be large relative to $\sqrt{\log n/T}$.
Assumption~\ref{reg} is used both to control the sample regressions and to keep
every proxy's correlation with the signal uniformly below one in absolute
value, after removing the controls. Pointwise imperfect correlation would
not suffice if a proxy became arbitrarily close to the signal as $T$ grew.
The constant $\delta_0$ depends on the concentration and covariance bounds. Section~\ref{simulation} reports finite sample performance at $\delta=1$.

With several signals, a proxy can combine their contributions and have
a larger partial score than any individual signal. The next result
therefore requires a signal to proxy dominance condition at every incomplete
subset of true signals.

\begin{theorem}
\label{th4}
Consider \eqref{dgp1} with a fixed number $k\geq1$ of signals. Suppose
that Assumptions~\ref{ass1}--\ref{ass_proj} and \ref{reg} hold, with \ref{ass_proj} and \ref{reg} imposed on the exact recovery class,
$n\to\infty$, and $n=O(T^\kappa)$ for some fixed $\kappa>0$.
Suppose that
\begin{equation}
\label{cond1}
\sqrt{\frac{T}{\log n}}\,
\min_{A\subsetneq\mathcal S_s}
\left\{
\max_{s\in\mathcal S_s\setminus A}|\psi_{s,A,T}|
-\max_{j\in{P}}|\psi_{j,A,T}|
\right\}\longrightarrow\infty,
\end{equation}
where the maximum over an empty proxy set is zero. Use the stagewise
threshold specified in Theorem~\ref{th3}. There is a finite constant
$\delta_0$ such that, for each fixed $\delta>\delta_0$, BMT admits
$k$ true signals and then stops with probability tending to one. Thus
$\Pr(\mathcal A_1)\to1$ and $\Pr(\widehat J=k)\to1$, and the
performance measure conclusions of Theorem~\ref{th3} hold.
\end{theorem}

The signal to proxy dominance condition~\eqref{cond1} requires the largest
absolute signal score to exceed the largest absolute proxy score after the
same conditioning set, with a difference large relative to sampling error. It is
stagewise in form and uniform over all proper subsets of the true signal
set, including $A=\varnothing$ and subsets that may never be selected.
Since the proxy maximum is nonnegative,
the condition also requires a detectable signal at every such subset.
Its margin exceeds both the stochastic error in comparing candidates and
the order of the admission threshold. Noise variables have zero population
partial score at these subsets by their definition and Assumption~\ref{ass1}.
After all signals have entered, every remaining candidate, including every
proxy, has zero population partial score. The uniform stochastic bound
covers this final test as well.

Only the number of signals is fixed in Theorem~\ref{th4}. The maximum
in \eqref{cond1} includes all proxies, however many there are. Assumption~\ref{reg}
does not impose an eigenvalue bound on their joint covariance matrix:
each covered regression contains the controls, a subset of signals, and
at most one proxy. It still excludes a proxy approaching the span of the
signals on which the regression conditions.

\begin{remark}
Assumption~\textbf{\ref{ass1}} permits conditional heteroskedasticity, but BMT
uses the conventional OLS $t$ statistic.  The selection results therefore rely
on uniform concentration of that statistic; they do not establish
heteroskedasticity robust or HAC validity for the high dimensional selection
step.  A robust or HAC covariance extension would concern only
fixed dimensional inference after exact model selection.
\end{remark}

\begin{remark}
The following examples concern covariates after removing the controls.
In each case, noise variables are orthogonal to the signal--proxy block.
The calculations are given after the proof of Theorem~\ref{th4}.

First, suppose that the signals are mutually orthogonal and each proxy
is correlated with one signal and orthogonal to the others. If these correlations are uniformly
bounded by $\bar r<1$ in absolute value, a proxy's partial correlation
with the outcome is either zero or its correlation with the remaining
associated signal times that signal's partial correlation. Under the
variance bounds, \eqref{cond1} follows from
$\sqrt{T/\log n}\min_{s\in\mathcal S_s}|\beta_{s,T}|\to\infty$.

Second, consider the equicorrelated factor design
\[
x_{jt}=\sqrt r\,f_t+\sqrt{1-r}\,e_{jt},\qquad
j\in\mathcal S_s\cup{P},\qquad 0<r\leq\bar r<1,
\]
where the factor and idiosyncratic components have unit variance and are
mutually uncorrelated. At least one remaining signal has a larger absolute
partial score than every proxy. The same minimum coefficient condition
implies \eqref{cond1}, without a restriction on the coefficient signs.
The bound $r\leq\bar r<1$ preserves nonzero idiosyncratic variation.

The conclusion need not hold with unequal factor loadings; the calculations following the proof of Theorem~\ref{th4} give a counterexample.
\end{remark}

The next corollary also allows the number of signals to grow. In that case,
the general approximation error of order $(1+K_T)\sqrt{\log n}$ grows relative
to the threshold. The proof therefore obtains a sharper bound for the test
made after all signals have entered.

\begin{corollary}\label{cor:growing_signals}
Consider \eqref{dgp1} with $k=k_T\to\infty$ and fixed control dimension
$\zeta$. Retain the noise orthogonality restriction in Section~2. Suppose
that Assumptions~\ref{ass1}--\ref{ass_proj} and \ref{reg} hold, with \ref{ass_proj} and \ref{reg} imposed on
$\mathcal C_T^{\mathrm{ER}}$. All constants in these assumptions and the
following bounds are uniform in $T$. Suppose that
$n\to\infty$, $n=O(T^\kappa)$ for some fixed $\kappa>0$, and $k_T=o(n)$.
Suppose also that
\[
 \sup_T\sum_{s\in\mathcal S_s}|\beta_{s,T}|<\infty,
 \qquad
 (1+k_T)\sqrt{\frac{\log n}{T}}\longrightarrow0,
 \qquad
 \sup_{j\notin\mathcal S_s}
 \|\gamma_{x_j,\mathcal S_s,T}\|_1\leq C
\]
for a constant $C<\infty$, and
\begin{equation}\label{eq:growing_signal_gap}
\frac{1}{1+k_T}\sqrt{\frac{T}{\log n}}\,
\min_{A\subsetneq\mathcal S_s}
\left\{
 \max_{s\in\mathcal S_s\setminus A}|\psi_{s,A,T}|
 -\max_{j\in{P}}|\psi_{j,A,T}|
\right\}\longrightarrow\infty.
\end{equation}
The maximum over an empty proxy set is zero. Use the threshold specified in
Theorem~\ref{th3}, with the same exponent at every stage. There is a finite
constant $\delta_0$ such that, for each fixed $\delta>\delta_0$,
\[
 \Pr(\mathcal A_1)\longrightarrow1,\qquad
 \Pr(\widehat J=k_T)\longrightarrow1.
\]
The performance measure conclusions of Theorem~\ref{th3} also hold.
\end{corollary}

The additional candidate projection bound is imposed only after all signals
have entered. Before that point, \eqref{eq:growing_signal_gap} ensures that a
signal has a larger absolute statistic than every remaining proxy or noise variable.
It does not require every intermediate noise statistic to fall below the
threshold. The final projection bound gives a sharper null statistic bound
and hence stopping with the same fixed exponent. The proof is in
Appendix~\ref{app:growing_signals_proof}.

Under the stronger bounded $\ell_1$ projection condition stated after
Assumption~\ref{reg}, imposed over the whole exact recovery class, the final
projection bound is automatic and \eqref{eq:growing_signal_gap} can be
replaced by \eqref{cond1}. The other conditions of
Corollary~\ref{cor:growing_signals} remain in force. This stronger condition
improves the statistic remainder to $O_p(\sqrt{\log n})$; it is not needed
for the fixed $k$ results.

These are sufficient restrictions on growing signal counts. In the local proxy Gaussian example in Appendix~\ref{app:growing_signals_proof}, equal coefficients $b/k_T$ satisfy the general bound when $k_T=o((T/\log n)^{1/4})$. The stronger all subset projection condition permits $k_T=o(\sqrt{T/\log n})$ in that example; both cases also require $k_T=o(n)$. Neither rate is claimed to be optimal. The additional conditions remain specific to this extension and do not change the fixed $k$ results or the inference theorem.

Retesting selected variables can give exact recovery even when a proxy enters before a signal. We consider two versions of BMT that use retesting. BMT-R1 runs the original procedure to its stopping point and then tests each included coefficient in the final joint regression, removing variables whose absolute test statistics do not exceed the threshold. BMT-R2 applies the same check to the previously selected variables after each new admission, before resuming the forward search. Both use BMT's threshold; BMT-R1 makes only one final check, while BMT-R2 allows a deleted variable to enter again. Their definitions and results are given in Appendix~\ref{app:bmt_retesting}.

Theorem~\ref{th:retesting_exact} shows that both versions recover the true model with probability tending to one when the signal and proxy counts are fixed, the smallest signal coefficient is large relative to $\sqrt{\log n/T}$, and the moment and design conditions hold over the larger class $\mathcal C_T^{\mathrm{AM}}$. Unlike Theorem~\ref{th4}, this result does not require the signal--proxy ordering in \eqref{cond1}, but it does require control of regressions containing several proxies and does not allow their number to grow. BMT-R1 can remove proxies from a model containing all signals, but cannot recover an omitted signal. BMT-R2 can also change subsequent admissions by removing an earlier proxy; Appendix~\ref{app:r2_rescue} gives an example with a weaker signal in which this distinction matters. In finite samples, retesting can also remove a genuine signal with a weak conditional contribution, as illustrated in Appendix~C. Given this tradeoff, we retain the simpler ordinary BMT procedure as the main method.

\subsection{Results for approximating model}\label{subsec:approx_model}

We next consider selection without imposing the signal to proxy dominance
condition~\eqref{cond1}. The following results keep both the signal count and
the proxy count fixed, retain the positive coefficient lower bound in \ref{ass10},
and impose \ref{ass_proj} and \ref{reg} on the larger class that permits proxies in the
conditioning set. Under these assumptions, all signals are retained and
noise is excluded, while proxies may also be selected. Let
$\mathcal{A}_{0}$ denote the event that the approximating model is recovered:
\begin{equation}
\mathcal{A}_{0}=\left\{
{\textstyle\sum\nolimits_{i\in\mathcal{S}_{s}}}
\widehat{\mathcal{L}}_{i}=k\right\}  \cap\left\{
{\textstyle\sum\nolimits_{i\in\mathcal{S}_{d}}}
\widehat{\mathcal{L}}_{i}=0\right\}  \text{,} \label{A0}%
\end{equation}
namely the event that all $k$ true signals are selected and no distant/noise variables are included. In other words, the procedure recovers the entire set $\mathcal{S}_{s}$ without any false positives from $\mathcal{S}_{d}$, although it may still select some pseudo/proxy signals in ${P}$.

Further, define the numbers of true positives and false positives (relative to the approximating model), by
\begin{equation}
TP_{n,T}=\sum_{i\in\mathcal{S}_{s}}\widehat{\mathcal{L}}_{i};
\qquad
FP_{n,T}^{(0)}=\sum_{i\in\mathcal{S}_{d}}\widehat{\mathcal{L}}_{i},
\end{equation}
and let $|\mathcal{S}_{s}|=k$ and $|\mathcal{S}_{d}|=n-k-k^{\ast}$. We then define the True Positive Rate, as well as False Positive Rate and False Discovery Rate relative to $\mathcal{A}_{0}$, as
\begin{equation}\label{eq:TPR_etc_Approximate}
TPR_{n,T}=\frac{TP_{n,T}}{k};\qquad
FPR_{n,T}^{(0)}=\frac{FP_{n,T}^{(0)}}{|\mathcal{S}_{d}|};\qquad
FDR_{n,T}^{(0)}=
\frac{FP_{n,T}^{(0)}}{TP_{n,T}+FP_{n,T}^{(0)}}.
\end{equation}
We set $FDR_{n,T}^{(0)}=0$ when its denominator is zero. Selections from ${P}$ do not enter $FPR_{n,T}^{(0)}$ and $FDR_{n,T}^{(0)}$, so that, unlike $TPR_{n,T}$, these measures evaluate performance with respect to the approximating model $\mathcal{A}_{0}$.\footnote{One can instead define $FPR_{n,T}^{(0)}$ relative to the true model, so that false positives include both proxy signals and pure noise. In that case, the contribution of a finite number of selected proxy signals is asymptotically negligible relative to $|{P}\cup\mathcal{S}_{d}|$, which grows with $n$. By contrast, this argument does not straightforwardly extend to $FDR_{n,T}^{(0)}$, since both the numerator and denominator involve the finite numbers of selected signals and proxies, and the ratio need not converge to zero unless the number of selected proxy signals is asymptotically negligible. Here we use the approximating model measures; the true model measures are defined in Section~\ref{subsec:true_model}.}

\begin{theorem}
\label{th1}Consider the DGP (\ref{dgp1}) with fixed numbers $k\geq1$ of
signals and $k^\ast\geq0$ of proxies. Suppose that Assumptions~\ref{ass1}--\ref{reg} hold
on the approximating model class $\mathcal C_T^{\mathrm{AM}}$, together
with the noise orthogonality restriction in Section~2. Let
$n,T\to\infty$ with $n=O(T^\kappa)$ for some fixed $\kappa>0$.
Use the threshold \eqref{cvf}, with fixed $0<\alpha_0<1$, $c\geq\alpha_0$, and
$f(m,\delta)=cm^\delta$, at every stage.
There is a finite constant $\delta_0$, depending on the maintained
moment, mixing, and design bounds, such that, for every fixed
$\delta>\delta_0$:
\begin{itemize}
\item[(a)] the number $\widehat J$ of admitted candidates satisfies
\begin{equation}
\Pr(\widehat J>k+k^\ast)=o(1);\label{Phat}
\end{equation}
\item[(b)] the approximating model event \eqref{A0} satisfies
\begin{equation}
\Pr(\mathcal A_0)=1-o(1);\label{pA0}
\end{equation}
\item[(c)] the rates defined in \eqref{eq:TPR_etc_Approximate} satisfy
\begin{equation}
E(TPR_{n,T})=1-o(1),\qquad TPR_{n,T}\to_p1,\label{tprn0}
\end{equation}
and
\begin{equation}
E(FPR_{n,T}^{(0)})+E(FDR_{n,T}^{(0)})=o(1),\qquad
FPR_{n,T}^{(0)},FDR_{n,T}^{(0)}\to_p0.\label{fprn0}
\end{equation}
\end{itemize}
Moreover, with probability tending to one, BMT stops at the first
stage after all signals have been admitted.
\end{theorem}

Theorem~\ref{th1} is a baseline approximating model result. It shows that
one at a time selection includes every signal and excludes noise with
probability tending to one, while permitting a finite number of proxies.
Proxies are not counted as false positives in \eqref{fprn0}.

\begin{theorem}
\label{th2}Under the conditions of Theorem~\ref{th1}, let
$\boldsymbol{\widetilde u}$ be the residual vector from least squares
on the always in controls and the final selected candidates. Let
$\boldsymbol{\widetilde\beta}_n$ contain the corresponding candidate
coefficient estimates, with zero entries for unselected candidates, and
let $\boldsymbol\beta_n$ be the true candidate coefficient vector.
Then
\begin{equation}
F_{\widetilde u}=T^{-1}\|\boldsymbol{\widetilde u}\|^2
=\sigma^2+O_p(T^{-1/2}),\label{forecast00}
\end{equation}
and
\begin{equation}
F_{\boldsymbol{\widetilde\beta}}
=\|\boldsymbol{\widetilde\beta}_n-\boldsymbol\beta_n\|^2
=O_p(T^{-1}).\label{fnormbeta}
\end{equation}
\end{theorem}

These are rates for the average squared residual and the squared
Euclidean coefficient error. The coefficient error itself is
$O_p(T^{-1/2})$. Thus a finite number of selected proxies does not change these rates.
The result does not assert equality with the oracle coefficient
distribution or give an out of sample forecasting bound.

Further discussion of the design conditions, threshold choice and early
stopping is given after the proof in Appendix~\ref{app:am_estimation_proof}.

\subsection{Post selection inference}\label{subsec:post_selection}

\begin{theorem}
\label{th:post_selection_oracle}
Consider \eqref{dgp1} with fixed $k$ and fixed control dimension $\zeta$.
Suppose that the conditions of Theorem~\ref{th3} or Theorem~\ref{th4} hold,
so that $\Pr(\mathcal A_1)\to1$. Order the controls first and the signals
second, and write
\[
q_t=q_{\mathcal S_s,t}=(\boldsymbol z_t',\boldsymbol x_{\mathcal S_s,t}')'.
\]
In the control basis specified in Section~\ref{subsec:assumptions}, suppose
in addition that the following deterministic limits exist and are positive
definite:
\[
 Q_{\mathcal S_s,T}=\overline E_T(q_tq_t')\longrightarrow Q,
 \qquad
 \Omega_T=\overline E_T(u_t^2q_tq_t')\longrightarrow\Omega.
\]
Let $V_{\mathcal S_s}$ be the lower right $k\times k$ block of
$Q^{-1}\Omega Q^{-1}$. Let $\boldsymbol{\widetilde\beta}_n$ be the
coefficient vector from the final BMT regression, with zeros at unselected
indices, as defined above. Then:
\begin{itemize}
\item[(a)] (Consistency)
\[
 \boldsymbol{\widetilde\beta}_{\mathcal S_s}
 -\boldsymbol\beta_{\mathcal S_s}\xrightarrow{p}0.
\]
For any sequence of indices $i\notin\mathcal S_s$,
$\widetilde\beta_i\xrightarrow{p}0$.

\item[(b)] (Asymptotic normality)
\[
 \sqrt T\bigl(\boldsymbol{\widetilde\beta}_{\mathcal S_s}
       -\boldsymbol\beta_{\mathcal S_s}\bigr)
 \xrightarrow{d}N(0,V_{\mathcal S_s}).
\]

\item[(c)] (Consistent covariance estimation)
Let $\widehat S$ be the selected candidate set and $\widehat u_t$ the
residual from the final regression, including the always in controls. Set
\[
 \widehat q_t=(\boldsymbol z_t',\boldsymbol x_{\widehat S,t}')',\qquad
 \widehat Q=T^{-1}\sum_{t=1}^T\widehat q_t\widehat q_t',\qquad
 \widehat\Omega_{\rm HC}
 =T^{-1}\sum_{t=1}^T\widehat u_t^2\widehat q_t\widehat q_t'.
\]
Embed the candidate block of
$\widehat Q^{-1}\widehat\Omega_{\rm HC}\widehat Q^{-1}$ in an $n\times n$
matrix by assigning zeros to unselected rows and columns. Its true signal
block, denoted $\widehat V_{\mathcal S_s}$, satisfies
\[
 \widehat V_{\mathcal S_s}\xrightarrow{p}V_{\mathcal S_s}.
\]
The covariance estimate for the unscaled coefficient estimator is
$\widehat V_{\mathcal S_s}/T$.
\end{itemize}
\end{theorem}

Theorem~\ref{th:post_selection_oracle} gives the same first order distribution
as OLS on the true signals and the always in controls. The two additional
limits concern this fixed dimensional regression only; they are not needed
for the selection results. Assumption~\ref{ass10} is not used. The signals
may shrink at the rates allowed by Theorems~\ref{th3}--\ref{th4}.
The conclusion is pointwise, not uniform near the model selection
boundary.\footnote{See \citet{LeebPoetscher2008} for the limitations of
uniform distributional approximations after model selection.}
Corollary~\ref{cor:growing_signals} does not supply a growing dimensional
version of this inference result. Excluded coefficients have no
nondegenerate normal limit under exact recovery.

In the homoskedastic special case $\Omega=\sigma^2Q$, the covariance reduces
to the signal block of $\sigma^2Q^{-1}$. The corresponding ordinary OLS
estimate, using the final residual variance, is then consistent.
For example, $E(u_t^2\mid q_t)=\sigma^2$ implies this simplification;
constant unconditional variance alone does not.

The score restrictions in Assumption~\ref{ass1} give
$E(q_tu_t\mid\mathcal F_{t-1})=\boldsymbol 0$ almost surely.
Since $q_tu_t$ is $\mathcal F_t$ measurable, it is a martingale difference
with respect to $(\mathcal F_t)$, so its nonzero lag covariances vanish. HC estimation therefore suffices even when the regressors
are persistent. A Bartlett HAC estimate with deterministic integer bandwidth
$L_T\to\infty$ and $L_T=o(\sqrt T)$ is also consistent for the same covariance;
its definition and proof are in Appendix~\ref{app:post_selection_proof}.
This does not relax \ref{ass1} or change the conventional statistics used for selection.

\subsection{Comparison with other methods}\label{subsec:lasso_comparison}
Variable selection results for the Lasso typically combine a design restriction with a lower bound on signal strength; see \citet{ZhaoYu2006,MeinshausenBuhlmann2006,Wainwright2009}. Partitioning the population regressor covariance into active and inactive blocks, the strong irrepresentable condition is
\[
 \bigl\|\boldsymbol\Sigma_{21}\boldsymbol\Sigma_{11}^{-1}
 \operatorname{sign}(\boldsymbol\beta_1)\bigr\|_\infty\leq 1-\eta,
 \qquad \eta>0.
\]
This condition concerns the actual sign vector. Uniformity over all sign vectors is a stronger requirement. The corresponding weak necessary condition replaces $1-\eta$ by $1$. Signal strength must also be sufficient relative to the penalty and sampling error. BMT has a detectability requirement too: condition~\eqref{cond1} requires a studentised score gap larger than $\sqrt{\log n/T}$ on every proper signal subset. Its coefficient weighted scores differ from the Lasso's sign weighted restriction. Their relative strength depends on the design and coefficients. In particular, the Lasso irrepresentable condition can allow nonzero signal--noise covariances. Our noise category instead imposes orthogonality to both signals and proxies after removing the controls.

OCMT also allows shrinking detectable coefficients; see \citet[Assumption~6]{chu2018}. Its net effects must be large relative to the testing threshold, a beta min type requirement. OCMT can retain proxies with nonzero net effects; exact recovery by BMT additionally requires at least one remaining signal to have a larger absolute population score than every proxy after conditioning on any incomplete subset of true signals, by a sufficient margin. OCMT explicitly allows proxy selection in its approximating model results and discusses a subsequent model selection step \citep{chu2018}. Our comparison below concerns OCMT without that additional step. The mixing extensions in Section~4.2 and the theoretical supplement of \citet{chu2018} show that temporal dependence alone does not distinguish the methods. Our assumptions retain score martingale difference restrictions and impose local mixing and bounds on residual second moments over the stated deterministic classes. These conditions do not give a general ordering of the dependence assumptions.

The error control targets also differ. A family wise error probability
concerns at least one false inclusion, whereas statistical FDR is the
expected proportion of false discoveries among the selected variables,
with proportion zero for an empty selection. The sample ratios denoted
by $FDR$ above are false discovery proportions; their expectations are
the corresponding false discovery rates. The approximating model ratio $FDR^{(0)}$ excludes selected proxies from both numerator and denominator. Our selection results make
the probability of any noise inclusion vanish under their conditions;
excluding proxies additionally requires the exact recovery conditions.
They do not establish finite sample family wise control at nominal level
$\alpha_0$ for every adaptive path. Knockoff methods target FDR under
their exchangeability and model conditions
\citep{barber2015knockoffs,candes2018modelx}. Ordered testing rules such
as ForwardStop target FDR for a sequence of hypotheses under the required
null $p$ value conditions \citep{gse2015}. Conventional BMT statistics
on a data dependent path cannot simply be substituted for those
$p$ values. Neither error criterion implies a uniform ranking of
selection methods. The results here concern selection at a given sample size. They do not provide time uniform error control for repeated use in real time monitoring.

Factor adjustment addresses a different source of dependence. \citet{FanKeWang2020} propose FarmSelect, a framework that estimates latent common factors and applies regularised model selection to the idiosyncratic components, while accounting for the estimated factors without penalisation. The framework can accommodate different penalty functions, including Lasso and SCAD. Its theoretical guarantees depend on the accuracy of factor estimation and conditions on the
factor adjusted design. When dependence is driven by a few pervasive factors, FarmSelect provides a relevant alternative to applying selection methods directly to the original regressors. However, local proxy dependence need not admit such a factor structure. In principle, a similar factor adjustment step could also be incorporated into BMT, although the theoretical properties of such an extension would require separate investigation.

\subsubsection{Composite proxies and coefficient heterogeneity}
\label{subsec:stylised_comparison}
Consider a fixed set $S=\{1,\ldots,k\}$ of independent standard Gaussian signals and a fixed finite set $P$ of proxies. For each $p\in P$, let
\[
 x_{pt}=c_p'x_{St}+\sqrt{1-\|c_p\|_2^2}\,e_{pt},
 \qquad \|c_p\|_2<1.
\]
The $e_{pt}$ are independent standard Gaussian variables, independent of the signals. Set $y_t=\sum_{i\in S}\beta_i x_{it}+u_t$, with fixed nonzero $\beta_i$ and independent $u_t\sim N(0,\sigma^2)$, $\sigma^2>0$. Observations are independent over $t$. Add independent standard Gaussian noise candidates, independent of these variables, so that the total number $n\asymp T^\kappa$ for a fixed $\kappa>0$. An intercept is included and there are no other always in controls. All design parameters are fixed as $T$ increases. This is a specified orthogonal signal design, not a rotation of a general selection problem.

For $A\subsetneq S$, write $R=S\setminus A$ and define
\begin{equation}\label{eq:composite_scores}
 M_A=\max_{i\in R}|\beta_i|,
 \qquad
 H_{p,A}=\frac{|c_{p,R}'\beta_R|}{\sqrt{1-\|c_{p,A}\|_2^2}}.
\end{equation}
These are covariances with the residual outcome divided by the residual candidate standard deviations. Dividing by the common residual outcome standard deviation gives the absolute partial correlations. Their ordering therefore also orders the population studentised scores.

\begin{theorem}\label{thm:Th7}
In this composite proxy design, suppose $H_{p,A}<M_A$ for every $p\in P$ and every $A\subsetneq S$. With the threshold in \eqref{cvf} and the same sufficiently large fixed exponent $\delta$ at every stage,
\[
 \Pr(\widehat S_{\mathrm{BMT}}=S)\longrightarrow1.
\]
For a proxy $p$, the Lasso irrepresentable quantity is
$|c_p'\operatorname{sign}(\beta_S)|$. If this exceeds one, the weak necessary irrepresentable condition fails. In that case, for the objective
$(2T)^{-1}\|y-Xb\|_2^2+\lambda_T\|b\|_1$, with an unpenalised intercept,
$\lambda_T\to0$ and $\sqrt T\lambda_T\to\infty$ imply
$\Pr(\widehat S_{\mathrm{Lasso}}=S)\to0$.
\end{theorem}

A simple nonempty separation region makes the role of heterogeneous coefficients explicit. Take one proxy, $k\geq2$,
$\beta_S=(1,\epsilon,\ldots,\epsilon)'$ and
$c=(c_1,a,\ldots,a)'$, where $0<c_1<1$, $0<\epsilon<1$ and $g=(k-1)a$. Then
\begin{equation}\label{eq:composite_g_interval}
 1-c_1<g<\min\left\{\sqrt{1-c_1^2},\frac{1-c_1}{\epsilon}\right\}
\end{equation}
implies positive definiteness, strict BMT ordering on every signal subset, and failure of the Lasso irrepresentable condition. The Lasso probability statement requires the tuning regime in Theorem~\ref{thm:Th7}. In contrast, if all $|\beta_i|=b>0$, then $H_{p,\varnothing}=b|c_p'\operatorname{sign}(\beta_S)|$. A strict violation of the irrepresentable condition makes the proxy outrank every signal at the first stage. Hence homogeneous coefficients cannot produce this BMT favourable separation merely by increasing support size. Proofs are given in Appendix~\ref{app:general_proof}.

The following special case gives the three regressor illustration. Such a proxy can represent an observable that combines information from several underlying drivers, as in factor based applications \citep{StockWatson2002}.
\begin{equation}\label{eq:dgp_wedge}
 y_t=x_{1t}+\alpha x_{2t}+u_t,\qquad
 x_{3t}=\rho(x_{1t}+x_{2t})+\sqrt{1-2\rho^2}\,e_{3t}.
\end{equation}
\begin{corollary}\label{thm:wedge}
Under the sampling, noise candidate and threshold conditions above, let
$0<\alpha<1$ and $0<\rho<1/\sqrt2$. BMT recovers $S=\{1,2\}$ with probability tending to one if
$\rho(1+\alpha)<1$. The Lasso irrepresentable condition fails if $2\rho>1$. In that case, the Lasso fails to recover $S$ with probability tending to one under the tuning regime of Theorem~\ref{thm:Th7}. Thus the separation interval is
\[
 \left(\frac12,\min\left\{\frac1{1+\alpha},\frac1{\sqrt2}\right\}\right).
\]
\end{corollary}

After $x_1$ enters, the signal and proxy scores in \eqref{eq:composite_scores} are $\alpha$ and $\alpha\rho/\sqrt{1-\rho^2}$. Positive definiteness makes the latter smaller. The same calculation applies if $x_2$ is conditioned on first. For $\alpha=0.2$, the upper endpoint is $1/\sqrt2$, not $0.83$. The proxy's initial covariance with $y$ is $\rho(1+\alpha)$; $2\rho$ is its sign weighted irrepresentable quantity. The same sign restriction is essential to this last expression.

The growing candidate pool ensures a diverging stopping threshold. A literal three candidate experiment with fixed nominal size does not give the stated probability one terminal stopping result. All comparisons here have fixed signal dimension. Appendix~\ref{app:other_comparisons} gives the corresponding OCMT, top $k$ SIS and Adaptive Lasso results, and a separately qualified near singular Gaussian knockoff calculation. These results do not assert uniform dominance over those methods or extend Corollary~\ref{cor:growing_signals} without checking its additional conditions.

SCAD and MCP have zero penalty derivative for sufficiently large coefficients, so the Lasso shrinkage argument above does not carry over directly \citep{fan2001,zha2010}. Neither method uniformly dominates BMT. Appendix~\ref{app:scad_mcp} gives explicit, separately specified designs in both directions: BMT can retain a weak conditional signal that the global unweighted folded concave criterion deletes, while that criterion can recover a fixed signal model in which irreversible BMT admits a proxy. The appendix also distinguishes the global minimiser from an oracle local solution and from a computed estimator. Its Lasso initialised LLA example concerns the stated initialisation and tuning, not every SCAD/MCP algorithm.

\section{Simulation Study}
\label{simulation}

This section evaluates the finite sample performance of BMT
relative to several competing variable selection procedures.
We consider a range of diagnostic measures, including the
Matthews Correlation Coefficient (MCC)
\citep{Matthews1975,ChiccoJurman2020}, the F1 score,
the True Discovery Rate (TDR), the False Discovery Rate
(FDR), the True Positive Rate (TPR), the False Positive
Rate (FPR), and Exact Recovery (ER), defined as the proportion of Monte Carlo replications in which the selected model coincides with the true model.

TDR, often referred to as ``precision'', measures the proportion of selected variables that are true signals,
while TPR, or ``recall'', captures the proportion of true signals that are correctly identified. FDR quantifies the share of selected variables that are nonsignals, and FPR reflects the share of nonsignals that are erroneously selected. The F1 score, defined as the harmonic mean of TDR and TPR, summarises the precision recall tradeoff. Exact recovery is a more stringent criterion, requiring all true signals to be selected (TPR = 1) and no nonsignals to be included (FPR = 0).

The MCC combines true positives, true negatives, false positives, and false negatives into a single balanced measure of classification performance. It ranges from $-1$ (complete disagreement) to $+1$ (perfect classification), with $0$ corresponding to performance no better than random guessing. MCC is particularly informative in sparse, high dimensional settings, where true signals are few relative to the total number of candidate regressors. For this reason, we rely primarily on MCC to assess variable selection accuracy, while reporting the remaining classification measures as complementary diagnostics. Further details are provided in the Appendix.

In addition, we report model size, defined as the average number of selected variables, the root mean squared error (RMSE) of the estimated coefficient vector, the root mean squared forecast error (RMSFE), and computational time.

We compare BMT with six alternative procedures: the OCMT method of \citet{chu2018},  the Generalised OCMT (GOCMT) procedure of \citet{Sharifvaghefi2025}, Lasso \citep{tib1996}, Adaptive Lasso \citep{zou2006adaptive}, the Minimax Concave Penalty (MCP) of \citet{zha2010}, and the Smoothly Clipped Absolute Deviation (SCAD) penalty of \citet{fan2001}. GOCMT extends OCMT by conditioning on estimated latent common factors to account for pervasive dependence among the candidate regressors. 

Lasso is implemented using MATLAB's built in \texttt{lasso} function with ten fold cross validation. The regularisation path is generated automatically using MATLAB's default settings, without imposing a user specified $\lambda$ grid. We consider two tuning rules: the minimum cross validation MSE criterion and the one standard error (1SE) rule, which favours more parsimonious specifications. We also consider a prefiltered version of Lasso (PF Lasso), which involves projecting out estimated latent common factors from the candidate regressors prior to variable selection, as in \citet{FanKeWang2020}. 
Adaptive Lasso is implemented in two stages, with
an initial Lasso estimator used to construct adaptive
penalty weights, followed by a weighted Lasso regression. We consider two variants, applying either the 1SE rule or the minimum MSE criterion consistently in both stages.

For MCP and SCAD, we employ the two step Local Linear
Approximation (2LLA) algorithm, following
\citet{FanXueZou2014}. Starting from a separately
tuned Lasso estimator, the nonconvex penalty is
locally approximated by a weighted $\ell_1$ penalty,
and the resulting weighted Lasso problem is solved
in two successive steps. We use the conventional
concavity parameters $\gamma=3$ for MCP and $a=3.7$
for SCAD. For both methods, we consider a grid of
100 logarithmically spaced regularisation parameters
ranging from $\lambda_{\max}$ to $0.01\lambda_{\max}$,
where $\lambda_{\max}=T^{-1}\max_j
|\mathbf{x}_j^{*\prime}\mathbf{y}^*|$, with
$\mathbf{x}_j^*$ and $\mathbf{y}^*$ denoting the
standardised regressors and centred dependent
variable after partialling out the lagged dependent
variable. The initial Lasso and final MCP or SCAD
regularisation parameters are selected separately
by minimising prediction MSE on an independent
validation sample, following post selection OLS
refitting on the original regression equation.
For all penalised methods, the lagged dependent variable is partialled out before penalised estimation and is subsequently included without penalisation in the post selection OLS regression, together with the selected covariates. We consider prefiltering for standard Lasso to examine whether accounting for latent common factors improves selection accuracy. The results below compare its performance with the unfiltered versions of BMT, Adaptive Lasso, MCP and SCAD.

\subsection{Design}

The true DGP is given by
\begin{equation}
y_{t}=\alpha y_{t-1} + \sum_{i=1}^{k} \beta_{i} x_{it} + \varsigma{u}_{t},
\end{equation}
where $u_{t}\sim i.i.d.N(0,1)$ $\forall t$. The autoregressive parameter $\alpha$
alternates between $0.4$ and $0.8$.

We consider the same number of pseudo signals, $k$, denoted $x_{k+1,t},
x_{k+2,t} \dots, x_{2k,t}$. Both sets of variables are generated as
\begin{equation}
x_{it} = \frac{\varepsilon_{it} + \nu_{g} g_{t} + \nu_{f} f_{t}}{\sqrt{1 +
\nu_{g}^{2} + \nu_{f}^{2}}}; \quad\text{for } i = 1, \dots, 2k.
\end{equation}
The remaining zero coefficient regressors are generated as
\begin{equation}
x_{it} = \frac{\varepsilon_{i-1,t} + \varepsilon_{it} + \nu_{f} f_{t}}{\sqrt{2
+ \nu_{f}^{2}}}; \quad\text{for } i > 2k.
\end{equation}
The terms $\varepsilon_{it}$ represent idiosyncratic
variations across variables, while $g_{t}$ and $f_{t}$ capture ``local'' and
common factors, respectively. The local factor $g_{t}$ affects only the true
signals and pseudo signals, whereas the common factor $f_{t}$ influences all variables.

$g_{t}$, $f_{t}$ and $\varepsilon_{it}$ are drawn as AR(1) processes with mean zero and variance one. In particular, we specify
\begin{equation}
z_t = \varrho z_{t-1} + \sqrt{1-\varrho^2}\,\upsilon_{z,t},
\qquad z \in \{\varepsilon_i, f, g\},
\end{equation}
where $\upsilon_{z,t}\sim i.i.d.N(0,1)$ for
$z\in\{\varepsilon_i,f,g\}$.

The inclusion of local and common factors serves two main purposes. First, it introduces multicollinearity among regressors, as is typical in real world high dimensional datasets where correlated signals and pseudo signals arise
due to shared latent structures. Second, it captures the common variation often observed in big data settings, where many covariates are driven by latent common factors and shared underlying shocks. Examples of such drivers include
macroeconomic trends, technological innovations, regulatory changes, demographic shifts, geopolitical developments, and sector wide shocks such as fluctuations in commodity prices, global supply chain disruptions, or monetary policy shifts. These factors can lead to strong comovement among otherwise distinct variables, thereby shaping the correlation structure of the candidate set of regressors. Finally, modeling $g_t$, $f_t$, and $\varepsilon_{it}$ as stationary AR(1) processes induces serial correlation over time and persistence typical of economic time series, so the design features both cross sectional and temporal dependence.

The signal to noise ratio of the model is defined as
\begin{equation}
SNR= \frac{var\left(  \sum_{i=1}^{k} \beta_{i} x_{it} \right)}{var\left( \varsigma u_{t}\right)} = \frac{D}{\varsigma^{2}}
\end{equation}
where $D = \left(  B + \gamma
b_{s}^{2} \right) / \left(  1 + \gamma\right)$,  $B =
\|\boldsymbol{\beta}\|^{2}$, $b_{s} = \boldsymbol{\beta}^{\prime}
\boldsymbol{\iota}$ with $\boldsymbol{\iota}=\left(  1,1,\dots,1\right)
^{\prime}$, and $\gamma=\nu_{f}^{2}+\nu_{g}^{2}$. Intuitively, $\gamma$ summarises how much of the regressors' variance is explained by common factors rather than idiosyncratic noise: $\nu_{f}^{2}$ captures the strength of a global common factor, whereas $\nu_{g}^{2}$ captures the strength of the local factor. Solving for $\varsigma$ yields
\begin{equation}
\varsigma=\sqrt{\frac{D}{SNR}}.
\end{equation}
It follows that $\varsigma$ adjusts
endogenously in response to changes in $\boldsymbol{\beta}$, and
$\gamma$, so as to maintain a constant signal to noise ratio.

We control the strength of factor induced dependence among the candidate regressors through the parameter $V=1+k\gamma$, where
$\gamma=\nu_g^2+\nu_f^2$ measures the combined contribution of the local and global factors. In particular, for any two distinct regressors $x_{it}$ and $x_{i't}$ with $i,i'\leq 2k$,
\begin{equation}
\operatorname{corr}(x_{it},x_{i't})
=\frac{\nu_g^2+\nu_f^2}
{1+\nu_g^2+\nu_f^2}
=\frac{\gamma}{1+\gamma}.
\end{equation}
Thus, larger values of $V$ imply stronger factor induced
dependence among signals and proxies, making it more difficult to distinguish genuine explanatory variables from correlated proxies and to identify their individual contributions. Solving for $\gamma$ yields $\gamma=(V-1)/k$. We specify a parameter $\pi$ that determines the proportion of $\gamma$ attributable to the global factor, such that $\nu_f^2=\pi(V-1)/k$ and $\nu_g^2=(1-\pi)(V-1)/k$. Hence,
\begin{equation}
\nu_f=\sqrt{\pi\frac{V-1}{k}},
\qquad
\nu_g=\sqrt{(1-\pi)\frac{V-1}{k}}.
\end{equation}
This parametrisation allows us to control both the overall strength of factor dependence, through $V$, and its allocation between the local and global factors, through $\pi$.
We consider $\pi\in\{0,0.5,1\}$, thereby allowing for three
distinct dependence structures. When $\pi=0$, we have $\nu_f=0$, so that all factor induced multicollinearity arises from the local factor $g_t$, which affects only the true signals and pseudo signals. In this case, the remaining noise
regressors are not exposed to a pervasive common factor. This configuration is closely related to the simulation framework of \citet{chu2018} and provides an important benchmark for assessing OCMT in an environment without global factor
dependence. When $\pi=0.5$, local and global factors contribute equally to multicollinearity, whereas $\pi=1$ corresponds to dependence driven entirely by the global factor $f_t$.

We set $SNR = 3$, $\varrho = 0.6$, and  $\delta=1$, which corresponds to a relatively mild choice of the critical value exponent.\footnote{This choice allows us to assess the performance of BMT in designs featuring both local and global latent factors, which generate proxy signals through different correlation channels. Larger values of $\delta$ impose more conservative thresholds and further reduce the probability of selecting proxies, but may also reduce the detection of true signals. Focusing on a mild value of $\delta$ therefore provides an assessment of BMT’s ability to control proxy selection under relatively low thresholds.}
 We set the signal coefficients to $\beta_{i} = 1 $ for all $i = 1, \dots, k $. We consider two levels of sparsity, $k\in\{1,4\}$, and two levels of factor dependence strength,
$V\in\{2,4\}$. Higher values of $V$ correspond to stronger factor induced dependence among the regressors, whether arising from local or global factors. For instance, when $k=4$, increasing $V$ from $2$ to $4$ raises the pairwise correlation
between the true and pseudo signals from $0.20$ to approximately $0.43$.

We vary the sample size $T$ and the total number of candidate regressors $n$ independently over $\{100,200,300\}$. There are $k$ signals, $k$ specified pseudo signals and $n-2k$ remaining zero coefficient regressors. Together with the two values of $\alpha$, two values of $V$ and three values of $\pi$, this gives 108 designs for each value of $k$, with 2,000 Monte Carlo replications per design.

\subsection{Results}

Tables~\ref{tab:MCC}--\ref{tab:CompTime_K} summarise the
performance of the eleven reported specifications across 108 simulation
designs for $k=4$.\footnote{To save space, we do not report results for $k=1$, as they are quantitatively very similar to those for $k=4$. Theorem~\ref{th3} gives an asymptotic exact recovery result for $k=1$ under its stated conditions. We therefore focus on the more challenging case of $k=4$.} 
We focus on variable selection accuracy, measured by
the Matthews Correlation Coefficient (MCC) and exact
recovery, coefficient estimation accuracy, measured
by RMSE, and computational efficiency.

The MCC results in Table~\ref{tab:MCC} show that BMT
outperforms the alternatives in terms of overall variable
selection accuracy. BMT attains a median MCC
of $0.987$, with a narrow interquartile range (IQR) of
$0.027$, and achieves $MCC>0.8$ in $98.1\%$ of designs.
It ranks first on average (mean rank $1.611$) and
achieves the highest MCC in $72.22\%$ of designs.

\begin{table}[th]
\caption{Summary Statistics for Matthews Correlation Coefficient (MCC)}
\label{tab:MCC}
\centering
\begin{tabular}{lccccccc}
\hline
& \textbf{Median} & \textbf{IQR} & \textbf{Min} & \textbf{Max}
& \textbf{Prop. $>$ 0.8} & \textbf{Rank} & \textbf{Prop. 1st} \\
\hline

BMT
& 0.987 & 0.027 & 0.772 & 0.993 & 98.10\% & 1.611 & 72.22\% \\

OCMT
& 0.584 & 0.492 & 0.000 & 0.911 & 14.80\% & 8.870 & 0.00\% \\

GOCMT
& 0.844 & 0.223 & 0.553 & 0.989 & 60.20\% & 5.602 & 0.00\% \\

Lasso-1SE
& 0.869 & 0.117 & 0.621 & 0.950 & 74.10\% & 6.213 & 0.00\% \\

Lasso-MSE
& 0.536 & 0.067 & 0.429 & 0.638 & 0.00\% & 9.694 & 0.00\% \\

PF Lasso-1SE
& 0.889 & 0.077 & 0.477 & 0.949 & 86.10\% & 5.917 & 0.00\% \\

PF Lasso-MSE
& 0.531 & 0.075 & 0.427 & 0.648 & 0.00\% & 9.833 & 0.00\% \\

Ad. Lasso-1SE
& 0.969 & 0.076 & 0.717 & 0.996 & 97.20\% & 3.009 & 12.04\% \\

Ad. Lasso-MSE
& 0.573 & 0.067 & 0.459 & 0.671 & 0.00\% & 9.111 & 0.00\% \\

MCP
& 0.959 & 0.025 & 0.849 & 0.978 & 100.00\% & 2.657 & 11.11\% \\

SCAD
& 0.953 & 0.029 & 0.836 & 0.977 & 100.00\% & 3.481 & 4.63\% \\

\hline
\end{tabular}

\vspace{0.5em}
\par
\begin{minipage}{0.99\linewidth}
\small
\textit{\textbf{Notes}:}
All summary statistics are computed across the 108 simulation
designs for $k=4$.
MCC ranges from $-1$ (perfect misclassification) to $+1$
(perfect selection), with $0$ indicating random guessing.
``Prop. $>$ 0.8'' reports the percentage of designs in which
$MCC>0.8$.
``Rank'' is the average rank, with rank 1 denoting the
best performing method; for MCC, higher values receive
better ranks.
``Prop. 1st'' reports the percentage of designs in which
a method attains the highest MCC.
\end{minipage}
\end{table}

Among the alternative procedures, MCP and SCAD exhibit strong and relatively stable selection performance, with median MCC values of $0.959$ and $0.953$, respectively, and narrow IQRs of $0.025$ and $0.029$. Both achieve $MCC>0.8$ in all simulation designs, although they rank first in only $11.11\%$ and $4.63\%$ of designs, respectively. Adaptive Lasso with the 1SE tuning rule also performs particularly well, attaining a median MCC of $0.969$ and ranking first in $12.04\%$ of designs. Its performance is, however, somewhat more variable than that of BMT, as reflected in its larger IQR ($0.076$ compared with $0.027$).

The performance of Lasso based procedures depends substantially on the tuning rule. Under the 1SE rule, standard Lasso and prefiltered Lasso attain median MCC values of $0.869$ and $0.889$, respectively, whereas the corresponding values under the minimum MSE rule fall to $0.536$ and $0.531$. A similar pattern emerges for Adaptive Lasso, whose median MCC declines from $0.969$ under the 1SE rule to $0.573$ under the minimum MSE rule. These results indicate that the more parsimonious specifications favoured by the 1SE rule yield substantially better overall selection accuracy in the designs considered.

OCMT exhibits substantially weaker and more variable selection performance, with a median MCC of $0.584$ and an IQR of $0.492$. As illustrated by the figures below, this deterioration is primarily associated with designs in which $\pi>0$, so that a global factor induces pervasive dependence among the candidate regressors. Accounting for this dependence
through GOCMT yields substantial improvements, raising the median MCC to $0.844$ and reducing the IQR to $0.223$. Nevertheless, GOCMT remains less accurate and considerably more variable than BMT, MCP, SCAD and Adaptive Lasso under the 1SE rule, and does not attain the highest MCC in any design.
Across these designs, BMT achieves the highest median MCC, the best average rank and the largest proportion of first place rankings.

The median MCCs of BMT, Adaptive Lasso-1SE, MCP and SCAD change less with the source of factor dependence than those of standard Lasso-1SE, although individual designs remain sensitive to it. For example, when $V=4$, increasing $\pi$ from $0$ to $1$ changes the median MCC of BMT, Adaptive Lasso-1SE, MCP and SCAD only marginally: from $0.986$ to $0.983$, $0.976$ to $0.964$, $0.958$ to $0.956$, and $0.954$ to $0.951$, respectively. By contrast, the median MCC of standard Lasso-1SE declines from $0.923$ to $0.727$, while prefiltering attenuates this decline, with the corresponding MCC falling from $0.924$ to $0.852$.

The exact recovery results in Table~\ref{tab:ExRec} provide further evidence of the strong selection performance of BMT. It achieves the highest median exact recovery rate of $0.889$, with an IQR of $0.175$, and a recovery rate exceeding $0.8$
in $67.6\%$ of designs. BMT also has the best average rank ($2.005$) and has a first place share of $42.13\%$.

Among the other methods, Adaptive Lasso with the 1SE tuning rule performs particularly well, attaining a median exact recovery rate of $0.842$ and ranking first in $32.41\%$ of designs. Its performance is, however, considerably more variable than that of BMT, as indicated by its
larger IQR of $0.364$. MCP and SCAD also achieve high median recovery rates of $0.795$ and $0.766$, respectively, with substantially lower dispersion (IQRs of $0.088$ and $0.122$).
MCP has a first place share of $21.76\%$, compared with $1.85\%$ for SCAD.

As with MCC, the tuning rule has a substantial effect on the performance of Lasso based procedures. Under the 1SE rule, standard and prefiltered Lasso achieve median exact recovery rates of $0.373$ and $0.466$, respectively, compared with
only $0.007$ and $0.008$ under the minimum MSE rule. For Adaptive Lasso, the corresponding rates are $0.842$ and $0.025$. Thus, the more parsimonious specifications favoured by the 1SE rule yield considerably higher exact recovery rates.

OCMT performs poorly, with a median exact recovery rate of zero. GOCMT improves this to $0.220$, although its performance remains highly variable (IQR $0.626$). It achieves a recovery rate above $0.8$ in $16.7\%$ of designs and ranks first
in only $1.85\%$.

BMT has the highest median exact recovery rate and the best average rank across the full design grid. Adaptive Lasso-1SE and MCP also perform well. At $T=300$, Adaptive Lasso-1SE has a higher median exact recovery rate than BMT ($0.933$ compared with $0.915$).

\begin{table}[th]
\caption{Summary Statistics for Exact Recovery}
\label{tab:ExRec}
\centering
\begin{tabular}{lccccccc}
\hline
& \textbf{Median} & \textbf{IQR} & \textbf{Min} & \textbf{Max}
& \textbf{Prop. $>$ 0.8} & \textbf{Rank} & \textbf{Prop. 1st} \\
\hline

BMT
& 0.889 & 0.175 & 0.186 & 0.936 & 67.60\% & 2.005 & 42.13\% \\

OCMT
& 0.000 & 0.019 & 0.000 & 0.426 & 0.00\% & 9.694 & 0.00\% \\

GOCMT
& 0.220 & 0.626 & 0.000 & 0.904 & 16.70\% & 6.560 & 1.85\% \\

Lasso-1SE
& 0.373 & 0.318 & 0.018 & 0.706 & 0.00\% & 6.273 & 0.00\% \\

Lasso-MSE
& 0.007 & 0.014 & 0.000 & 0.046 & 0.00\% & 9.537 & 0.00\% \\

PF Lasso-1SE
& 0.466 & 0.226 & 0.030 & 0.720 & 0.00\% & 5.690 & 0.00\% \\

PF Lasso-MSE
& 0.008 & 0.016 & 0.000 & 0.036 & 0.00\% & 9.287 & 0.00\% \\

Ad. Lasso-1SE
& 0.842 & 0.364 & 0.136 & 0.980 & 57.40\% & 2.579 & 32.41\% \\

Ad. Lasso-MSE
& 0.025 & 0.027 & 0.002 & 0.116 & 0.00\% & 8.352 & 0.00\% \\

MCP
& 0.795 & 0.088 & 0.352 & 0.882 & 47.20\% & 2.491 & 21.76\% \\

SCAD
& 0.766 & 0.122 & 0.308 & 0.870 & 27.80\% & 3.532 & 1.85\% \\

\hline
\end{tabular}

\vspace{0.5em}
\par
\begin{minipage}{0.99\linewidth}
\small
\textit{\textbf{Notes}:}
All summary statistics are computed across the 108 simulation
designs for $k=4$. Exact recovery is the proportion of
replications in which the exact model is recovered.
``Prop. $>$ 0.8'' reports the percentage of designs in which
the exact recovery rate exceeds 0.8.
``Rank'' is the average rank, with rank 1 denoting the
best performing method; for exact recovery, higher values
receive better ranks.
``Prop. 1st'' reports each method's share of first place rankings, expressed as a percentage.
\end{minipage}
\end{table}

Table~\ref{tab:RMSE} reports the RMSE of the estimated
coefficient vector.\footnote{RMSE is defined as $\sqrt{r^{-1}\sum_{j=1}^{r}
\|\widetilde{\boldsymbol{\beta}}_n^{(j)}
-\boldsymbol{\beta}_n\|^2}$, where $\widetilde{\boldsymbol{\beta}}_n^{(j)}$ denotes the full dimensional post selection coefficient vector, containing OLS estimates for selected regressors and zeros elsewhere, and $r$ is the number of Monte Carlo replications.} BMT achieves the lowest median RMSE ($0.324$), followed by MCP ($0.339$), SCAD ($0.347$) and Adaptive Lasso under the 1SE rule ($0.353$). BMT also has the best average rank ($2.685$) and achieves the lowest RMSE in $35.19\%$ of designs, compared with $18.52\%$ for both MCP and SCAD and $12.96\%$ for Adaptive Lasso-1SE.
MCP and SCAD exhibit somewhat lower dispersion, with IQRs of $0.179$ and $0.186$, respectively, compared with $0.261$ for BMT.

The 1SE rule again yields substantially lower RMSE than the minimum MSE rule across all Lasso variants. OCMT exhibits considerable instability, with a median RMSE of $0.711$, an IQR of $0.999$ and an exceptionally large maximum of $898.602$. Accounting for common factor dependence through GOCMT  substantially improves estimation accuracy, reducing the
median RMSE to $0.406$ and the IQR to $0.411$. GOCMT also achieves the lowest RMSE in $13.89\%$ of designs, although its
performance remains more variable than that of BMT, MCP and SCAD. The results therefore favour BMT in terms of coefficient estimation accuracy, while MCP and SCAD exhibit somewhat greater stability across the simulation designs.

\begin{table}[th]
\caption{Summary Statistics for RMSE}
\label{tab:RMSE}
\centering
\begin{tabular}{lcccccc}
\hline
& \textbf{Median} & \textbf{IQR} & \textbf{Min} & \textbf{Max}
& \textbf{Rank} & \textbf{Prop. 1st} \\
\hline

BMT
& 0.324 & 0.261 & 0.200 & 1.329 & 2.685 & 35.19\% \\

OCMT
& 0.711 & 0.999 & 0.250 & 898.602 & 7.954 & 0.93\% \\

GOCMT
& 0.406 & 0.411 & 0.196 & 1.789 & 4.750 & 13.89\% \\

Lasso-1SE
& 0.456 & 0.294 & 0.244 & 1.164 & 6.093 & 0.00\% \\

Lasso-MSE
& 0.776 & 0.361 & 0.505 & 1.512 & 8.861 & 0.00\% \\

PF Lasso-1SE
& 0.444 & 0.339 & 0.253 & 1.639 & 6.324 & 0.00\% \\

PF Lasso-MSE
& 0.836 & 0.362 & 0.512 & 1.653 & 10.065 & 0.00\% \\

Ad. Lasso-1SE
& 0.353 & 0.341 & 0.194 & 1.206 & 3.806 & 12.96\% \\

Ad. Lasso-MSE
& 0.774 & 0.412 & 0.495 & 1.528 & 9.556 & 0.00\% \\

MCP
& 0.339 & 0.179 & 0.219 & 0.943 & 2.824 & 18.52\% \\

SCAD
& 0.347 & 0.186 & 0.218 & 0.889 & 3.083 & 18.52\% \\

\hline
\end{tabular}

\vspace{0.5em}
\par
\begin{minipage}{0.99\linewidth}
\small
\textit{\textbf{Notes}:}
All summary statistics are computed across the 108 simulation
designs for $k=4$.
``Rank'' is the average rank, with rank 1 denoting the
best performing method; for RMSE, lower values receive
better ranks.
``Prop. 1st'' reports the percentage of designs in which
a method attains the lowest RMSE.
\end{minipage}
\end{table}

Finally, Table~\ref{tab:CompTime_K} highlights the computational efficiency of BMT. Its median running time is $1.6$ minutes for 2,000 Monte Carlo replications, compared with $3.0$ minutes for GOCMT and $4.6$ minutes for OCMT. BMT is the
fastest method in $96.3\%$ of designs, with an average rank of $1.037$. Standard Lasso requires approximately $13$ minutes, increasing to $18$ minutes for prefiltered Lasso and $40$ minutes for Adaptive Lasso. MCP and SCAD are substantially more computationally demanding, with median running times of approximately eight hours. These results indicate that BMT combines high selection accuracy with a substantial  computational advantage, particularly relative to nonconvex
penalisation methods.

\begin{table}[th]
\caption{Summary Statistics for Computational Time}
\label{tab:CompTime_K}
\centering
\begin{tabular}{lcccccc}
\hline
& \textbf{Median} & \textbf{IQR} & \textbf{Min} & \textbf{Max}
& \textbf{Rank} & \textbf{Prop. 1st} \\
\hline

BMT
& 1.6 & 1.7 & 0.3 & 5.0 & 1.037 & 96.30\% \\

OCMT
& 4.6 & 7.8 & 0.5 & 122.1 & 3.417 & 3.70\% \\

GOCMT
& 3.0 & 3.5 & 0.7 & 11.7 & 2.380 & 0.00\% \\

Lasso-1SE
& 13.3 & 26.9 & 3.1 & 81.7 & 4.287 & 0.00\% \\

Lasso-MSE
& 13.2 & 26.8 & 3.2 & 81.7 & 4.213 & 0.00\% \\

PF Lasso-1SE
& 17.9 & 38.1 & 4.5 & 111.5 & 6.333 & 0.00\% \\

PF Lasso-MSE
& 17.9 & 38.1 & 4.6 & 111.4 & 6.389 & 0.00\% \\

Ad. Lasso-1SE
& 40.0 & 81.1 & 9.2 & 245.0 & 8.065 & 0.00\% \\

Ad. Lasso-MSE
& 40.2 & 81.9 & 9.7 & 249.9 & 8.880 & 0.00\% \\

MCP
& 505.2 & 1104.0 & 84.9 & 6396.8 & 10.972 & 0.00\% \\

SCAD
& 479.1 & 1045.9 & 79.8 & 6169.3 & 10.028 & 0.00\% \\

\hline
\end{tabular}

\vspace{0.5em}
\par
\begin{minipage}{0.99\linewidth}
\small
\textit{\textbf{Notes}:}
All summary statistics are computed across the 108 simulation
designs for $k=4$.
Computational time is measured in minutes and refers to the
total time required to complete 2,000 Monte Carlo replications.
``Rank'' is the average rank, with rank 1 denoting the fastest
method; lower computational times receive better ranks.
``Prop. 1st'' reports the percentage of designs in which
a method attains the lowest computational time.
\end{minipage}
\end{table}

Figures~\ref{fig:MCC_K4}--\ref{fig:RMSE_K4}
provide a visual comparison of MCC, exact recovery
and RMSE across different sample dimensions and
multicollinearity regimes. The figures show principal component prefiltered
Lasso with the 1SE rule and Adaptive Lasso with the 1SE rule at both stages.
The MCC results show that BMT achieves consistently high selection accuracy across the different configurations, with performance generally improving as $T$ increases. Although MCP and SCAD occasionally attain higher MCC values in smaller samples,
BMT exhibits a clearer advantage as the time dimension increases, with MCC values approaching one even under substantial multicollinearity and pervasive factor dependence. MCP and SCAD also achieve high selection accuracy, while Lasso generally performs less favourably,
particularly when $\pi=1$. GOCMT exhibits greater sensitivity to the source of multicollinearity. When $\pi=0$, its performance is very similar to that of
OCMT, with MCC declining as $T$ increases in some configurations, particularly when $\mathrm{V}=2$. This behaviour is associated with the selection of additional pseudo signals, as reflected in the relatively low true discovery rates reported in the Appendix. When $\pi>0$,
GOCMT generally improves upon OCMT, although its performance remains more variable than that of BMT. The behaviour of OCMT differs markedly across factor configurations. When $\pi=0$, its MCC remains relatively high, particularly for $\mathrm{V}=2$, whereas the introduction of global factor dependence ($\pi>0$) leads
to a substantial deterioration, especially when $\mathrm{V}=4$. This reflects the tendency of OCMT's marginal tests to select correlated pseudo signals alongside the true regressors, with the resulting false positives becoming more prevalent as $T$ increases.

The exact recovery results provide a more stringent assessment of variable selection. BMT's recovery rate generally increases with $T$, reaching high levels across most configurations. MCP and SCAD display similar improvements with $T$, although their recovery rates typically remain below those of BMT. Lasso performs less favourably, with recovery rates that remain substantially lower even as the sample size increases. OCMT achieves very low exact recovery rates across most configurations, including some in which its MCC is relatively high. This illustrates that satisfactory overall classification
performance does not necessarily imply recovery of the exact set of true signals. GOCMT achieves considerably higher recovery rates than OCMT in several designs, but its
performance remains sensitive to the factor configuration and sample dimensions.

The RMSE results provide further insight into the consequences of variable selection for coefficient estimation. Across most configurations, RMSE declines as $T$ increases. MCP and SCAD tend to achieve lower RMSE values than BMT when $T=100$, whereas BMT generally performs more favourably as the time dimension increases, with all three methods achieving relatively low estimation errors in larger samples. GOCMT and Lasso also exhibit declining estimation errors, although their relative performance varies across designs. In contrast, OCMT displays considerable instability when $\pi>0$, particularly under $\mathrm{V}=4$. In these configurations, RMSE can increase
dramatically, with occasional extreme values that are several orders of magnitude larger than those of the alternative procedures. Such behaviour is consistent with the inclusion of highly correlated pseudo signals and selected models close to the sample size, which can make post selection coefficient estimates unstable.

The RMSFE results in the Appendix show similar forecast errors for BMT, Adaptive Lasso-1SE, MCP and SCAD, with their ordering varying across designs. Better variable selection therefore does not imply a uniform forecasting advantage.

Detailed numerical results for MCC, exact recovery, RMSE, model size and the remaining evaluation metrics (F1, TDR, FDR, TPR, FPR and RMSFE) are reported in the Appendix. To conserve space, we report detailed results only for prefiltered Lasso and Adaptive Lasso under the 1SE tuning rule, as these  specifications generally outperform their counterparts under the minimum MSE rule. The results show that BMT achieves high variable selection and estimation accuracy across a wide range
of dimensions and dependence structures, including settings characterised by strong factor induced dependence. BMT is also straightforward to implement and computationally efficient, offering substantial savings in computational time relative to the alternative methods.

    \begin{landscape}
    \thispagestyle{empty} 
    \centering
    \begin{minipage}{0.95\linewidth}
      \vspace{-0.5em}
    \begin{minipage}{0.48\linewidth}
    \centering
    \vspace{-0.4em}
    V = 4, $\pi=0$
    \includegraphics[width=\linewidth,height=0.28\textheight,keepaspectratio]
{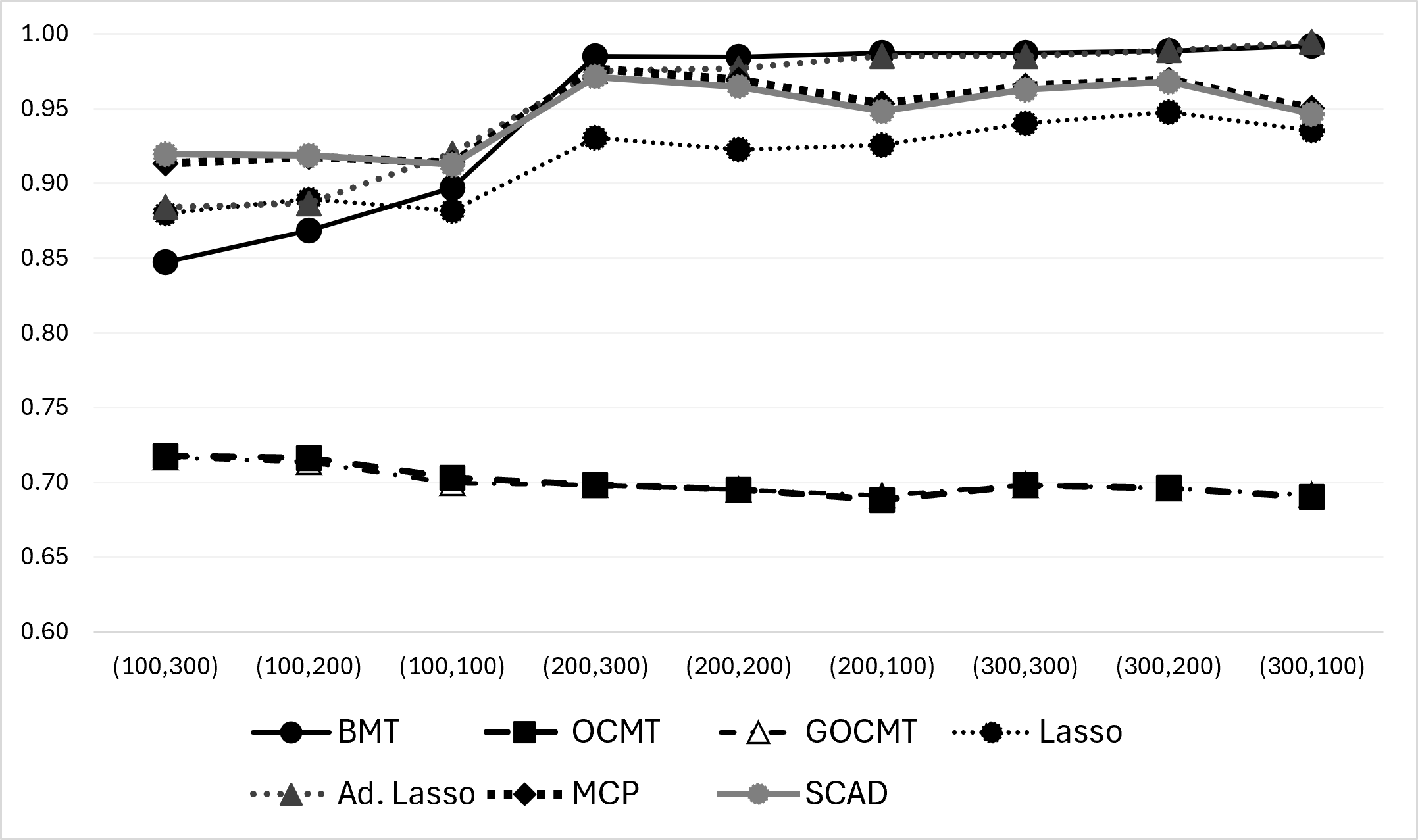}\\
    \end{minipage}
    \hfill
    \begin{minipage}{0.48\linewidth}
    \centering
    \vspace{-0.4em}
    V = 2, $\pi=0$
    \includegraphics[height=0.28\textheight]{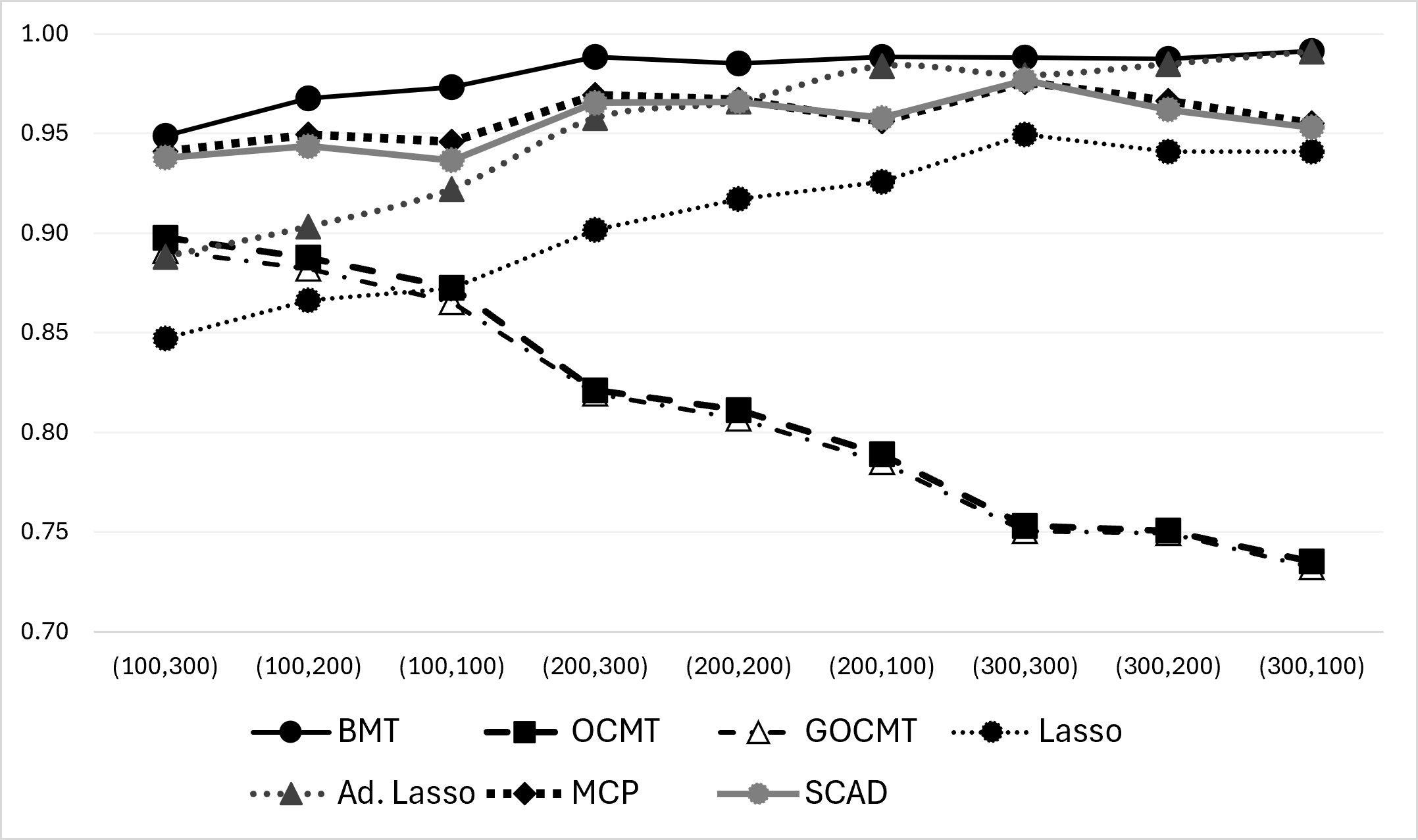}\\
    \end{minipage}
    \begin{minipage}{0.48\linewidth}
    \centering
    \vspace{0.4em}
    V = 4, $\pi=0.5$
    \includegraphics[height=0.28\textheight]{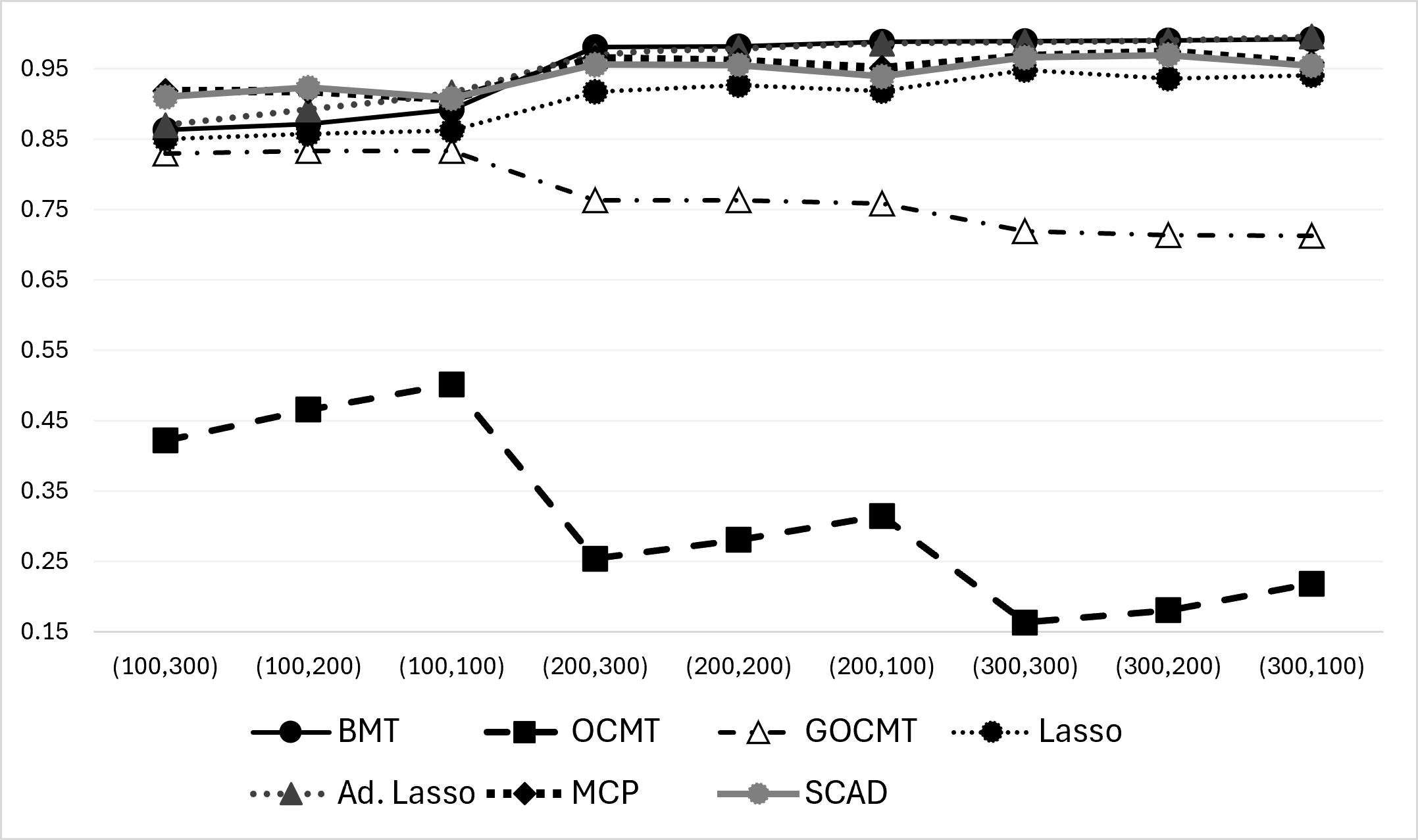}\\
    \end{minipage}
    \hfill
    \begin{minipage}{0.48\linewidth}
    \centering
    \vspace{0.4em}
    V = 2, $\pi=0.5$
    \includegraphics[height=0.28\textheight]{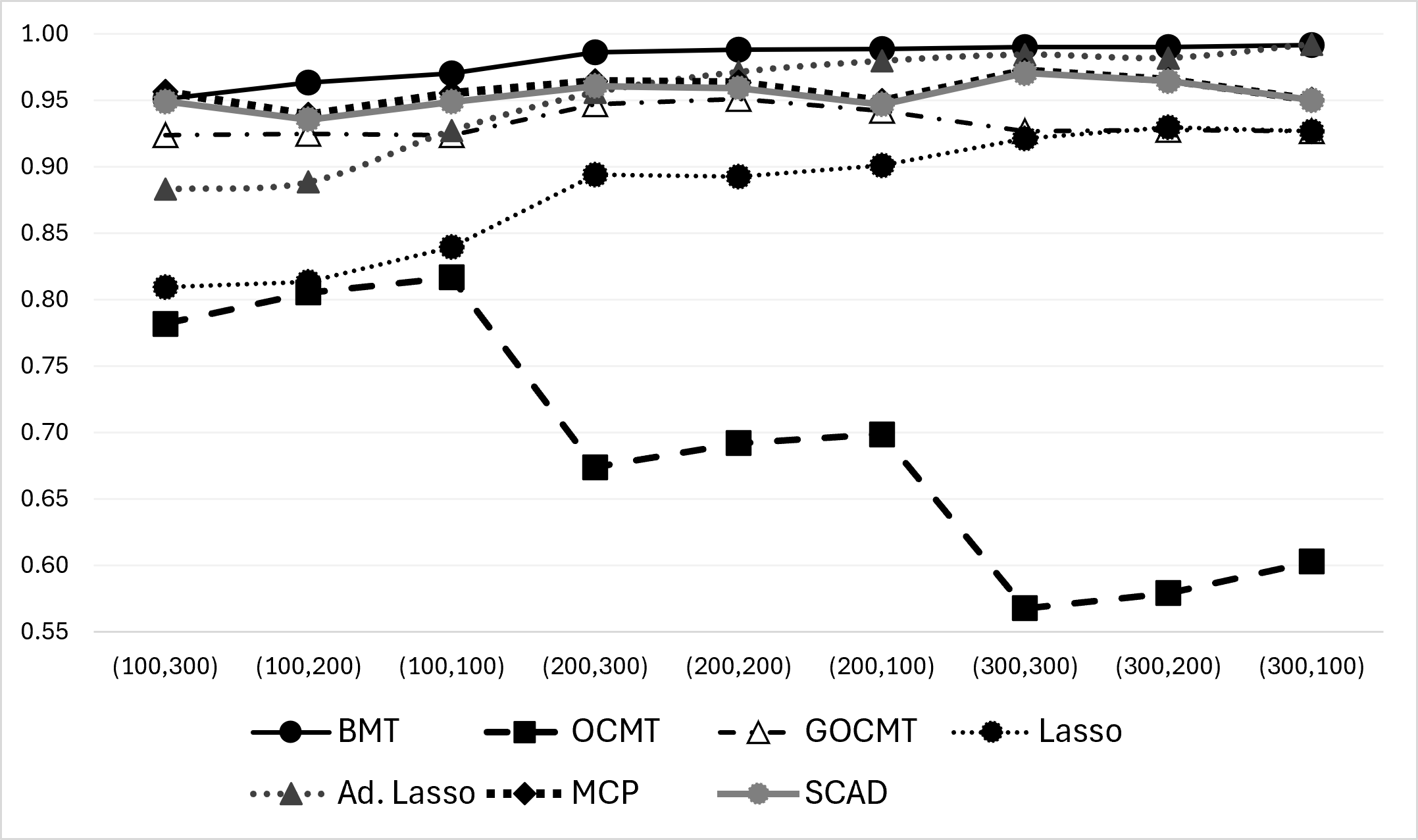}\\
    \end{minipage}
    \begin{minipage}{0.48\linewidth}
    \centering
    \vspace{0.2em}
    V = 4, $\pi=1$
    \includegraphics[height=0.28\textheight]{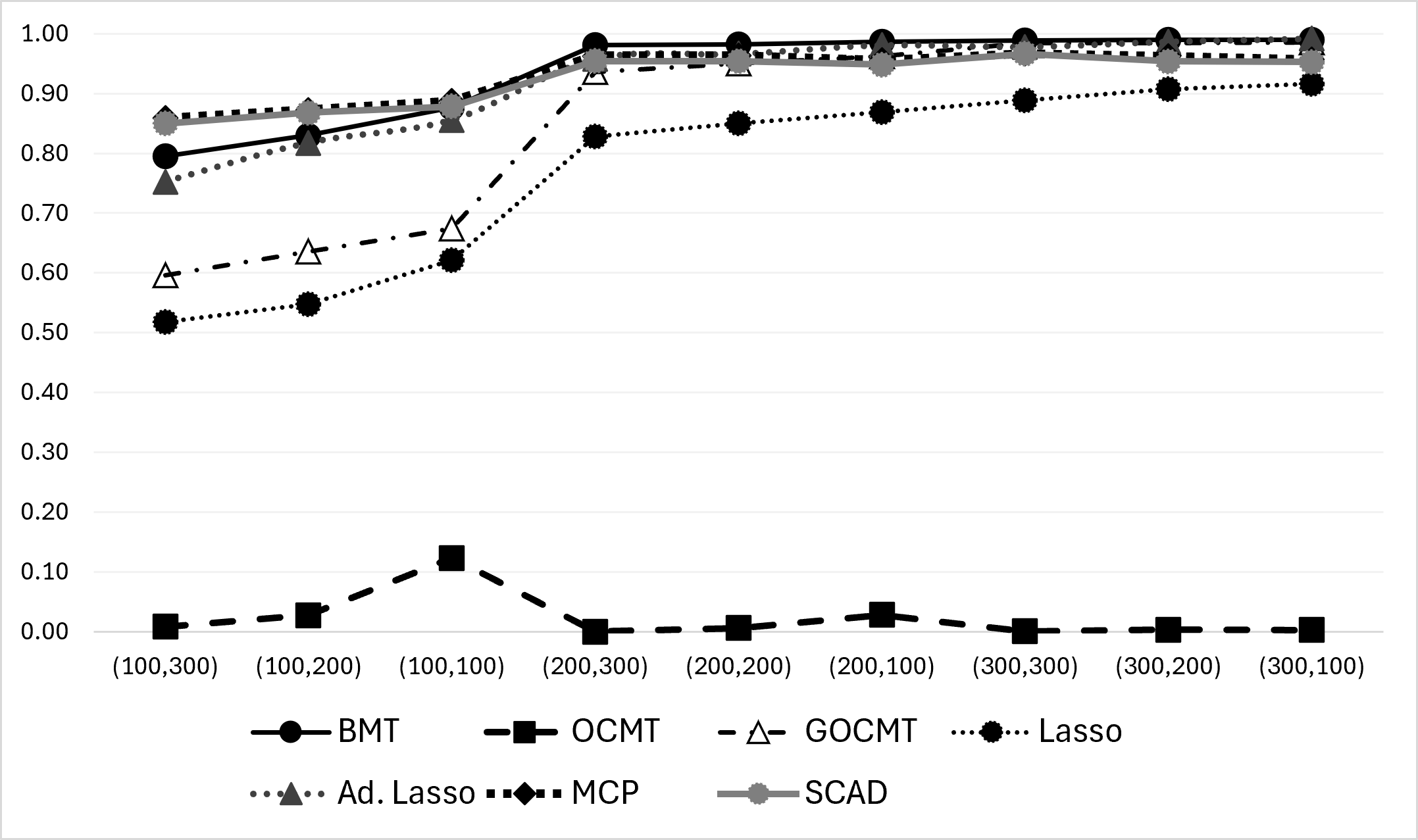}\\
    \end{minipage}
    \hfill
    \begin{minipage}{0.48\linewidth}
    \centering
    \vspace{0.2em}
    V = 2, $\pi=1$
    \includegraphics[height=0.28\textheight]{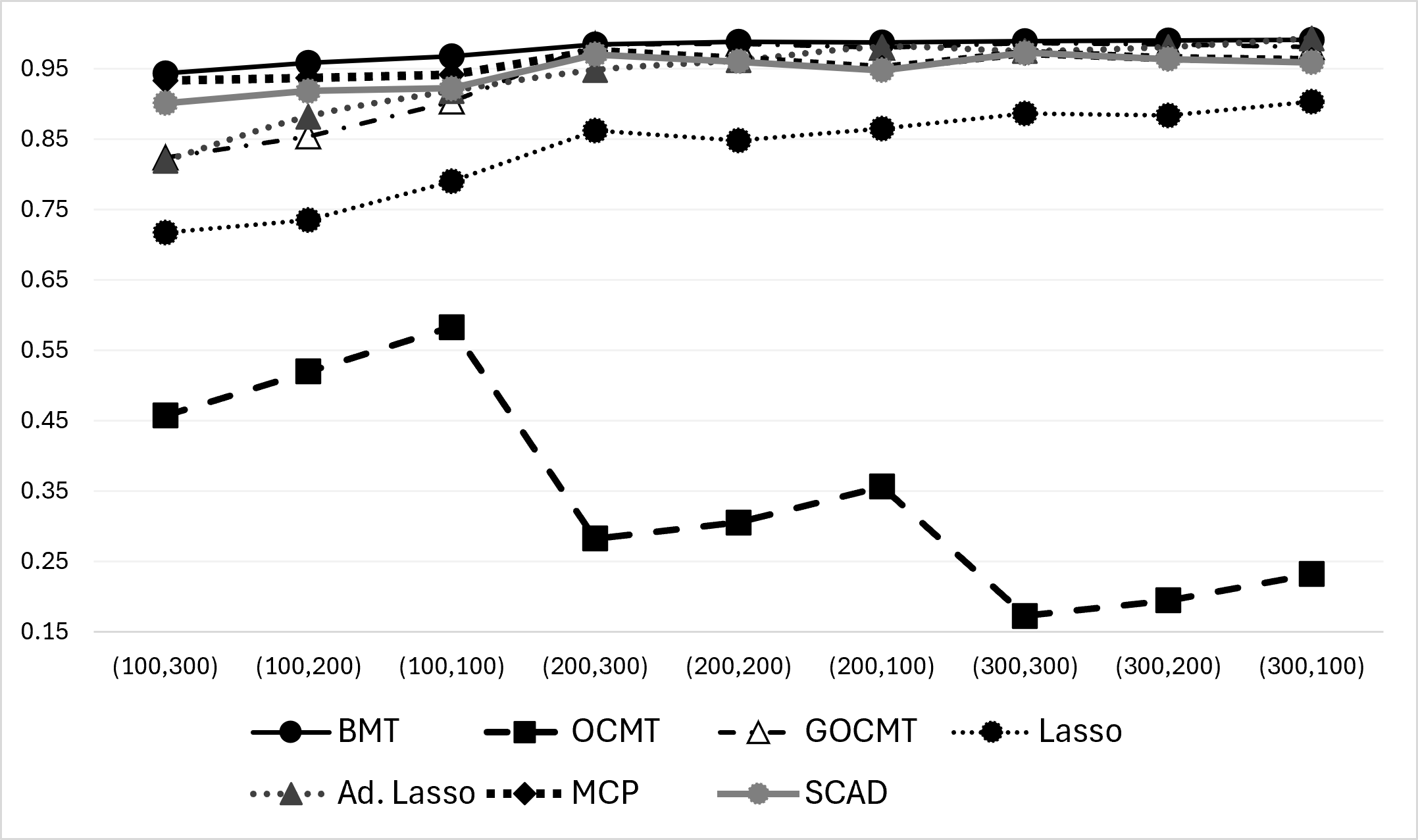}\\
    \end{minipage}
    \vspace{-0.4em}
    \captionof{figure}{Visual Summary of MCC Performance for $\alpha=0.8$}
    \label{fig:MCC_K4}
    \end{minipage}
    \end{landscape}

 \begin{landscape}
    \thispagestyle{empty} 
    \centering
    \begin{minipage}{0.95\linewidth}
      \vspace{-0.5em}
    \begin{minipage}{0.48\linewidth}
    \centering
    \vspace{-0.4em}
    V = 4, $\pi=0$
    \includegraphics[width=\linewidth,height=0.28\textheight,keepaspectratio]
{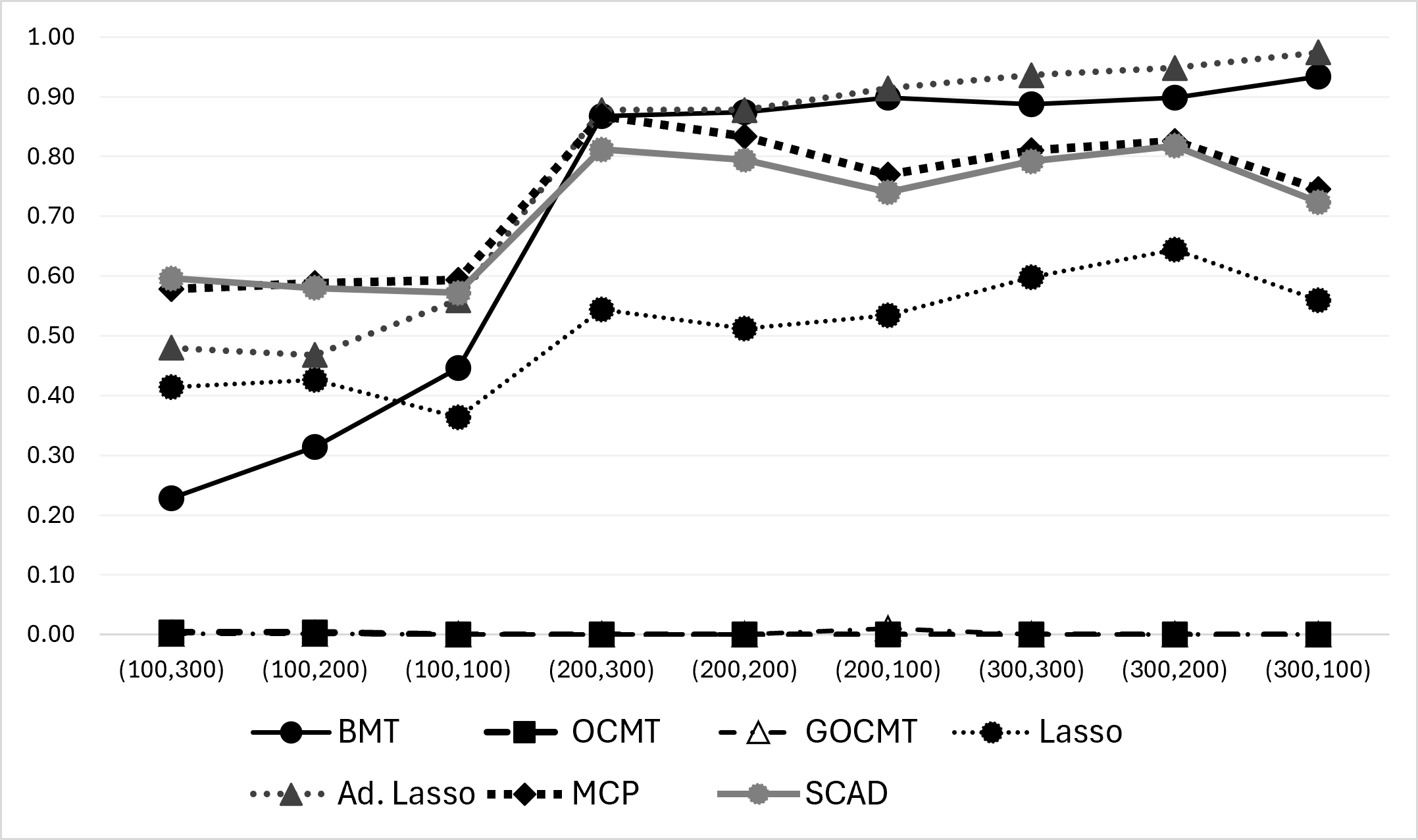}\\
    \end{minipage}
    \hfill
    \begin{minipage}{0.48\linewidth}
    \centering
    \vspace{-0.4em}
    V = 2, $\pi=0$
    \includegraphics[height=0.28\textheight]{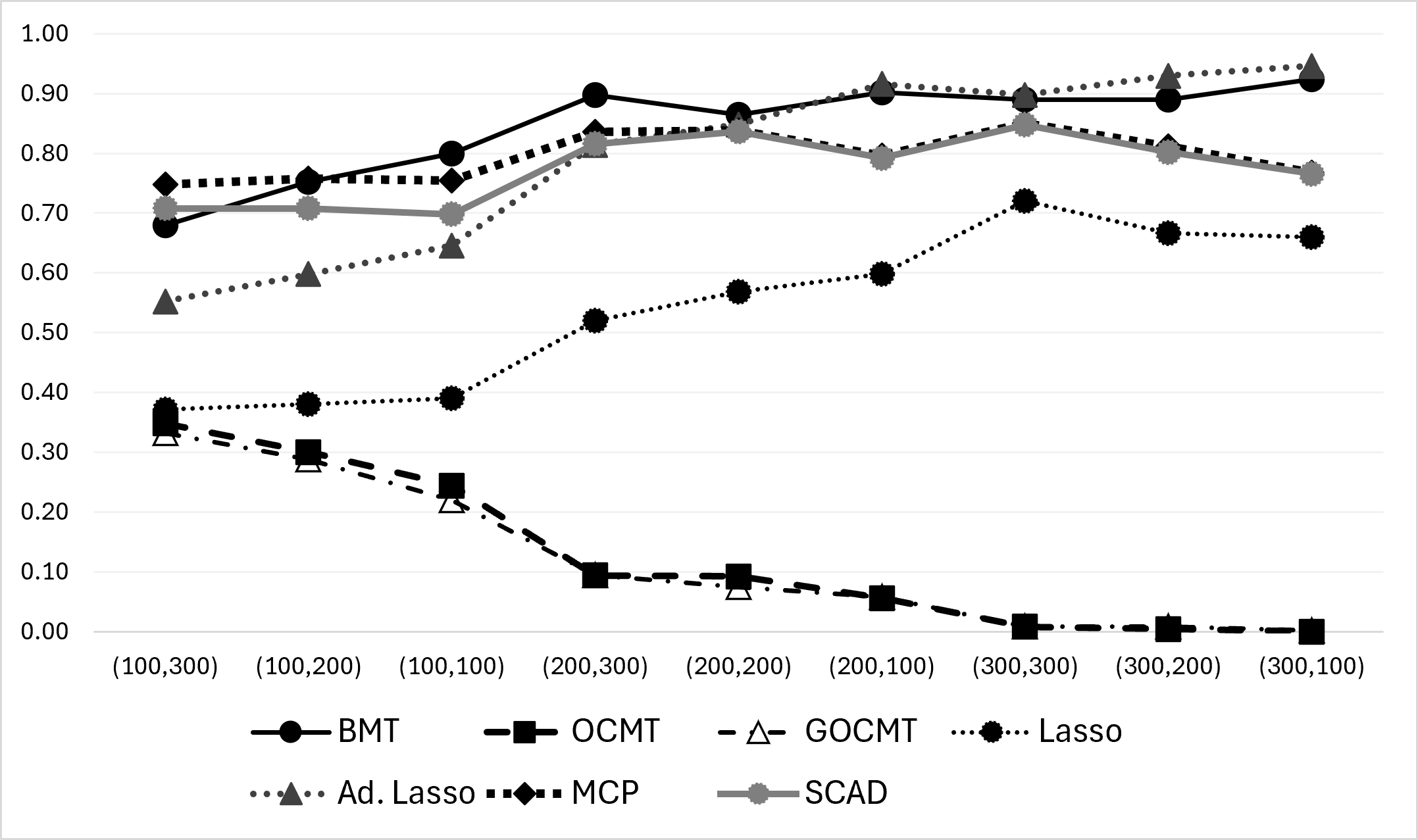}\\
    \end{minipage}
    \begin{minipage}{0.48\linewidth}
    \centering
    \vspace{0.4em}
    V = 4, $\pi=0.5$
    \includegraphics[height=0.28\textheight]{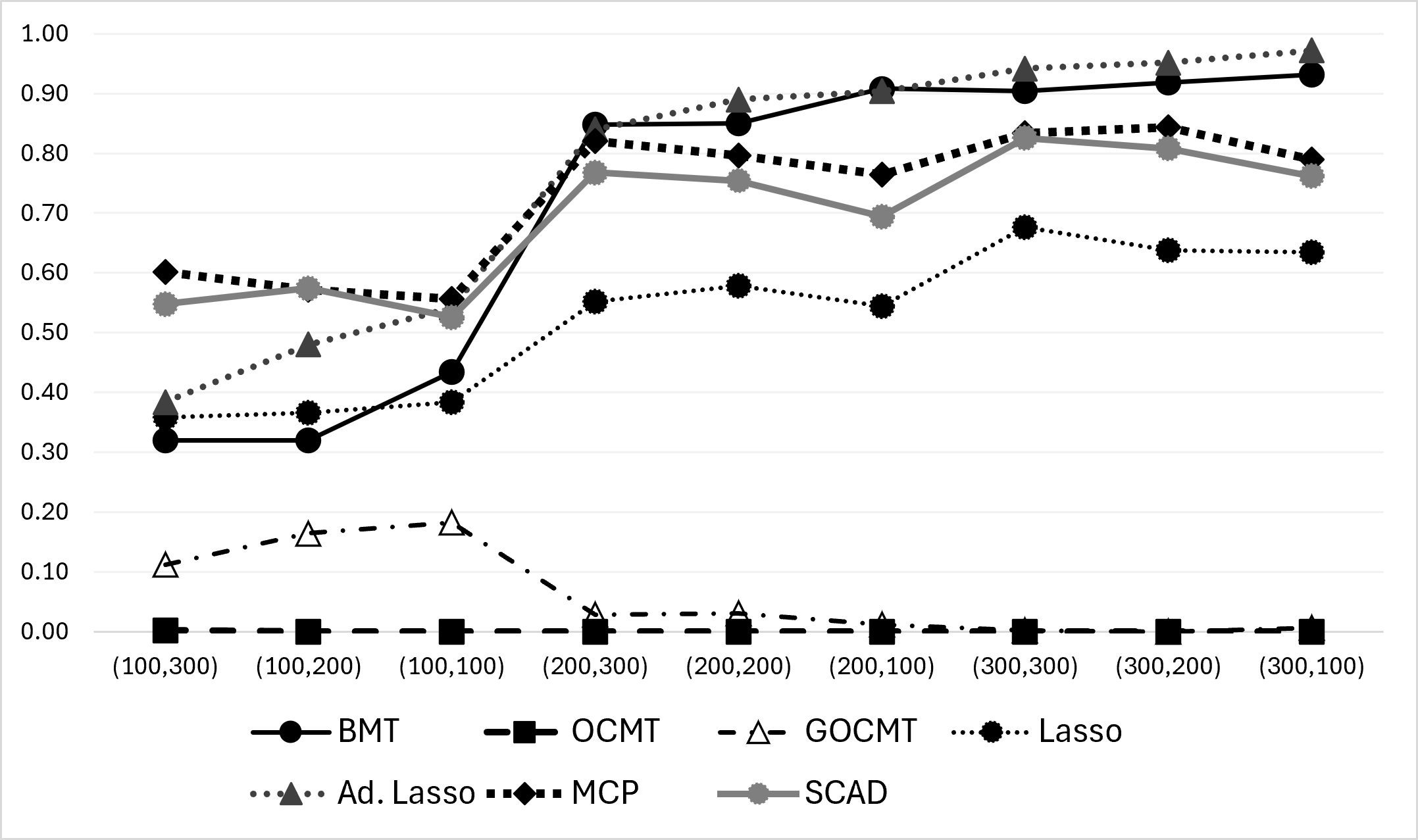}\\
    \end{minipage}
    \hfill
    \begin{minipage}{0.48\linewidth}
    \centering
    \vspace{0.4em}
    V = 2, $\pi=0.5$
    \includegraphics[height=0.28\textheight]{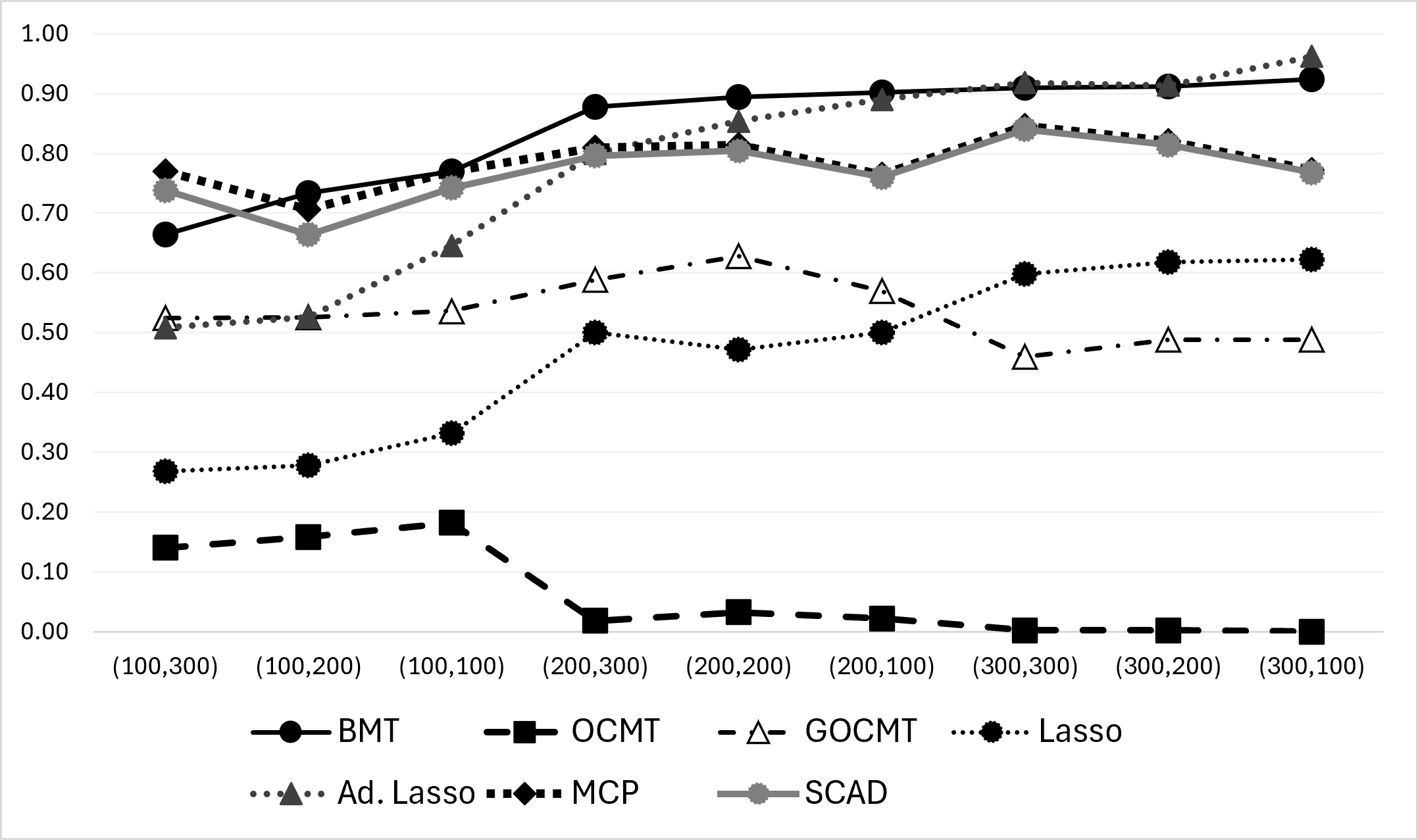}\\
    \end{minipage}
    \begin{minipage}{0.48\linewidth}
    \centering
    \vspace{0.2em}
    V = 4, $\pi=1$
    \includegraphics[height=0.28\textheight]{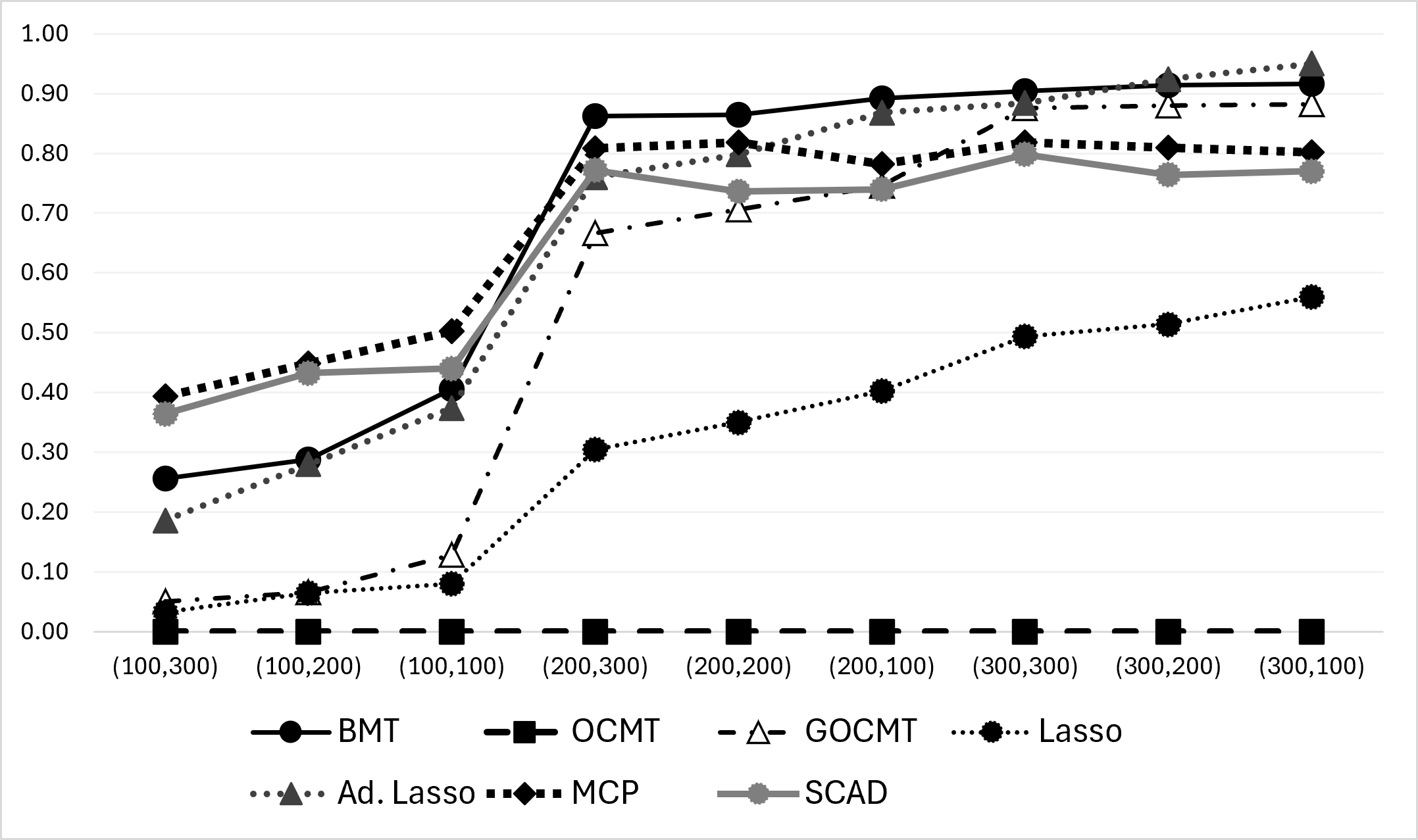}\\
    \end{minipage}
    \hfill
    \begin{minipage}{0.48\linewidth}
    \centering
    \vspace{0.2em}
    V = 2, $\pi=1$
    \includegraphics[height=0.28\textheight]{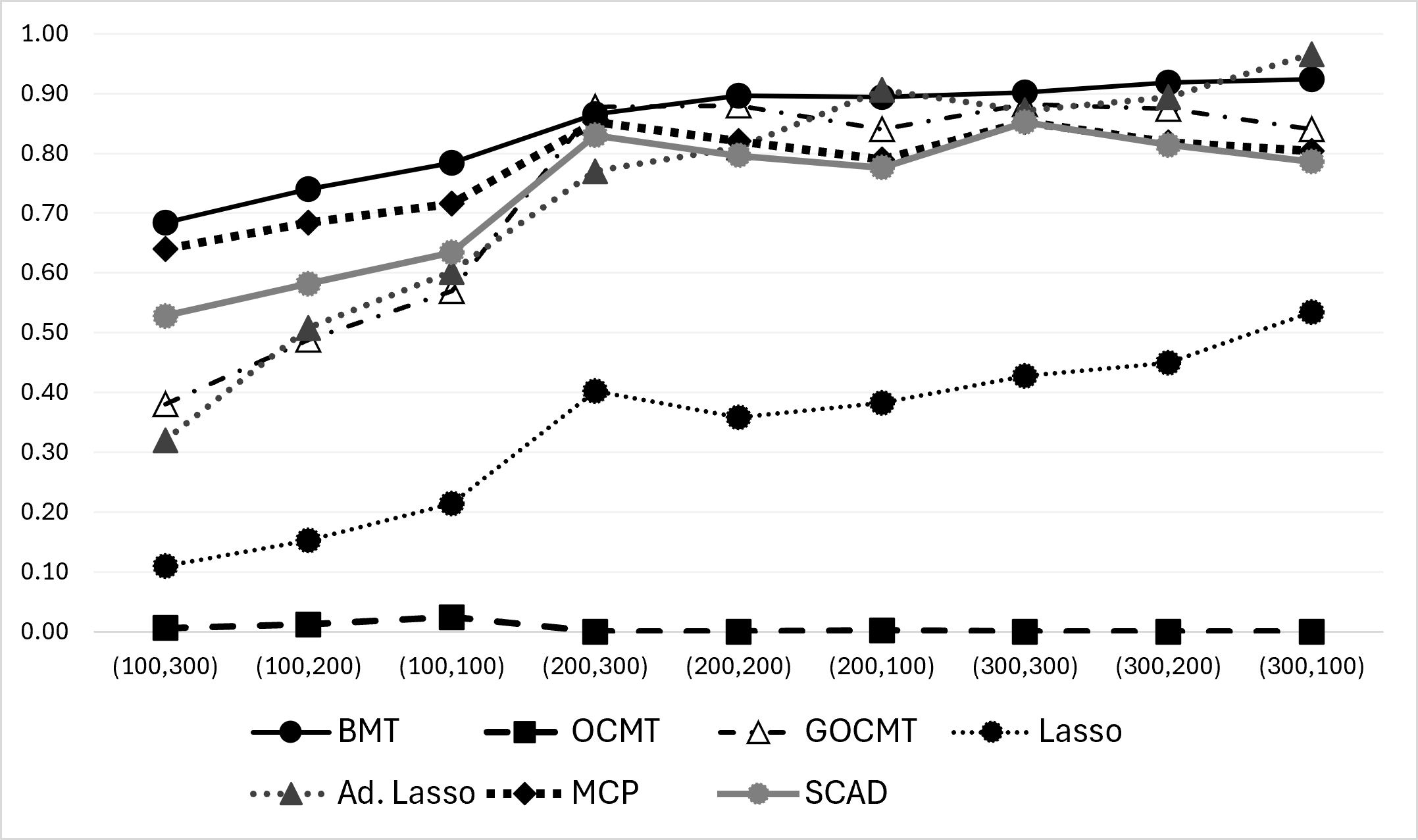}\\
    \end{minipage}
    \vspace{-0.4em}
    \captionof{figure}{Visual Summary of Performance on Exact Recovery for $\alpha=0.8$}
    \label{fig:ExRec_K4}
    \end{minipage}
    \end{landscape}

 \begin{landscape}
    \thispagestyle{empty} 
    \centering
    \begin{minipage}{0.95\linewidth}
      \vspace{-0.5em}
    \begin{minipage}{0.48\linewidth}
    \centering
    \vspace{-0.4em}
    V = 4, $\pi=0$
    \includegraphics[width=\linewidth,height=0.28\textheight,keepaspectratio]
{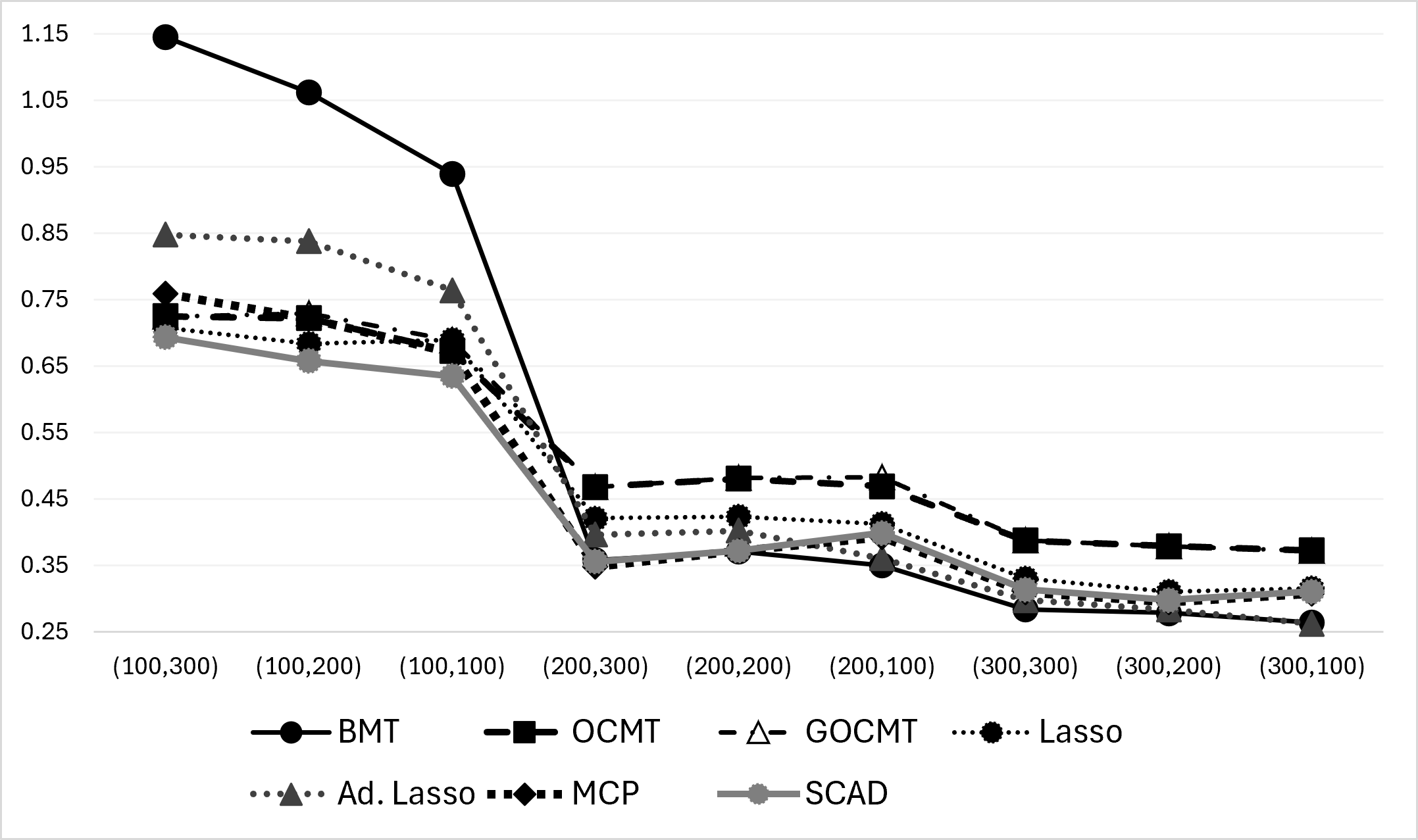}\\
    \end{minipage}
    \hfill
    \begin{minipage}{0.48\linewidth}
    \centering
    \vspace{-0.4em}
    V = 2, $\pi=0$
    \includegraphics[height=0.28\textheight]{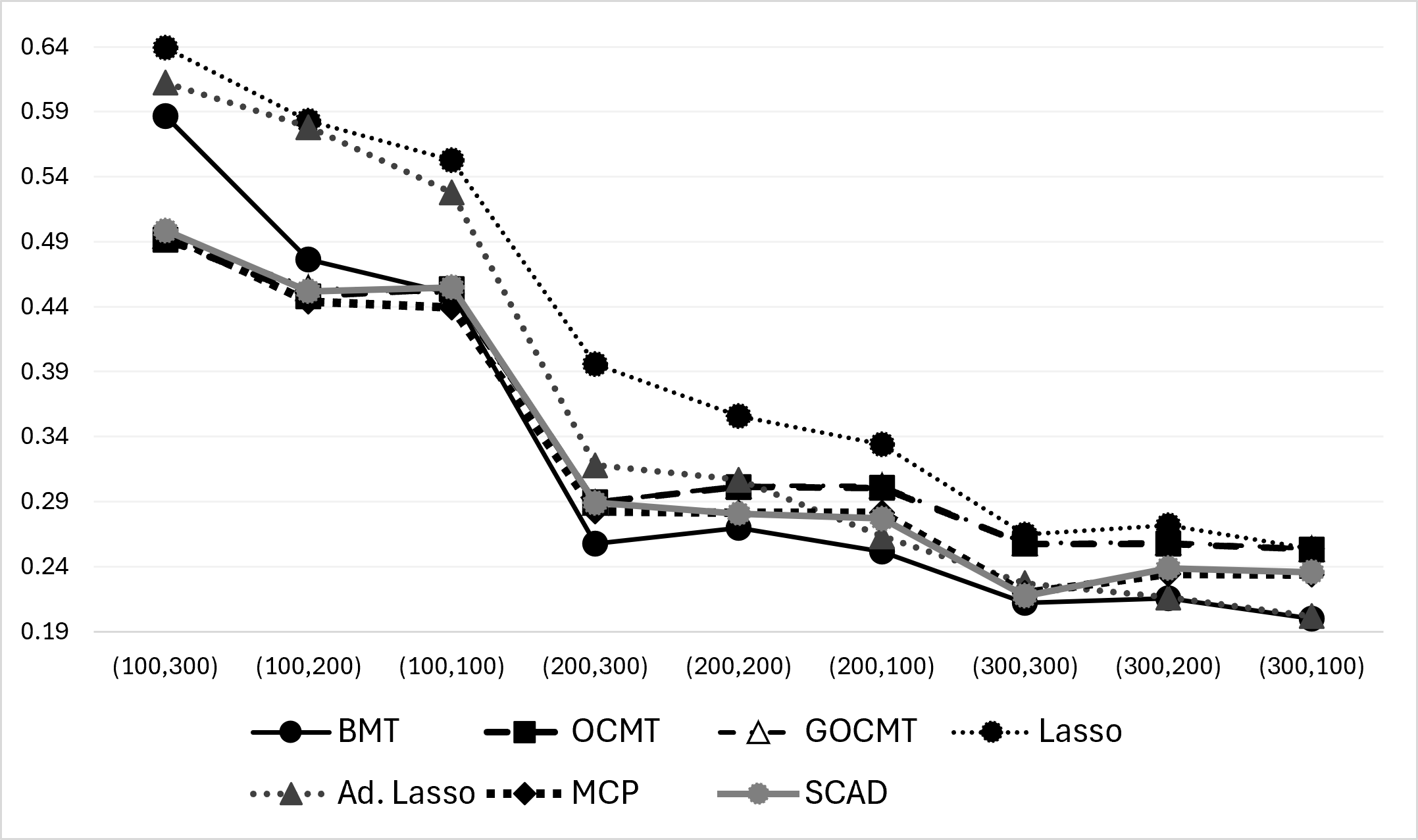}\\
    \end{minipage}
    \begin{minipage}{0.48\linewidth}
    \centering
    \vspace{0.4em}
    V = 4, $\pi=0.5$
    \includegraphics[height=0.28\textheight]{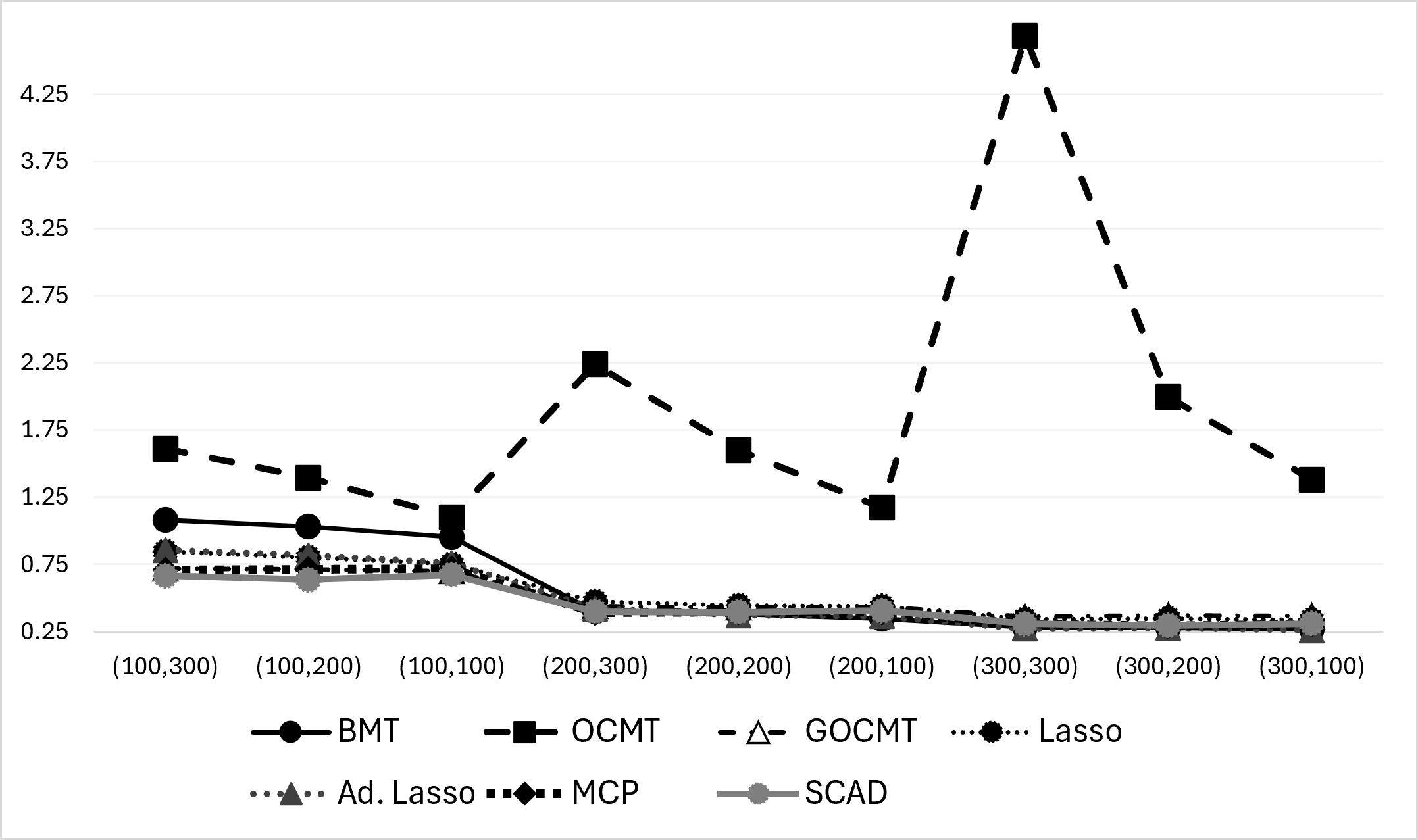}\\
    \end{minipage}
    \hfill
    \begin{minipage}{0.48\linewidth}
    \centering
    \vspace{0.4em}
    V = 2, $\pi=0.5$
    \includegraphics[height=0.28\textheight]{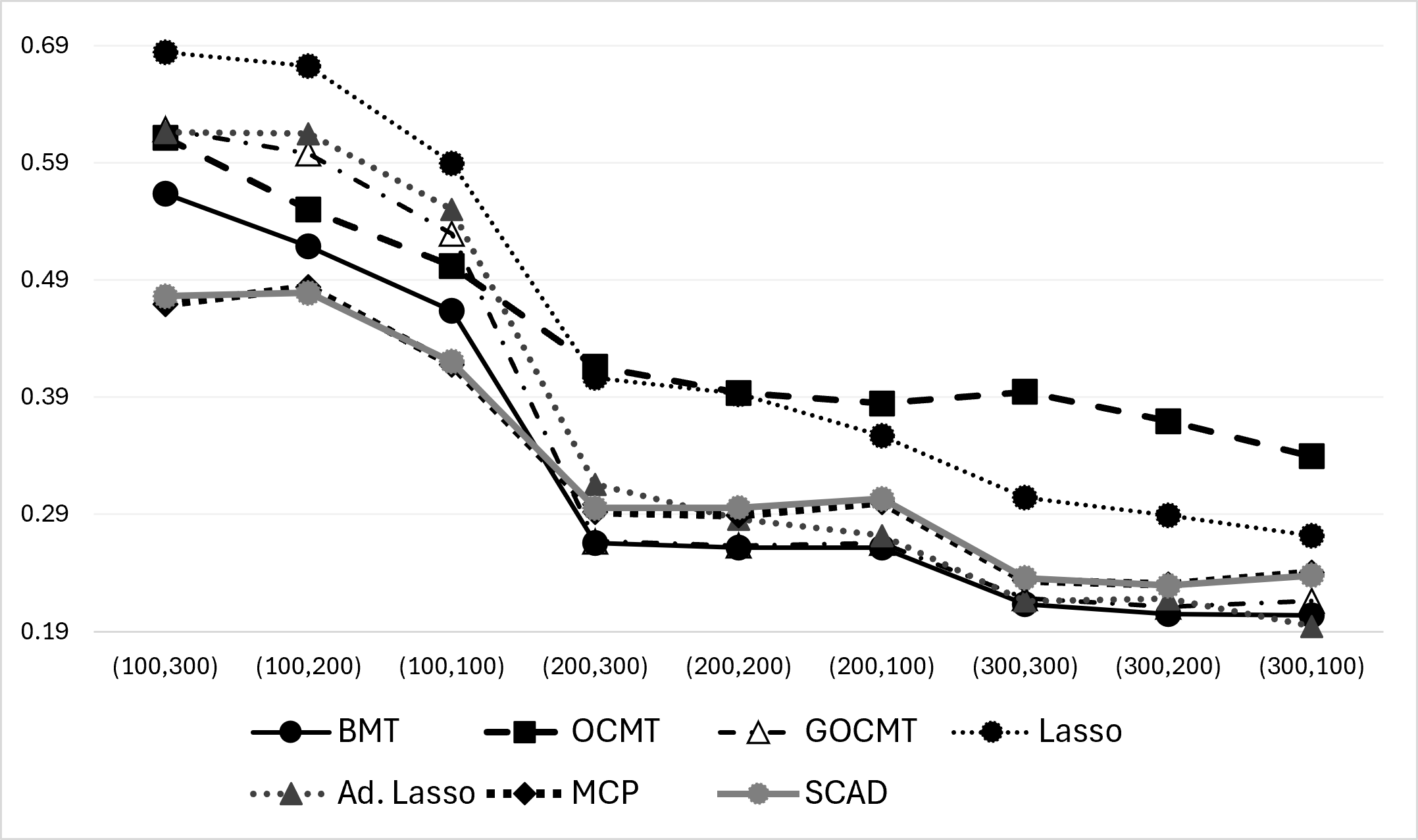}\\
    \end{minipage}
    \begin{minipage}{0.48\linewidth}
    \centering
    \vspace{0.2em}
    V = 4, $\pi=1$
    \includegraphics[height=0.28\textheight]{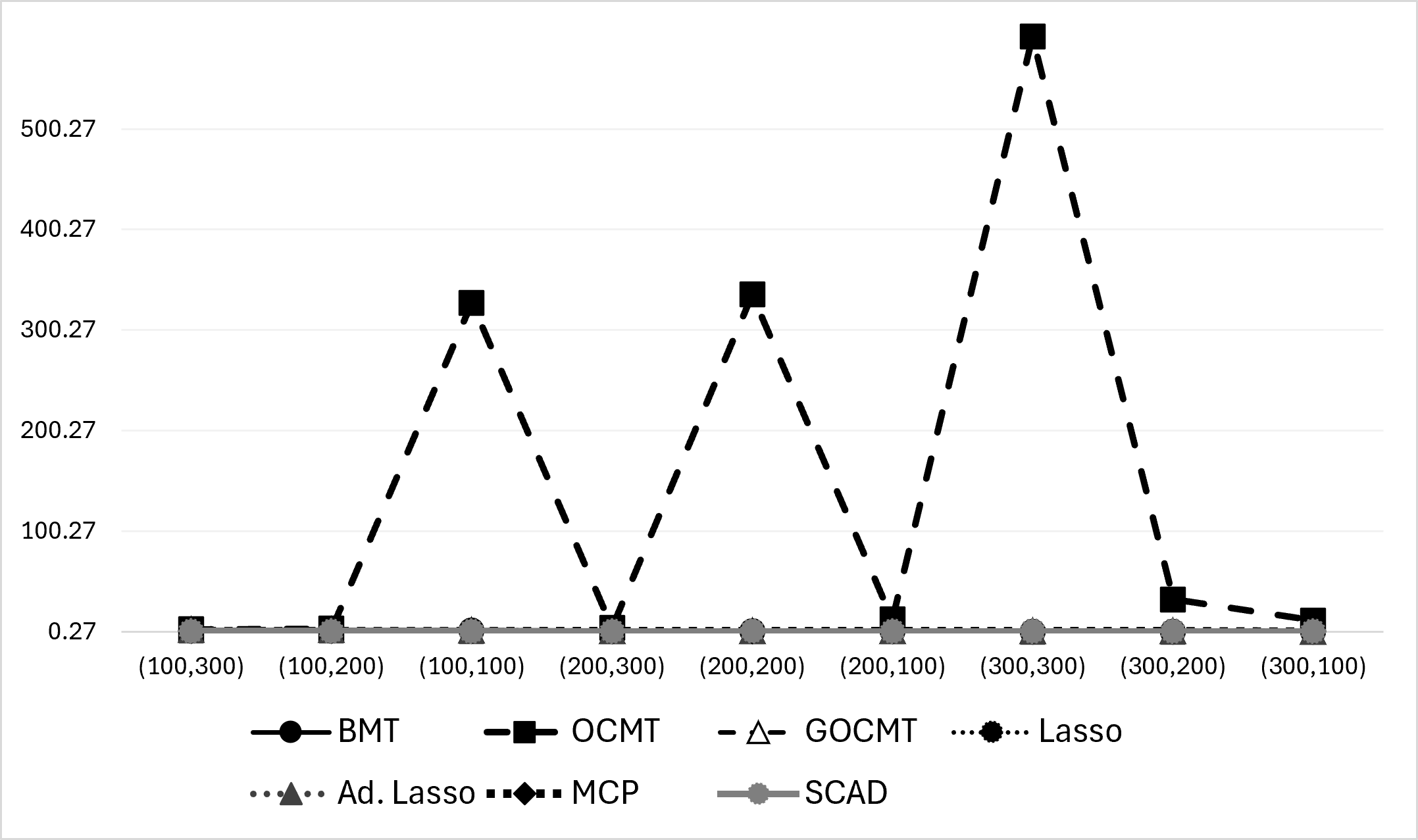}\\
    \end{minipage}
    \hfill
    \begin{minipage}{0.48\linewidth}
    \centering
    \vspace{0.2em}
    V = 2, $\pi=1$
    \includegraphics[height=0.28\textheight]{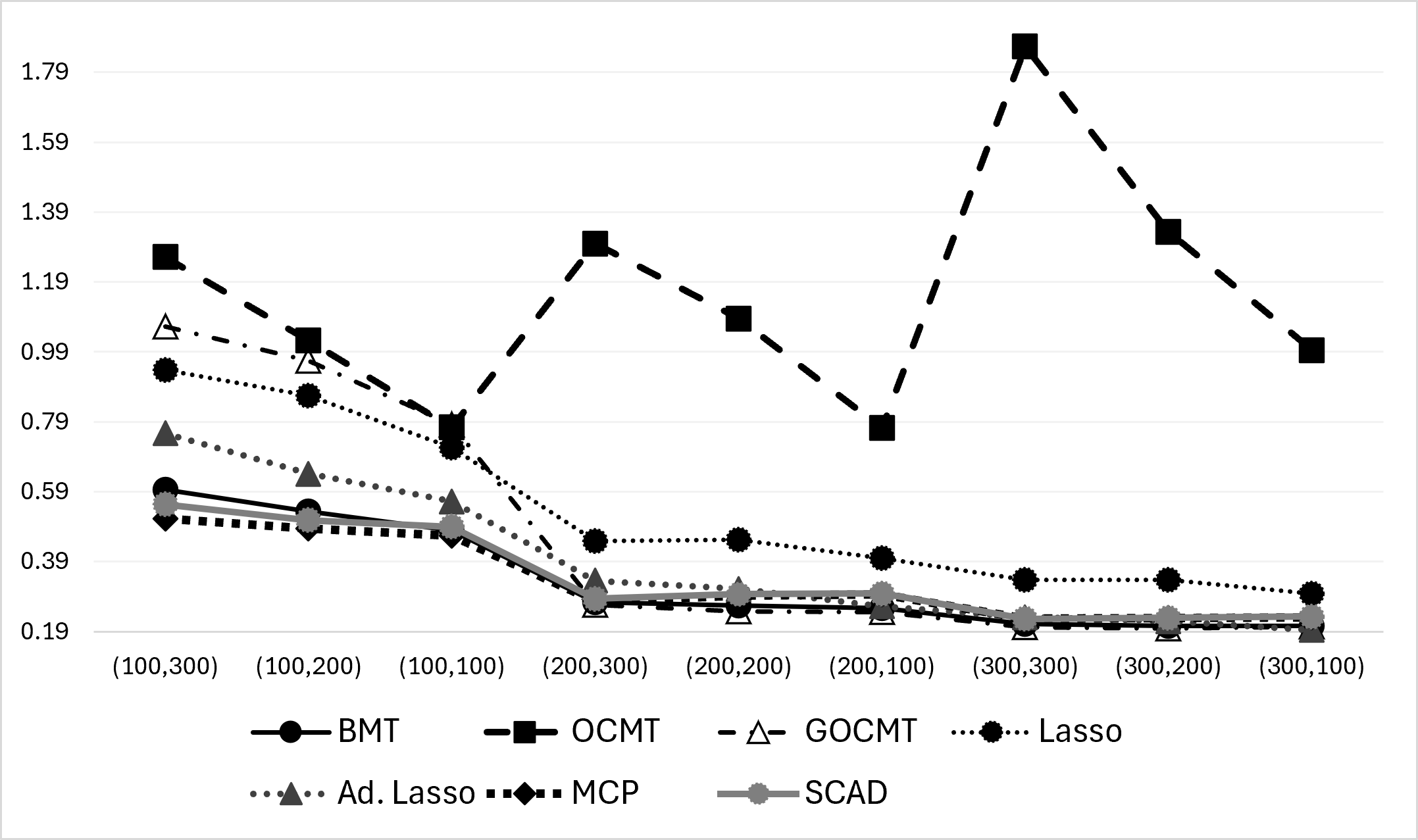}\\
    \end{minipage}
    \vspace{-0.4em}
    \captionof{figure}{Visual Summary of RMSE Performance for $\alpha=0.8$}
    \label{fig:RMSE_K4}
    \end{minipage}
    \end{landscape}

\section{Empirical Illustrations}
\label{empirics}

In this section we present two empirical illustrations that highlight the potential of BMT as a model selection procedure.

\subsection{Corporate Emissions}

Climate change and pollution sit at the center of the global policy agenda, not only because they shape long run environmental outcomes, but because they map directly into quality of life metrics such as air and water safety, morbidity risks, productivity, and property values. Regulators, investors, and communities increasingly need credible, data driven evidence on what drives emissions in order to design effective penalties, target enforcement, and allocate capital responsibly.

In what follows, we analyse annual emissions and model them as a function of (i) penalties imposed on emitting facilities and (ii) indicators of economic activity and broader macro financial conditions. While overall economic activity is likely correlated with aggregate emissions, it is far less clear which specific drivers, such as sectors, energy intensity, trade exposure, investment cycles, capacity utilization, or other factors, are most influential.

Methodologically, this is a high dimensional problem. The dependent variable is aggregate annual emissions for the US economy (total releases, in million pounds) observed over 32 years, while the pool of candidate economic drivers runs into the hundreds, so that $n \gg T$. These covariates include preselected variables, such as penalties imposed on emitting facilities, as well as a large set of indicators of economic activity and broader macro financial conditions. In this illustration, we use our model selection procedure to isolate a sparse, interpretable set of predictors and to examine whether enforcement penalties relate to aggregate emissions once we control for the relevant components of economic activity.

Our environmental data are drawn from the U.S. Environmental Protection Agency’s (EPA) Toxics Release Inventory (TRI), a public database that records facility level releases of listed toxic chemicals. Facilities in TRI covered sectors that exceed reporting thresholds are required to report annual releases by medium (air, water and land). Although the data are self reported, TRI applies internal quality assurance and quality control (QA/QC) checks and provides rich facility level attributes (e.g.\ location and industry classification). For our purposes, we use TRI to construct the dependent variable. Our outcome of interest is aggregate annual air emissions for the US economy, defined as the sum of reported releases to air across all facilities, measured in pounds (lb), which yields a national time series over the sample period.

Penalty measures for emitting facilities are derived from EPA enforcement data, with the primary measure being Total Penalty Assessed (the dollar amount assessed in all settlements for concluded enforcement actions). This measure decomposes into penalties from Administrative-Formal (AFR) and Judicial (JDC) actions.\footnote{These penalty fields are defined in EPA’s ECHO/ICIS data dictionaries and reflect assessed amounts (not necessarily amounts collected). See the EPA ECHO data dictionary at \url{https://echo.epa.gov/help/reports/dfr-data-dictionary?check_logged_in=1}.}
 In what follows, we use Total Penalty Assessed and treat it as the aggregate (AFR + JDC).

Indices of economic activity and broader macro financial conditions are obtained from the FRED-QD (Federal Reserve Economic Data - Quarterly Database) \newline introduced by \mbox{\cite{McCrackenNg2021}}. This publicly available dataset is updated monthly through the FRED database. In what follows, the 2025-03 vintage is used, which covers data through 2024Q4. The dataset comprises 245 U.S. macroeconomic time series, most of which begin in 1959Q1.
The dataset follows the structure of \cite{stock_disentangling_2012}, but with broader coverage. The variables are classified into 14 thematic blocks: (i) NIPA output and income aggregates, (ii) industrial production and capacity utilization, (iii) labor market indicators, (iv) housing activity, (v) inventories, orders, and sales, (vi) price level and inflation measures, (vii) earnings and productivity, (viii) interest rates and term spreads, (ix) monetary and credit aggregates, (x) household balance sheet ratios, (xi) exchange rates, (xii) survey based and sentiment indicators, (xiii) stock market variables, and (xiv) nonhousehold (corporate and government) balance sheet items. Full details of these variables are provided in \cite{McCrackenNg2021}.

In addition to the variables described above, the regression specification includes a lagged dependent variable to capture persistence, a deterministic time trend to account for the secular decline in emissions, and the first principal component from the FRED-QD database, extracted from the full set of economic and macro financial indicators. This component summarises common variation across these series and helps mitigate collinearity among the regressors.

The model is specified as in Eq.~\eqref{dgp1}:
\begin{equation}\label{eq:dgp_application}
y_{t}=\mathbf{a}^{\prime}\boldsymbol{z}_{t}+\sum_{i=1}^{k}
\beta_{i}x_{it}+u_{t}, \quad  t=1,\dots,T \ (T=31),
\end{equation}
where $y_{t}$ denotes the natural logarithm of aggregate U.S.\ air emissions in year $t$. The vector of preselected controls $\boldsymbol{z}_{t}=\left(z_{1,t}, z_{2,t}, z_{3,t}, z_{4,t} \right)^{\prime}$ is $4\times 1$ and consists of: the lagged dependent variable $y_{t-1}$, the natural logarithm of lagged Total Penalty, a deterministic time trend $t$, and the first principal component of the FRED-QD macro financial indicators. The remaining candidate regressors $\{x_{it}\}_{i=1}^{n}$ are the $n=245$ individual FRED-QD indicators, with $n\gg T$, and we denote by $\mathcal{S}_{s}\subset\{1,\dots,n\}$ the (unknown) set of true signals.

We compare BMT, OCMT, GOCMT, Lasso with PCA prefiltering (Lasso-PF), Adaptive Lasso, MCP and SCAD. For Lasso-PF and Adaptive Lasso, we report both the one standard error (1SE) and minimum MSE cross validation rules.

Table \ref{tab:air1} below presents results for the benchmark specification in Eq. \eqref{eq:dgp_application}.

\begin{table}[t]
\caption{Air Emissions}
\label{tab:air1}
\centering
\resizebox{\textwidth}{!}{%
\begin{tabular}{lccccccccc}
\hline
& BMT & OCMT & GOCMT
& \multicolumn{2}{c}{Lasso-PF}
& \multicolumn{2}{c}{Adaptive Lasso}
& MCP & SCAD \\
& & & & 1SE & MSE & 1SE & MSE & & \\
\hline

$z_{1}$ (lagged emissions)
& 0.663*** & 0.790*** & 0.540**
& 0.636*** & 0.629***
& 0.636*** & 0.635***
& 0.636*** & 0.638*** \\
& (0.167) & (0.176) & (0.268)
& (0.185) & (0.173)
& (0.185) & (0.158)
& (0.185) & (0.181) \\

$z_{2}$ (lagged penalty)
& 0.003 & 0.002 & 0.003
& 0.003 & 0.003
& 0.003 & 0.003
& 0.003 & 0.004** \\
& (0.002) & (0.003) & (0.003)
& (0.002) & (0.002)
& (0.002) & (0.002)
& (0.002) & (0.002) \\

$z_{3}$ (trend)
& -0.005 & -0.002 & -0.005
& -0.005 & -0.002
& -0.005 & -0.004
& -0.005 & -0.005* \\
& (0.003) & (0.003) & (0.004)
& (0.004) & (0.004)
& (0.004) & (0.004)
& (0.004) & (0.003) \\

$z_{4}$ (first pc)
& 0.001 & -0.002 & -0.000
& 0.002 & 0.002
& 0.002 & 0.004***
& 0.002 & -0.001 \\
& (0.001) & (0.003) & (0.002)
& (0.001) & (0.001)
& (0.001) & (0.001)
& (0.001) & (0.001) \\

\textit{permitmw}
& 0.238*** & 2.656** & 0.331
& & & & & & 0.093 \\
& (0.054) & (1.055) & (0.237)
& & & & & & (0.217) \\

\textit{prfix}
& & 0.078 & & & & & & & \\
& & (0.398) & & & & & & & \\

\textit{uslah}
& & 0.014 & & & & & & & \\
& & (0.259) & & & & & & & \\

\textit{houst}
& & 6.405* & 0.046
& & & & & & \\
& & (3.627) & (0.551)
& & & & & & \\

\textit{permit}
& & -9.288** & -1.636
& & & & & & \\
& & (4.064) & (1.143)
& & & & & & \\

\textit{houstmw}
& & -2.185** & &
& & & & &  \\
& & (0.949) & &
& & & & &  \\

\textit{housts}
& & -2.969 & &
& & & & &  \\
& & (2.192) & &
& & & & &  \\

\textit{houstw}
& & -1.650 & -0.104
& & & & & & \\
& & (1.053) & (0.411)
& & & & & & \\

\textit{revolslx}
& & & 0.427
& & & & & & \\
& & & (0.441)
& & & & & & \\

\textit{permits}
& & 4.744** & 0.890
& & & & & & \\
& & (2.407) & (0.688)
& & & & & & \\

\textit{permitw}
& & 2.437** & 0.660
& & & & & & 0.084 \\
& & (1.217) & (0.543)
& & & & & & (0.179) \\

\textit{gfdebtnx}
& & & & & &
& $.043$ & & \\
& & & & & &
& ($0.19$) & & \\

\textit{tlbsnnbbdix}
& & & & &
$-.533$**
& & $-0.348$ & & \\
& & & & &
($0.239$)
& & ($0.256$) & & \\

\textit{lns13023569}
& & & & & & & & & -0.227 \\
& & & & & & & & & (0.147) \\

\hline
$T$
& 31 & 31 & 31 & 31 & 31 & 31 & 31 & 31 & 31 \\

No. selected
& 1 & 10 & 7 & 0 & 1 & 0 & 2 & 0 & 3 \\

RMSE
& 0.056 & \textbf{0.043} & 0.051
& 0.064 & 0.061
& 0.064 & 0.057
& 0.064 & 0.050 \\

BIC
& -69.750 & -56.159 & -55.135
& -65.368 & -65.293
& -65.368 & -65.152
& -65.368 & \textbf{-70.031} \\

RMSFE-1 year
& \textbf{0.006} & 0.025 & 0.065
& 0.027 & 0.028
& 0.027 & 0.014
& 0.027 & 0.020 \\

RMSFE-2 years
& \textbf{0.007} & 0.119 & 0.069
& 0.036 & 0.023
& 0.036 & 0.016
& 0.036 & 0.016 \\

RMSFE-4 years
& \textbf{0.037} & 0.093 & 0.088
& 0.072 & 0.074
& 0.072 & 0.053
& 0.072 & 0.086 \\

\hline
\end{tabular}%
}
\par
\begin{center}
\hspace*{-0.95cm}%
\begin{minipage}{1.10\textwidth}
{\footnotesize
\linespread{0.85}\selectfont
Notes: $z_1$--$z_4$ are included in all models; other regressors are selected using the indicated methods. No. selected excludes the preselected controls. Lasso-PF uses PCA prefiltering; 1SE and MSE denote the one standard error and minimum MSE cross validation rules. Blank entries denote excluded variables. Coefficients are post selection OLS estimates, with HC1 robust standard errors in parentheses. RMSE is $\sqrt{RSS/T}$; BIC is based on the post selection model. RMSFE-$h$ evaluates the final $h=1,2,4$ observations, holding the full sample selected specification fixed and reestimating its coefficients on the preceding observations. Reported statistics are conditional on the selected specifications.
\textit{* p$<$0.10, ** p$<$0.05, *** p$<$0.01.}
\par
}
\end{minipage}
\end{center}
\end{table}

Under BMT (column `BMT'), the autoregressive coefficient on lagged emissions, $z_{1}$, is estimated at 0.663 and is highly statistically significant. Conditional on a deterministic trend, this relatively large coefficient points to strong persistence in U.S. air emissions. By contrast, the coefficient on the lagged penalty variable ($z_{2}$) is positive, very small in magnitude and statistically insignificant.

The coefficient on the trend variable ($z_{3}$) is negative, reflecting that emissions exhibit a gradual downward trajectory over the sample period. This aligns with the broad historical trend in U.S. emissions, reflecting structural changes in the economy and energy production.

Among the additional regressors in the candidate set, only one variable, \textit{permitmw} (new private housing units authorized by building permits in the Midwest), is selected, with a positive and highly significant coefficient (0.238). This finding is consistent with the well documented role of the Midwest as a leading emitter within the U.S. If treated as a separate country, the region would be the fifth largest emitter globally \citep{Kvam2025Midwest}. The Midwest’s emissions profile is driven by several structural factors: high transportation demand, heavy reliance on coal and natural gas for electricity generation, widespread natural gas use in residential heating and buildings, and emissions intensive agricultural practices in the Corn Belt. Against this backdrop, new housing activity in the Midwest is naturally associated with higher emissions, making the positive coefficient on permitmw both statistically and economically meaningful.

Turning to the OCMT method (column `OCMT'), the results for the preselected covariates $z_{1}$-$z_{4}$ are broadly similar to those obtained under BMT. Whereas BMT selects only \textit{permitmw}, OCMT selects it along with nine additional covariates. These include \textit{permit}, new private housing units authorised in the United States, and \textit{permits} and \textit{permitw}, the corresponding measures for the South and West. OCMT also selects housing starts in the United States (\textit{houst}), Midwest (\textit{houstmw}), South (\textit{housts}) and West (\textit{houstw}), together with \textit{uslah} (employment in leisure and hospitality) and \textit{prfix} (real residential fixed investment). Neither of the last two coefficients is statistically significant in this specification.

OCMT can admit several regressors in one stage, whereas BMT admits at most one. The joint selection of related housing indicators does not by itself establish that any has a zero coefficient in the DGP, nor does this comparison establish whether BMT omits relevant variables.

To examine the dependence among the selected regressors, we calculate pairwise correlations after partialling out the four compulsory controls in equation~(34). Specifically, each selected candidate is regressed on the compulsory controls, and the correlations are computed between the resulting residuals. Of the nine additional variables selected by OCMT, eight have absolute partial correlations exceeding 0.91 with \texttt{permitmw}, the sole additional regressor selected by BMT. These variables comprise residential fixed investment and several national and regional housing starts and building permits indicators. The exception is \texttt{uslah}, whose partial correlation
with \texttt{permitmw} is $-0.386$.
These results indicate substantial overlap in the information contained in the selected housing indicators, even after accounting for the compulsory controls.

GOCMT selects seven additional candidates, mainly housing related indicators. Lasso-PF and Adaptive Lasso select no additional candidates under the 1SE rule; the minimum MSE rule selects one and two candidates, respectively. MCP also selects no additional candidates. SCAD selects three additional variables: \textit{permitmw}, \textit{permitw} and \textit{lns13023569}.

OCMT achieves the lowest in sample RMSE (0.043). SCAD has the lowest BIC ($-70.031$), followed closely by BMT ($-69.750$). BMT selects only one additional candidate and has the lowest RMSFE over the final one, two and four years of the sample, at 0.006, 0.007 and 0.037, respectively. These forecast comparisons are conditional on the full sample selected specifications, with coefficients reestimated using the observations preceding each forecasting period.

\subsection{Phillips curves and inflation forecasting}
\label{empirics2}

Reduced form Phillips curves relate future inflation to past inflation and measures of economic activity. This link provides a basis for inflation forecasting, which can be extended to broader sets of macroeconomic predictors \citep{StockWatson1999, atkeson2001phillips, stock2008phillips}. Our second application examines variable selection in this setting, allowing labour market, monetary and financial indicators to compete for inclusion. The baseline comparison includes BMT, OCMT, GOCMT, Lasso with and without PCA prefiltering, Adaptive Lasso, MCP and SCAD. We then examine how the forecast comparisons change when inflation terms are imposed as controls.

Let $\pi_t$ denote quarterly inflation constructed from the Consumer Price Index (CPIAUCSL). At each horizon $h\in\{1,2,4\}$, we estimate the direct predictive regression
\begin{equation} \label{eq:phillips_curve}
    y_{t,h} = \alpha_h + \sum_{\ell=0}^{p-1}\gamma_{\ell,h}\Delta\pi_{t-\ell} + \sum_{j=1}^{N} \beta_{j,h} x_{j,t} + \epsilon_{t,h}.
\end{equation}
For the change in inflation, the target is $y_{t,h}=\pi_{t+h}-\pi_{t+h-1}$. For average inflation, the target is $y_{t,h}=h^{-1}\sum_{s=1}^{h}\pi_{t+s}-\pi_t$. Adding $\pi_t$ to the forecast of the second target gives a forecast of average inflation over the next $h$ quarters. The two targets coincide at $h=1$.

The predictors $x_{j,t}$ are drawn from the \textit{FRED-QD} quarterly database \citep{McCrackenNg2021}. The candidate set contains $N=200$ macroeconomic indicators and excludes CPIAUCSL itself. The intercept and inflation terms are included without selection in every method and in the post selection OLS regression. Selection applies only to the macroeconomic candidates. Labour market indicators are among these candidates; no particular slack or inflation expectations measure is imposed as a control. The selected coefficients describe predictive associations.

The restriction $p=0$ omits the inflation terms from \eqref{eq:phillips_curve}. Tables \ref{tab:phillips_fullsample} and \ref{tab:phillips_oos} report full sample selection and recursive forecast performance under this restriction. We compare it with $p=1,2,4$, using $\Delta\pi_t,\ldots,\Delta\pi_{t-p+1}$ as compulsory controls, and with a separate specification using four inflation level terms, $\pi_t,\ldots,\pi_{t-3}$. The controls remain in the regression even when no macroeconomic candidate is selected.

For each target and horizon, forecasts are evaluated over 36 common recursive origins. Variable selection and coefficient estimation are repeated as the estimation window expands, using training outcomes observed by each forecast origin. Forecast performance is measured by the out of sample coefficient of determination, $R^2_{\mathrm{OOS}}$, and RMSFE relative to a zero forecast of the target. For average inflation, this benchmark forecasts future average inflation at the current inflation rate. The restricted comparison also includes historical mean and separate four lag autoregressive forecasts. The latter comparison does not condition macroeconomic variable selection on inflation lags.

\begin{table}[htbp]
\centering
\caption{Inflation Forecasting: Full Sample Selection without Inflation Controls}
\label{tab:phillips_fullsample}
\small
\setlength{\tabcolsep}{4pt}
\begin{tabular}{l r r r >{\raggedright\arraybackslash}p{7.0cm}}
\toprule
\textbf{Method} & \textbf{Size} & \textbf{RMSE} & \textbf{BIC} & \textbf{Selected Variables} \\
\midrule
BMT & 2 & 0.004979 & \textbf{-1630.6} & \texttt{M2REAL}, \texttt{TB3MS} \\
OCMT & 8 & 0.005082 & -1589.7 & \texttt{USEHS}, \texttt{USLAH}, \texttt{USSERV}, \texttt{COMPRNFB}, \texttt{RCPHBS}, \texttt{BOGMBASEREALx}, \texttt{M2REAL}, \texttt{GFDEBTNx} \\
GOCMT & 2 & 0.004979 & \textbf{-1630.6} & \texttt{TB3MS}, \texttt{M2REAL} \\
\addlinespace[2pt]
Lasso (1SE) & 0 & 0.005748 & -1580.4 & None \\
Lasso (MinMSE) & 3 & 0.005676 & -1569.7 & \texttt{HWIx}, \texttt{TNWMVBSNNCBBDIx}, \texttt{TLBSNNBBDIx} \\
Lasso-PF (1SE) & 0 & 0.005748 & -1580.4 & None \\
Lasso-PF (MinMSE) & 1 & 0.005683 & -1579.8 & \texttt{TNWMVBSNNCBBDIx} \\
AdLasso (1SE,1SE) & 0 & 0.005748 & -1580.4 & None \\
AdLasso (MinMSE,MinMSE) & 0 & 0.005748 & -1580.4 & None \\
\addlinespace[2pt]
MCP & 9 & 0.004565 & -1629.9 & \texttt{DPIC96}, \texttt{IPB51220SQ}, \texttt{LNS13023569}, \texttt{DGOERG3Q086SBEA}, \texttt{DRCARG3Q086SBEA}, \texttt{RCPHBS}, \texttt{TB3MS}, \texttt{M2REAL}, \texttt{TARESAx} \\
SCAD & 28 & 0.003845 & -1600.9 & \texttt{GPDIC1}, \texttt{DPIC96}, \texttt{IPB51220SQ}, \texttt{CUMFNS}, \texttt{LNS13023557}, \texttt{LNS13023569}, \texttt{HOUSTNE}, \texttt{HOUSTS}, \texttt{CMRMTSPLx}, \texttt{AMDMNOx}, \texttt{AMDMUOx}, \texttt{DGDSRG3Q086SBEA}, \texttt{DSERRG3Q086SBEA}, \texttt{DNDGRG3Q086SBEA}, \texttt{DFDHRG3Q086SBEA}, \texttt{DGOERG3Q086SBEA}, \texttt{DRCARG3Q086SBEA}, \texttt{RCPHBS}, \texttt{TB3MS}, \texttt{M2REAL}, \texttt{TARESAx}, \texttt{B021RE1Q156NBEA}, \texttt{TOTRESNS}, \texttt{NONBORRES}, \texttt{DTCOLNVHFNM}, \texttt{CLAIMSx}, \texttt{PERMITNE}, \texttt{TNWMVBSNNCBx} \\
\bottomrule
\end{tabular}
\par\vspace{3pt}
\begin{minipage}{\textwidth}
\footnotesize\textit{Notes:}
The dependent variable is the one-quarter-ahead change in inflation,
$\Delta\pi_{t+1}$. Full-sample estimates use a common sample of $T=212$ observations and 200 candidate regressors, after excluding observations with missing or nonfinite values in the dependent variable or any candidate regressor. The intercept is the only compulsory regressor. Size is the number of selected candidate variables, excluding the intercept. RMSE and BIC are computed from post-selection OLS. PF denotes principal component prefiltering and AdLasso denotes Adaptive Lasso; MinMSE and 1SE denote the minimum-MSE and one-standard-error penalty
selection rules, respectively.
\end{minipage}
\end{table}

\begin{table}[htbp]
\centering
\caption{Inflation Forecasting without Inflation Controls}
\label{tab:phillips_oos}
\scriptsize
\setlength{\tabcolsep}{3.5pt}
\resizebox{\textwidth}{!}{%
\begin{tabular}{l *{9}{r}}
\toprule
& \multicolumn{3}{c}{$h=1$}
& \multicolumn{3}{c}{$h=2$}
& \multicolumn{3}{c}{$h=4$} \\
\cmidrule(lr){2-4}\cmidrule(lr){5-7}\cmidrule(lr){8-10}
\textbf{Method}
& $R^2_{\mathrm{OOS}}$ (\%) & Avg. size & Rel. RMSFE
& $R^2_{\mathrm{OOS}}$ (\%) & Avg. size & Rel. RMSFE
& $R^2_{\mathrm{OOS}}$ (\%) & Avg. size & Rel. RMSFE \\
\midrule
\multicolumn{10}{l}{\textit{Change in inflation}} \\
\addlinespace[2pt]
BMT & 38.58 & 2.00 & 0.784 & -1.10 & 1.08 & 1.005 & -13.74 & 1.33 & 1.066 \\
OCMT & 29.11 & 4.78 & 0.842 & -0.17 & 0.00 & 1.001 & -3.14 & 0.50 & 1.016 \\
GOCMT & 31.01 & 2.44 & 0.831 & -0.17 & 0.00 & 1.001 & -4.36 & 0.22 & 1.022 \\
\addlinespace[2pt]
Lasso (1SE) & -0.53 & 0.00 & 1.003 & -0.17 & 0.00 & 1.001 & -0.29 & 0.00 & 1.001 \\
Lasso (MinMSE) & 24.53 & 0.94 & 0.869 & -67.18 & 1.25 & 1.293 & -1.46 & 0.33 & 1.007 \\
Lasso-PF (1SE) & -0.53 & 0.00 & 1.003 & -0.17 & 0.00 & 1.001 & -0.29 & 0.00 & 1.001 \\
Lasso-PF (MinMSE) & -55.32 & 14.67 & 1.246 & -18.48 & 2.61 & 1.088 & -0.32 & 0.28 & 1.002 \\
AdLasso (1SE,1SE) & -0.53 & 0.00 & 1.003 & -0.17 & 0.00 & 1.001 & -0.29 & 0.00 & 1.001 \\
AdLasso (MinMSE,MinMSE) & 35.74 & 1.56 & 0.802 & -58.28 & 1.08 & 1.258 & -2.56 & 0.33 & 1.013 \\
\addlinespace[2pt]
MCP & 19.33 & 5.86 & 0.898 & -3.30 & 10.69 & 1.016 & 1.24 & 0.56 & 0.994 \\
SCAD & 29.65 & 9.47 & 0.839 & -8.37 & 14.72 & 1.041 & -0.96 & 0.17 & 1.005 \\
\addlinespace[2pt]
AR(4) & 13.08 & 4.00 & 0.932 & -1.63 & 4.00 & 1.008 & -1.88 & 4.00 & 1.009 \\
Mean & -0.53 & 0.00 & 1.003 & -0.17 & 0.00 & 1.001 & -0.29 & 0.00 & 1.001 \\
\midrule
\multicolumn{10}{l}{\textit{Average inflation}} \\
\addlinespace[2pt]
BMT & 38.58 & 2.00 & 0.784 & 38.71 & 1.06 & 0.783 & 41.33 & 4.28 & 0.766 \\
OCMT & 29.11 & 4.78 & 0.842 & 29.51 & 6.00 & 0.840 & 43.78 & 6.61 & 0.750 \\
GOCMT & 31.01 & 2.44 & 0.831 & 29.41 & 3.69 & 0.840 & 50.45 & 8.17 & 0.704 \\
\addlinespace[2pt]
Lasso (1SE) & -0.53 & 0.00 & 1.003 & -0.56 & 0.00 & 1.003 & -8.69 & 0.11 & 1.043 \\
Lasso (MinMSE) & 24.53 & 0.94 & 0.869 & -14.09 & 2.19 & 1.068 & -1.92 & 3.53 & 1.010 \\
Lasso-PF (1SE) & -0.53 & 0.00 & 1.003 & -0.56 & 0.00 & 1.003 & -216.35 & 14.11 & 1.779 \\
Lasso-PF (MinMSE) & -55.32 & 14.67 & 1.246 & -153.78 & 7.56 & 1.593 & -78.99 & 31.67 & 1.338 \\
AdLasso (1SE,1SE) & -0.53 & 0.00 & 1.003 & -0.56 & 0.00 & 1.003 & -2.46 & 0.08 & 1.012 \\
AdLasso (MinMSE,MinMSE) & 35.74 & 1.56 & 0.802 & -14.12 & 1.89 & 1.068 & -16.73 & 3.92 & 1.080 \\
\addlinespace[2pt]
MCP & 19.33 & 5.86 & 0.898 & 39.89 & 4.56 & 0.775 & 47.78 & 8.06 & 0.723 \\
SCAD & 29.65 & 9.47 & 0.839 & 12.50 & 12.00 & 0.935 & 56.64 & 10.97 & 0.658 \\
\addlinespace[2pt]
AR(4) & 13.08 & 4.00 & 0.932 & 15.70 & 4.00 & 0.918 & 13.42 & 4.00 & 0.930 \\
Mean & -0.53 & 0.00 & 1.003 & -0.56 & 0.00 & 1.003 & -1.03 & 0.00 & 1.005 \\
\bottomrule
\end{tabular}%
}
\par\vspace{3pt}
\begin{minipage}{\textwidth}
\footnotesize\textit{Notes:}
Each target--horizon combination uses 36 recursive forecast origins.
The selected models contain no own inflation controls; AR(4) is estimated separately.
The upper panel forecasts $\pi_{t+h}-\pi_{t+h-1}$; the lower panel forecasts
$h^{-1}\sum_{s=1}^{h}\pi_{t+s}-\pi_t$.
The benchmark is a zero forecast of the corresponding target: no change in
inflation in the upper panel, and average future inflation equal to current
inflation in the lower panel.
$R^2_{\mathrm{OOS}}=100(1-\mathrm{MSFE}/\mathrm{MSFE}_0)$, where
$\mathrm{MSFE}_0$ is the benchmark mean squared forecast error.
Relative RMSFE is the ratio of the method's RMSFE to the benchmark RMSFE;
values below one indicate lower forecast errors.
Average size is the mean number of selected candidate regressors across
forecast origins, excluding the intercept; AR(4) counts its four lags.
``Mean'' uses the expanding window sample mean of the corresponding target
at each forecast origin.
\end{minipage}
\end{table}

In the restricted specification, Table \ref{tab:phillips_fullsample} shows that BMT selects real M2, which is nominal M2 deflated by the CPI, and the three month Treasury bill rate. GOCMT selects the same two variables, and the two methods attain the lowest BIC. MCP and SCAD select nine and 28 variables, respectively, with both specifications including real M2 and the three-month Treasury bill rate. Although MCP and SCAD achieve lower in-sample RMSE, these improvements come at the cost of substantially larger models. Notably, BMT achieves a slightly lower BIC than MCP (by 0.7 units), despite selecting only two variables compared with nine under MCP. SCAD achieves the lowest in-sample RMSE but is considerably less parsimonious. OCMT selects eight variables, including real M2 but not the Treasury bill rate, while the Lasso variants select between zero and three candidates. In the full sample specification, BMT selects no labour market variable, whereas OCMT includes three sectoral employment series.

BMT has the highest one-quarter-ahead-out-of-sample $R^2$ in both panels of Table \ref{tab:phillips_oos}, at 38.58\%, with an average of two selected variables. For quarterly inflation changes, this advantage does not extend to longer horizons. No method improves on the zero change benchmark at $h=2$, and only MCP does so at $h=4$, with $R^2_{\mathrm{OOS}}=1.24\%$.

For average inflation, BMT performs particularly well at $h = 2$, recording an $R^2_{\mathrm{OOS}}$ of 38.71\%, only marginally below MCP's 39.89\%, while selecting considerably fewer predictors on average (1.06 compared with 4.56). At $h=4$, SCAD, GOCMT, MCP and OCMT all have higher $R^2_{\mathrm{OOS}}$ than BMT. SCAD has the highest value, 56.64\%, compared with 41.33\% for BMT; their average selected sizes are 10.97 and 4.28.

BMT also has lower forecast errors than the AR(4) benchmark at $h=1$ and for average inflation at all three horizons. These comparisons are descriptive summaries over the 36 forecast origins.

Table \ref{tab:phillips_controls} compares the forecasting performance of BMT, OCMT and GOCMT under alternative specifications of the compulsory inflation controls. To ensure  comparability, all specifications use the same estimation sample and 36 forecast origins. This requires starting the estimation sample three observations later than in Table \ref{tab:phillips_oos}, while keeping the forecast evaluation period unchanged. The table also reports results without inflation controls, re-estimated over this common sample, so that differences across columns reflect the inclusion of compulsory inflation terms rather than changes in the estimation period. The corresponding results for the remaining selection procedures, including BMT-R1 and BMT-R2, together with a controls-only forecast,
are provided in Appendix~\ref{app:inflation_controls}. For each specification, we report $R^2_{\mathrm{OOS}}$ in per cent, with the average number of selected macroeconomic predictors across the 36 forecast origins shown in parentheses, excluding the intercept and compulsory inflation controls.

\begin{table}[!htbp]
\centering
\caption{Inflation Forecasting with Compulsory Inflation Controls}
\label{tab:phillips_controls}
\small
\setlength{\tabcolsep}{3pt}
\begin{tabular}{l*{5}{r}}
\toprule
& No controls & \shortstack{1 inflation\\change}
& \shortstack{2 inflation\\changes}
& \shortstack{4 inflation\\changes}
& \shortstack{4 inflation\\levels} \\
\midrule
\multicolumn{6}{l}{\textit{$h=1$: both targets}} \\
\addlinespace[2pt]
BMT & 38.03 (2.00) & 43.00 (2.00) & -23.31 (1.28) & -19.76 (1.28) & -7.95 (1.28) \\
OCMT & 28.90 (5.19) & 35.95 (3.22) & -17.78 (0.58) & -9.85 (0.92) & -3.19 (4.00) \\
GOCMT & 31.15 (2.22) & 40.20 (2.47) & -158.31 (0.61) & -91.52 (1.08) & -64.47 (4.00) \\
\addlinespace[4pt]
\multicolumn{6}{l}{\textit{Change in inflation: $h=2$}} \\
\addlinespace[2pt]
BMT & -2.07 (1.06) & -10.03 (1.00) & 0.59 (1.00) & 1.22 (1.00) & 3.45 (1.00) \\
OCMT & -0.17 (0.00) & -3.94 (0.22) & -2.09 (0.00) & 1.15 (1.06) & 3.75 (1.61) \\
GOCMT & -0.17 (0.00) & -3.71 (0.00) & -4.18 (0.31) & 1.22 (1.00) & 3.45 (1.00) \\
\addlinespace[4pt]
\multicolumn{6}{l}{\textit{Change in inflation: $h=4$}} \\
\addlinespace[2pt]
BMT & -0.57 (1.44) & -2.82 (1.17) & -3.07 (1.17) & -10.84 (1.03) & -9.62 (1.17) \\
OCMT & -0.19 (0.25) & -9.00 (0.92) & -14.92 (1.03) & -9.25 (0.72) & -2.54 (0.56) \\
GOCMT & -4.50 (0.22) & -1.70 (1.03) & -4.17 (1.33) & -2.23 (0.81) & -1.29 (0.11) \\
\addlinespace[4pt]
\multicolumn{6}{l}{\textit{Average inflation: $h=2$}} \\
\addlinespace[2pt]
BMT & 36.93 (1.08) & 45.86 (2.00) & 21.41 (1.64) & 22.77 (1.69) & 29.80 (1.25) \\
OCMT & 27.99 (6.25) & 45.03 (2.53) & 27.66 (3.42) & -55.16 (4.75) & 28.25 (2.81) \\
GOCMT & 27.36 (3.69) & 42.02 (2.00) & -79.11 (1.72) & -67.62 (2.58) & -25.79 (2.50) \\
\addlinespace[4pt]
\multicolumn{6}{l}{\textit{Average inflation: $h=4$}} \\
\addlinespace[2pt]
BMT & 39.55 (4.28) & 47.06 (3.56) & -67.46 (2.25) & -66.21 (2.31) & 20.47 (1.83) \\
OCMT & 44.04 (6.31) & 44.71 (6.92) & 37.75 (6.11) & 6.91 (10.03) & -65.16 (10.92) \\
GOCMT & 33.55 (7.39) & 61.32 (3.78) & 63.11 (5.78) & 60.49 (12.11) & 11.27 (10.11) \\
\bottomrule
\end{tabular}
\par\vspace{3pt}
\begin{minipage}{\textwidth}
\footnotesize\textit{Notes:} Entries are $R^2_{\mathrm{OOS}}$ in per cent, with the average number of selected macroeconomic candidates in parentheses. All columns use the same initial estimation sample and 36 forecast origins. The no control column includes the intercept. The other columns also include the indicated number of inflation change terms, starting at $\Delta\pi_t$, or inflation level terms, starting at $\pi_t$. These terms are compulsory and are excluded from the reported candidate counts. All selected models are refitted by OLS with the intercept and inflation controls retained. The benchmark is a zero forecast of the corresponding target.
\end{minipage}
\end{table}

With one inflation change control, BMT improves at $h=1$: its $R^2_{\mathrm{OOS}}$ rises from 38.03\% to 43.00\%. For average inflation, it rises from 36.93\% to 45.86\% at $h=2$, and from 39.55\% to 47.06\% at $h=4$. For future quarterly inflation changes, the corresponding results remain negative at $h=2$ and $h=4$.

These gains do not extend to larger control blocks. With two or four inflation change controls, BMT has negative out of sample $R^2$ at $h=1$ and for average inflation at $h=4$. Four inflation level controls also reduce its performance at $h=1$ and for average inflation at both longer horizons, relative to the matched no control specification. The effect of inflation controls therefore depends on the target, horizon and control block. The one term results are favourable to BMT, but do not establish an overall improvement. No lag order is selected from these evaluation results.

The relative comparisons also remain mixed. With one inflation change control, GOCMT attains 61.32\% for average inflation at $h=4$, compared with 47.06\% for BMT. Their average candidate counts are 3.78 and 3.56, respectively. These counts exclude the compulsory regressors. For example, with four inflation change controls, BMT selects an average of 2.31 candidates at this horizon, compared with 4.28 without inflation controls; including the intercept and controls gives 7.31 and 5.28 estimated coefficients, respectively. Fewer selected candidates can therefore accompany a larger regression.

\section{Concluding Remarks}\label{conclusions}

High dimensional regression analysis sits at the centre of modern empirical work in statistics, econometrics and machine learning, underpinning applications that range from macroeconomic forecasting to genomics and text as data. Existing approaches include regularisation methods such as Ridge and Lasso and greedy algorithms such as stepwise procedures, while multiple testing based procedures such as OCMT offer an alternative route to parsimony. 

This paper has proposed Boosting with Multiple Testing (BMT), which combines forward stepwise variable addition with a multiple testing threshold. Like OCMT, it tests candidates conditional on the variables already included. BMT admits at most one candidate at each stage and then recomputes the remaining scores. This can prevent further proxy admissions, although ordinary BMT retains every variable already selected.

BMT recovers the exact true model under the stated signal strength, signal to proxy dominance, moment and covariance conditions. The fixed signal results allow a growing proxy pool; a separate corollary permits growing signal counts under stronger conditions. Under their own assumptions with fixed signal and proxy counts, the approximating model results retain all signals and exclude noise while allowing some proxies. They also give residual variance and coefficient error rates. Oracle inference is established separately after exact selection in fixed dimension. The comparisons identify designs favouring BMT over the Lasso under specified tuning, and the appendix gives examples in both directions for SCAD and MCP. For fixed signal and proxy counts, the appendix also establishes exact recovery by BMT-R1 and BMT-R2 under detectable coefficients and design bounds covering proxy containing models, without the signal--proxy ordering required by ordinary BMT's exact recovery theorem.

The Monte Carlo results show favourable finite sample performance across the reported sample sizes and factor dependence configurations. BMT has the highest median MCC and exact recovery rate and the lowest median coefficient RMSE among the compared specifications. Adaptive Lasso under the 1SE rule, MCP and SCAD also perform well, and no method dominates across all designs and criteria. The results support BMT as an alternative when regressors are numerous, strongly correlated and potentially dynamic. Under the exact recovery and inference conditions, post selection least squares has the oracle first order distribution.


Future work could extend BMT to panel data and models with endogenous regressors, using instrumental variables or control functions. Extensions to nonlinear models, including quantile and count regression, also merit investigation.

\bibliographystyle{apalike}
\bibliography{boosting}

@article{StockWatson1999,
  author  = {Stock, James H. and Watson, Mark W.},
  title   = {Forecasting Inflation},
  journal = {Journal of Monetary Economics},
  volume  = {44},
  pages   = {293--335},
  year    = {1999}
}

@article{StockWatson2002,
  author  = {Stock, James H. and Watson, Mark W.},
  title   = {Macroeconomic Forecasting Using Diffusion Indexes},
  journal = {Journal of Business \& Economic Statistics},
  volume  = {20},
  number  = {2},
  pages   = {147--162},
  year    = {2002}
}

@Article{ChenChen2008,
  Title                    = {Extended Bayesian information criteria for model selection with large model spaces},
  Author                   = {J. Chen and Z. Chen},
  Journal                  = {Biometrika},
  Year                     = {2008},
  Pages                    = {759--771},
  Volume                   = {95}
}

@Article{fit2014,
  Title                    = {Optimal Inference After Model Selection},
  Author                   = {W. Fithian and D. Sun and J. Taylor},
  Journal                  = {arXiv:1410.2597v4},
  Year                     = {2014},
  OPTPages                    = {65--70},
  OPTVolume                   = {6}
}

@Article{tib2014,
  Title                    = {Exact post-selection inference for sequential regression procedures},
  Author                   = {R. J. Tibshirani and J. Taylor and R. Lockhart and R. Tibshirani},
  Journal                  = {arXiv:1401.3889},
  Year                     = {2014},
  OPTPages                    = {65--70},
  OPTVolume                   = {6}
}

@Article{gse2015,
  Title                    = {Sequential selection procedures and false discovery rate control},
  Author                   = { M. G. {G'Sell} and S. Wager and A. Chouldechova and R. Tibshirani},
  Journal                  = {Journal of the Royal Statistical Society: Series B},
  Year                     = {2016},
  Pages                    = {423--444},
  Volume                   = {78}
}

@article{li2015,
  author  = {Li, Ang and Barber, Rina Foygel},
  title   = {Accumulation Tests for FDR Control in Ordered Hypothesis Testing},
  journal = {Journal of the American Statistical Association},
  year    = {2017},
  volume  = {112},
  number  = {518},
  pages   = {837--849},
  doi     = {10.1080/01621459.2016.1180989}
}

@article{kozbur2017,
  author  = {Kozbur, Damian},
  title   = {Testing-Based Forward Model Selection},
  journal = {American Economic Review},
  year    = {2017},
  volume  = {107},
  number  = {5},
  pages   = {266--269},
  doi     = {10.1257/aer.p20171039}
}

@article{kozbur2020,
  author  = {Kozbur, Damian},
  title   = {Analysis of Testing-Based Forward Model Selection},
  journal = {Econometrica},
  year    = {2020},
  volume  = {88},
  number  = {5},
  pages   = {2147--2173},
  doi     = {10.3982/ECTA16273}
}

@Article{fan2013,
  Title                    = {Asymptotic equivalence of regularization methods in thresholded parameter space},
  Author                   = {J. Fan and J. Lv},
  Journal                  = {Journal of the American Statistical Association},
  Year                     = {2013},
  Pages                    = {1044--1061},
  Volume                   = {108}
}

@Article{LengEtal2006,
  Title                    = {A Note on the Lasso and Related Procedures in Model Selection},
  Author                   = {C. Leng and Y. Lin and G. Wahba},
  Journal                  = {Statistica Sinica},
  Year                     = {2006},
  Pages                    = {1273--1284},
  Volume                   = {16}
}

@Article{hua2008,
  author = {J. Huang and  J. Horowitz  and S. Ma},
  title = {Asymptotic properties of bridge estimators in sparse high-dimensional regression models} ,
  journal =      {Annals of Statistics},
  year = {2008} ,
  OPTkey =       {},
  volume =       {36},
  OPTnumber =       {2},
  OPTmonth =     {},
  pages =        {587--613},
  OPTnote =      {},
  OPTannote =    {}
}

@Article{buh2006,
  Title                    = {Boosting for High-Dimensional Linear Models},
  Author                   = {P. B{\"u}hlmann},
  Journal                  = {The Annals of Statistics},
  Year                     = {2006},
  Number                   = {2},
  Pages                    = {559--583},
  Volume                   = {34}
}

@Book{HastieEtal2009,
  Title                    = {The Elements of Statistical Learning},
  Author                   = {Hastie, T. and Tibshirani, R. and Friedman, J.},
  Publisher                = {Springer},
  Year                     = {2009},
  Edition   = {2nd}
}

@Article{can2007,
  Title                    = {The Dantzig Selector: Statistical Estimation When $p$ is Much Larger Than $n$},
  Author                   = {E. Candes and T. Tao},
  Journal                  = {Annals of Statistics},
  Year                     = {2007},
  Pages                    = {2313--2351},
  Volume                   = {35}
}

@Article{efr2004,
  Title                    = {Least Angle Regression},
  Author                   = {B. Efron and T. Hastie and I Johnstone and R. Tibshirani},
  Journal                  = {Annals of Statistics},
  Year                     = {2004},
  Pages                    = {407--499},
  Volume                   = {32}
}

@Article{fan2001,
  Title                    = {Variable Selection via Nonconcave Penalized Likelihood and Its Oracle Properties},
  Author                   = {J. Fan and R. Li},
  Journal                  = {Journal of the American Statistical Association},
  Year                     = {2001},
  Pages                    = {1348--1360},
  Volume                   = {96}
}

@article{FanKeWang2020,
  author  = {Fan, Jianqing and Ke, Yuan and Wang, Kaizheng},
  title   = {Factor-Adjusted Regularized Model Selection},
  journal = {Journal of Econometrics},
  year    = {2020},
  volume  = {216},
  number  = {1},
  pages   = {71--85},
  doi     = {10.1016/j.jeconom.2020.01.006}
}

@article{FanXueZou2014,
  author  = {Fan, Jianqing and Xue, Lingzhou and Zou, Hui},
  title   = {Strong Oracle Optimality of Folded Concave Penalized Estimation},
  journal = {The Annals of Statistics},
  year    = {2014},
  volume  = {42},
  number  = {3},
  pages   = {819--849},
  doi     = {10.1214/13-AOS1198}
}

@Article{fan2008,
  Title                    = {Sure Independence Screening for Ultra-High Dimensional Feature Space},
  Author                   = {J. Fan and J. Lv},
  Journal                  = {Journal of Royal Statistical Society B},
  Year                     = {2008},
  Pages                    = {849--911},
  Volume                   = {70}
}

@Article{fan2010,
  Title                    = {Sure Independence Screening in Generalized Linear Models with NP-Dimensionality},
  Author                   = {J. Fan and R. Song},
  Journal                  = {Annals of Statistics},
  Year                     = {2010},
  Pages                    = {3567--3604},
  Volume                   = {38}
}

@Article{fri2001,
  Title                    = {Greedy Function Approximation: A Gradient Boosting Machine},
  Author                   = {J. Friedman},
  Journal                  = {Annals of Statistics},
  Year                     = {2001},
  Pages                    = {1189--1232},
  Volume                   = {29}
}

@Article{tib1996,
  Title                    = {Regression Shrinkage and Selection via the Lasso},
  Author                   = {R. Tibshirani},
  Journal                  = {Journal of the Royal Statistical Society B},
  Year                     = {1996},
  Pages                    = {267--288},
  Volume                   = {58}
}

@article{zou2006adaptive,
  author  = {Zou, Hui},
  title   = {The Adaptive Lasso and Its Oracle Properties},
  journal = {Journal of the American Statistical Association},
  year    = {2006},
  volume  = {101},
  number  = {476},
  pages   = {1418--1429},
  doi     = {10.1198/016214506000000735}
}

@article{barber2015knockoffs,
  author  = {Barber, Rina Foygel and Cand{\`e}s, Emmanuel J.},
  title   = {Controlling the False Discovery Rate via Knockoffs},
  journal = {The Annals of Statistics},
  year    = {2015},
  volume  = {43},
  number  = {5},
  pages   = {2055--2085},
  doi     = {10.1214/15-AOS1337}
}

@article{candes2018modelx,
  author  = {Cand{\`e}s, Emmanuel and Fan, Yingying and Janson, Lucas and Lv, Jinchi},
  title   = {Panning for Gold: {`M}odel-{X} Knockoffs for High-Dimensional Controlled Variable Selection},
  journal = {Journal of the Royal Statistical Society: Series B (Statistical Methodology)},
  year    = {2018},
  volume  = {80},
  number  = {3},
  pages   = {551--577},
  doi     = {10.1111/rssb.12265}
}

@Article{zha2010,
  Title                    = {Nearly Unbiased Variable Selection under Minimax Concave Penalty},
  Author                   = {C. H. Zhang},
  Journal                  = {Annals of Statistics},
  Year                     = {2010},
  Pages                    = {894--942},
  Volume                   = {38}
}

@Article{zou2005,
  Title                    = {Regularization and Variable Selection via the Elastic Net},
  Author                   = {H. Zou and T. Hastie},
  Journal                  = {Journal of the Royal Statistical Society B},
  Year                     = {2005},
  Pages                    = {301--320},
  Volume                   = {67}
}

@article{kap2010, title={TESTS OF THE MARTINGALE DIFFERENCE HYPOTHESIS USING BOOSTING AND RBF NEURAL NETWORK APPROXIMATIONS}, volume={26}, DOI={10.1017/S0266466609990612}, number={5}, journal={Econometric Theory}, author={Kapetanios, George and Blake, Andrew P.}, year={2010}, pages={1363–1397}}

@article{LeebPoetscher2008,
  author  = {Leeb, Hannes and P{\"o}tscher, Benedikt M.},
  year    = {2006},
  title   = {Can One Estimate the Conditional Distribution of Post-Model-Selection Estimators?},
  journal = {The Annals of Statistics},
  volume  = {34},
  number  = {5},
  pages   = {2554--2591}
}

@article{Sharifvaghefi2025,
  author  = {Sharifvaghefi, Mahrad},
  title   = {Variable selection in linear regressions with possibly all strongly correlated covariates},
  journal = {Econometric Reviews},
  year    = {2025},
  doi     = {10.1080/07474938.2024.2446864}
}

@article{McCrackenNg2021,
  author  = {Michael W. McCracken and Serena Ng},
  title   = {FRED-QD: A Quarterly Database for Macroeconomic Research},
  journal = {Federal Reserve Bank of St. Louis Review},
  year    = {2021},
  volume  = {103},
  number  = {1},
  pages   = {1--44},
  doi     = {10.20955/r.103.1-44}
}

@article{stock_disentangling_2012,
    title = {Disentangling the {Channels} of the 2007-09 {Recession}},
    volume = {43},
    url = {https://EconPapers.repec.org/RePEc:bin:bpeajo:v:43:y:2012:i:2012-01:p:81-156},
    number = {1 (Spring)},
    urldate = {2025-07-17},
    journal = {Brookings Papers on Economic Activity},
    author = {Stock, James H. and Watson, Mark},
    year = {2012},
    note = {Publisher: Economic Studies Program, The Brookings Institution},
    pages = {81--156},
}

@Article{chu2018,
  Title                    = {A One Covariate at a Time, Multiple Testing Approach to Variable
Selection in High-Dimensional Linear Regression Models},
  Author                   = {A. Chudik and G. Kapetanios and M. H. Pesaran},
  Journal                  = {Econometrica},
  Year                     = {2018},
  Pages                    = {1479-1512},
  Volume                   = {86}
}

@article{den2022, 
title={Estimation of Time-Varying Covariance Matrices for Large Datasets}, 
author={Dendramis, Yiannis and Giraitis, Liudas and Kapetanios, George},
volume={37}, 
number={6}, 
journal={Econometric Theory}, 
 year={2021}, 
pages={1100–1134}}

@misc{Kvam2025Midwest,
  title={The Midwest leads U.S. emissions, here’s why},
  author={Kvam, Isak},
  year={2025},
  month={Mar},
  day={25},
  howpublished={Fresh Energy},
  url={https://fresh-energy.org/the-midwest-leads-u-s-emissions-heres-why},
}

@article{ZhaoYu2006,
  author  = {Zhao, Peng and Yu, Bin},
  title   = {On Model Selection Consistency of {Lasso}},
  journal = {Journal of Machine Learning Research},
  year    = {2006},
  volume  = {7},
  pages   = {2541--2563}
}

@article{MeinshausenBuhlmann2006,
  author  = {Meinshausen, Nicolai and B{\"u}hlmann, Peter},
  title   = {High-Dimensional Graphs and Variable Selection with the {Lasso}},
  journal = {Annals of Statistics},
  year    = {2006},
  volume  = {34},
  number  = {3},
  pages   = {1436--1462}
}

@article{Wainwright2009,
  author  = {Wainwright, Martin J.},
  title   = {Sharp Thresholds for High-Dimensional and Noisy Sparsity Recovery Using $\ell_1$-Constrained Quadratic Programming ({Lasso})},
  journal = {IEEE Transactions on Information Theory},
  year    = {2009},
  volume  = {55},
  number  = {5},
  pages   = {2183--2202},
  doi     = {10.1109/TIT.2009.2016018}
}

@article{Matthews1975,
  author  = {Matthews, Brian W.},
  title   = {Comparison of the predicted and observed secondary structure of {T4} phage lysozyme},
  journal = {Biochimica et Biophysica Acta (BBA) - Protein Structure},
  year    = {1975},
  volume  = {405},
  number  = {2},
  pages   = {442--451},
  doi     = {10.1016/0005-2795(75)90109-9}
}

@article{ChiccoJurman2020,
  author  = {Chicco, Davide and Jurman, Giuseppe},
  title   = {The advantages of the {Matthews} correlation coefficient ({MCC}) over {F1} score and accuracy in binary classification evaluation},
  journal = {BMC Genomics},
  year    = {2020},
  volume  = {21},
  number  = {1},
  pages   = {1--13},
  doi     = {10.1186/s12864-019-6413-7}
}

@Article{atkeson2001phillips,
  author  = {Atkeson, Andrew and Ohanian, Lee E.},
  title   = {Are Phillips curves useful for forecasting inflation?},
  journal = {Federal Reserve Bank of Minneapolis Quarterly Review},
  year    = {2001},
  volume  = {25},
  number  = {1},
  pages   = {2--11}
}

@InCollection{stock2008phillips,
  author    = {Stock, James H. and Watson, Mark W.},
  title     = {Phillips curve inflation forecasts},
  booktitle = {Understanding Inflation and the Implications for Monetary Policy},
  publisher = {MIT Press},
  year      = {2009},
  editor    = {Fuhrer, Jeff and Kodrzycki, Yolanda and Little, Jane Sneddon and Olivei, Giovanni},
  pages     = {99--202}
}

@Article{IngLai2011,
 author={Ching-Kang Ing and Tze Leung Lai},
 title={A stepwise regression method and consistent model selection for high-dimensional sparse linear models},
 journal={Statistica Sinica}, year={2011}, volume={21}, pages={1473--1513}, doi={10.5705/ss.2010.081}
}

@Article{IngEtAl2017,
 author={Ching-Kang Ing and Tze Leung Lai and Milan Shen and Ka Wai Tsang and Shu-Hui Yu},
 title={Multiple Testing in Regression Models With Applications to Fault Diagnosis in the Big Data Era},
 journal={Technometrics}, year={2017}, volume={59}, number={3}, pages={351--360}, doi={10.1080/00401706.2016.1236755}
}

@Article{ZambomKim2022,
 author={Adriano Zanin Zambom and Jongwook Kim},
 title={New Algorithms and Software for Significance Controlled Variable Selection},
 journal={Statistics, Optimization and Information Computing}, year={2022}, volume={10}, number={3}, pages={949--967}, doi={10.19139/soic-2310-5070-1520}
}

@Article{zouli2008,
  Author = {Hui Zou and Runze Li},
  Title = {One-step sparse estimates in nonconcave penalized likelihood models},
  Journal = {Annals of Statistics},
  Year = {2008}, Volume = {36}, Number = {4}, Pages = {1509--1533},
  DOI = {10.1214/009053607000000802}
}


\clearpage
\begin{bibunit}[apalike]
\gdef\bmtbibprefix{appendix.}

\setcounter{page}{1}
\setcounter{footnote}{0}

\begin{titlepage}
\centering

{\LARGE\bfseries Model Selection in High Dimensional Linear Regression
using Boosting with Multiple Testing\par}
\vspace{0.6em}
{\Large\bfseries Online Appendix\par}

\vspace{1.5em}
{\large
\begin{tabular}{c}
George Kapetanios$^{1}$\\
Vasilis Sarafidis$^{2}$\\
Alexia Ventouri$^{1}$
\end{tabular}
\par}

\vspace{0.8em}
{\small $^{1}$Department of Banking and Finance, King's College London, Bush House, 30 Aldwych, London WC2B 4BG, UK\par}

{\small $^{2}$Brunel Business School, Brunel University London, Eastern Gateway, Kingston Lane, Uxbridge UB8 3PN, UK\par}

{\small *Corresponding author: George Kapetanios. Email: george.kapetanios@kcl.ac.uk\par}

\raggedright
\normalsize

\bigskip
\textbf{Abstract:} 
This supplementary material provides the theoretical derivations and complete proofs of all results presented in the paper. It also reports detailed results from the Monte Carlo study, as well as additional tables and figures that complement the evidence discussed in the main text.

\end{titlepage}
\thispagestyle{plain}

\renewcommand{\thetable}{OA.\arabic{table}}
\setcounter{table}{0}
\renewcommand{\theequation}{OA.\arabic{equation}}
\setcounter{equation}{0}

\renewcommand{\thetable}{OA.\arabic{table}}
\setcounter{table}{0}

\setcounter{remark}{0}
\renewcommand{\theremark}{A.\arabic{remark}}

\renewcommand{\thefigure}{OA.\arabic{figure}}
\setcounter{figure}{0}

\setcounter{section}{0}
\setcounter{subsection}{0}

\setcounter{section}{0}
\renewcommand{\thesection}{Appendix \Alph{section}}
\renewcommand{\thesubsection}{\thesection.\arabic{subsection}}
\renewcommand{\thesubsubsection}{\thesubsection.\arabic{subsubsection}}

\clearpage
\appendix

\section{Theoretical Derivations}

\subsection{Technical Lemmas}\label{app:lemmas}

We first give the auxiliary lemmas used in the proofs, stating the
assumptions for each result separately. Throughout the appendix, generic
constants $C,C_0,C_1,\dots$ may change from line to line.

\subsubsection{A maximal inequality for heterogeneous strongly mixing processes}

For ease of reference we reproduce the Bernstein type inequality for sums
\begin{equation*}
\mbox{$S_{T}=T^{-1/2} \sum_{t=1}^{T} (v_{t}-E v_{t})$}
\end{equation*}
\noindent
of heterogeneous $\alpha$ mixing variables $v_t$, shown in \citet{den2022}.

Following \citet{den2022}, write $v_t\in\mathcal{E}(s)$, $s>0$, when
\begin{equation}
\max_{t}E \exp(a|v_{t}|^s)<\infty\quad \mbox{for some  $a>0$}.  \label{e:th1n}
\end{equation}
\addtocounter{equation}{1}

For this tail class, let
\begin{align}  \label{e:Bern1}
& f_{t}(\gamma_{1},\gamma_{2},c,r)=c_{0}\left\{\exp \left(%
-c_{1}r^{\gamma_{1}}\right)+ \exp\left(-c_{2}\left(\frac{r\sqrt t}{\log^{2} t}%
\right)^{\gamma_{2}}\right)\right\}, \quad r>0, \,\,\,t\ \ge2,
\end{align}
where $\gamma_{1}>0$, $\gamma_{2}>0$, and the positive constants $c=(c_{0},
c_{1}, c_{2})$ do not depend on $r$ or $t$.

\begin{lemma}
\label{lem1}Let $v_t-Ev_t$ be strongly mixing with coefficient
$\alpha_v(h)\leq C_v\xi_v^h$, $h\geq1$, where $C_v<\infty$ and
$0<\xi_v<1$, and let $v_t\in\mathcal E(s)$ for some $s>0$. Then, for all
$r>0$ and $T\geq2$,
\begin{equation}
P(|S_T|\geq r)\leq f_T(2,\gamma,c,r),
\qquad \gamma=\frac{s}{s+1},
\label{e:i}
\end{equation}
where $c$ does not depend on $r$ or $T$.
\end{lemma}
\addtocounter{equation}{1}

\noindent
Lemma~\ref{lem1} follows directly from the probability inequalities for
heterogeneous strongly mixing processes established in
\citet{den2022} and is therefore stated without proof; the results below
use only the exponential tail case stated here.

\subsubsection{Auxiliary lemmas for the BMT procedure}

For exact recovery, we consider histories with no admission outside
$\mathcal S_s$, whereas for the approximating model target we restrict admissions
to $\mathcal S_s\cup{P}$. On any such history, $A=\widehat S_{r-1}$
belongs to the corresponding deterministic class.  For each remaining
candidate $i$ at the start of stage $r$, the
notation in the lemmas below is
\[
\boldsymbol q^{(r)}_{\cdot t}=q_{A,t},\qquad
\boldsymbol\gamma_i^{(r)}=\gamma_{x_i,A,T},\qquad
\widetilde x^{(r)}_{i,t}=x_{it}^{A}.
\]
At stage~1, $A=\varnothing$ and $q_{\varnothing,t}=\boldsymbol z_t$, so the
conditioning vector includes any always in controls.  The assumptions are
therefore imposed on the deterministic classes, not on a matrix selected
along the realised path.

\begin{lemma}\label{lem:score}
Suppose that \eqref{dgp1} and Assumptions~\ref{ass1}--\ref{ass_proj} and \ref{reg} hold, $k$ is fixed,
$n\to\infty$, and
$n=O(T^\kappa)$ for some fixed $\kappa>0$. Write
$\widehat E_T(h_t)=T^{-1}\sum_{t=1}^T h_t$. For $v_t,w_t$ ranging over
$y_t-\mathbf a'\boldsymbol z_t,u_t,x_{1t},\ldots,x_{nt}$ and the
components of the control basis specified in Section~3.1, let
\[
\varepsilon_T=\max_{v,w}
 |(\widehat E_T-\overline E_T)(v_tw_t)|.
\]
There are positive constants $C$ and $c=(c_0,c_1,c_2)$, independent of
$n,T,r$, such that
\[
\Pr(\sqrt T\,\varepsilon_T>r)
 \leq C(n+\zeta+2)^2 f_T(2,\gamma,c,r),\qquad
\gamma=\frac{\min(s_x,s_u)}{2+\min(s_x,s_u)},\quad r>0.
\]
Consequently, $\varepsilon_T=O_p(\sqrt{\log n/T})$. Moreover,
\[
\max_{(A,j)\in\mathcal C_T}
 \left|T^{-1/2}\sum_{t=1}^T x_{jt}^A u_t\right|
 =O_p\!\left(\sqrt{(1+K_T)\log n}\right).
\]
For fixed $K_T$, the last bound is $O_p(\sqrt{\log n})$.
\end{lemma}

\begin{proof}
Every product of two components of
$(\boldsymbol z_t',x_{1t},\ldots,x_{nt},u_t)'$ is covered by the local
mixing condition in Assumption~\ref{ass1A_mixingall}. Assumption~\ref{ass2} and Cauchy--Schwarz give
uniform exponential moments of order $\min(s_x,s_u)/2$ for these products.
Lemma~\ref{lem1} and a union bound therefore give the displayed probability
bound for their centred sample averages, with constants uniform in the
component indices.

The control free outcome is
$y_t-\mathbf a'\boldsymbol z_t=\sum_{s\in\mathcal S_s}\beta_s x_{st}+u_t$.
Since $k$ is fixed and the signal coefficients are bounded above in absolute
value as in Section~2, products involving this outcome expand into a fixed
number of primitive products with bounded coefficients. Adjusting the
constants gives the same probability bound without using a positive lower
bound on the nonzero coefficients. Since the controls have been removed,
no moment bound on the raw outcome or bound on $\mathbf a$ is required,
and the argument does not require mixing of the full candidate vector
or of a growing residualisation.

For sufficiently large fixed $L$, substitute $r=L\sqrt{\log n}$ in the
probability bound. Its first term tends to zero when $c_1L^2>2$.
The second tends to zero because
$n=O(T^\kappa)$ and
$(\sqrt{T\log n}/\log^2T)^\gamma/\log n\to\infty$.
This proves the second moment rate.

Assumption~\ref{ass1} gives $\overline E_T(x_{jt}u_t)=0$ and
$\overline E_T(q_{A,t}u_t)=0$. Hence, on the event defining
$\varepsilon_T$,
\[
\left|\widehat E_T(x_{jt}^A u_t)\right|
 \leq(1+\|\gamma_{x_j,A,T}\|_1)\varepsilon_T
 \leq C\sqrt{1+K_T}\,\varepsilon_T,
\]
where $\zeta$ is fixed and Assumption~\ref{ass_proj} bounds the $\ell_2$ norm.
The same event covers every $(A,j)\in\mathcal C_T$; no union over
realised stages is needed.
\end{proof}

\begin{lemma}\label{lem:tstat}
Under the assumptions of Lemma~\ref{lem:score}, let a hat denote the sample
version of a population second moment or Schur complement in Section~3.1.
Suppose also that
\[
(1+K_T)\sqrt{\log n/T}\longrightarrow0.
\]
Uniformly over the deterministic class $\mathcal C_T$, the sample
conditioning and candidate moment matrices are nonsingular with probability
tending to one. On this event the conventional OLS statistic satisfies
\[
t_{j\mid A}
 =\sqrt{T-d_A}\,
   \frac{\widehat\rho_{j,A}}{\sqrt{1-\widehat\rho_{j,A}^2}},
\qquad d_A=\zeta+|A|+1,
\]
whenever the sample outcome residual variance after adding $j$ is positive.
This variance is positive uniformly with probability tending to one.
Furthermore,
\[
\max_{(A,j)\in\mathcal C_T}
 |t_{j\mid A}-\sqrt T\,\psi_{j,A,T}|
 =O_p\!\left((1+K_T)\sqrt{\log n}\right).
\]
In particular, the remainder is $O_p(\sqrt{\log n})$ for fixed $K_T$.
\end{lemma}

\begin{proof}
Write $\widehat Q_A=\widehat E_T(q_{A,t}q_{A,t}')$ and
$\widehat g_{w,A}=\widehat E_T(q_{A,t}w_t)$. Entrywise control gives
\[
\|\widehat Q_A-Q_{A,T}\|_2
 \leq(\zeta+|A|)\varepsilon_T.
\]
The corresponding bound for the moment matrix including candidate $j$ is
$d_A\varepsilon_T$. By Assumption~\ref{reg}, their smallest eigenvalues are
therefore bounded below by a positive constant on an event whose
probability tends to one, simultaneously over the class. On this event,
\[
\|\widehat Q_A^{-1}\|_2\leq C,\qquad
\|\widehat Q_A^{-1}-Q_{A,T}^{-1}\|_2
 \leq C(1+K_T)\varepsilon_T.
\]

Fix $(A,j)$. Since $q_{A,t}$ contains the controls, replacing $y_t$
by $y_t-\mathbf a'\boldsymbol z_t$ leaves all population and sample
residuals, Schur complements, and candidate $t$ statistics unchanged.
For the following calculation take
$v_t,w_t\in\{x_{jt},y_t-\mathbf a'\boldsymbol z_t\}$, and define their
projection coefficients by \eqref{ux}. The latter outcome has a bounded
second moment by \ref{ass2}, the fixed signal count, and the coefficient bound in
Section~2. Its projection coefficient also has bounded
$\ell_2$ norm: $\gamma_w'Q_{A,T}\gamma_w\leq\overline E_T(w_t^2)$
and \ref{reg} bounds $\lambda_{\min}(Q_{A,T})$ below. Write these coefficients
as $\gamma_w$ and put
$e_w=\widehat g_{w,A}-\widehat Q_A\gamma_w$.
Population projection orthogonality gives
$e_w=(\widehat E_T-\overline E_T)(q_{A,t}w_t^A)$.
Expanding the sample Schur complement around these population projections
gives the exact cancellation
\[
\widehat s_{vw\mid A}-s_{vw\mid A,T}
 =(\widehat E_T-\overline E_T)(v_t^Aw_t^A)
   -e_v'\widehat Q_A^{-1}e_w.
\]
The preceding projection bound, Assumption~\ref{ass_proj}, and $\|\gamma_w\|_1\leq
\sqrt{\zeta+|A|}\,\|\gamma_w\|_2$ imply
\begin{align*}
 |(\widehat E_T-\overline E_T)(v_t^Aw_t^A)|
 &\leq(1+\|\gamma_v\|_1)(1+\|\gamma_w\|_1)\varepsilon_T
 \leq C(1+K_T)\varepsilon_T,\\
 \|e_w\|_2
 &\leq\sqrt{\zeta+|A|}(1+\|\gamma_w\|_1)\varepsilon_T
 \leq C(1+K_T)\varepsilon_T.
\end{align*}
It follows that each of the candidate variance, outcome variance, and
candidate--outcome covariance has error bounded by
\[
C\{(1+K_T)\varepsilon_T+(1+K_T)^2\varepsilon_T^2\}.
\]
Since $(1+K_T)\varepsilon_T=o_p(1)$, Assumption~\ref{ass_proj} then gives
\begin{align*}
 \max_{(A,j)\in\mathcal C_T}
 |\widehat\rho_{j,A}-\rho_{j,A,T}|
 &=O_p\!\left((1+K_T)\sqrt{\log n/T}\right),\\
 \max_{(A,j)\in\mathcal C_T}
 |\widehat\omega_{j,A}^2-\omega_{j,A,T}^2|
 &=O_p\!\left((1+K_T)\sqrt{\log n/T}\right).
\end{align*}
Both sample residual variances, as well as
$1-\rho_{j,A,T}^2=\omega_{j,A,T}^2/v_{y,A,T}$ and its sample version,
remain uniformly bounded below by a positive constant.
The derivative of $r/\sqrt{1-r^2}$ is therefore bounded on the range
in question.

For the exact identity, let $M_A$ be the sample residual maker for the
always in controls and $x_A$. Frisch--Waugh--Lovell gives
\[
\widehat\theta_{j\mid A}
 =\frac{\boldsymbol x_j'M_A\boldsymbol y}
        {\boldsymbol x_j'M_A\boldsymbol x_j},\qquad
\mathrm{RSS}_{j\mid A}
 =\boldsymbol y'M_A\boldsymbol y\,(1-\widehat\rho_{j,A}^2).
\]
Using $\widehat\sigma_{j\mid A}^2=\mathrm{RSS}_{j\mid A}/(T-d_A)$
in the conventional standard error yields the stated identity. In
particular, the residual variance is specific to both the candidate and
the conditioning set; it is not replaced by $\sigma^2=E(u_t^2)$.

Finally,
\[
|\sqrt{T-d_A}-\sqrt T|
 =\frac{d_A}{\sqrt{T-d_A}+\sqrt T}
 \leq C(1+K_T)/\sqrt T.
\]
The smooth map bound and this degrees of freedom correction prove the
expansion. The argument includes $A=\varnothing$: $M_A$ then removes
only the always in controls. If there are no such controls, $M_A=I_T$
and the projection terms are empty.
\end{proof}

Under the bounded $\ell_1$ condition stated after Assumption~\ref{reg}, the two
bounds in the Schur complement calculation improve to
$C\varepsilon_T$ and $C\sqrt{1+K_T}\,\varepsilon_T$, respectively.
The Schur complement error is then
$C\{\varepsilon_T+(1+K_T)\varepsilon_T^2\}$.
Under the same inversion condition, the uniform $t$ statistic remainder
is $O_p(\sqrt{\log n})$. This calculation keeps $k$ fixed as in
Assumption~\ref{ass10}; it does not establish the second moment rate for a growing
number of signals. The separate argument in Appendix~\ref{app:growing_signals_proof}
supplies that rate under a bounded sum of absolute signal coefficients.

The probability bound in Lemma~\ref{lem:score} gives more than an
$O_p$ statement. For fixed $K_T$, choose a sufficiently large fixed
$L$ so that $\varepsilon_T\leq L\sqrt{\log n/T}$ with probability
tending to one. The preceding deterministic bounds then give a finite
$B$ such that
\[
\max_{(A,j)\in\mathcal C_T}
 |t_{j\mid A}-\sqrt T\,\psi_{j,A,T}|
 \leq B\sqrt{\log n}
\]
on an event whose probability tends to one. Since
$c_{\alpha_0}(m,\delta)/\sqrt{\log n}\to\sqrt{2\delta}$ when $m=n-O(1)$,
$\delta>B^2/2$ is a sufficient comparison for candidates with
$\psi_{j,A,T}=0$. The constant depends on the maintained moment,
mixing, and design bounds. This is not standard normal calibration of
the conventional statistic under conditional heteroskedasticity.
In the population notation of Section~3.1, the restriction in Section~2 is
$\overline E_T(x_{jt}^{\varnothing}x_{\ell t}^{\varnothing})=0$ for
$j\in\mathcal S_d$ and $\ell\in\mathcal S_s\cup{P}$. Projecting further
on any subset of signals and proxies therefore leaves each noise
variable orthogonal to the remaining signal component. Together with
Assumption~\ref{ass1}, this gives $\psi_{j,A,T}=0$ for noise candidates in
either deterministic class. Theorem~\ref{th1} uses this property to prove noise exclusion and
stopping on the approximating model class.

\subsection{Proof of Theorems}

\subsubsection{Proof of Theorem~\ref{th3}}\label{app:exact_one_proof}

Apply Lemmas~\ref{lem:score}--\ref{lem:tstat} to
$\mathcal C_T^{\mathrm{ER}}$, where $K_T=1$. For a sufficiently large
fixed $L$, the event $\varepsilon_T\leq L\sqrt{\log n/T}$ has
probability tending to one: its complement has probability at most
\[
C(n+\zeta+2)^2 f_T(2,\gamma,c,L\sqrt{\log n})=o(1).
\]
On this event, all covered regressions are well defined and, for all
sufficiently large $n,T$,
\[
\max_{(A,j)\in\mathcal C_T^{\mathrm{ER}}}
|t_{j\mid A}-\sqrt T\,\psi_{j,A,T}|
\leq B\sqrt{\log n}
\]
for a fixed constant $B<\infty$. Choose $\delta>B^2/2$. The threshold
satisfies
\[
\frac{c_{\alpha_0}(m,\delta)}{\sqrt{\log n}}\longrightarrow\sqrt{2\delta}>B
\]
uniformly for $n-1\leq m\leq n$.

At $A=\varnothing$, the superscript $\varnothing$ denotes the population
residual after projection on the compulsory controls $\boldsymbol z_t$ alone,
as defined in \eqref{ux}. Since the error $u_t$ is orthogonal to these
controls, residualising the DGP gives
\[
y_t^{\varnothing}=\beta_{1,T}x_{1t}^{\varnothing}+u_t,\qquad
\rho_{1,\varnothing,T}
=\frac{\beta_{1,T}\sqrt{v_{1,\varnothing,T}}}
       {\sqrt{v_{y,\varnothing,T}}}.
\]
For a proxy $j$, let $r_{j,T}$ be its population correlation with
$x_1$ after removing the controls. Then
\[
\rho_{j,\varnothing,T}=r_{j,T}\rho_{1,\varnothing,T}.
\]
The local eigenvalue bound in \ref{reg} implies that the conditional covariance matrix
of $(x_1,x_j)$ given the controls has eigenvalues bounded above and below.
Its Schur complement is therefore bounded below, and
$1-r_{j,T}^2$ is uniformly bounded away from zero. Thus
$|r_{j,T}|\leq\bar r<1$ for every proxy, with a fixed $\bar r$.

The map $a\mapsto a/\sqrt{1-a^2}$ is increasing on $[0,1)$ and
has derivative at least one. Hence, uniformly over proxies,
\[
|\psi_{1,\varnothing,T}|-|\psi_{j,\varnothing,T}|
\geq(1-\bar r)|\rho_{1,\varnothing,T}|
\geq C_1|\beta_{1,T}|
\]
for some $C_1>0$, by \ref{ass_proj} and \ref{reg}. Noise candidates have zero population
score because they are orthogonal to the true signal after removing the
controls and are orthogonal to $u_t$ by \ref{ass1}. If there are no proxies, the
same bound holds with the proxy maximum set to zero.

Condition~\eqref{eq:one_signal_detectability} now makes the signal's
statistic larger in absolute value than every other candidate's statistic
on the same event. It also gives
\[
|t_{1\mid\varnothing}|
\geq\sqrt T\,|\psi_{1,\varnothing,T}|-B\sqrt{\log n}
>c_{\alpha_0}(n,\delta)
\]
for all sufficiently large $n,T$. Thus the first stage admits the true
signal. This includes the threshold test before admission.

After the first admission, $A=\{1\}$ denotes projection on both the
compulsory controls and the selected signal $x_{1t}$. Consequently,
$y_t^{\{1\}}=u_t$ in the population projection. Assumption~\ref{ass1}
gives $c_{j,A,T}=\psi_{j,A,T}=0$ for every remaining candidate. Their
statistics are bounded by $B\sqrt{\log n}<c_{\alpha_0}(n-1,\delta)$ on the
same event, so BMT stops. This proves \eqref{eq:ER_one_signal} without
conditioning a sampling distribution on the selected set.

Each performance measure is in $[0,1]$. On $\mathcal A_1$ they take
the values $1,0,0,1$ in the order $TPR,FPR,FDR,TDR$. Therefore each
expected absolute deviation from its stated limit is bounded by
$\Pr(\mathcal A_1^c)=o(1)$. The expectation and probability conclusions
follow.

\subsubsection{Proof of Theorem~\ref{th4}}\label{app:exact_general_proof}

The same argument as in the proof of Theorem~\ref{th3}, now with
$K_T=k$ fixed, gives an event of probability tending to one on which
all covered regressions are well defined and
\[
\max_{(A,j)\in\mathcal C_T^{\mathrm{ER}}}
|t_{j\mid A}-\sqrt T\,\psi_{j,A,T}|
\leq B\sqrt{\log n}
\]
for a fixed $B<\infty$. Choose $\delta>B^2/2$. Uniformly for
$n-k\leq m\leq n$, the threshold divided by $\sqrt{\log n}$
then tends to $\sqrt{2\delta}>B$. The remainder of the proof is on
this event.

For any $A\subseteq\mathcal S_s$ and $j\in\mathcal S_d$, the
control residualised noise variable is orthogonal to every signal.
Further projection on $x_A$ leaves it unchanged. The DGP and \ref{ass1}
therefore give $c_{j,A,T}=\psi_{j,A,T}=0$. Its statistic is below
the threshold at every covered subset, including immediately before
a possible first noise admission.

At $A=\varnothing$, the difference between the largest absolute signal
statistic and the largest absolute proxy statistic (set to zero when
there are no proxies) is at least
\[
\sqrt T\left\{
\max_{s\in\mathcal S_s}|\psi_{s,\varnothing,T}|
-\max_{j\in{P}}|\psi_{j,\varnothing,T}|
\right\}-2B\sqrt{\log n}.
\]
This is positive by \eqref{cond1}. Since the proxy maximum is
nonnegative, the largest signal statistic is also larger than
$c_{\alpha_0}(n,\delta)$. Noise candidates are below that threshold by the
preceding step. Thus the largest candidate statistic belongs to a true
signal and passes the admission test. Ties between signals do not affect
this conclusion.

The same two inequalities hold for every $A\subsetneq\mathcal S_s$,
with the maxima over $\mathcal S_s\setminus A$ and
${P}$, and with threshold $c_{\alpha_0}(n-|A|,\delta)$.
Indeed, \eqref{cond1} is uniform over these subsets and its margin
divided by $\sqrt{\log n/T}$ tends to infinity. Starting from the
first admission, induction therefore gives a further signal at every
stage until all $k$ signals have entered. The argument never assumes
that the realised set is a signal subset after an erroneous admission;
it establishes that no such admission occurs on the stated event.

At $A=\mathcal S_s$, the population outcome residual is $u_t$.
Assumption~\ref{ass1} gives $c_{j,A,T}=\psi_{j,A,T}=0$ for every remaining
candidate, including all proxies. The class
$\mathcal C_T^{\mathrm{ER}}$ contains this terminal set. Every
remaining statistic is therefore bounded by
$B\sqrt{\log n}<c_{\alpha_0}(n-k,\delta)$, so the procedure stops after
$k$ admissions. This proves $\Pr(\mathcal A_1)\to1$ and
$\Pr(\widehat J=k)\to1$. The bounded measure argument at the end
of the proof of Theorem~\ref{th3} proves the remaining conclusions.

For the examples, the covariance calculations apply after projection on the controls.
Primitive models satisfying \ref{ass1}--\ref{ass2} can be obtained by taking no controls,
independent observations, and independent centred Gaussian factors,
idiosyncratic components, regression errors and noise variables.
Take the noise and error variances bounded above and below. For the
first two examples, assume bounded signal coefficients satisfying
$\sqrt{T/\log n}\min_s|\beta_{s,T}|\to\infty$.

In the local block example, standardise the mutually orthogonal signals
and the proxies to unit variance. Write a proxy as
\[
x_{jt}=r_jx_{s(j),t}+\sqrt{1-r_j^2}\,e_{jt},\qquad
0<|r_j|\leq\bar r<1,
\]
with $e_j$ orthogonal to all signals. For
$A\subsetneq\mathcal S_s$, a proxy associated with a selected signal
has $c_{j,A,T}=0$. Otherwise it is orthogonal to $x_A$ and
$\rho_{j,A,T}=r_j\rho_{s(j),A,T}$. The increasing map used in the
proof of Theorem~\ref{th3} therefore gives
\[
\max_{s\notin A}|\psi_{s,A,T}|-
\max_{j\in{P}}|\psi_{j,A,T}|
\geq(1-\bar r)\max_{s\notin A}|\rho_{s,A,T}|
\geq C_2\min_s|\beta_{s,T}|
\]
for a constant $C_2>0$; signal maxima here range over
$\mathcal S_s\setminus A$. A covered moment matrix has an identity
signal block and at most one proxy with a uniformly positive residual
variance. Its eigenvalues are therefore uniformly bounded above and
below. \ref{ass_proj} follows from these bounds and the positive error variance.
No bound on the covariance matrix of all proxies is needed.

For the equicorrelated design, put $a=|A|$, $d=1-r$, and
\[
\tau_a=\frac{r(1-r)}{1-r+ar}.
\]
After projection on $a$ signals, every remaining signal and proxy has
variance $d+\tau_a$, and any two distinct remaining variables have
covariance $\tau_a$. Consequently,
\[
c_{s,A,T}=d\beta_{s,T}
 +\tau_a\sum_{\ell\in\mathcal S_s\setminus A}\beta_{\ell,T},
\qquad
c_{j,A,T}=\tau_a\sum_{\ell\in\mathcal S_s\setminus A}\beta_{\ell,T}
\quad(j\in{P}).
\]
If the sum of the remaining coefficients is nonzero, at least one of
those coefficients has the same sign as the sum. Its covariance exceeds
the proxy covariance in absolute value by $d|\beta_{s,T}|$. If the sum
is zero, every proxy covariance is zero and the same conclusion holds
for any remaining signal. Thus the largest absolute signal covariance
exceeds the proxy covariance in absolute value by at least
$d\min_{s\in\mathcal S_s\setminus A}|\beta_{s,T}|$.
All candidate variances are equal, so the partial correlation gap is
at least this quantity divided by $\sqrt{(d+\tau_a)v_{y,A,T}}$.
The absolute studentised score gap is at least as large. Since
$d\geq1-\bar r>0$, this proves \eqref{cond1}. Every covered moment matrix has dimension
at most $k+1$, smallest eigenvalue at least $1-\bar r$, and bounded
largest eigenvalue. The outcome residual variance is bounded below
by the error variance, so \ref{ass_proj} and \ref{reg} hold uniformly even when the proxy
count grows.

This conclusion does not extend to arbitrary factor loadings. For example, take three
signals $x_{st}=(f_t+e_{st})/\sqrt2$, a proxy
$x_{pt}=(7f_t+e_{pt})/\sqrt{50}$, unit variance uncorrelated components,
and $y_t=x_{1t}+x_{2t}+x_{3t}+u_t$, with an orthogonal unit variance
error. At the first stage the signal scores are $2/\sqrt3$, whereas
the proxy score is $21/\sqrt{259}>2/\sqrt3$. The covariance matrix
is positive definite, but \eqref{cond1} fails. A common factor structure
alone therefore does not imply exact recovery.

In this unequal loading example the three signal covariance matrix has diagonal entries one and off diagonal entries $1/2$.
The proxy has variance one and covariance $7/10$ with each signal.
The signal block has eigenvalues $1/2,1/2,2$, and the proxy's variance
after projection on all three signals is
\[
1-\left(\frac{7}{10}\right)^2\frac{3}{2}
=\frac{53}{200}>0.
\]
Thus the full four variable covariance matrix is positive definite;
all of its principal submatrices satisfy the fixed dimensional 
bounds. With a unit variance orthogonal error, $v_{y,\varnothing,T}=7$,
$c_{s,\varnothing,T}=2$, and $c_{p,\varnothing,T}=21/10$.
Since $\psi_j=c_j/\sqrt{v_jv_y-c_j^2}$, these moments give the scores stated above. Independent noise variables can be
added as $n$ grows without changing these comparisons or the local bounds.

\subsubsection{Proof of Corollary~\ref{cor:growing_signals}}\label{app:growing_signals_proof}

Here $K_T=k_T$. We first establish the moment bound for the growing signal
count. Assumption~\ref{ass1A_mixingall} supplies uniform mixing bounds for every product of
two primitive components of
$(\boldsymbol z_t',x_{1t},\ldots,x_{nt},u_t)'$. Assumption~\ref{ass2} supplies
uniform exponential moments of order $\min(s_x,s_u)/2$ for these products.
Lemma~\ref{lem1} therefore applies with constants independent of the
component indices and $T$. A union bound makes the probability that any
centred primitive sample average exceeds $L\sqrt{\log n/T}$ in absolute value at most
\[
 C(n+\zeta+2)^2 f_T(2,\gamma,c,L\sqrt{\log n})=o(1),\qquad
 \gamma=\frac{\min(s_x,s_u)}{2+\min(s_x,s_u)},
\]
for a sufficiently large fixed $L$, since $n=O(T^\kappa)$.
On the resulting event, products involving
\[
 w_t=y_t-\mathbf a'\boldsymbol z_t
     =\sum_{s\in\mathcal S_s}\beta_{s,T}x_{st}+u_t
\]
expand into sums of primitive products. The sums of absolute coefficients
are uniformly bounded by the assumed bound on
$\sum_s|\beta_{s,T}|$, also for $w_t^2$. Thus, after increasing $L$, the
same event gives
\[
 \varepsilon_T\leq L\sqrt{\log n/T},
\]
with $\varepsilon_T$ defined as in Lemma~\ref{lem:score}. The argument uses
only the local mixing assumption; it does not require mixing of a growing
candidate vector. Assumption~\ref{ass2} and the coefficient bound also give a
uniformly bounded second moment for $w_t$.

On this event, every matrix estimation error has spectral norm at most
$C(1+k_T)\varepsilon_T=o(1)$. Assumption~\ref{reg} therefore gives uniform sample
inverse bounds. The population projection of $w_t$ has bounded $\ell_2$
norm, since its second moment is bounded and the conditioning matrix
has a smallest eigenvalue uniformly bounded away from zero. The exact Schur complement
calculation in the proof of Lemma~\ref{lem:tstat} then bounds all relevant
second moment errors by
\[
 C\{(1+k_T)\varepsilon_T+(1+k_T)^2\varepsilon_T^2\}.
\]
Assumption~\ref{ass_proj} preserves the residual variance bounds and bounds the derivative
of $r/\sqrt{1-r^2}$. Moreover, $d_A\leq\zeta+k_T+1=o(T)$ and
\[
 |\sqrt{T-d_A}-\sqrt T|\leq C(1+k_T)/\sqrt T.
\]
The exact statistic identity consequently gives, on the same event,
\[
 \max_{(A,j)\in\mathcal C_T^{\mathrm{ER}}}
 |t_{j\mid A}-\sqrt T\,\psi_{j,A,T}|
 \leq B(1+k_T)\sqrt{\log n}
\]
for a fixed $B<\infty$ and all sufficiently large $n,T$. This one primitive
event covers every signal subset.

At any proper signal subset, \eqref{eq:growing_signal_gap} makes the
largest absolute signal statistic exceed the largest absolute proxy statistic,
with the latter maximum set to zero if there are no proxies: subtracting
twice the displayed error bound still leaves a positive gap. Since the
proxy maximum is nonnegative, the signal maximum also exceeds
$B(1+k_T)\sqrt{\log n}$. Noise candidates have zero population scores by
the noise definition and \ref{ass1}, so none has a larger absolute statistic. Finally,
$k_T=o(n)$ implies
\[
 \frac{c_{\alpha_0}(n-|A|,\delta)}{\sqrt{\log n}}
 \longrightarrow\sqrt{2\delta}
 \quad\text{uniformly over }A\subseteq\mathcal S_s.
\]
The largest absolute signal statistic therefore passes the admission test, including when
$A=\varnothing$. Induction gives $k_T$ successive signal admissions.
This argument proves that no noise is admitted; it does not assert that
all its intermediate statistics lie below the threshold.

For terminal stopping, put $A=\mathcal S_s$ and write
$\widehat Q_A=\widehat E_T(q_{A,t}q_{A,t}')$ and
\[
 e_j=\widehat E_T(q_{A,t}x_{jt})
       -\widehat Q_A\gamma_{x_j,A,T}.
\]
The terminal projection bound, population projection orthogonality, and \ref{ass1}
imply, uniformly over the remaining candidates,
\[
 \|e_j\|_2\leq C\sqrt{1+k_T}\,\varepsilon_T,\qquad
 \|\widehat E_T(q_{A,t}u_t)\|_2
       \leq C\sqrt{1+k_T}\,\varepsilon_T,\qquad
 |\widehat E_T(x_{jt}^A u_t)|\leq C\varepsilon_T.
\]
Frisch--Waugh--Lovell gives the exact numerator identity
\[
 T^{-1}\boldsymbol x_j'M_A\boldsymbol u
 =\widehat E_T(x_{jt}^Au_t)
       -e_j'\widehat Q_A^{-1}\widehat E_T(q_{A,t}u_t).
\]
Its absolute value is at most
$C\{\varepsilon_T+(1+k_T)\varepsilon_T^2\}$. The candidate residual
variance is uniformly bounded below. At this terminal set the outcome
residuals are those from projecting $u_t$, and
\[
 T^{-1}\boldsymbol u'M_{A,j}\boldsymbol u
 =\widehat E_T(u_t^2)
 -\widehat E_T(r_{A,j,t}u_t)'
       \widehat G_{A,j}^{-1}\widehat E_T(r_{A,j,t}u_t)
 \geq\sigma^2-\varepsilon_T-C(1+k_T)\varepsilon_T^2.
\]
Here $M_{A,j}$ removes the controls, the signals in $A$, and candidate $j$,
and $\widehat G_{A,j}$ is the corresponding sample moment matrix. Assumption~\ref{ass_proj}
at $A=\mathcal S_s$ bounds $\sigma^2=v_{y,A,T}$ away from zero, so the
lower bound is uniformly positive for all sufficiently large $n,T$. The conventional
statistic is therefore bounded by
\[
 \max_{j\notin\mathcal S_s}|t_{j\mid\mathcal S_s}|
 \leq C\sqrt T\{\varepsilon_T+(1+k_T)\varepsilon_T^2\}
 \leq B_0\sqrt{\log n}
\]
on the same event, for a fixed $B_0<\infty$. Choose
$\delta>B_0^2/2$. The final threshold then exceeds this bound and BMT
stops. This proves the two probability conclusions, with the remaining claims
following from the bounded measure argument in the proof of
Theorem~\ref{th3}.

Under the stronger bounded $\ell_1$ projection condition, the all subset
Schur complement error improves to
$C\{\varepsilon_T+(1+k_T)\varepsilon_T^2\}$. The maintained inversion
condition and degrees of freedom bound then give a uniform statistic
remainder bounded by $B\sqrt{\log n}$ on the same event. Repeating the
ordering and stopping arguments with \eqref{cond1} proves the sharper
alternative.

For an example with a growing number of signals, take no controls and independent observations. Let the signals, regression
error, and all idiosyncratic components be independent centred standard
Gaussians, and let each proxy have the form
\[
 x_{jt}=r_jx_{s(j),t}+\sqrt{1-r_j^2}\,e_{jt},\qquad
 0<|r_j|\leq\bar r<1,
\]
with $s(j)\in\mathcal S_s$. Independent standard Gaussian noise variables
may also be included. Set $\beta_{s,T}=b/k_T$ for every signal, with fixed $b>0$. Every
covered matrix contains an identity signal block and at most one
proxy, so its eigenvalues are uniformly bounded above and below even as
$k_T$ grows. The coefficient $\ell_1$ norm is $b$, each candidate projection
has at most one nonzero coefficient, and all relevant outcome residual
variances lie between positive fixed constants. Thus \ref{ass1}--\ref{ass_proj} and \ref{reg} and the
terminal projection condition hold uniformly.

After any proper signal subset, a remaining signal has partial correlation
of order $1/k_T$, while each proxy's is either zero or $r_j$ times that of
its associated remaining signal. The calculation for the local block
example after Theorem~\ref{th4} therefore bounds the score gap below by
$C/k_T$. The general condition \eqref{eq:growing_signal_gap} holds when
\[
 k_T=o\!\left((T/\log n)^{1/4}\right),\qquad k_T=o(n).
\]
This example also satisfies the stronger all subset projection condition:
the outcome projection on any signal subset has coefficient $\ell_1$ norm
at most $b$. The sharper alternative hence permits
\[
 k_T=o\!\left(\sqrt{T/\log n}\right),\qquad k_T=o(n).
\]
These are sufficient rates from the two bounds, not an optimality claim.

\begin{proposition}
\label{prop}Consider the DGP \eqref{dgp1} under the assumptions of
Theorem~\ref{th1}. In the population procedure, replace the sample statistics
by the population scores $|\psi_{j,A,T}|$, admit their maximiser if its score
is positive, and stop otherwise. The number $J_0$ of admissions up to and
including the last true signal satisfies $k\leq J_0\leq k+k^\ast$.
\end{proposition}

\subsubsection{Proof of Proposition \ref{prop}\label{prop1}}

Let $A\subseteq\mathcal S_s\cup{P}$ be the indices already selected,
and suppose that $S=\mathcal S_s\setminus A$ is nonempty. The vector of
remaining signal covariances is
\[
(c_{s,A,T})_{s\in S}
=\overline E_T(\boldsymbol x_{S,t}^{A}\boldsymbol x_{S,t}^{A\prime})
 \boldsymbol\beta_S.
\]
The residual covariance matrix is positive definite by \ref{reg}. Since
$\boldsymbol\beta_S\neq\boldsymbol0$, at least one remaining signal has a
nonzero covariance and hence a positive absolute population score. The
noise scores are zero by the orthogonality restriction in Section~2 and \ref{ass1}.
Thus the population procedure cannot stop before all signals have entered
and cannot admit noise. Each admission adds one distinct signal or proxy,
which gives the stated bounds. Once all signals have entered, the outcome
residual is $u_t$, and all remaining population scores are zero by \ref{ass1}.

\subsubsection{Proof of Theorem~\ref{th1}}\label{app:am_selection_proof}

Use Lemmas~\ref{lem:score}--\ref{lem:tstat} on
$\mathcal C_T^{\mathrm{AM}}$, for which $K_T=k+k^\ast$ is fixed.
For a sufficiently large fixed $L$, the event
$\varepsilon_T\leq L\sqrt{\log n/T}$ has probability tending to one.
Indeed, its complement has probability at most
\[
C(n+\zeta+2)^2
 f_T(2,\gamma,c,L\sqrt{\log n})=o(1).
\]
On this event, for all sufficiently large $n,T$, all covered sample
regressions are well defined and
\[
\max_{(A,j)\in\mathcal C_T^{\mathrm{AM}}}
 |t_{j\mid A}-\sqrt T\,\psi_{j,A,T}|
 \leq B\sqrt{\log n}
\]
for a fixed $B<\infty$. Take $\delta>B^2/2$. Uniformly for
$n-k-k^\ast\leq m\leq n$, the threshold divided by $\sqrt{\log n}$
tends to $\sqrt{2\delta}>B$.

First consider a noise candidate $j\in\mathcal S_d$ and any
$A\subseteq\mathcal S_s\cup{P}$. After projection on the
controls, $x_j^{\varnothing}$ is orthogonal to every component of the
signal and proxy block. Further projection on $x_A$ therefore leaves
it unchanged. Assumption~\ref{ass1} also gives orthogonality to $u_t$, and hence
$c_{j,A,T}=\psi_{j,A,T}=0$. Thus every such noise statistic lies below
the threshold on the same event. This argument applies to any covered
history, including the history immediately before a possible first
noise admission.

Next suppose that $A$ omits at least one signal, and write
$S=\mathcal S_s\setminus A$. For this calculation let
\[
V_{S\mid A,T}=\overline E_T(\boldsymbol x_{S,t}^{A}
                              \boldsymbol x_{S,t}^{A\prime}).
\]
The full signal--proxy moment matrix with the controls satisfies \ref{reg}.
Its Schur complement for $S$ after $A$ therefore has smallest
eigenvalue bounded below by a positive constant $c_G$.
This also follows directly by minimising the full quadratic form
over the coefficients on $q_A$. The DGP and score orthogonality give
\[
(c_{s,A,T})_{s\in S}=V_{S\mid A,T}\boldsymbol\beta_S,
\qquad
\|(c_{s,A,T})_{s\in S}\|_2
 \geq c_G\|\boldsymbol\beta_S\|_2.
\]
Since $k$ is fixed and $\min_{s\in\mathcal S_s}|\beta_s|\geq b>0$,
at least one remaining signal has $|c_{s,A,T}|\geq c_Gb/\sqrt{k}$.
The upper variance bounds in \ref{ass_proj} imply that its $|\psi_{s,A,T}|$ is
bounded below by a constant independent of $A,n,T$. Its statistic
therefore exceeds every stagewise threshold for all sufficiently large
$n,T$, because $\log n=o(T)$.

Starting from $A=\varnothing$, BMT consequently cannot stop while a
signal remains, and cannot admit noise. Each successful stage adds a
distinct signal or proxy, so all signals have entered after at most
$k+k^\ast$ admissions. Once $A$ contains all signals, the population
outcome residual is $u_t$. Every remaining candidate then has zero
partial coefficient by Assumption~\ref{ass1}. The uniform bound applies to this
terminal test as well, and BMT stops. Hence the event above implies
$\mathcal A_0$ and $k\leq\widehat J\leq k+k^\ast$, proving
\eqref{Phat}--\eqref{pA0} and the stopping assertion. No assumption
is imposed on histories after an erroneous admission.

Finally, all three performance measures lie in $[0,1]$. On
$\mathcal A_0$, they equal $1,0,0$, respectively. Consequently,
\[
0\leq1-E(TPR_{n,T})\leq\Pr(\mathcal A_0^c),\qquad
E(FPR_{n,T}^{(0)})+E(FDR_{n,T}^{(0)})
 \leq2\Pr(\mathcal A_0^c).
\]
These bounds and Markov's inequality prove
\eqref{tprn0}--\eqref{fprn0}.

\subsubsection{Proof of Theorem~\ref{th2}}\label{app:am_estimation_proof}

Consider the finite collection of deterministic sets
$\mathcal S_s\subseteq A\subseteq\mathcal S_s\cup{P}$.
Let $Q_A$ denote the sample matrix with rows $q_{A,t}'$, and let
$b_A$ be its true coefficient vector, containing $\mathbf a$, the
signal coefficients, and zeros on the proxies. For each such set,
\[
\boldsymbol y=Q_A b_A+\boldsymbol u,\qquad
\widehat b_A-b_A
 =(T^{-1}Q_A'Q_A)^{-1}T^{-1}Q_A'\boldsymbol u.
\]
The inverse is $O_p(1)$ uniformly over this finite collection by \ref{reg}
and Lemma~\ref{lem:score}. Each component of $q_{A,t}u_t$ is a
martingale difference by \ref{ass1}. Its second moment is uniformly bounded by
\ref{ass2}, so its sample average is $O_p(T^{-1/2})$. There are only a fixed
number of controls, signals, and proxies. It follows that
\[
\max_{\mathcal S_s\subseteq A\subseteq\mathcal S_s\cup{P}}
 \|\widehat b_A-b_A\|_2=O_p(T^{-1/2}).
\]
This bound concerns estimation errors and does not require bounded
coefficients on normalised controls.

Let $P_A=Q_A(Q_A'Q_A)^{-1}Q_A'$. The same score and inverse bounds give
$\boldsymbol u'P_A\boldsymbol u=O_p(1)$ uniformly over the collection.
Thus its least squares residual vector satisfies
\[
T^{-1}\|\widehat{\boldsymbol u}_A\|^2
 =T^{-1}\boldsymbol u'\boldsymbol u
    -T^{-1}\boldsymbol u'P_A\boldsymbol u
 =T^{-1}\boldsymbol u'\boldsymbol u+O_p(T^{-1}).
\]
Under \ref{ass1}--\ref{ass2}, the covariances of $u_t^2$ are uniformly summable. Hence
\[
\operatorname{Var}\!\left(T^{-1}\sum_{t=1}^T u_t^2\right)=O(T^{-1}).
\]
Together with $E(u_t^2)=\sigma^2$, this gives
\[
T^{-1}\boldsymbol u'\boldsymbol u=\sigma^2+O_p(T^{-1/2}).
\]

On $\mathcal A_0$, the selected set belongs to this finite collection.
The preceding uniform bounds apply to its residuals and candidate
coefficients without conditioning their distributions on selection.
For either scaled estimation error, its tail probability is bounded
by the corresponding finite collection tail probability plus
$\Pr(\mathcal A_0^c)=o(1)$. This proves
\eqref{forecast00} and \eqref{fnormbeta}. A rate for
$\Pr(\mathcal A_0^c)$ is not needed for these $O_p$ conclusions.

The deterministic class covers every subset of signals and proxies,
including the set used in the stopping test. The population bounds in
Assumption~\ref{reg} and uniform second moment convergence give
the required sample invertibility.
On samples where a required regression is undefined, the corresponding
statistics and estimates may be set to zero. The proof shows that this
event has probability $o(1)$, so these assignments do not affect the
stated convergence rates in probability.

\begin{remark}[On the choice of $\delta$ in Theorem~\ref{th1}]
For the fixed dimensional approximating model class, the proof of
Lemma~\ref{lem:tstat} gives a constant $B$ such that the uniform
$t$ statistic error is at most $B\sqrt{\log n}$ with probability
tending to one. Since $c_{\alpha_0}(n-O(1),\delta)/\sqrt{\log n}\to\sqrt{2\delta}$,
one may take $\delta_0=B^2/2$. This is a sufficient threshold condition;
it does not provide a universal numerical choice of $\delta$ or a
finite sample family wise error bound at level $\alpha_0$.
\end{remark}

If $T\asymp n^{\kappa_1}$, the rates in \eqref{forecast00} and \eqref{fnormbeta}
are $O_p(n^{-\kappa_1/2})$ and $O_p(n^{-\kappa_1})$, respectively.

\begin{remark}[Early stopping and comparison with OCMT]
\label{rem:am_early_stopping}
Once all signals have entered, every remaining proxy has zero population
partial coefficient. Theorem~\ref{th1} therefore implies that no further
proxy is admitted with probability tending to one. In contrast, an
OCMT marginal screening step can admit a proxy whose marginal coefficient
is nonzero. The theorem does not require signals to precede every proxy,
or establish a strictly positive probability that some proxy remains
unselected; those conclusions require a comparison of their scores.
\end{remark}

\subsubsection{Proof of Theorem 5}\label{app:post_selection_proof}

Where a required regressor cross product matrix is singular, set the
coefficient and covariance estimates to zero, using the same rule for the
oracle and selected regressions. The regression identities below are
applied on samples where the required inverses exist.

Write $\mathbf Z$ for the always in control matrix and
$M_Z=I-\mathbf Z(\mathbf Z'\mathbf Z)^{-1}\mathbf Z'$, with $M_Z=I$ when
there are no controls. The oracle signal estimator is
\begin{equation}\label{eq:oracle_def}
 \widehat{\boldsymbol\beta}^{\,o}_{\mathcal S_s}
 =(\mathbf X_{\mathcal S_s}'M_Z\mathbf X_{\mathcal S_s})^{-1}
   \mathbf X_{\mathcal S_s}'M_Z\boldsymbol y.
\end{equation}
This is the signal block of OLS on $q_t$, using the control basis specified
in the theorem throughout. The coefficient on that basis need not be
bounded, since the always in component is fitted exactly.

Put $s_t=q_tu_t$ and $d=\zeta+k$. Each coordinate of $s_t$ is a control score
or a true signal score. The corresponding restrictions in Assumption~\ref{ass1}
give $E(s_t\mid\mathcal F_{t-1})=\boldsymbol 0$ almost surely. Since $s_t$ is
$\mathcal F_t$ measurable, it is a martingale difference with respect to
$(\mathcal F_t)$. This uses the score restrictions, not the error restriction
conditional on $\mathcal F_{t-1}^{u}$. Since $d$ is fixed, \ref{ass2} gives uniform moments of
every fixed order for $q_t$ and $s_t$. Each entry of $q_tq_t'$ and
$s_ts_t'$ involves at most two candidate coordinates, the controls and
$u_t$. Thus \ref{ass1A_mixingall} and \ref{ass2} give uniformly summable covariances for their centred
entries. It follows that
\[
 \widehat Q^o:=T^{-1}\sum_{t=1}^Tq_tq_t'
       =Q_{\mathcal S_s,T}+O_p(T^{-1/2})\xrightarrow{p}Q,
 \qquad
 T^{-1}\sum_{t=1}^Ts_ts_t'=\Omega_T+O_p(T^{-1/2})\xrightarrow{p}\Omega.
\]
The entries of
$s_ts_t'-E(s_ts_t'\mid\mathcal F_{t-1})$ are also martingale differences
with uniformly bounded second moments. Consequently,
\[
 T^{-1}\sum_{t=1}^T E(s_ts_t'\mid\mathcal F_{t-1})\xrightarrow{p}\Omega.
\]
For every $\varepsilon>0$ and some $\eta>0$,
\[
 E\left\{T^{-1}\sum_{t=1}^T
 E\bigl(\|s_t\|^2\mathbf1\{\|s_t\|>\varepsilon\sqrt T\}
       \mid\mathcal F_{t-1}\bigr)\right\}
 \leq C\varepsilon^{-\eta}T^{-\eta/2}\longrightarrow0.
\]
Markov's inequality gives conditional Lindeberg. The martingale central
limit theorem, applied to each fixed linear combination, therefore yields
$T^{-1/2}\sum_ts_t\xrightarrow{d}N(0,\Omega)$.
The joint oracle coefficient error is
$(\widehat Q^o)^{-1}T^{-1}\sum_ts_t$. Taking its signal block proves
\begin{equation}\label{eq:oracle_limits}
 \widehat{\boldsymbol\beta}^{\,o}_{\mathcal S_s}
       -\boldsymbol\beta_{\mathcal S_s}=O_p(T^{-1/2}),\qquad
 \sqrt T(\widehat{\boldsymbol\beta}^{\,o}_{\mathcal S_s}
       -\boldsymbol\beta_{\mathcal S_s})
       \xrightarrow{d}N(0,V_{\mathcal S_s}).
\end{equation}

Let $b_T$ collect the control and signal coefficients, and let
$\widehat b^o$ be their joint oracle OLS estimator. Then
$h_T=\widehat b^o-b_T=O_p(T^{-1/2})$ and the oracle residual is
$\widehat u_t^o=u_t-q_t'h_T$. Hence
\begin{align*}
 &\left\|T^{-1}\sum_tq_tq_t'
          \{(\widehat u_t^o)^2-u_t^2\}\right\|\\
 &\quad\leq 2\|h_T\|T^{-1}\sum_t\|q_t\|^3|u_t|
       +\|h_T\|^2T^{-1}\sum_t\|q_t\|^4=o_p(1),
\end{align*}
because both sample averages are $O_p(1)$ by \ref{ass2}. The oracle HC covariance
estimate therefore converges to $Q^{-1}\Omega Q^{-1}$.
Also,
\[
 T^{-1}\sum_t(\widehat u_t^o)^2
 =T^{-1}\sum_tu_t^2-h_T'\widehat Q^o h_T\xrightarrow{p}\sigma^2.
\]
Dividing the residual sum of squares by $T-d$ gives the same limit.
Thus the ordinary covariance estimate converges to $\sigma^2Q^{-1}$,
which is the required covariance when $\Omega=\sigma^2Q$.

For the HAC alternative, define in the final selected regression
\[
 \widehat s_t=\widehat q_t\widehat u_t,\qquad
 \widehat\Gamma_h=T^{-1}\sum_{t=h+1}^T
                      \widehat s_t\widehat s_{t-h}',\qquad 0\leq h\leq L_T,
\]
and replace $\widehat\Omega_{\rm HC}$ by
\[
 \widehat\Omega_{\rm HAC}=\widehat\Gamma_0+
 \sum_{h=1}^{L_T}\left(1-\frac{h}{L_T+1}\right)
             (\widehat\Gamma_h+\widehat\Gamma_h').
\]
To check consistency in the oracle regression, first use $s_t$ in place
of $\widehat s_t$ and denote the resulting lag matrix by $\Gamma_h^o$.
For each $h\geq1$, the entries of $s_ts_{t-h}'$ are martingale differences
in $t$. Their second moments are uniformly bounded, so
$E\|\Gamma_h^o\|^2\leq C/T$, uniformly over $1\leq h<T$.
The sum of the nonzero lag matrices is therefore
$O_p(L_T/\sqrt T)=o_p(1)$. Replacing $s_t$ by $q_t\widehat u_t^o$ changes
each lag matrix by at most
$C_T(\|h_T\|+\|h_T\|^2)$, uniformly in $h$, where $C_T=O_p(1)$.
This follows from Cauchy--Schwarz and the full sample averages of
$\|q_t\|^4$ and $\|q_t\|^2u_t^2$. Summing over the bandwidth still gives
$o_p(1)$. Thus the oracle HAC estimate also converges to $\Omega$.
No serial correlation of the scores is allowed by this argument.

Let $\widehat S$ denote the candidate set selected by BMT. Then
\begin{equation}\label{eq:A1_def}
 \mathcal A_1=\{\widehat S=\mathcal S_s\}.
\end{equation}
On this event the final regression and the oracle regression contain the
same signals and controls. With the same zero assignments in both
regressions, we obtain
\begin{equation}\label{eq:post_equals_oracle}
 \boldsymbol{\widetilde\beta}_{\mathcal S_s}
       =\widehat{\boldsymbol\beta}^{\,o}_{\mathcal S_s},\qquad
 \widetilde\beta_i=0\quad(i\notin\mathcal S_s).
\end{equation}
The oracle sample second moment matrix converges in probability to the
positive definite matrix $Q$, so it is nonsingular with probability tending
to one. Since $\Pr(\mathcal A_1)\to1$, the selected regression is also
nonsingular with probability tending to one. The zero assignments therefore
do not affect the preceding probability and distribution limits. For
$\varepsilon>0$,
\begin{equation}\label{eq:consistency_bound}
 \Pr\bigl(\|\boldsymbol{\widetilde\beta}_{\mathcal S_s}
                  -\boldsymbol\beta_{\mathcal S_s}\|>\varepsilon\bigr)
 \leq\Pr\bigl(\|\widehat{\boldsymbol\beta}^{\,o}_{\mathcal S_s}
                  -\boldsymbol\beta_{\mathcal S_s}\|>\varepsilon\bigr)
          +\Pr(\mathcal A_1^c)\longrightarrow0.
\end{equation}
For an inactive index, the corresponding probability is at most
$\Pr(\mathcal A_1^c)$. For any bounded Lipschitz function $g$,
\begin{align}
 &\left|Eg\!\left(\sqrt T(\boldsymbol{\widetilde\beta}_{\mathcal S_s}
                       -\boldsymbol\beta_{\mathcal S_s})\right)
 -Eg\!\left(\sqrt T(\widehat{\boldsymbol\beta}^{\,o}_{\mathcal S_s}
                       -\boldsymbol\beta_{\mathcal S_s})\right)\right|
 \nonumber\\
 &\hspace{25mm}\leq2\|g\|_\infty\Pr(\mathcal A_1^c)\longrightarrow0.
 \label{eq:BL_bound}
\end{align}
This proves the normal limit without a rate requirement on the selection
error probability. Each post selection covariance estimate likewise equals
its oracle counterpart on $\mathcal A_1$. Its probability of differing from
the stated limit by more than $\varepsilon$ is bounded by the corresponding
oracle probability plus $\Pr(\mathcal A_1^c)$. This proves covariance
consistency, including the homoskedastic and Bartlett alternatives under
their stated conditions.

\subsubsection{Proof of Theorem~\ref{thm:Th7} and the composite proxy calculations}
\label{app:general_proof}
For any $A\subseteq S$, population residualisation gives
\[
 y_t^A=\sum_{i\in R}\beta_i x_{it}+u_t,
 \qquad
 x_{pt}^A=\sum_{i\in R}c_{pi}x_{it}+\sqrt{1-\|c_p\|_2^2}\,e_{pt}.
\]
Thus $\operatorname{Var}(x_{pt}^A)=1-\|c_{p,A}\|_2^2$ and
$\operatorname{Cov}(x_{pt}^A,y_t^A)=c_{p,R}'\beta_R$. Each remaining signal has variance one and covariance $\beta_i$ with $y_t^A$, proving \eqref{eq:composite_scores}. The finite signal--proxy covariance matrix is positive definite: its Schur complement given the signals is the diagonal matrix with entries $1-\|c_p\|_2^2>0$. Independent noise candidates add an identity block.

The local moment and residual variance bounds hold uniformly over the signal subsets and candidates. Gaussianity gives the required tail bounds, independence over time gives the local mixing and score conditions, and the independent noise block satisfies the residual orthogonality restriction. The intercept has fixed dimension. Since $k$, $P$, the coefficients and the proxy loadings are fixed, strict ordering on the finite collection of signal--proxy comparisons gives a positive minimum studentised score gap. It therefore implies \eqref{cond1}. Theorem~\ref{th4} gives signal admission, noise exclusion and terminal stopping with one sufficiently large fixed $\delta$. This also explains why the growing noise pool is specified even in a small signal--proxy example.

For completeness, the tuning specific Lasso conclusion follows directly from the active and inactive optimality equations. Centre the observations to remove the unpenalised intercept and write $Q_T=X_S'X_S/T$. On the event of exact support recovery, the active equations give
\[
 \widehat\beta_S-\beta_S
 =Q_T^{-1}\{X_S'u/T-\lambda_T\operatorname{sign}(\widehat\beta_S)\}.
\]
Since $Q_T\to I_k$ and $k$ is fixed, the right side is $O_p(T^{-1/2}+\lambda_T)$ on this event, irrespective of the fitted signs. Fixed nonzero coefficients then imply the true signs with probability tending to one on the event. The inactive equation for proxy $p$ requires
\[
 \left|\frac{X_p'M_Su}{T\lambda_T}
       +\frac{X_p'X_S}{T}Q_T^{-1}\operatorname{sign}(\beta_S)\right|\leq1,
 \qquad M_S=I-X_S(X_S'X_S)^{-1}X_S'.
\]
The first term is $o_p(1)$ because $\sqrt T\lambda_T\to\infty$, and the second tends to $c_p'\operatorname{sign}(\beta_S)$. A strict absolute limit above one makes this inequality impossible with probability tending to one. This proves the stated failure of exact support recovery. It does not establish that the proxy itself must be selected under every tuning sequence.

To verify \eqref{eq:composite_g_interval}, put $a=g/(k-1)$. If $A$ does not contain signal~1 and contains $m$ weak signals, the ordering condition is
\[
 \{c_1+\epsilon a(k-1-m)\}^2+ma^2<1.
\]
As a function of $m\in[0,k-1]$ its left side is convex, so it suffices to check the endpoints. These are $(c_1+\epsilon g)^2<1$ and $c_1^2+g^2/(k-1)<1$. Both follow from \eqref{eq:composite_g_interval}; the latter also proves positive definiteness. If $A$ contains signal~1 and $m$ weak signals, put $q=k-1-m\geq1$. Ordering is equivalent to
\[
 a^2(q^2+m)<1-c_1^2.
\]
Here $q^2+m=q^2+k-1-q\leq(k-1)^2$, so $g^2<1-c_1^2$ suffices. This covers every proper signal subset. Finally $c_1+g>1$ violates the Lasso condition. Both upper endpoints in \eqref{eq:composite_g_interval} exceed $1-c_1$, establishing nonemptiness. With homogeneous coefficient magnitudes, the first stage identity stated in the text proves the converse. A strict proxy advantage persists in the sample and exceeds the entry threshold, so irreversible BMT cannot recover exactly in that case.

For Corollary~\ref{thm:wedge}, the three proper subsets are $\varnothing$, $\{1\}$ and $\{2\}$. Their largest signal scores are $1$, $\alpha$ and $1$, and their proxy scores are $\rho(1+\alpha)$, $\alpha\rho/\sqrt{1-\rho^2}$ and $\rho/\sqrt{1-\rho^2}$. The stated conditions make all three comparisons strict. The smallest eigenvalue of the three regressor covariance is $1-\sqrt2\rho>0$, and the same sign irrepresentable quantity is $2\rho$. Theorem~\ref{thm:Th7} now gives both conclusions.

\subsection{Other selection methods}
\label{app:other_comparisons}
We retain the fixed dimensional Gaussian signal--proxy design, independent growing noise pool, intercept and threshold regime of Section~\ref{subsec:stylised_comparison}, except for the near singular sequence specified in the knockoff subsection. The comparisons concern exact support recovery, which is a different target from sure screening, prediction or false discovery rate control.

For the equal loading case with one proxy $z_t$, write
$c=\rho\mathbf 1_k$, $\beta_1>\cdots>\beta_k>0$, and impose $0<\rho<1/\sqrt k$. For a proper subset $A\subsetneq S$, set $B_R=\sum_{i\in R}\beta_i$, $M_A=\max_{i\in R}\beta_i$, and define
\begin{equation}\label{eq:bmt_general}
 \rho_*:=\min\left\{\frac1{\sqrt k},\
 \min_{A\subsetneq S}\frac{M_A}{\sqrt{B_R^2+|A|M_A^2}}\right\}.
\end{equation}
Then $0<\rho<\rho_*$ ensures positive definiteness and BMT ordering on every signal subset. In particular, after signals $1,\ldots,j-1$ enter, the correct comparison is
\[
 \beta_j>\frac{\rho\sum_{i=j}^k\beta_i}{\sqrt{1-(j-1)\rho^2}}.
\]
The residual variance denominator is needed in this comparison. The statements describe a family of fixed $k$ designs, not a growing $k$ consistency result.

\subsubsection{Comparison with OCMT}
\label{app:ocmt_comparison}
OCMT tests individual candidates and retains all that exceed its multiple testing threshold; later stages can admit signals with zero initial net effects \citep{chu2018}. In the designs below the proxy already has a nonzero first stage net effect. Its first stage retention suffices to rule out exact recovery by the procedure that retains these admissions, without any claim about a subsequent pruning variant or the accuracy of prediction.

\begin{proposition}
\label{prop:bmt_ocmt_three_app}
Consider \eqref{eq:dgp_wedge} under the common sampling and noise pool conditions, with $0<\alpha<1$ and $0<\rho<1/\sqrt2$. Suppose the first stage OCMT threshold $c_T$ satisfies $c_T\to\infty$ and $c_T=o(\sqrt T)$. Then OCMT retains all three signal--proxy variables with probability tending to one. If also $\rho<1/(1+\alpha)$, BMT recovers $\{1,2\}$ with probability tending to one. Thus a separation interval is
\[
 \left(0,\min\{(1+\alpha)^{-1},2^{-1/2}\}\right).
\]
\end{proposition}
\begin{proof}
The three marginal covariances are $1$, $\alpha$ and $\rho(1+\alpha)$. All are fixed and positive, so their conventional marginal statistics diverge at rate $\sqrt T$ and exceed $c_T$. This is a retention statement for these three variables; the rate restrictions on $c_T$ alone need not exclude all growing noise candidates. BMT recovery follows from Corollary~\ref{thm:wedge}, with its own threshold conditions.
\end{proof}

\begin{theorem}
\label{thm:bmt_ocmt_high_app}
In the common fixed $k$ design with $k\geq2$, $\beta_1>\cdots>\beta_k>0$ and $c=\rho\mathbf 1_k$, impose $0<\rho<1/\sqrt k$ and the OCMT threshold conditions above. The proxy is retained by OCMT with probability tending to one. If $0<\rho<\rho_*$, BMT recovers $S$ with probability tending to one.
\end{theorem}
\begin{proof}
The proxy's marginal covariance with $y_t$ is $\rho\sum_{i=1}^k\beta_i>0$, so its first stage statistic diverges faster than the threshold. For BMT, \eqref{eq:bmt_general} gives positive definiteness and
\[
 \frac{\rho B_R}{\sqrt{1-|A|\rho^2}}<M_A
\]
for every proper $A$. Apply Theorem~\ref{thm:Th7}.
\end{proof}

\subsubsection{Comparison with Sure Independence Screening}
\label{app:sis_comparison}
SIS ranks absolute marginal correlations and retains a chosen number $d$ of variables \citep{fan2008}. We consider the choice $d=k$, equal to the true support size, so the results do not rule out sure screening with $d>k$ followed by another selector.

\begin{proposition}
\label{prop:bmt_sis_three_app}
In the design of Proposition~\ref{prop:bmt_ocmt_three_app}, let SIS retain the top two candidates. If $\rho>\alpha/(1+\alpha)$, its selected set is $\{1,3\}$ with probability tending to one. If also $\rho<\min\{(1+\alpha)^{-1},2^{-1/2}\}$, BMT recovers $\{1,2\}$ with probability tending to one. The resulting open interval is nonempty for $0<\alpha<1$.
\end{proposition}
\begin{proof}
The proxy's covariance $\rho(1+\alpha)$ exceeds the weaker signal's covariance $\alpha$, while the first signal's covariance is $1>\alpha$. All candidates have unit population variance. The maximum absolute sample correlation of the independent noise variables is $O_p(\sqrt{\log n/T})=o_p(1)$, so it cannot change this strict ranking. Corollary~\ref{thm:wedge} proves the BMT claim. Both stated upper endpoints exceed $\alpha/(1+\alpha)$.
\end{proof}

\begin{theorem}
\label{thm:bmt_sis_high_app}
In the design of Theorem~\ref{thm:bmt_ocmt_high_app}, top $k$ SIS excludes the weakest signal and includes the proxy with probability tending to one if
$\rho>\beta_k/\sum_{i=1}^k\beta_i$. BMT recovers $S$ with probability tending to one if $\rho<\rho_*$. The interval
\[
 \left(\frac{\beta_k}{\sum_{i=1}^k\beta_i},\rho_*\right)
\]
is nonempty.
\end{theorem}
\begin{proof}
The strict population ranking and the uniform vanishing of noise correlations prove the SIS statement. Theorem~\ref{thm:Th7} proves the BMT statement. To show nonemptiness, put $B=\sum_{i=1}^k\beta_i$. For $A=\varnothing$, $B/\beta_k>B/\beta_1$. For $a=|A|\geq1$, write $x=B_R/M_A\geq1$. Then
\[
 B/\beta_k\geq x+a>\sqrt{x^2+a}.
\]
Thus $\beta_k/B$ is below every subset bound in \eqref{eq:bmt_general}. It is also below $1/k<1/\sqrt k$, proving the claim.
\end{proof}

\subsubsection{Comparison with Gaussian Model-X knockoffs}
\label{app:knockoff_comparison}
Knockoffs compare original variables with synthetic counterparts using antisymmetric statistics to control the false discovery rate \citep{barber2015knockoffs,candes2018modelx}, which is a different target from exact support recovery. We examine Gaussian Model-X knockoffs with the equicorrelated construction and a signed marginal statistic; the near singular calculation below concerns this specific procedure and does not establish a BMT advantage.

Keep $k\geq2$ and $\beta_1>\cdots>\beta_k>0$ fixed, but replace the fixed equal loading by $0<\rho_T<1/\sqrt k$. The covariance matrix $\Sigma_T$ of the signals, proxy and independent noise candidates has smallest eigenvalue $1-\sqrt k\rho_T$. Construct Gaussian knockoffs independently across observations, conditionally independent of $y$ given $X$, with joint covariance
\[
 \begin{pmatrix}\Sigma_T&\Sigma_T-s_TI\\
 \Sigma_T-s_TI&\Sigma_T\end{pmatrix},
 \qquad s_T=\min\{2(1-\sqrt k\rho_T),1\}.
\]
For each original candidate $v$, use
\begin{equation}\label{eq:knockoff_stat_full}
 W_v=T^{-1/2}\sum_{t=1}^T(v_t-\widetilde v_t)y_t.
\end{equation}
At a fixed target $q\in(0,1)$, define
\begin{equation}\label{eq:tau_plus_def}
 \tau_+=\min\left\{t\in\{|W_v|:|W_v|>0\}:
 \frac{1+\#\{v:W_v\leq-t\}}{\#\{v:W_v\geq t\}\vee1}\leq q\right\},
\end{equation}
with $\min\varnothing=+\infty$, and
\begin{equation}\label{eq:KO_selected_model}
 \widehat S_{\mathrm{KO}}=\{v:W_v\geq\tau_+\}.
\end{equation}
The statistic changes sign when a variable and its knockoff are swapped and the other statistics remain unchanged. It is therefore an admissible antisymmetric statistic, although this alone gives no power guarantee.

\begin{theorem}
\label{thm:bmt_knockoff_local}
Suppose $T(1-\sqrt k\rho_T)\to c\in[0,\infty)$ and put
$\sigma_y^2=\sigma^2+\sum_{j=1}^k\beta_j^2$. Then
\[
 \limsup_{T\to\infty}\Pr(\widehat S_{\mathrm{KO}}=S)
 \leq\Phi\left(\frac{\beta_k\sqrt c}{\sigma_y}\right)<1.
\]
For $q<1/k$, exact recovery by knockoff$+$ is impossible at every sample size. Along the stated near singular sequence, the uniform lower bound required by the BMT recovery theorems fails. No BMT consistency conclusion follows from those theorems on this sequence.
\end{theorem}
\begin{proof}
Let $D_t=x_{kt}-\widetilde x_{kt}$ and $A_t=D_ty_t$. Joint Gaussian moments give
\[
 E(A_t)=s_T\beta_k,\qquad
 \operatorname{Var}(A_t)=2s_T\sigma_y^2+s_T^2\beta_k^2.
\]
The pair $(D_t/\sqrt{s_T},y_t)$ is jointly Gaussian with bounded covariance entries, so $(A_t-EA_t)/\sqrt{s_T}$ has a uniformly bounded fourth moment and variance tending to $2\sigma_y^2>0$. The rowwise independent triangular array Lyapunov theorem therefore gives
\[
 \frac{W_{x_k}}{\sqrt{s_T}}
 \ \Rightarrow\ N(\beta_k\sqrt{2c},2\sigma_y^2).
\]
The summands are products of Gaussian variables, not themselves Gaussian; the limit follows from the stated moment argument. Exact recovery requires $W_{x_k}>0$, giving the probability bound by continuity of the limit. If exactly $k$ variables are selected at a finite threshold, the knockoff$+$ ratio is at least $1/k$, proving impossibility when $q<1/k$.

Finally the local block containing all $k$ signals and the proxy has smallest eigenvalue $1-\sqrt k\rho_T\to0$, contrary to the fixed lower bound in Assumption~\ref{reg}. At the last signal only stage the proxy to signal score ratio is
$\rho_T/\sqrt{1-(k-1)\rho_T^2}\to1$. Thus the previous fixed gap argument cannot be used either.
\end{proof}

The bound can reflect the deliberately near singular design and the chosen feature statistic. With a fixed positive  bound, $s_T$ is bounded away from zero and the mean of $W_{x_k}$ grows as $\sqrt T$; this argument then gives no nontrivial upper bound below one. Neither the theorem nor the $q<1/k$ constraint is a failure of FDR control. We do not infer a general ranking of BMT and Model-X knockoffs from it.

\subsubsection{Comparison with Adaptive Lasso}
\label{app:adalasso_highdim_discussion}
Adaptive Lasso uses data dependent positive penalty weights \citep{zou2006adaptive}, which can change the proxy restriction substantially. Consider the objective
\[
 (2T)^{-1}\|y-Xb\|_2^2+\lambda_T\sum_j w_{j,T}|b_j|,
\]
with an unpenalised intercept. The following necessary condition applies when the weights have finite positive probability limits on the signals and the proxy, and does not cover a consistent pilot whose proxy weight diverges.

\begin{proposition}
\label{prop:adalasso_curse_support}
In the common composite proxy design with one proxy, suppose
$w_{i,T}\to_p a_i\in(0,\infty)$ for $i\in S$ and
$w_{p,T}\to_p a_p\in(0,\infty)$, while $\lambda_T\to0$ and
$\sqrt T\lambda_T\to\infty$. A necessary condition for exact support consistency is
\begin{equation}\label{eq:weighted_irrep_proxy}
 \frac{|\sum_{i\in S}c_i a_i\operatorname{sign}(\beta_i)|}{a_p}\leq1.
\end{equation}
If the left side is strictly above one, exact support probability tends to zero. For positive coefficients and equal loadings $c_i=\rho>0$, suppose further that
\begin{equation}\label{eq:proxy_weight_not_large}
 a_p\leq M\min_{i\in S}a_i,\qquad 0<M<\infty.
\end{equation}
Then consistency requires
\begin{equation}\label{eq:adalasso_support_curse}
 \rho\leq M/k.
\end{equation}
BMT recovers $S$ with probability tending to one whenever $0<\rho<\rho_*$, independently of these penalty weights.
\end{proposition}
\begin{proof}
On exact support, the active equations are those in the proof of Theorem~\ref{thm:Th7}, with the sign vector replaced by
$\operatorname{diag}(w_{S,T})\operatorname{sign}(\widehat\beta_S)$. Bounded weights and $\lambda_T\to0$ again give the true signs on this event with probability tending to one. The proxy equation, divided by $\lambda_T w_{p,T}$, converges to the quantity in \eqref{eq:weighted_irrep_proxy}; its score noise term is $o_p(1)$. A strict limit above one is impossible. In the positive equal loading case, the quantity is $\rho\sum_i a_i/a_p\geq\rho k/M$, proving the necessary bound. Equality is not settled by this argument. The BMT conclusion follows from \eqref{eq:bmt_general} and Theorem~\ref{thm:Th7}.
\end{proof}

\begin{corollary}
\label{cor:adalasso_geometric}
Let $\beta_j=b r^{j-1}$, where $b>0$, $0<r<1$, and $k\geq2$ is fixed. Under the assumptions of Proposition~\ref{prop:adalasso_curse_support}, including \eqref{eq:proxy_weight_not_large}, a sufficient separation interval is
\[
 \frac{M}{k}<\rho<\frac1{\sqrt{k+(1-r)^{-2}}}.
\]
It is nonempty whenever $k^2>M^2\{k+(1-r)^{-2}\}$. On this interval BMT recovers $S$ with probability tending to one, and the stated weighted estimator has exact support probability tending to zero.
\end{corollary}
\begin{proof}
For any remaining subset, its successive indices differ by at least one, so $B_R/M_A\leq\sum_{h=0}^{|R|-1}r^h<(1-r)^{-1}$. Since $|A|\leq k-1$, the displayed upper bound on $\rho$ is below every subset bound in \eqref{eq:bmt_general} and below $1/\sqrt k$. Proposition~\ref{prop:adalasso_curse_support} proves the claims.
\end{proof}

Positive definiteness forces equal proxy loadings below $1/\sqrt k$; an order one loading bound as $k\to\infty$ would be invalid. The corollary describes separate fixed $k$ experiments. A growing $k$ version would additionally need the moment, detectability and stopping conditions of Corollary~\ref{cor:growing_signals}, particularly because the weakest geometric coefficient shrinks rapidly.

For weights $w_j=|\widetilde\beta_j|^{-\gamma}$, $\gamma>0$, a pilot consistent on the signals and converging to zero on the proxy instead gives $w_p\to_p\infty$ and finite active weights. That case is outside the proposition and can remove the weighted proxy obstruction. Its selection properties require the pilot and tuning conditions appropriate to that estimator. Thus the proposition is a conditional result about specified limiting weights, not an impossibility theorem for Adaptive Lasso.

\subsection{SCAD, MCP and local linear approximation}
\label{app:scad_mcp}

The examples below distinguish the global penalised estimator from the
local solution returned by a particular algorithm. We use deterministic,
centred columns satisfying $x_j'x_j/T=1$, with $n\asymp T^\xi$,
$0<\xi<1$, and independent Gaussian errors
$u\sim N(0,\sigma^2 I_T)$, $\sigma^2>0$, and include an intercept.
The noise columns have zero coefficients and are orthogonal both to the
signal--proxy block and to one another. Since $n\leq T-1$ eventually,
the specified moment matrices can be constructed in the orthogonal
complement of the constant vector. We prove these fixed design
comparisons directly, without invoking or changing the random design
assumptions of the main results.

For the unweighted sample criterion, write
\[
 Q_T(b)=\min_{\alpha\in\mathbb R}\frac1{2T}
       \|y-\alpha\mathbf1-Xb\|_2^2+\sum_{j=1}^n p_{\lambda_T}(|b_j|).
\]
We use the SCAD and MCP penalties of \citet{fan2001} and \citet{zha2010},
with marginally standardised columns and no rescaling of the penalty
after residualisation. Both have $p_\lambda(0)=0$, with derivatives
\[
 p'_{\lambda,a}(t)=
 \begin{cases}
 \lambda,&0\leq t\leq\lambda,\\
 (a\lambda-t)/(a-1),&\lambda<t\leq a\lambda,\\
 0,&t>a\lambda,
 \end{cases}
 \qquad
 p'_{\lambda,\gamma}(t)=(\lambda-t/\gamma)_+,
\]
respectively, with fixed $a>2$ and $\gamma>1$. Let
\[
 (\tau,\kappa_p,K_p)=
 \begin{cases}
 (a,(a-1)^{-1},(a+1)/2),&\text{SCAD},\\
 (\gamma,\gamma^{-1},\gamma/2),&\text{MCP}.
 \end{cases}
\]
The derivative is zero beyond $\tau\lambda$, $\kappa_p$ is the
maximum penalty concavity, and the plateau is $K_p\lambda^2$.

The first two comparisons use
\begin{equation}\label{eq:scad_comparison_tuning}
 q_B=\sqrt\delta>1,\qquad
 \lambda_T=A\sigma\sqrt{\frac{2\log n}{T}},\qquad A>1.
\end{equation}
BMT uses \eqref{cvf} with the same fixed $\delta$ at every stage,
including the first; hence
$c_{\alpha_0}(n-O(1),\delta)=q_B\sqrt{2\log n}\{1+o(1)\}$.

If $v>\kappa_p$, then
$v(q-z)^2/2+p_\lambda(|q|)$ is strictly convex: its curvature on
each smooth interval is at least $v-\kappa_p$, its derivative is
continuous at nonzero knots, and it has an upward derivative jump at
zero. Therefore zero is its unique global minimiser if
$v|z|<\lambda$. In particular, this holds for a unit  coordinate
when $|z|<\lambda$, because $1>\kappa_p$ for both penalties.

\begin{proposition}
\label{prop:scad_deletion}
Suppose the fixed design experiment above has
\[
 T^{-1}X'X=\begin{pmatrix}1&r\\r&1\end{pmatrix}\oplus I_{n-2},
 \quad 0<r<1,\quad v=1-r^2,\qquad
 y=B x_1+\beta_{2,T}x_2+u,
\]
where $B>0$ is fixed and
\begin{equation}\label{eq:scad_deletion_region}
 \beta_{2,T}=\frac{C\sigma}{\sqrt v}
                  \sqrt{\frac{2\log n}{T}},\qquad
 \max\{(a-1)^{-1},\gamma^{-1}\}<v,\qquad
 q_B<C<\frac A{\sqrt v}.
\end{equation}
Under \eqref{eq:scad_comparison_tuning}, BMT returns $\{1,2\}$,
whereas each unique global SCAD and MCP solution returns $\{1\}$,
with probability tending to one.
\end{proposition}

\begin{proof}
At the first stage the conventional statistics satisfy
\[
 \frac{t_{1,1}}{\sqrt T}\xrightarrow{p}\frac B\sigma,
 \qquad
 \frac{t_{2,1}}{\sqrt T}\xrightarrow{p}
 \frac{rB}{\sqrt{\sigma^2+B^2v}}.
\]
The first limit is strictly larger, while the noise numerators have no
signal component and their maximum studentised magnitude is
$O_p(\sqrt{\log n})$. It follows that $x_1$ enters first and crosses the
threshold. Conditional on $x_1$, the statistic for $x_2$ is noncentral
$t$ with noncentrality $C\sqrt{2\log n}$ and $T-3$ degrees of
freedom. Dividing it by $\sqrt{2\log n}$ gives convergence in
probability to $C>q_B$.

At either deterministic model $\{1\}$ or $\{1,2\}$, each noise
statistic has a centred Gaussian numerator and an independent
noncentral or central chi squared denominator. The noncentral
denominator gives no larger absolute tail than the central $t$.
With $n\asymp T^\xi$, a union bound at a fixed multiple of
$\sqrt{2\log n}$ strictly between one and $q_B$ puts all these statistics below
the BMT threshold with probability tending to one, since $q_B>1$.
Thus $x_2$ enters next and BMT stops. These bounds concern deterministic
models; they do not condition on a preceding selection event.

The penalised criterion separates into the signal block and individual
noise coordinates. Put $s=X_{\{1,2\}}'u/T$ and
$\theta=(B,\beta_{2,T})'$. Comparison of a signal block global
minimiser with $\theta$ gives, for $d=b-\theta$,
\[
 \frac{1-r}{2}\|d\|_2^2
 \leq\|s\|_2\|d\|_2+2K_p\lambda_T^2.
\]
Since $s=O_p(T^{-1/2})=o_p(\lambda_T)$, every such minimiser is
within $O_p(\lambda_T)$ of $\theta$. Its first coefficient is
therefore in the flat penalty region with probability tending to one.
This step does not presume convexity of the two dimensional criterion.

Minimising the quadratic over the first coordinate then gives, apart
from constants, the lower bound
\[
 \frac v2(q-\widetilde\beta_{2,T})^2+p_{\lambda_T}(|q|),
 \qquad
 \widetilde\beta_{2,T}
 =\beta_{2,T}+\frac{x_2'M_{[\mathbf1,x_1]}u}{Tv},
\]
where $M_D$ denotes orthogonal projection off the columns of $D$.
The correction is $O_p(T^{-1/2})$, and
$v\beta_{2,T}/\lambda_T=C\sqrt v/A<1$. The scalar calculation
makes $q=0$ uniquely optimal with probability tending to one. The
profiled first coefficient there is
$B+r\beta_{2,T}+x_1'u/T$, which is also flat. The lower bound is
attained at that candidate, so it identifies the unique global
signal block minimiser.

For the noise coordinates, $z_j=x_j'u/T$ are independent
$N(0,\sigma^2/T)$ variables. The Gaussian union bound and $A>1$
give $\max_{j\geq3}|z_j|<\lambda_T$ with probability tending to
one. The unit scalar calculation makes their unique minimisers
zero, completing the full sample argument.
\end{proof}

For $a=3.7$, $\gamma=3$, $r=0.70$, $q_B=A=1.10$ and $C=1.30$,
$v=0.51>1/2.7$ and $q_B<C<A/\sqrt v$. Writing
$s_T=\sigma\sqrt{2\log n/T}$, the coefficient thresholds are
\[
 \underbrace{\frac{q_B}{\sqrt v}s_T}_{1.5403s_T}
 <\underbrace{\beta_{2,T}}_{1.8204s_T}
 <\underbrace{\frac A v s_T}_{2.1569s_T}.
\]
Also $v\beta_{2,T}/\lambda_T\simeq0.8440$. The leading testing
and penalty constants are equal, so the interval reflects their
different conditional scaling. This gives deletion of a signal, rather
than replacement by a proxy, for the stated unweighted objective;
conditioning set dependent penalty weights would change the comparison.
The local signal is only a fixed multiple of the testing scale, so
the direct Gaussian proof above is needed instead of the main theorem's
stronger sufficient gap condition.

\begin{proposition}
\label{prop:scad_reverse}
In the same fixed design experiment, suppose
\[
 T^{-1}X'X=
 \begin{pmatrix}1&0&\rho\\0&1&\rho\\\rho&\rho&1\end{pmatrix}
 \oplus I_{n-3},\qquad
 \frac12<\rho<\frac1{\sqrt2},\qquad y=B(x_1+x_2)+u,
\]
where $B>0$ is fixed. Under \eqref{eq:scad_comparison_tuning}, BMT
admits proxy 3 first and fails exact recovery with probability tending
to one. Each unique global SCAD and MCP solution equals least squares
on $\{1,2\}$ with probability tending to one.
\end{proposition}

\begin{proof}
The finite block has minimum eigenvalue $1-\sqrt2\rho>0$, and its
proxy score is $2\rho B>B$, the score of either signal. Since the
columns have equal norms, the ordering of absolute conventional
statistics agrees with that of absolute sample correlations. The
strict proxy advantage therefore persists in the sample, and its
statistic is of order $\sqrt T$, exceeding the threshold. As BMT
never deletes an admitted variable, exact recovery fails.

The noise coefficients are zero by the preceding argument. For the
finite block, the score error is $O_p(T^{-1/2})$. Comparison with
$(B,B,0)'$ and positive definiteness put every global minimiser within
$O_p(\lambda_T+T^{-1/2})$ of that vector. Both true coefficients
are therefore flat. Put $v_p=1-2\rho^2>0$ and
$z=x_3'M_{[\mathbf1,x_1,x_2]}u/T=o_p(\lambda_T)$.
After profiling the signals, the difference from the oracle criterion
is bounded below by
\[
 \frac{v_p}2 q^2-zq+p_{\lambda_T}(|q|).
\]
For $0<|q|\leq\tau\lambda_T$, penalty concavity gives
$p_{\lambda_T}(|q|)\geq(K_p/\tau)\lambda_T|q|$, so the
difference is positive when $|z|<(K_p/\tau)\lambda_T$.
For $|q|>\tau\lambda_T$, it is at least
$K_p\lambda_T^2-z^2/(2v_p)>0$ with probability tending to one.
Thus the proxy is uniquely zero. The oracle least squares coefficients
are flat and attain the bound, making them the unique global sample
solution. No convexity condition on $v_p$ is needed.
\end{proof}

These two fixed signal count examples establish comparisons in both
directions at explicit tuning. They do not establish uniform dominance,
compare optimally tuned procedures, or identify the output of a
particular SCAD or MCP algorithm.

\subsubsection{Lasso initialised LLA}
\label{app:fc_lla}

The global criterion does not determine which solution a local algorithm
returns. We therefore consider LLA initialised by the Lasso at the same
penalty level, with the weighted Lasso solved exactly at each iteration.
The following example uses this initialisation and update rule
\citep{zouli2008}:
\[
 b^{(m)}=\argmin_b\left\{\ell(b)+\sum_j
 p'_{\lambda}(|b_j^{(m-1)}|)|b_j|\right\},
\]
where $\ell$ is the population quadratic loss below or the sample loss in
$Q_T$, as specified in each statement. The intercept remains unpenalised.

Let the three-column core have second-moment matrix
\[
 \Sigma_\rho=
 \begin{pmatrix}1&0&\rho\\0&1&\rho\\\rho&\rho&1\end{pmatrix},
 \qquad \frac12<\rho<\frac1{\sqrt2},
\]
and true coefficient $b^*=(B,c\lambda,0)'$, where $B>0$ is fixed and
$\lambda\downarrow0$. Write $\tau=a$ for SCAD and $\tau=\gamma$ for MCP,
and use the population objective
\[
 Q_\lambda(b)=\frac{\sigma^2}{2}
 +\frac12(b-b^*)'\Sigma_\rho(b-b^*)
 +\sum_{j=1}^3p_\lambda(|b_j|).
\]
Suppose
\begin{equation}\label{eq:fc_lla_region}
 \max\left\{\tau+1,
       \frac{1-\rho+\tau(1-\rho^2)}{\rho}\right\}
 <c<\frac{1-\rho}{1-2\rho^2}.
\end{equation}
This interval is nonempty for $\rho$ sufficiently close to $1/\sqrt2$.
Define
\[
 r_L=\frac{\rho c-(1-\rho)}{1-\rho^2},\qquad
 r_W=\frac{\rho c}{1-\rho^2},\qquad
 m_L=c-\rho r_L,\qquad
 m_W=\frac{c(1-2\rho^2)}{1-\rho^2}.
\]

\begin{proposition}\label{prop:scad_lla}
For sufficiently small $\lambda$, the unique population SCAD or MCP
global minimiser is $b^*$. The unique population Lasso minimiser is
\[
 b^L=(B-\lambda-\rho r_L\lambda,\ 0,\ r_L\lambda)',
\]
whereas one LLA iteration from $b^L$ gives
\[
 b^W=(B-\rho r_W\lambda,\ 0,\ r_W\lambda)'.
\]
Subsequent iterations remain at $b^W$, which is a strict nonoracle local
minimum of the folded concave criterion, with
\begin{equation}\label{eq:fc_lla_gap}
 Q_\lambda(b^W)-Q_\lambda(b^*)
 =\frac{c^2\lambda^2}{2}
   \frac{1-2\rho^2}{1-\rho^2}>0.
\end{equation}
In the centred deterministic design, growing noise construction above,
take this three column core, retain $q_B>1$, and replace the penalty
sequence in \eqref{eq:scad_comparison_tuning} by
$\lambda_T/\sqrt{\log n/T}\to\infty$, with $\lambda_T\to0$.
With probability tending to one, BMT and the unique global sample
SCAD/MCP minimiser select $\{1,2\}$, while the sample Lasso and its LLA
fixed point select $\{1,3\}$. The first LLA update is the restricted OLS
fit on $\{1,3\}$ and is a strict local minimum.
\end{proposition}

\begin{proof}
First consider the population objective, with the common error variance
term omitted. Its value at $b^*$ is $O(\lambda^2)$. Positive definiteness
therefore puts every global minimiser within $O(\lambda)$ of $b^*$, so its
first coefficient is in the flat penalty region. Profiling this
coefficient leaves
\[
 F(b_2,b_3)=\tfrac12(c\lambda-b_2-\rho b_3)^2
 +\tfrac12(1-2\rho^2)b_3^2+p_\lambda(|b_2|)+p_\lambda(|b_3|).
\]
An increasing concave penalty with $p(0)=0$ is subadditive: concavity gives
$p(x)\ge xp(x+y)/(x+y)$ and the analogous inequality for $p(y)$.
Consequently
$p_\lambda(|b_2+\rho b_3|)\le p_\lambda(|b_2|)+p_\lambda(|b_3|)$.
Writing $q=b_2+\rho b_3$ yields
\[
 F(b_2,b_3)\ge
 \tfrac12(c\lambda-q)^2+p_\lambda(|q|)
 +\tfrac12(1-2\rho^2)b_3^2.
\]
The univariate objective is strictly convex because the maximum penalty
concavity is less than one, and its derivative is zero at $q=c\lambda$,
which lies in the flat region. Thus the lower bound is uniquely attained at
$q=c\lambda$, $b_3=0$, with profiled first coefficient $B$, proving
the global, not merely local, optimality of $b^*$.

For the Lasso candidate supported on $\{1,3\}$, the active KKT equations
are
\[
 b_1-B+\rho b_3=-\lambda,\qquad
 \rho(b_1-B)+b_3-\rho c\lambda=-\lambda.
\]
Their solution is $b^L$. The omitted coordinate score is $m_L\lambda$,
where $0<m_L<1$ follows from the upper bound in
\eqref{eq:fc_lla_region}. The lower bound gives $r_L>\tau$; the first
coefficient also exceeds $\tau\lambda$ for small $\lambda$.
Thus LLA assigns weights $(0,\lambda,0)$. The restricted normal equations
give $b^W$, whose omitted score is $m_W\lambda<\lambda$.
Both weighted Lasso solutions are unique because $\Sigma_\rho$ is positive
definite. Since $r_W>r_L>\tau$, the next weights remain $(0,\lambda,0)$.

To check a strict local minimum of the original criterion, keep the two
active coordinates in their flat regions and profile their quadratic
loss. For a small omitted coefficient $t$, the objective change is
\[
 \tfrac12\frac{1-2\rho^2}{1-\rho^2}t^2
 -m_W\lambda t+p_\lambda(|t|).
\]
Since $p_\lambda(|t|)\ge\lambda|t|-\kappa_p t^2/2$, this is strictly
positive for sufficiently small nonzero $t$. On the active support the
quadratic is strictly convex. The two candidate points have the same
penalty, so their loss difference is the residual variance from projecting
$c\lambda x_2$ on $(x_1,x_3)$, giving \eqref{eq:fc_lla_gap}.

For the sample statement, the orthogonal noise block separates from the
core objective. A Gaussian maximum bound makes every noise score smaller
than $\lambda_T$ with probability tending to one. Each univariate
SCAD/MCP objective then has its unique minimum at zero; the same holds for
the noise coordinates of the Lasso and the stated LLA updates.
On the fixed core, sample score perturbations are
$O_p(T^{-1/2})=o_p(\lambda_T)$. The basic inequality puts all global
minimisers within $O_p(\lambda_T)$ of $b^*$. After profiling the strong
coefficient, dividing the weak and proxy coefficients by $\lambda_T$,
subtracting constants and dividing the criterion by $\lambda_T^2$,
the resulting criterion converges uniformly on bounded sets to $F$ with
$\lambda=1$, whose unique minimiser is $(c,0)$. The quadratic bound
confines all minimisers to a common bounded set with arbitrarily high
probability. Hence all global minimisers satisfy
$b_2/\lambda_T\to_p c$ and $b_3/\lambda_T\to_p0$.

Both true coefficients are now in the flat region. Profiling them leaves,
relative to restricted OLS on $\{1,2\}$,
\[
 \tfrac12(1-2\rho^2)t^2-e_Tt+p_{\lambda_T}(|t|),
 \qquad
 e_T=T^{-1}(x_3-\rho x_1-\rho x_2)'u
       =o_p(\lambda_T).
\]
For $|t|\le\eta\lambda_T$, with sufficiently small fixed $\eta$, the
penalty lower bound used above makes this difference strictly positive
for every nonzero $t$, with probability tending to one. Coefficient
closeness puts every global minimiser in this neighbourhood. Its proxy
coefficient is therefore exactly zero, and its signal coefficients are
the unique restricted OLS fit.

The strict Lasso KKT and flat region margins persist under the
$o_p(\lambda_T)$ core score perturbations. Its first LLA update is
restricted OLS on $\{1,3\}$; the omitted score remains below
$\lambda_T$ and both active coefficients remain flat. The preceding
local minimum argument applies unchanged. The sample objective gap,
divided by $\lambda_T^2$, converges to the positive constant
$c^2(1-2\rho^2)/\{2(1-\rho^2)\}$ from \eqref{eq:fc_lla_gap}.

Finally, BMT selects $x_1$ first since $B>\rho(B+c\lambda_T)$ for large
$T$. After conditioning on $x_1$, the conventional statistics satisfy
\[
 \frac{t_{2\mid\{1\}}}{\sqrt T\lambda_T}\xrightarrow{p}\frac c\sigma,
 \qquad
 \frac{t_{3\mid\{1\}}}{\sqrt T\lambda_T}
 \xrightarrow{p}\frac{\rho c}{\sigma\sqrt{1-\rho^2}}.
\]
The residual error scales tend to $\sigma$ because $\lambda_T\to0$.
The first limit is strictly larger, and
$\sqrt T\lambda_T/\sqrt{\log n}\to\infty$. Thus the signal outranks
the proxy and the noise candidates and crosses the threshold.
After both
signals enter, the remaining partial coefficients are zero. The growing
noise pool and the stated Gaussian threshold bound give terminal stopping
with probability tending to one.
\end{proof}

For the calibration $\rho=0.69$, $c=5$, SCAD $a=3.7$ and MCP
$\gamma=3$, we obtain
\[
 r_L=5.9935,\quad m_L=0.8645,\quad
 r_W=6.5852,\quad m_W=0.4562,
\]
and an objective gap of $1.1405\lambda^2$. At $B=1$ and $\lambda=0.05$,
the retained values are
\begin{center}
\begin{tabular}{lccc}
\hline
 & $Q_\lambda(b^*)-\sigma^2/2$
 & $Q_\lambda(b^W)-\sigma^2/2$ & Difference\\
\hline
SCAD & 0.0117500 & 0.0146012 & 0.0028512\\
MCP & 0.0075000 & 0.0103512 & 0.0028512\\
\hline
\end{tabular}
\end{center}
These values were checked over every signed quadratic piece and boundary
face of the three dimensional objective, as a numerical check rather
than a substitute for the proof. The global argument at this calibration
also has the required flat first coefficient, since every global minimiser
satisfies
$|b_1-B|\le2\sqrt{K_p(\Sigma_\rho^{-1})_{11}}\lambda$.
At the stated calibration this leaves $b_1-\tau\lambda$ above $0.307$
for SCAD and $0.444$ for MCP.

The oracle local solution results of \citet{fan2001} do not guarantee the
output of an arbitrary algorithm. The LLA results of \citet{zouli2008}
and \citet{FanXueZou2014} require an appropriate initial estimator and
additional regularity and score conditions. In the latter paper's
Theorem~1, the signal bound is $\beta_{\min}>(\tau+1)\lambda$,
alongside a maximum norm localisation condition. The calibration above
satisfies the signal bound, but its Lasso initialiser sets a coefficient
of $5\lambda$ to zero and puts about $5.99\lambda$ on the proxy.
Its maximum norm error therefore exceeds the localisation bound
$\lambda$ for these penalties. Corollary~3 of that paper also gives
conditions under which a suitably tuned Lasso initialiser succeeds;
failure of Lasso support recovery alone does not imply LLA failure.
The example concerns the stated same penalty initialisation. It is not
a failure of every local solver, a global optimisation method, or the
specified PLUS path studied by \citet{zha2010}.

\subsection{BMT with retesting}
\label{app:bmt_retesting}

This appendix examines two extensions of BMT that allow previously selected regressors to be removed.
The two retesting procedures differ in when they reconsider selected
variables. BMT-R1 checks the included variables once BMT has stopped,
whereas BMT-R2 checks the previously selected variables after each new
admission. In both cases, the checks use conventional OLS statistics
from the joint regression, including the always in controls, and retain
the forward threshold and its exponent. This allows a proxy admitted
early in the search to be removed when its conditional contribution
becomes small.

\subsubsection{The two procedures}
\label{app:retesting_algorithms}
Let $N=\{1,\ldots,n\}$ and
$b(A)=c_{\alpha_0}(n-|A|,\delta)$ for a forward state $A\ne N$.
As in BMT, admission requires the largest $|t_{j\mid A}|$,
$j\in N\setminus A$, to exceed $b(A)$ strictly, with the smallest
index breaking ties. For a joint model $B$ and $j\in B$, define
\[
 t^-_{j,B}=t_{j\mid B\setminus\{j\}}.
\]
The retest therefore concerns the coefficient of $j$ conditional on all
other variables in $B$ and the controls, rather than a marginal test
or a regression on residuals from a model that already includes $j$.
The controls remain in every fitted model.

For BMT-R1, let $F$ be the set returned by BMT. The final check uses
its stopping threshold
$b_1(F)=c_{\alpha_0}(\max\{1,n-|F|\},\delta)$ and returns
\begin{equation}\label{eq:r1_rule}
 \widehat S_{R1}=\{j\in F:|t^-_{j,F}|>b_1(F)\}.
\end{equation}
All retests are calculated in $F$ before any deletion, after which the
variables that fail are removed together and the retained model is
refitted. There is no further forward search or backward check. The
empty set remains empty; if every candidate has entered, the convention
uses $c_{\alpha_0}(1,\delta)$, the last admission threshold.

BMT-R2 starts from $A=\varnothing$ and applies the BMT admission and
stopping rule at each current set $A$. Following an admission $j$, the
joint regression on $B=A\cup\{j\}$ determines
\begin{equation}\label{eq:r2_rule}
 D(A,j)=\{i\in A:|t^-_{i,B}|\leq b(A)\},\qquad
 A^{\rm next}=B\setminus D(A,j).
\end{equation}
The old selected variables are checked once, using the same numerical
cutoff as the admission just made; the new entrant is not retested at
this stage. All failures are removed together, and the forward search
resumes after refitting, with eligible set $N\setminus A^{\rm next}$.
Deleted variables are therefore eligible for readmission. The next
forward threshold reflects the new number of remaining candidates,
without a separate exit level or exponent.

A state after retesting denotes the set after this single retest, not the
outcome of repeated deletion to a fixed point. The procedure records
these sets, beginning with the empty set, and stops with a repeated state
flag if a set recurs, returning that set. Hence there can be at most
$2^n$ distinct states after retesting before ordinary stopping or a flagged
repetition. Ordinary stopping also applies when the eligible set is
empty. A singular fit, nonpositive residual variance or nonpositive
residual degrees of freedom produces a fit failure flag. The conditions
below ensure that neither flag occurs with probability tending to one.

The retested family is $F$ for BMT-R1 and $A$ for BMT-R2, but in each case the
cutoff is inherited from BMT's remaining candidate family rather than
recalibrated to the number of retests. Equality with that cutoff leads
to deletion. These are selection rules: they do not supply conditional
post selection $p$ values or a nominal finite sample family wise error
guarantee for adaptively chosen models.

\subsubsection{Exact recovery}
\label{app:retesting_event}
Let $S=\mathcal S_s$ and
$\lambda_{j,A}=\sqrt T\,\psi_{j,A,T}$, where the population score is
defined in Section~3.1. Both procedures require control of entry and
retesting statistics on a class specified before selection. Let
$\mathcal C_T^R$ be such a deterministic class, containing all entry
pairs and all pairs $(B\setminus\{i\},i)$ used in the retests, and
suppose that, for a deterministic $r_T>0$,
\begin{equation}\label{eq:retesting_event}
 \Pr(\mathcal E_T)\longrightarrow1,\qquad
 \mathcal E_T=\left\{\text{all covered fits are defined},\quad
 \max_{(A,j)\in\mathcal C_T^R}|t_{j\mid A}-\lambda_{j,A}|\leq r_T\right\}.
\end{equation}
Lemma~\ref{lem:tstat} applies when its moment and local design conditions
hold on this enlarged class. For bounded conditioning size, the
concentration argument in its proof gives
$r_T=L\sqrt{\log n}$ for a sufficiently large fixed $L$, with the
probability in \eqref{eq:retesting_event} tending to one. The fixed
constant follows from that concentration argument, rather than from
an $O_p(\sqrt{\log n})$ assertion alone. Where proxies can enter, the
local conditions must cover joint models containing them; bounds on
signal only histories do not suffice. These requirements are local
to the results in this appendix and do not alter the general
assumptions. We consider a bounded number of signals and proxies,
without a growing $k$ extension.

The following result gives sufficient conditions directly in terms of
the data generating process. Unlike the main exact recovery results,
it does not require signals to dominate proxies at admission.

\begin{theorem}\label{th:retesting_exact}
Consider \eqref{dgp1} with fixed numbers $k\geq1$ of signals and
$k^*\geq0$ of proxies. Suppose that Assumptions~\ref{ass1}--\ref{ass_proj} and \ref{reg} hold on
$\mathcal C_T^{\mathrm{AM}}$, together with the noise orthogonality
restriction in Section~2. The nonzero coefficients may depend on $T$;
they are uniformly bounded above in absolute value and satisfy
\begin{equation}\label{eq:retesting_beta_min}
 \sqrt{\frac{T}{\log n}}\,
 \beta_{\min,T}\longrightarrow\infty,
 \qquad
 \beta_{\min,T}=\min_{s\in S}|\beta_{s,T}|.
\end{equation}
This condition replaces Assumption~\ref{ass10} for this theorem. Let $n\to\infty$
with $n=O(T^\kappa)$ for a fixed $\kappa>0$, and use the thresholds and
retesting rules in Appendix~\ref{app:retesting_algorithms}, with fixed
$0<\alpha_0<1$ and $c\geq\alpha_0$.
There is a finite constant $\delta_0$ such that, for every fixed
$\delta>\delta_0$,
\[
 \Pr(\widehat S_{R1}=S)\longrightarrow1,
 \qquad
 \Pr(\widehat S_{R2}=S)\longrightarrow1.
\]
With probability tending to one, BMT-R2 visits no state after retesting twice, all required
regressions have nonsingular regressor cross product matrices, positive
residual variances and positive residual degrees of freedom, and it makes
at most $2^{k+k^*}-1$ admissions. Neither result requires
the signal to proxy dominance condition~\eqref{cond1}.
\end{theorem}

\begin{proof}
Let $H=S\cup P$, $h=k+k^*$ and $a_T=\sqrt{\log n/T}$.
The class $\mathcal C_T^{\mathrm{AM}}$ covers forward regressions at
every subset of $H$, as well as the leave one variable out regressions
used in the joint retests. Lemmas~\ref{lem:score} and~\ref{lem:tstat},
including their fixed constant concentration bound, give an event with
probability tending to one on which all these fits are defined and
\[
 \max_{(A,j)\in\mathcal C_T^{\mathrm{AM}}}
 |t_{j\mid A}-\sqrt T\,\psi_{j,A,T}|
 \leq L\sqrt{\log n}
\]
for a sufficiently large fixed $L$. The lemmas do not require a
positive lower bound on the nonzero coefficients. Uniformly over
$|A|\leq h$, both $b(A)/\sqrt{\log n}$ and
$b_1(A)/\sqrt{\log n}$ tend to $\sqrt{2\delta}$, so a sufficiently
large fixed $\delta$ makes the cutoffs exceed $L\sqrt{\log n}$
eventually. The remainder of the proof is on this event.

Let $V(A)=s_{yy\mid A,T}$ be the population outcome residual variance.
If $M=S\setminus A$ is nonempty, score orthogonality in Assumption~\ref{ass1}
gives the vector of partial covariances
\[
 (c_{s,A,T}:s\in M)
   =\Sigma_{M\mid A,T}\boldsymbol\beta_{M,T},
 \qquad
 \Sigma_{M\mid A,T}
   =\overline E_T(\boldsymbol x_{M,t}^{A}
                         \boldsymbol x_{M,t}^{A\prime}).
\]
Assumption~\ref{reg} on the full signal--proxy class bounds the eigenvalues of
this Schur complement above and away from zero. Since $k$ is fixed,
there are constants independent of $A$ and $T$ such that
\[
 \max_{s\in M}|c_{s,A,T}|\geq C_1\beta_{\min,T},
 \qquad
 \max_{s\in M}|\psi_{s,A,T}|\geq C_2\beta_{\min,T}.
\]
The second bound uses
$\psi_{s,A,T}=c_{s,A,T}/
\{v_{s,A,T}^{1/2}\omega_{s,A,T}\}$ and the upper bounds in \ref{ass_proj}.
By \eqref{eq:retesting_beta_min}, a missing signal therefore prevents
stopping. The Section~2 orthogonality restriction gives
$\psi_{d,A,T}=0$ for every noise variable at these states, so no
noise variable is admitted.

For BMT-R1, ordinary BMT keeps adding members of $H$ until all signals
are included, at which point every remaining candidate has zero
population partial coefficient and fails the admission threshold.
Its final set thus satisfies $S\subseteq F\subseteq H$, after at
most $h$ admissions. In the joint regression on $F$, the signal
coefficients equal $\beta_{s,T}$ and the proxy coefficients equal
zero. The eigenvalue bounds and \eqref{eq:retesting_beta_min} make
the signal retest statistics exceed the cutoff, whereas the error
bound makes every proxy retest fail it. The single terminal deletion
therefore returns $S$.

For BMT-R2, consider an admission $j$ from a state $A$ missing a signal,
and let $B=A\cup\{j\}$. Since $j$ maximises the absolute sample
statistic, the uniform error bound implies
\[
 |\psi_{j,A,T}|
 \geq\max_{s\in S\setminus A}|\psi_{s,A,T}|-2La_T
 \geq C_3\beta_{\min,T}
\]
eventually. The reduction in population residual variance from this
admission consequently satisfies
\begin{equation}\label{eq:r2_entry_gain}
 V(A)-V(B)
  =\frac{c_{j,A,T}^2}{v_{j,A,T}}
  =\omega_{j,A,T}^2\psi_{j,A,T}^2
  \geq C_4\beta_{\min,T}^2.
\end{equation}
Let $D\subseteq A$ be the set of variables deleted in the single joint
retest, and denote the coefficient of $i$ in the population regression
on $B$ by $\theta_{i,B,T}$. For every $i\in D$, the test rule and
the uniform error bound imply
\[
 |\theta_{i,B,T}|
 =\frac{\omega_{i,B\setminus\{i\},T}}
        {v_{i,B\setminus\{i\},T}^{1/2}}
       |\psi_{i,B\setminus\{i\},T}|
 \leq C_5\frac{b(A)+L\sqrt{\log n}}{\sqrt T}
 \leq C_6 a_T.
\]
These coefficients all refer to the same joint model $B$, before
any deletion. The population projection identity for simultaneous
removal gives
\begin{equation}\label{eq:r2_deletion_loss}
 V(B\setminus D)-V(B)
 =\boldsymbol\theta_{D,B,T}'
     \Sigma_{D\mid B\setminus D,T}
     \boldsymbol\theta_{D,B,T}
 \leq C_7a_T^2.
\end{equation}
Here the conditional covariance matrix removes the controls and the
retained variables. Its upper eigenvalue is bounded by \ref{reg} and
$|D|\leq h$ is fixed; the expression is zero when $D$ is empty.
Combining \eqref{eq:r2_entry_gain} and \eqref{eq:r2_deletion_loss}
with \eqref{eq:retesting_beta_min} yields
\[
 V(B\setminus D)<V(A)
\]
for all sufficiently large $T$. Population residual variance thus
strictly decreases at every BMT-R2 transition from an incomplete model,
even if an old signal is deleted or a deleted variable reenters.
No intermediate signal retention or no readmission condition has
been imposed.

There are only $2^h$ possible selected sets. Strict decrease rules
out repetition, while a detectable missing signal rules out stopping
at an incomplete set. At the first admission for which $B$ contains
$S$, the entrant is the last missing signal. In this correctly
specified joint model, every signal coefficient equals its true
coefficient and every old proxy coefficient is zero. The terminal
retest comparison used for BMT-R1 therefore retains all signals and
deletes all old proxies, leaving $S$. Every remaining candidate then
has zero partial coefficient, so BMT-R2 stops. Its states after retesting are
distinct subsets of $H$, starting at the empty set, which gives at
most $2^h-1$ admissions. All fits are defined on the event used
throughout, completing the probability statements.
\end{proof}

The design condition in Theorem~\ref{th:retesting_exact} is stronger
than the signal only condition in Theorem~\ref{th4}: it covers every
signal--proxy subset and keeps the number of proxies fixed. In return,
signals need not dominate proxies at admission. The fixed lower
coefficient bound in Theorem~\ref{th1} is a special case of
\eqref{eq:retesting_beta_min}, while the latter also permits signals
to shrink more slowly than the testing scale. For BMT-R2, the proof
compares population residual variances: the gain from an admission
exceeds the loss from simultaneous deletion. It is not a finite sample
descent argument for a penalised residual sum of squares.

For an example satisfying these conditions, let $x_1,x_2,z$ and
$u/\sigma$ be independent standard Gaussian variables, independent
across observations, with $\sigma>0$, and set
\[
 x_3=\frac35x_1+\frac35x_2+\frac{\sqrt7}{5}z,
 \qquad y=x_1+x_2+u.
\]
Include an intercept and $n-3$ independent standard Gaussian noise
candidates, independent of this block and also independent across
observations, with $n\to\infty$ and
$n=O(T^\kappa)$. Gaussianity and temporal independence give \ref{ass1}--\ref{ass2};
the fixed positive definite signal--proxy covariance matrix and the
independent noise candidates give \ref{ass_proj} and \ref{reg} on
$\mathcal C_T^{\mathrm{AM}}$. Noise orthogonality holds and
$\beta_{\min,T}=1$. All three candidates have unit variance, but
the proxy has marginal covariance $6/5$ with $y$, compared with one
for each signal, and therefore enters first with probability tending
to one. Conditional on the proxy, either signal has
residual standardised covariance $7/20>0$; conditional on the proxy
and one signal, the remaining signal has conditional variance
$7/16>0$. Both signals are consequently admitted. The proxy's joint
coefficient after the first signal enters is $15/16$, so BMT-R2 does
not remove it at that stage; after both signals enter, its coefficient
is zero. Thus ordinary BMT returns $\{1,2,3\}$ whereas both BMT-R1 and
BMT-R2 return $\{1,2\}$, with probability tending to one for a
sufficiently large fixed $\delta$. This illustrates exact recovery
by retesting when the signal to proxy dominance condition in
Theorem~\ref{th4} fails.

For brevity, write $\lambda^-_{i,B}=\lambda_{i,B\setminus\{i\}}$.
The next proposition states the terminal retesting argument in terms
of population scores. It separates the power needed to retain signals
from the threshold needed to remove inactive variables, and can also
be used when the coefficient condition in
Theorem~\ref{th:retesting_exact} is unavailable.

\begin{proposition}\label{prop:r1_recovery}
For every sample on which the fits are defined,
$\widehat S_{R1}\subseteq\widehat S_{BMT}$. Suppose a deterministic
collection $\mathfrak F_T$ satisfies
\[
 \Pr(F\in\mathfrak F_T)\to1,\qquad
 S\subseteq F\subseteq S\cup P\quad(F\in\mathfrak F_T),
\]
and \eqref{eq:retesting_event} covers its joint model retests. Suppose,
for every $F\in\mathfrak F_T$,
\begin{align*}
 |\lambda^-_{s,F}|-r_T&>b_1(F) &&(s\in S),\\
 |\lambda^-_{i,F}|+r_T&<b_1(F) &&(i\in F\setminus S).
\end{align*}
Then $\Pr(\widehat S_{R1}=S)\to1$.
\end{proposition}
\begin{proof}
The subset statement follows directly from \eqref{eq:r1_rule}. On the
intersection of the coverage event and $\mathcal E_T$, the reverse
triangle inequality makes every signal pass, while the triangle
inequality makes every inactive included variable fail. The single
simultaneous deletion therefore returns exactly $S$.
\end{proof}

In the correctly specified linear model, $\lambda^-_{i,F}=0$ for an
inactive $i$ once $F$ contains $S$, under the stated regressor--error
orthogonality and nonsingularity. The uniform error bound guarantees
their removal when $r_T<b_1(F)$. For bounded $|F|$,
$b_1(F)/\sqrt{\log n}\to\sqrt{2\delta}$, so a sufficiently large
fixed $\delta$ provides this inequality. Under the full signal--proxy
design conditions, coefficients bounded away from zero give signal
noncentralities of order $\sqrt T$ uniformly over the finite
collection of possible $F$. This is the fixed coefficient case of
Theorem~\ref{th:retesting_exact}, which also covers shrinking
coefficients satisfying \eqref{eq:retesting_beta_min}. If that
condition fails, or the joint design becomes nearly singular, the
retention inequalities require a separate check. Regardless of these
conditions, BMT-R1 cannot add an omitted signal: its subset relation to
BMT holds sample by sample, not only asymptotically.

\subsubsection{Conditions on the BMT-R2 path}
\label{app:r2_progress}
Theorem~\ref{th:retesting_exact} allows proxies to enter before
signals without requiring old signals to be retained at every intermediate
stage. For weaker
signals, its comparison of admission gains and deletion losses may
not apply. The following alternative states sufficient conditions
on a deterministic graph of possible BMT-R2 paths, so they can be checked
without conditioning on the realised sample path.
Let $\mathfrak R_T$ be a deterministic collection of subsets of $S\cup P$,
containing $\varnothing$ and $S$, and with maximum size $K<n$.
Suppose $S\nsubseteq A$ for every
$A\in\mathfrak R_T\setminus\{S\}$. The possible entrants at $A$ are
\[
 J(A)=\left\{j\in N\setminus A:
 |\lambda_{j,A}|+r_T>b(A),\quad
 |\lambda_{j,A}|\geq\max_{\ell\in N\setminus A}|\lambda_{\ell,A}|-2r_T
 \right\}.
\]
For $j\in J(A)$ and $B=A\cup\{j\}$, a possible deletion set
$D\subseteq A$ satisfies
\begin{equation}\label{eq:r2_possible_deletions}
 \{i\in A:|\lambda^-_{i,B}|+r_T\leq b(A)\}
 \subseteq D\subseteq
 \{i\in A:|\lambda^-_{i,B}|-r_T\leq b(A)\}.
\end{equation}
Each such pair $(j,D)$ defines an edge $A\to B\setminus D$.
The resulting graph may include transitions that the sample algorithm
never takes, since the inequalities allow all decisions compatible
with the error bound in \eqref{eq:retesting_event}.

\begin{proposition}\label{prop:r2_recovery}
Suppose \eqref{eq:retesting_event} covers every entry pair at
$A\in\mathfrak R_T$ and every old variable retest in the graph above.
For all sufficiently large $T$, impose the following local conditions.
\begin{enumerate}
\item For every $A\ne S$ in $\mathfrak R_T$,
$\max_{s\in S\setminus A}|\lambda_{s,A}|-r_T>b(A)$.
For every noise candidate $d\notin S\cup P$ and every such $A$,
$|\lambda_{d,A}|+r_T<b(A)$.
\item For every graph admission $B=A\cup\{j\}$, old signals satisfy
$|\lambda^-_{s,B}|-r_T>b(A)$, $s\in A\cap S$.
If $S\subseteq B$, all old proxies satisfy
$|\lambda^-_{i,B}|+r_T<b(A)$, $i\in A\cap P$.
\item Every graph edge ends in $\mathfrak R_T$. There is an integer
function $\Gamma_T:\mathfrak R_T\to\{0,\ldots,L_T\}$ such that an
edge $A\to A'$ with $A'\cap S=A\cap S$ satisfies
$\Gamma_T(A')<\Gamma_T(A)$.
\item At $S$, $\max_{j\in N\setminus S}|\lambda_{j,S}|+r_T<b(S)$.
\end{enumerate}
Then, with probability tending to one, BMT-R2 returns $S$ without revisiting
a state after retesting, and all required regressions have nonsingular regressor
cross product matrices, positive residual variances and positive residual
degrees of freedom. On $\mathcal E_T$ it makes at most
$(|S|+1)(L_T+1)-1$ admissions.
\end{proposition}
\begin{proof}
On $\mathcal E_T$, condition~1 prevents ordinary stopping while a
signal is missing and excludes every noise variable. A sample maximiser
that passes the entry threshold belongs to $J(A)$: its two defining
inequalities follow from the error bound applied to the sample and
population maximisers. The actual deletion set satisfies
\eqref{eq:r2_possible_deletions}, so closure keeps the next state in
$\mathfrak R_T$. Condition~2 retains all old signals, while the new
entrant is protected by the algorithm. If this admission completes
$S$, the entrant is a signal and all old proxies are deleted, leaving
exactly $S$.

Along every edge before reaching $S$, the pair
$(|S\setminus A|,\Gamma_T(A))$ strictly decreases in lexicographic
order: its first coordinate never increases, and condition~3 decreases
the second whenever the first is unchanged. Since there are at most
$(|S|+1)(L_T+1)$ possible pairs, no state after retesting is revisited. The path
cannot continue indefinitely or stop with a missing signal, and hence
reaches $S$, where condition~4 stops further admissions. All required
fits are defined on $\mathcal E_T$, establishing both the admission
bound and the probability statement.
\end{proof}

Condition~3 does not follow merely from deleting variables with small
individual statistics. Nor does simultaneous deletion inherit a
sample penalised residual sum of squares argument for one at a time
backward elimination. The population decrease established in
Theorem~\ref{th:retesting_exact} instead uses
\eqref{eq:retesting_beta_min}, which need not hold here. One way to
verify condition~3 is to check that the finite graph has no directed
cycle within any fixed selected signal set; the length of the longest
remaining path in that subgraph then supplies $\Gamma_T$. A simpler
sufficient check rules out readmission of a deleted proxy before
another signal enters. Because the entrant is protected and old
signals are retained, a cycle at a fixed signal set would have to
delete and later readmit some proxy. Readmission after a further
signal remains possible. Both checks concern every graph edge, not
just one favourable path.

Noise has the same meaning as the distant variable category in the
main paper. Its zero after control cross products with every signal
and proxy imply $\lambda_{d,A}=0$ on the covered signal--proxy states;
this is not replaced by a small nonzero bound. Proxy noncentralities,
by contrast, need not vanish while signals are missing, which is why
the proposition imposes retention, deletion and progress conditions
at those intermediate models.

\begin{proposition}\label{prop:retesting_transfer}
Suppose, with probability tending to one, ordinary BMT follows a
signal only path and stops at $S$. Suppose additionally that on this
path each old selected signal passes the joint model retest after every
admission, using that stage's cutoff, and that every signal passes the
BMT-R1 retest in $S$. Then BMT, BMT-R1 and BMT-R2 all return $S$ with probability
tending to one. If the fixed dimensional inference conditions of
Theorem~\ref{th:post_selection_oracle} also hold, its oracle inference
conclusion applies to the refitted BMT-R1 and BMT-R2 estimators.
\end{proposition}
\begin{proof}
BMT-R1 deletes nothing from $S$, while induction over the ordinary BMT
admissions shows that BMT-R2 has the same current set, threshold and
maximiser at every stage. The retention condition deletes no old
signal, so the common path is strictly increasing and has the same
terminal stopping test. On exact selection, each refitted estimator
equals the oracle least squares estimator. Since the probability of
disagreement tends to zero, the pointwise inference conclusion
transfers to both retested procedures.
\end{proof}

The additional retention condition in this proposition can be checked
using $|\lambda^-_{s,B}|-r_T$. Forward signal--proxy dominance alone
does not establish it, because an old signal's coefficient in a larger
model comes from a different partial regression. The same qualification
applies when transferring the main single signal, multiple signal or
comparison results through this proposition. The inference conclusion
remains fixed dimensional.

\subsubsection{A signal recovered only by BMT-R2}
\label{app:r2_rescue}
BMT-R2 can also recover a weak signal that is missed by both BMT and BMT-R1,
because deleting a proxy changes the model used in subsequent entry
tests. The following fixed design example uses the centred Gaussian
experiment of Appendix~\ref{app:scad_mcp}, with $n\asymp T^\xi$,
$0<\xi<1$. Let $x_1,x_2,x_3,z$ be mutually orthogonal centred columns
with squared norm $T$, together with the proxy
\[
 x_4=\frac35x_1+\frac35x_2+\frac{\sqrt7}{10}x_3
                         +\frac{\sqrt{21}}{10}z.
\]
The remaining $n-4$ columns are orthogonal to this block and each
other, and only $x_1,x_2,x_3$ have nonzero coefficients in
\[
 y=x_1+x_2+\beta_{3,T}x_3+u,\quad
 u\sim N(0,\sigma^2I_T),\quad
 \beta_{3,T}=\frac{31}{20}\sigma\sqrt{\frac{2\log n}{T}},\quad
 \delta=\frac{49}{25}.
\]
An intercept is always included, and $\alpha_0\in(0,1)$ and
$c\geq\alpha_0$ are fixed as in \eqref{cvf}. All columns have unit
empirical second moment, while the proxy's innovation variance
$21/100>0$ makes the four-column second-moment matrix positive definite.
Since $n\leq T-1$ eventually, the construction fits in the centred
observation space.

\begin{proposition}\label{prop:r2_rescue}
In this experiment, with probability tending to one,
\[
 \widehat S_{BMT}=\{1,2,4\},\qquad
 \widehat S_{R1}=\{1,2\},\qquad
 \widehat S_{R2}=\{1,2,3\}.
\]
BMT-R2 deletes the old proxy immediately after the second strong signal
enters, and the deleted proxy fails the next entry threshold.
\end{proposition}
\begin{proof}
Let $q=7/5$ and $C=31/20$, so that the cutoff divided by
$\sqrt{2\log n}$ tends to $q$ at every bounded size forward or retest
state. The vanishing weak coefficient does not affect comparisons
between fixed nonzero limiting scores. In particular, the proxy's
marginal covariance with $y$ tends to
$6/5$, whereas each strong signal's tends to one. The columns have
equal norms, so the proxy has the largest absolute marginal partial
correlation and enters first with probability tending to one.
After conditioning on the proxy, the absolute residual standardised
covariance for either strong signal tends to
\[
 \frac{1-(3/5)(6/5)}{\sqrt{1-9/25}}=\frac7{20}=0.35.
\]
For $x_3$ it tends to
$(6/5)\sqrt{7/100}/\sqrt{93/100}<0.330$.
The common outcome residual variance gives the same ranking for
partial correlations and hence for absolute statistics, so a strong
signal enters next. Symmetry of the two strong signals makes their
order immaterial. After one of them has entered,
the other has limiting residual standardised covariance
$\sqrt{7/16}>0.66$, whereas $x_3$ has absolute value
\[
 \frac{(3/5)(\sqrt7/10)/(16/25)}
      {\sqrt{1-(7/100)/(16/25)}}<0.264.
\]
The remaining strong signal therefore enters next, with an entry
statistic of order $\sqrt T$ and a nonzero leading constant. The
joint statistics of the previously selected strong signals have the
same order, as does the proxy's joint statistic after the first strong
admission. Consequently no earlier BMT-R2 deletion occurs.

Conditional on both strong signals, the proxy has variance $7/25$
and correlation $1/2$ with $x_3$. The joint statistic of the proxy
in $\{1,2,4\}$ and the entry statistic of $x_3$ given this model
therefore satisfy, respectively,
\[
 \frac{t^-_{4,\{1,2,4\}}}{\sqrt{2\log n}}
       \ \xrightarrow{p}\ \frac C2=\frac{31}{40}<q,
 \qquad
 \frac{t_{3\mid\{1,2,4\}}}{\sqrt{2\log n}}
       \ \xrightarrow{p}\ \frac{\sqrt3 C}{2}
       =\frac{31\sqrt3}{40}<q.
\]
Both strong signals pass their joint retests, but BMT stops without
$x_3$. BMT-R1 removes the proxy at its final check and cannot add $x_3$,
whereas BMT-R2 removes the old proxy as soon as the second strong signal
enters. At the resulting state $A=\{1,2\}$, the next two relevant
scores satisfy
\[
 \frac{t_{3\mid A}}{\sqrt{2\log n}}\xrightarrow{p} C>q,
 \qquad
 \frac{t_{4\mid A}}{\sqrt{2\log n}}\xrightarrow{p} C/2<q.
\]
Thus $x_3$ enters and the old strong signals pass their retests,
while the proxy cannot immediately reenter. At $S=\{1,2,3\}$,
every inactive candidate has zero partial coefficient.

To establish noise exclusion and terminal stopping jointly with these
comparisons, note that there are finitely many core conditioning
models. At each model, the orthogonal noise numerators divided by
$\sigma\sqrt T$ are standard Gaussian, and their largest absolute
value is at most
$(1+\epsilon)\sqrt{2\log n}$ with probability tending to one, for
any fixed $\epsilon>0$. Uniformly over these noise candidates, their
single column residual variance reduction is
$O_p(\log n/T)$, and the fitted residual variance converges to the
model's population residual variance, which is at least $\sigma^2$.
Taking $0<\epsilon<q-1$, with a smaller margin if necessary for
that variance convergence, ensures that no noise statistic exceeds
the cutoff.
The finite core statistics have $O_p(1)$ errors about their Gaussian
noncentral leading terms, so their strict margins of order
$\sqrt{\log n}$ give all the stated decisions. Ordinary BMT and BMT-R1
follow the paths described above, while BMT-R2 visits the distinct states after retesting
$\varnothing,\{4\},\{4,i\},\{1,2\},S$, where $i=1$ or $2$.
No repetition or undefined fit occurs on this event.
\end{proof}

The weak coefficient in this three signal example is a fixed multiple
of the testing scale, so it satisfies neither the main theorem's
diverging gap condition nor \eqref{eq:retesting_beta_min}. The result
is established by the direct Gaussian calculation above, not by
Theorem~\ref{th:retesting_exact}. It shows that interleaved retesting
can recover an omitted signal, without implying that BMT-R2 uniformly
improves on BMT or BMT-R1.

Forward--backward selection and multiplicity adjusted stepwise methods
have a substantial literature. \citet{IngLai2011} study consistent
selection along an orthogonal greedy path; \citet{IngEtAl2017} study
multiple testing in sparse regression for fault diagnosis, including
backward elimination and family wise error control.
\citet{ZambomKim2022} provide forward, backward and stepwise procedures
with multiple testing corrections. The contribution here is the analysis
of these two retesting schedules with BMT's cutoff, including the
proxy removal and signal rescue conditions. It is not a claim to the
first forward--backward or multiplicity adjusted selection method.

\newpage
\section{Simulation Evidence}

\renewcommand{\thetable}
{B\arabic{table}} \setcounter{table}{0}

This appendix provides detailed results for the finite sample performance of all methods considered in Section~4 across different settings. We start with a
description of the metrics of comparison being employed.

Let $TP$, $FP$, $TN$ and $FN$ denote the number of true positives, false positives, true negatives and false negatives, respectively. We employ the following performance measures:

\begin{landscape}
\begin{table}[ht]
\centering
\renewcommand{\arraystretch}{1.3}
\caption{Performance measures used to evaluate model selection and forecasting accuracy.}
\begin{tabular}{|c|c|c|c|}
\hline
\multicolumn{4}{|c|}{\textbf{Performance Measures}} \\
\hline
\textbf{Name} & \textbf{Definition} & \textbf{Metric} & \textbf{Selection Criterion} \\
\hline
TPR & $\frac{\text{TP}}{\text{TP} + \text{FN}}$ & Share of true signals correctly identified & Higher TPR \\
\hline
FPR & $\frac{\text{FP}}{\text{FP} + \text{TN}}$ & Share of nonsignals incorrectly selected  & Lower FPR \\
\hline
TDR & $\frac{\text{TP}}{\text{TP} + \text{FP}}$ & Share of selected variables that are true signals & Higher TDR \\
\hline
FDR & $\frac{\text{FP}}{\text{TP} + \text{FP}}$ & Share of selected variables that are nonsignals & Lower FDR \\
\hline
MCC & $\frac{\text{TP} \cdot \text{TN} - \text{FP} \cdot \text{FN}}{\sqrt{(\text{TP}+\text{FP})(\text{TP}+\text{FN})(\text{TN}+\text{FP})(\text{TN}+\text{FN})}}$ & Overall classification quality & Higher MCC \\
\hline
F1 Score & $ 2 \times \frac{\text{TDR} \times \text{TPR}}{\text{TDR} + \text{TPR}}$ & Harmonic mean of TPR and TDR & Higher F1 Score \\
\hline
Exact Recovery
& $\frac{1}{r}\sum_{j=1}^{r}
\mathbbm{1}\!\left\{\widehat{\mathcal S}^{(j)}=\mathcal S_0\right\}$
& Share of replications recovering the exact model
& Higher Exact Recovery \\
\hline
$\widehat{K}$ & No. of variables selected & Estimated model size & Closer to true $K$ \\
\hline
RMSE & $\sqrt{\frac{1}{r} \sum_{j=1}^{r} \left\| \widetilde{\boldsymbol{\beta}}_{n}^{(j)} - \boldsymbol{\beta}_{n} \right\|^2}$ &In sample estimation error & Lower RMSE \\
\hline
RMSFE & $\sqrt{ \frac{1}{r} \sum_{j=1}^{r} \left( \frac{1}{S} \sum_{t=T+1}^{T+S} (y_t^{(j)} - \widehat{y}_t^{(j)})^2 \right) }$ & Out of sample forecast error & Lower RMSFE \\
\hline
\end{tabular}
\vspace{0.5em}
\parbox{\linewidth}{
\small
\textit{\textbf{Notes}:} MCC denotes Matthews Correlation Coefficient, also known as the $\phi$ coefficient in statistics. It measures the correlation between predicted and true binary outcomes, reflecting how well the selected outcomes align with the actual ones. MCC ranges from $-1$ (perfect misclassification) to $+1$ (perfect selection), with $0$ indicating random guessing. The F1 Score, defined as the harmonic mean of TPR and TDR, balances sensitivity and precision. A high F1 indicates effective signal detection while controlling for false positives. MCC uses all four classification counts (TP, FP, TN, FN), whereas F1 does not use true negatives; see e.g. Chicco and Jurman (2020). TPR and FPR each describe one aspect of selection performance. For instance, a method that indiscriminately selects all variables yields a true positive rate of 1, but also a false positive rate of 1. In general, there is often a tradeoff between FPR and TPR: methods that correctly identify more true signals tend to be less conservative, increasing the likelihood of selecting nonsignals and thus raising the false positive rate. $\widetilde{\boldsymbol{\beta}}_{n}^{(j)}$ is the $n$ dimensional vector that coincides with the post selection (OLS) estimates on the selected coordinates and is zero elsewhere. For RMSFE, model selection and post selection OLS estimation are performed using the first $T$ observations. The selected specification and
estimated coefficients are then held fixed when forecasting the subsequent $S=1$ observation. The one step ahead forecast uses the observed lagged dependent variable and the corresponding regressor values.
}
\end{table}
\end{landscape}


\vfill
\clearpage

\begin{landscape}
\thispagestyle{empty} 
\centering
\begin{minipage}{0.95\linewidth}
\vspace{-0.5em}
\begin{minipage}{0.48\linewidth}
\centering
\vspace{-0.4em}
    V = 4, $\pi=0$
\includegraphics[width=\linewidth,height=0.28\textheight,keepaspectratio]
{figures_n/MCC_VIF_4_pi_0_alpha_08.png}\\
\end{minipage}
\hfill
\begin{minipage}{0.48\linewidth}
\centering
\vspace{-0.4em}
    V = 2, $\pi=0$
\includegraphics[height=0.28\textheight]{figures_n/MCC_VIF_2_pi_0_alpha_08.png}\\
\end{minipage}
\begin{minipage}{0.48\linewidth}
\centering
\vspace{0.4em}
    V = 4, $\pi=0.5$
\includegraphics[height=0.28\textheight]{figures_n/MCC_VIF_4_pi_0.5_alpha_08.png}\\
\end{minipage}
\hfill
\begin{minipage}{0.48\linewidth}
\centering
\vspace{0.4em}
    V = 2, $\pi=0.5$
\includegraphics[height=0.28\textheight]{figures_n/MCC_VIF_2_pi_0.5_alpha_08.png}\\
\end{minipage}
\begin{minipage}{0.48\linewidth}
\centering
\vspace{0.2em}
    V = 4, $\pi=1$
\includegraphics[height=0.28\textheight]{figures_n/MCC_VIF_4_pi_1_alpha_08.png}\\
\end{minipage}
\hfill
\begin{minipage}{0.48\linewidth}
\centering
\vspace{0.2em}
    V = 2, $\pi=1$
\includegraphics[height=0.28\textheight]{figures_n/MCC_VIF_2_pi_1_alpha_08.png}\\
\end{minipage}
\vspace{-0.4em}
\captionof{figure}{Visual Summary of MCC Performance for $\alpha=0.8$}
\end{minipage}
\end{landscape}

\begin{landscape}
\thispagestyle{empty} 
\centering
\begin{minipage}{0.95\linewidth}
\vspace{-0.5em}
\begin{minipage}{0.48\linewidth}
\centering
\vspace{-0.4em}
    V = 4, $\pi=0$
\includegraphics[width=\linewidth,height=0.28\textheight,keepaspectratio]
{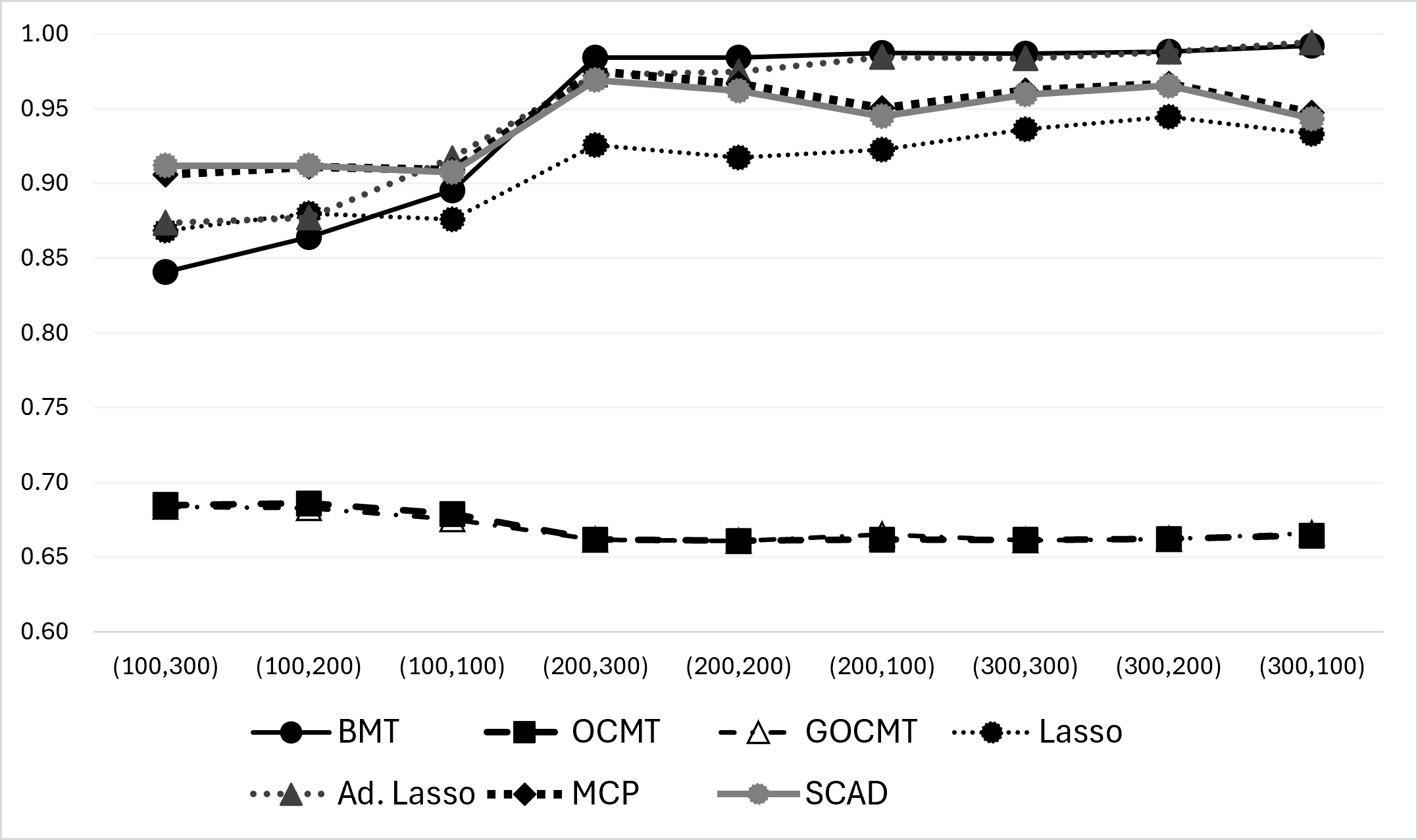}\\
\end{minipage}
\hfill
\begin{minipage}{0.48\linewidth}
\centering
\vspace{-0.4em}
    V = 2, $\pi=0$
\includegraphics[height=0.28\textheight]{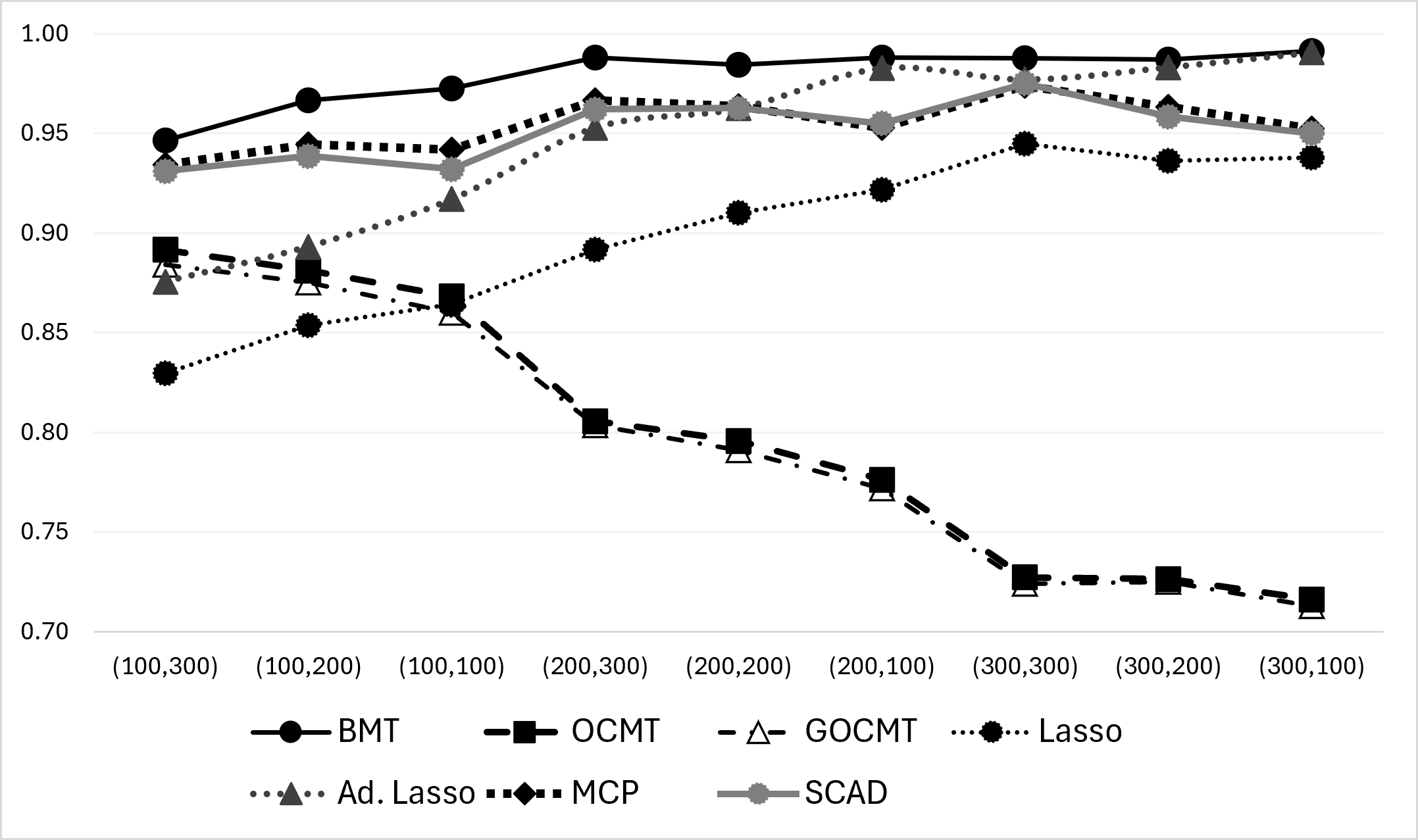}\\
\end{minipage}
\begin{minipage}{0.48\linewidth}
\centering
\vspace{0.4em}
    V = 4, $\pi=0.5$
\includegraphics[height=0.28\textheight]{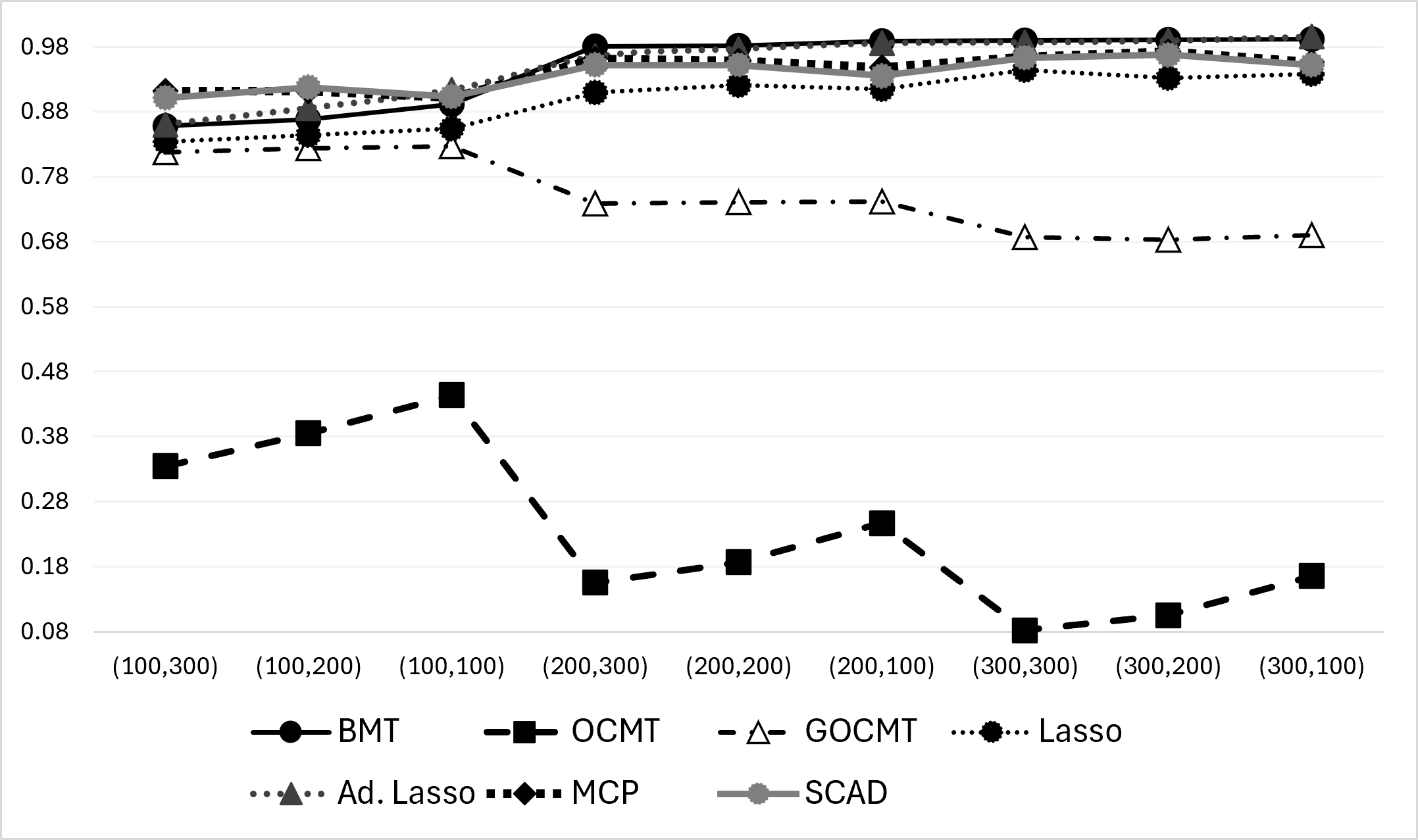}\\
\end{minipage}
\hfill
\begin{minipage}{0.48\linewidth}
\centering
\vspace{0.4em}
    V = 2, $\pi=0.5$
\includegraphics[height=0.28\textheight]{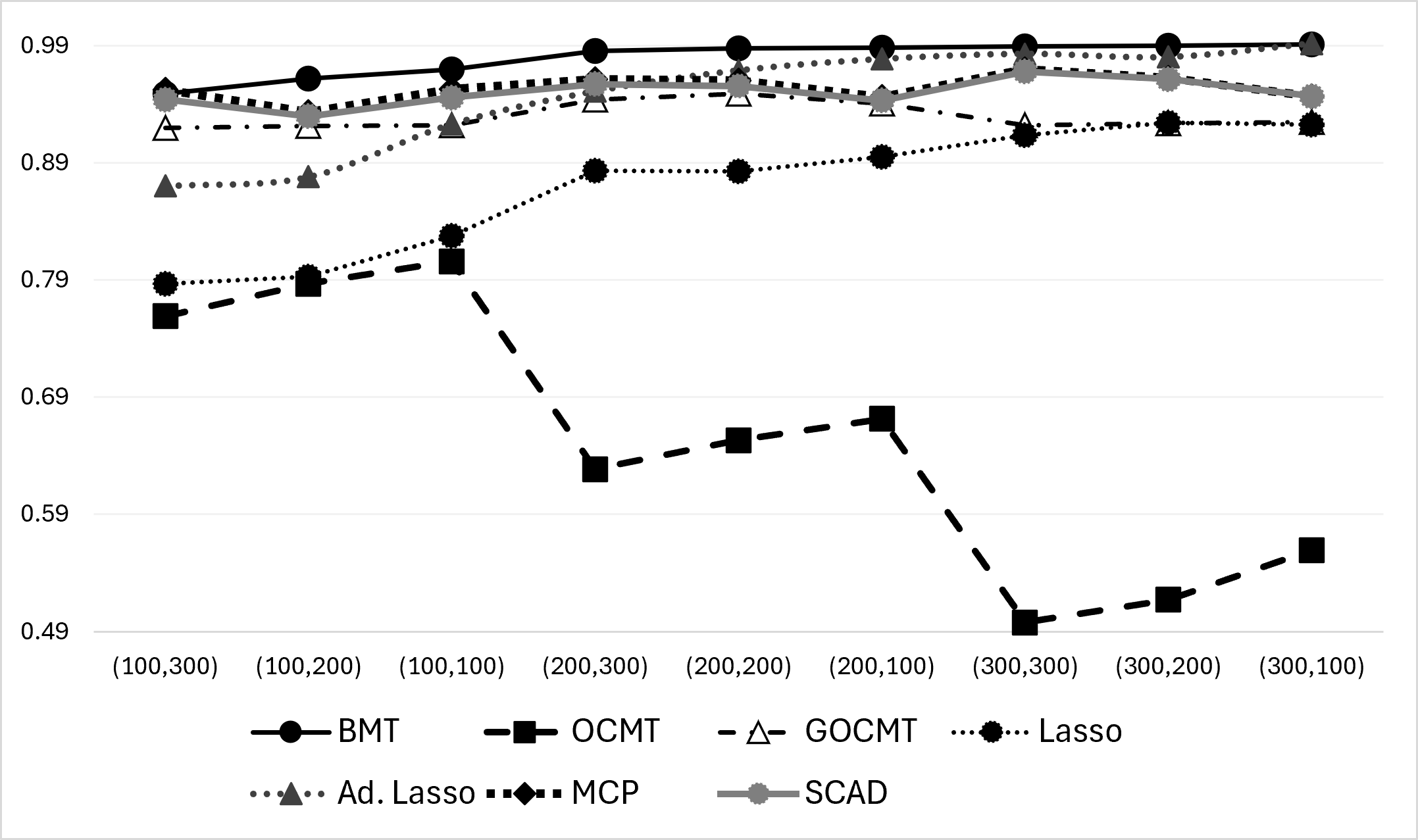}\\
\end{minipage}
\begin{minipage}{0.48\linewidth}
\centering
\vspace{0.2em}
    V = 4, $\pi=1$
\includegraphics[height=0.28\textheight]{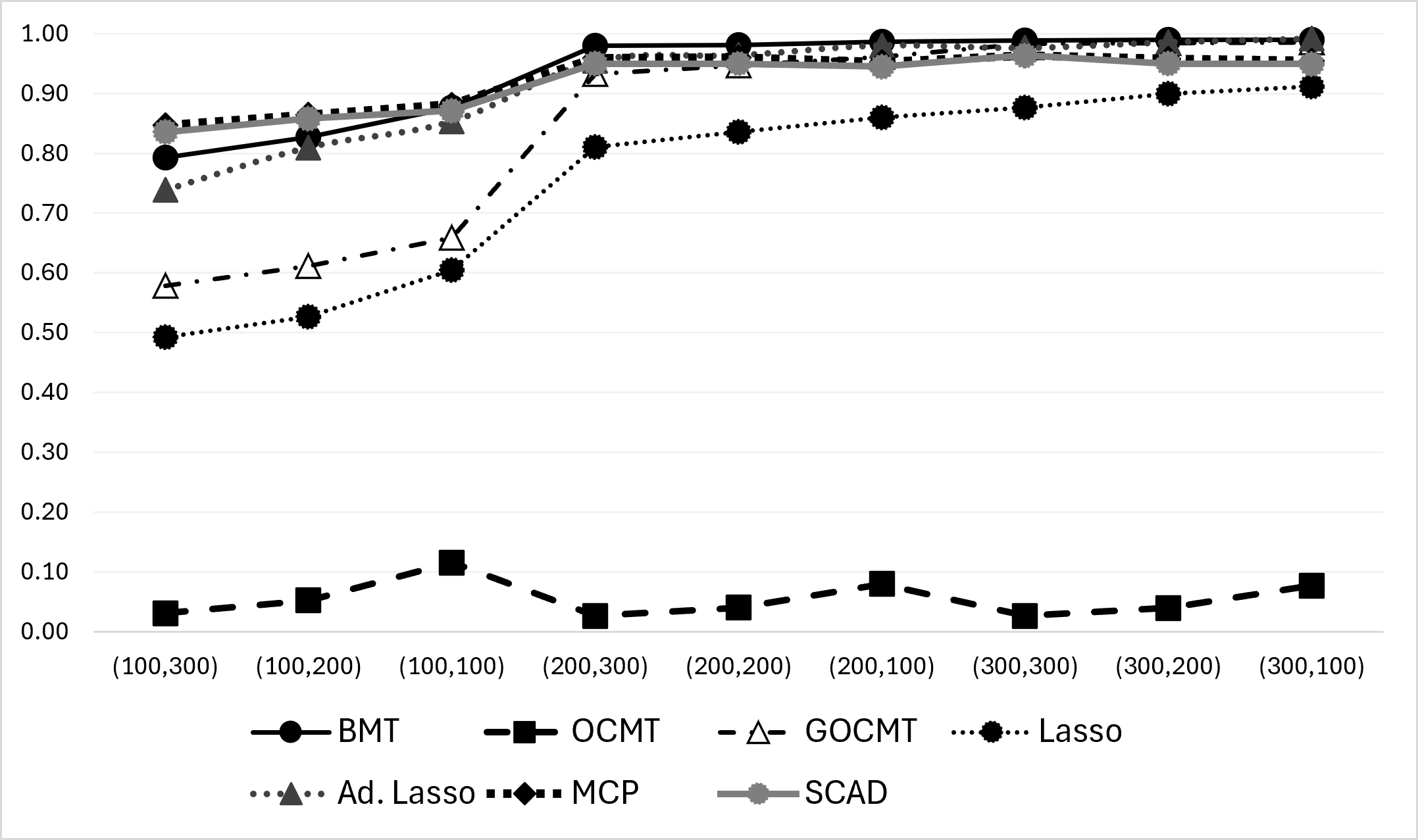}\\
\end{minipage}
\hfill
\begin{minipage}{0.48\linewidth}
\centering
\vspace{0.2em}
    V = 2, $\pi=1$
\includegraphics[height=0.28\textheight]{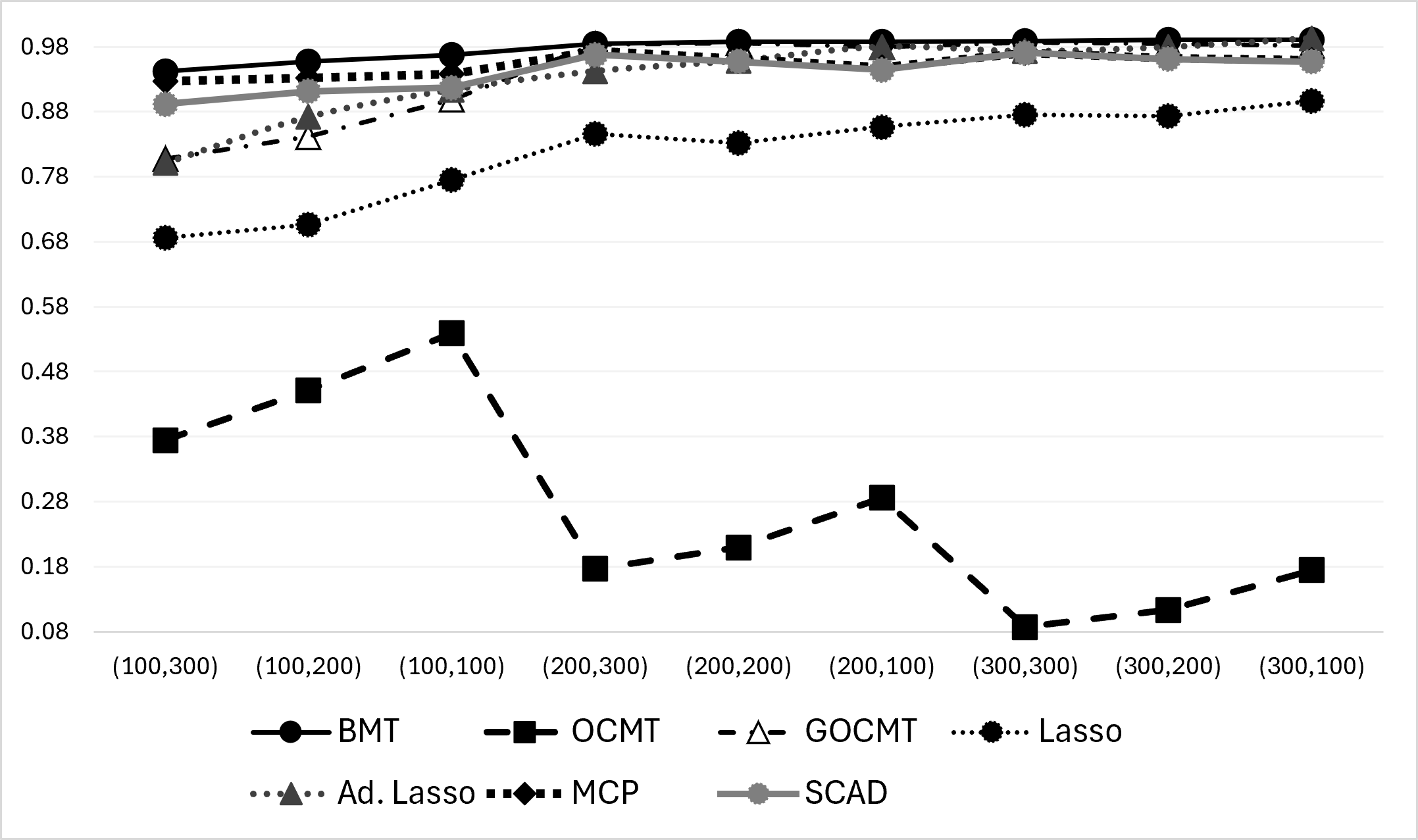}\\
\end{minipage}
\vspace{-0.4em}
\captionof{figure}{Visual Summary of F1 Performance for $\alpha=0.8$}
\end{minipage}
\end{landscape}

\begin{landscape}
\thispagestyle{empty} 
\centering
\begin{minipage}{0.95\linewidth}
\vspace{-0.5em}
\begin{minipage}{0.48\linewidth}
\centering
\vspace{-0.4em}
    V = 4, $\pi=0$
\includegraphics[width=\linewidth,height=0.28\textheight,keepaspectratio]
{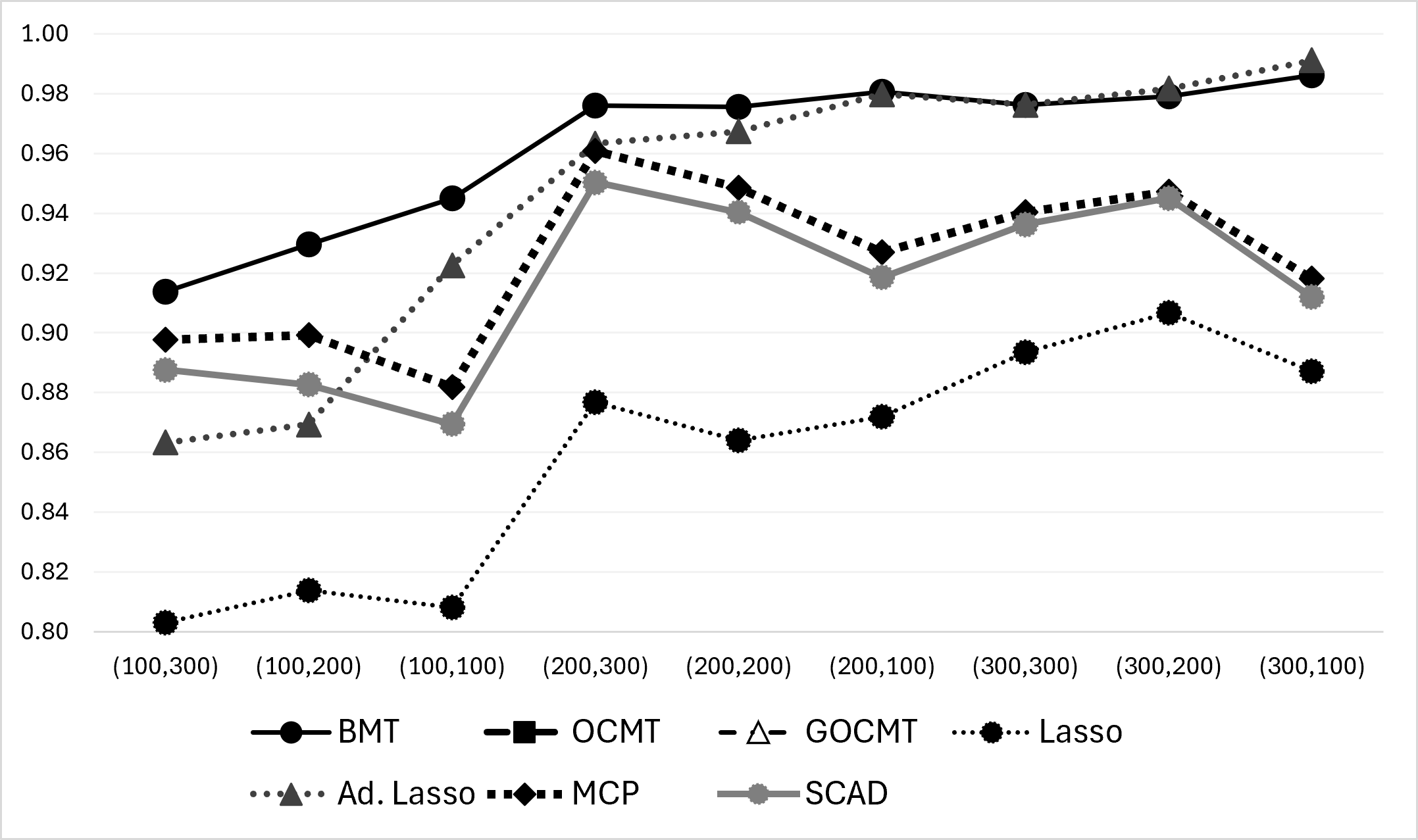}\\
\end{minipage}
\hfill
\begin{minipage}{0.48\linewidth}
\centering
\vspace{-0.4em}
    V = 2, $\pi=0$
\includegraphics[height=0.28\textheight]{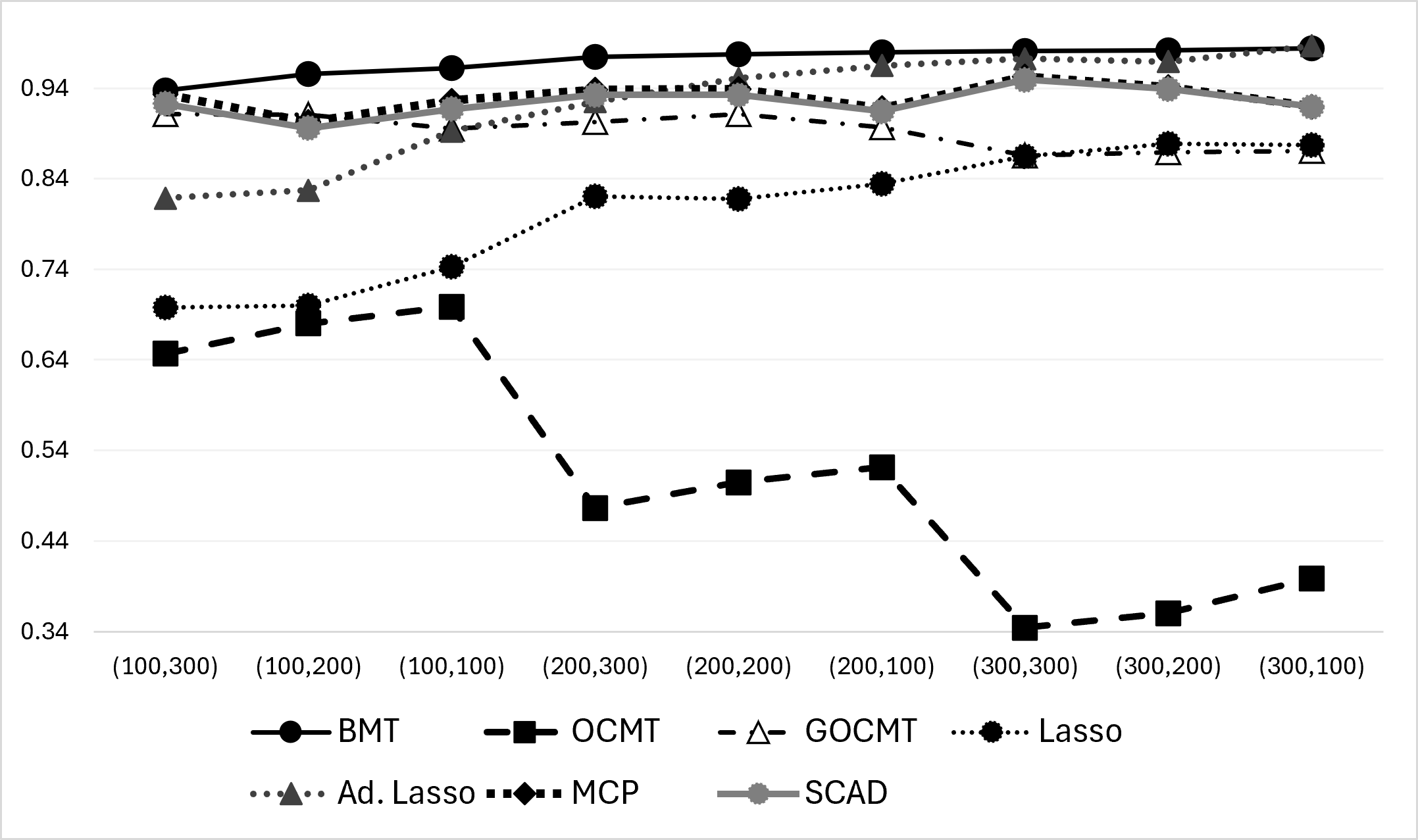}\\
\end{minipage}
\begin{minipage}{0.48\linewidth}
\centering
\vspace{0.4em}
    V = 4, $\pi=0.5$
\includegraphics[height=0.28\textheight]{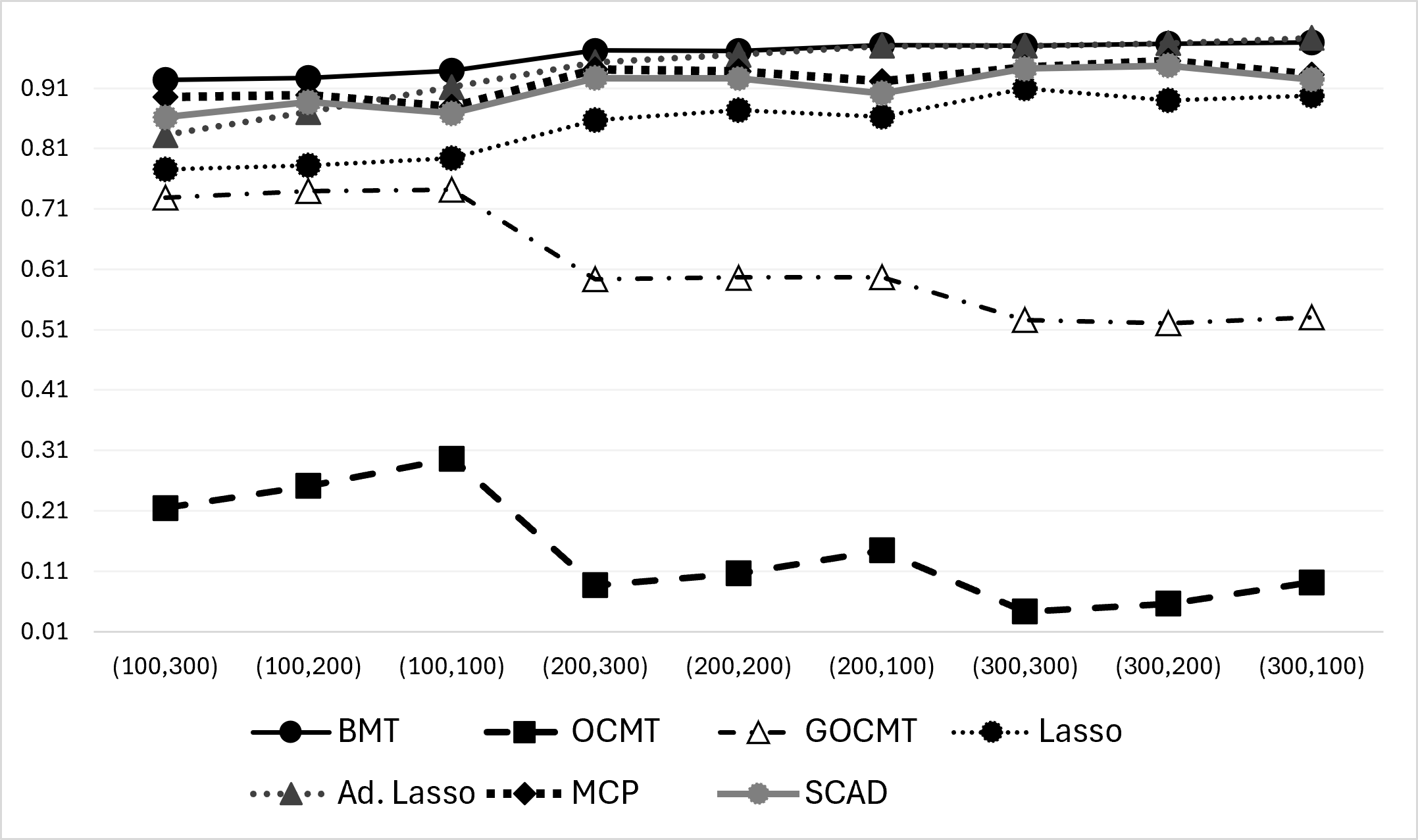}\\
\end{minipage}
\hfill
\begin{minipage}{0.48\linewidth}
\centering
\vspace{0.4em}
    V = 2, $\pi=0.5$
\includegraphics[height=0.28\textheight]{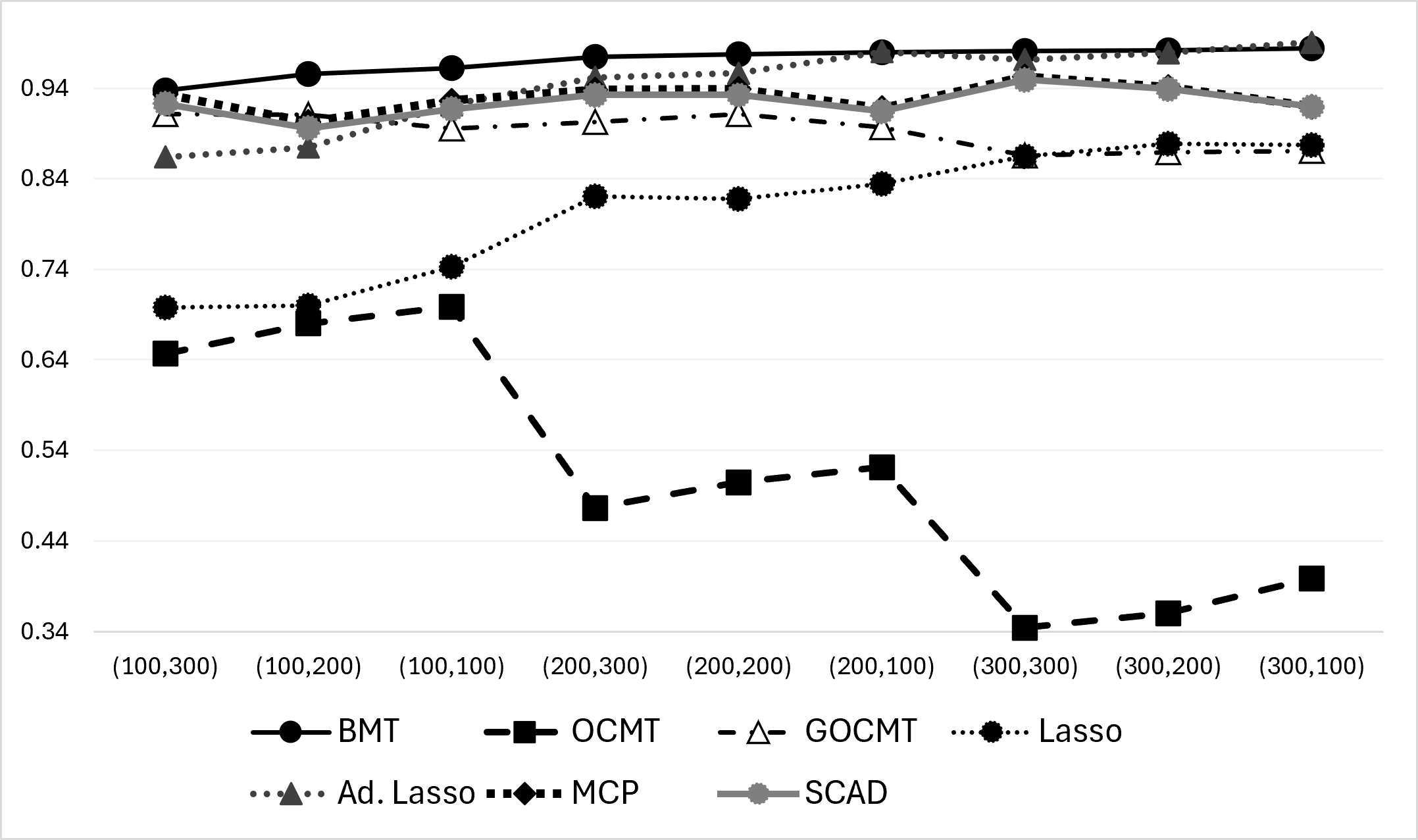}\\
\end{minipage}
\begin{minipage}{0.48\linewidth}
\centering
\vspace{0.2em}
    V = 4, $\pi=1$
\includegraphics[height=0.28\textheight]{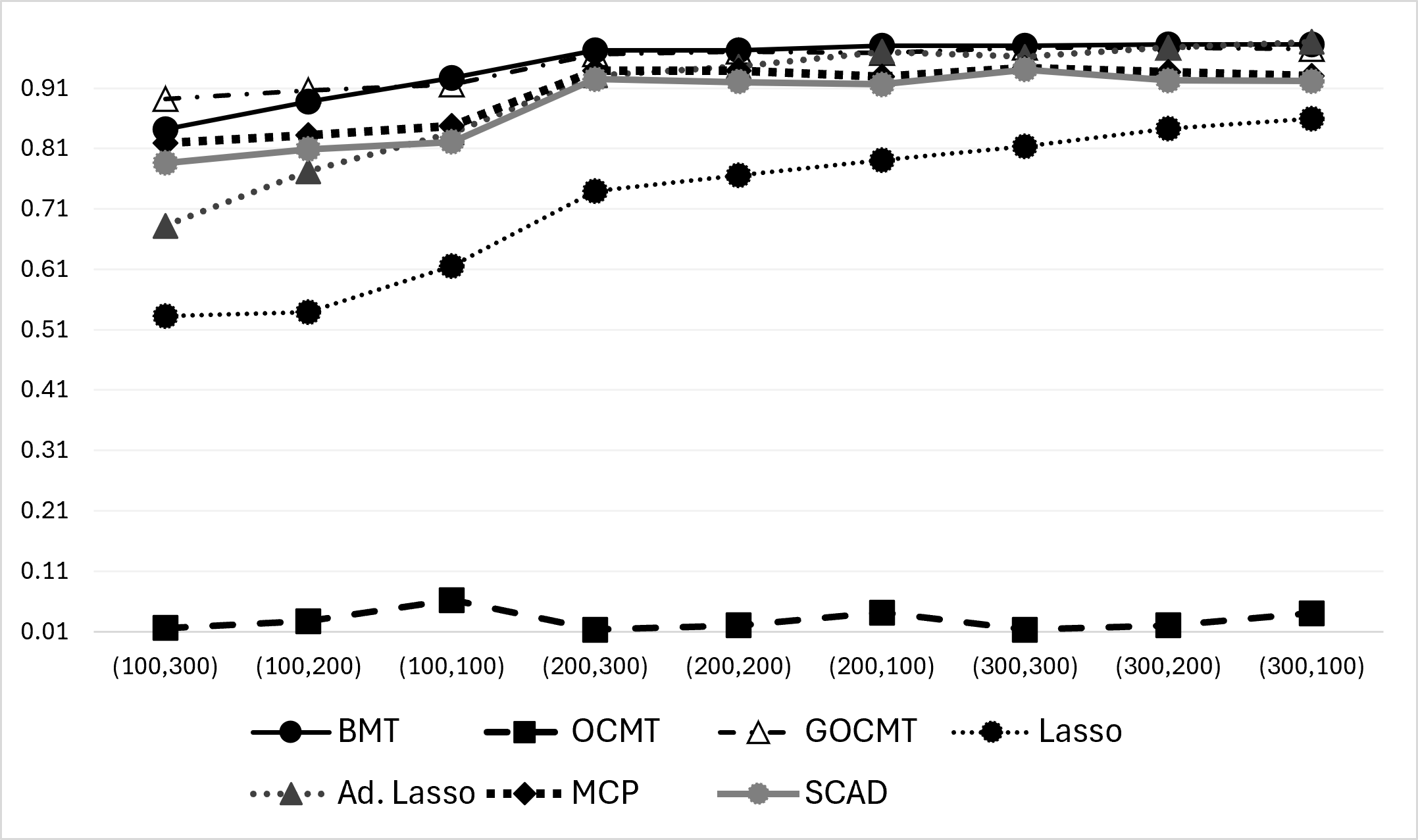}\\
\end{minipage}
\hfill
\begin{minipage}{0.48\linewidth}
\centering
\vspace{0.2em}
    V = 2, $\pi=1$
\includegraphics[height=0.28\textheight]{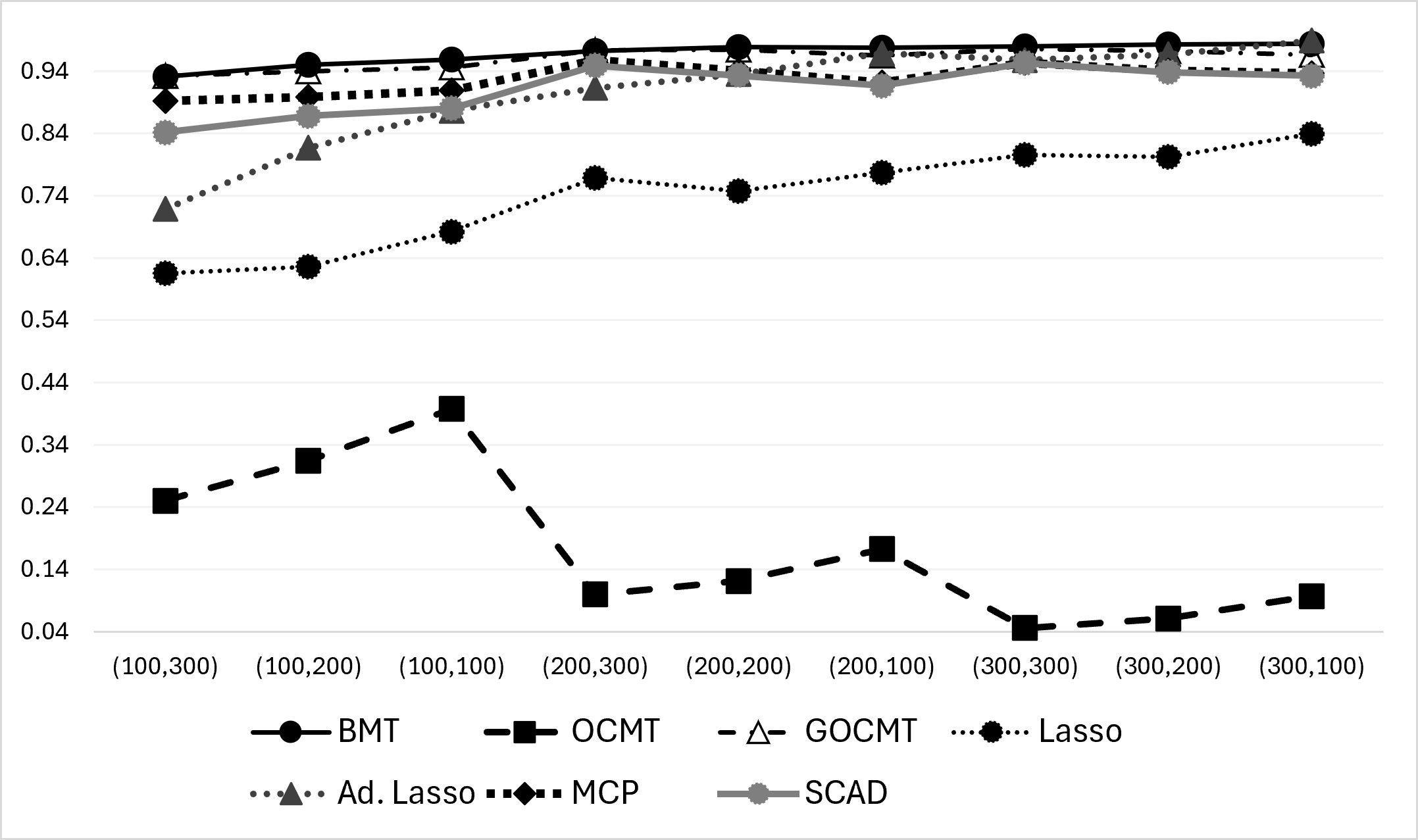}\\
\end{minipage}
\vspace{-0.4em}
\captionof{figure}{Visual Summary of TDR Performance for $\alpha=0.8$}
\end{minipage}
\end{landscape}

\begin{landscape}
\thispagestyle{empty} 
\centering
\begin{minipage}{0.95\linewidth}
\vspace{-0.5em}
\begin{minipage}{0.48\linewidth}
\centering
\vspace{-0.4em}
    V = 4, $\pi=0$
\includegraphics[width=\linewidth,height=0.28\textheight,keepaspectratio]{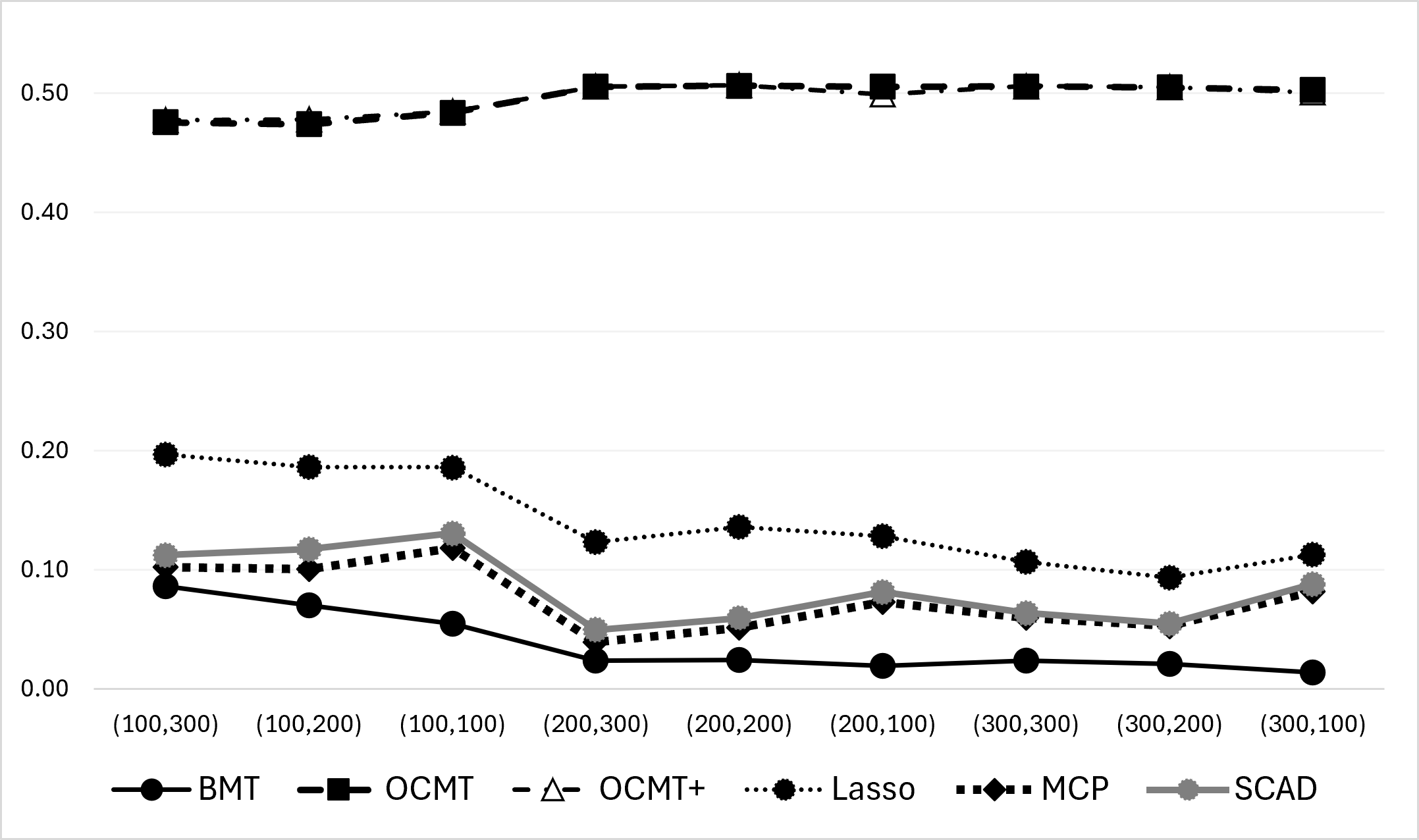}\\
\end{minipage}
\hfill
\begin{minipage}{0.48\linewidth}
\centering
\vspace{-0.4em}
    V = 2, $\pi=0$
\includegraphics[height=0.28\textheight]{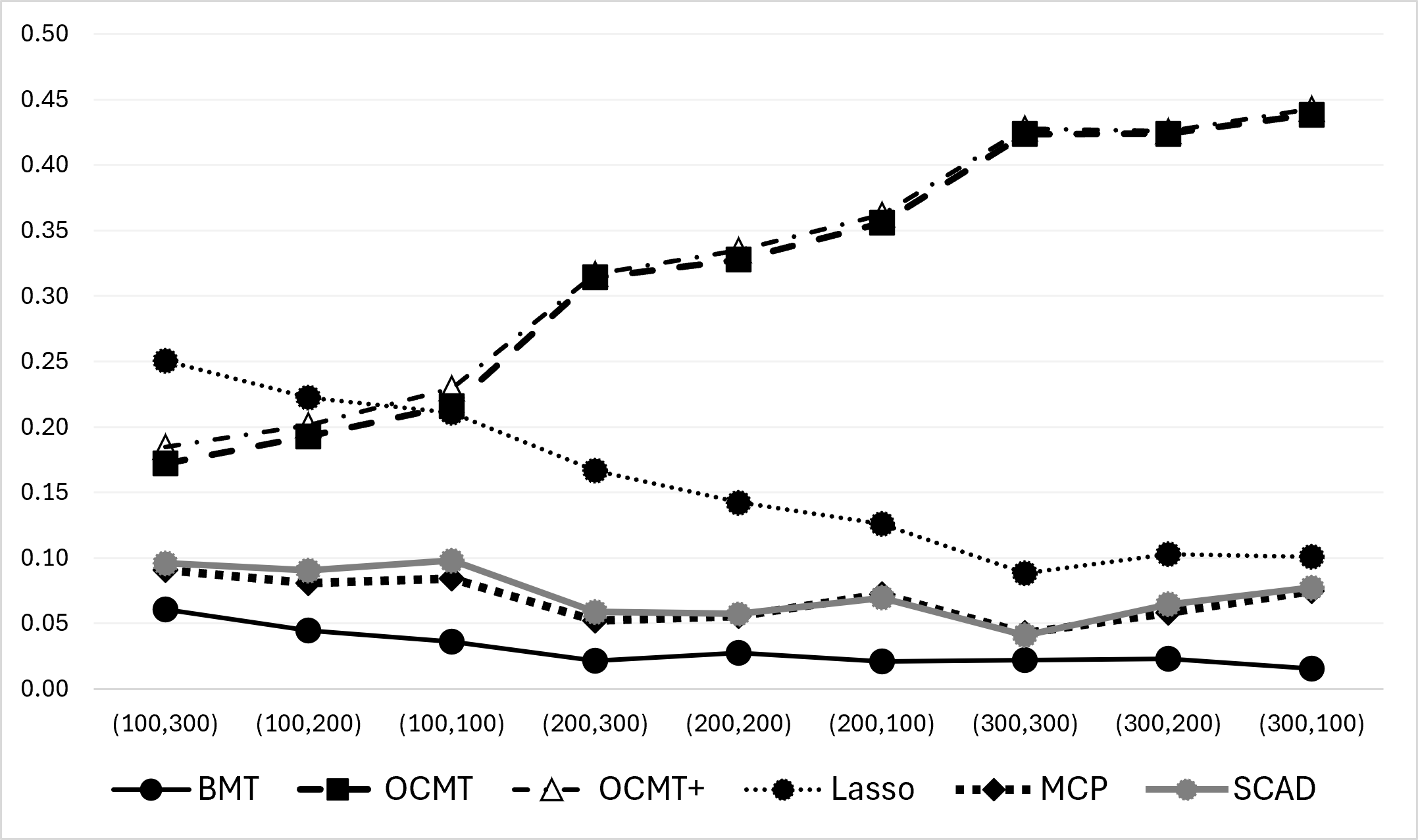}\\
\end{minipage}
\begin{minipage}{0.48\linewidth}
\centering
\vspace{0.4em}
    V = 4, $\pi=0.5$
\includegraphics[height=0.28\textheight]{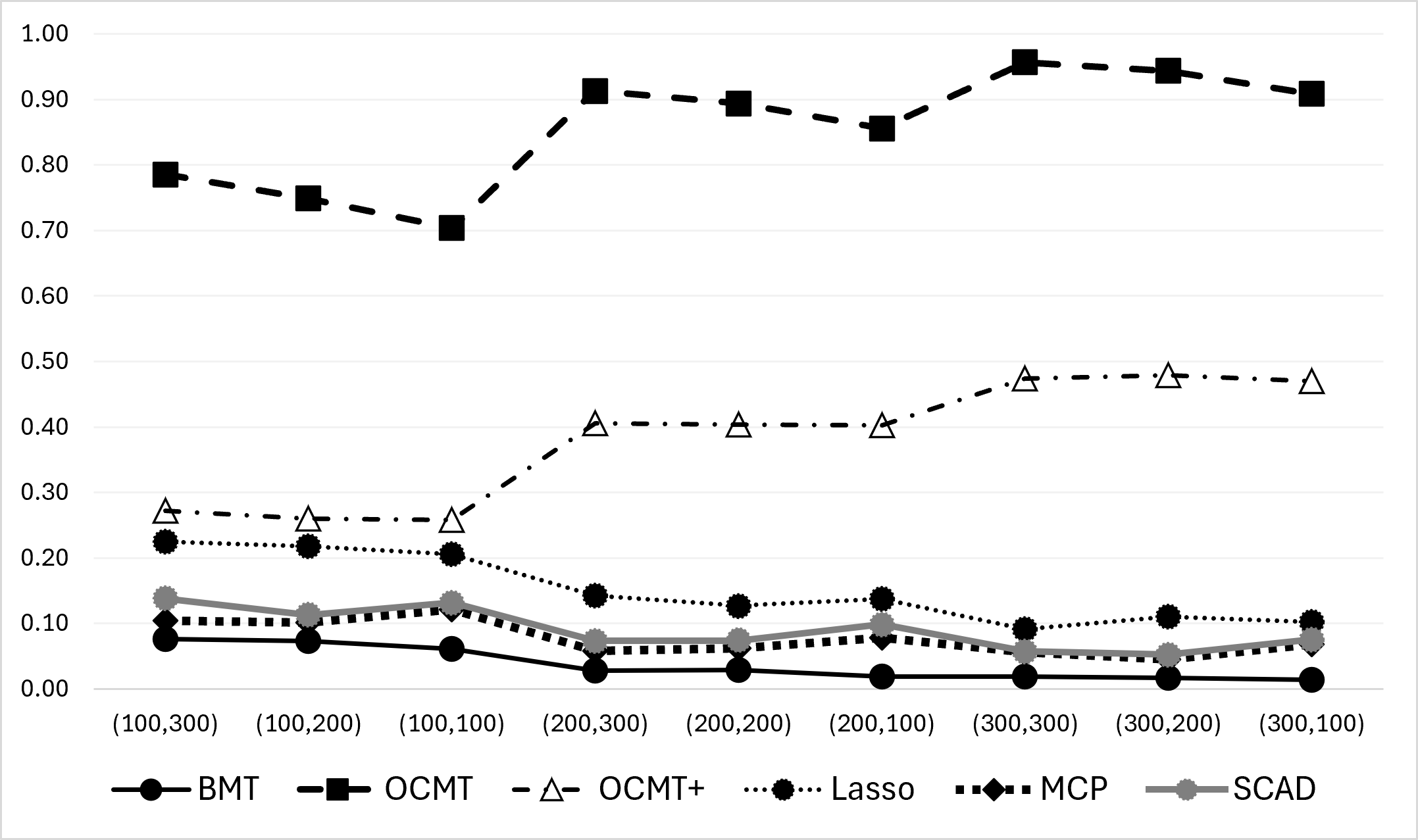}\\
\end{minipage}
\hfill
\begin{minipage}{0.48\linewidth}
\centering
\vspace{0.4em}
    V = 2, $\pi=0.5$
\includegraphics[height=0.28\textheight]{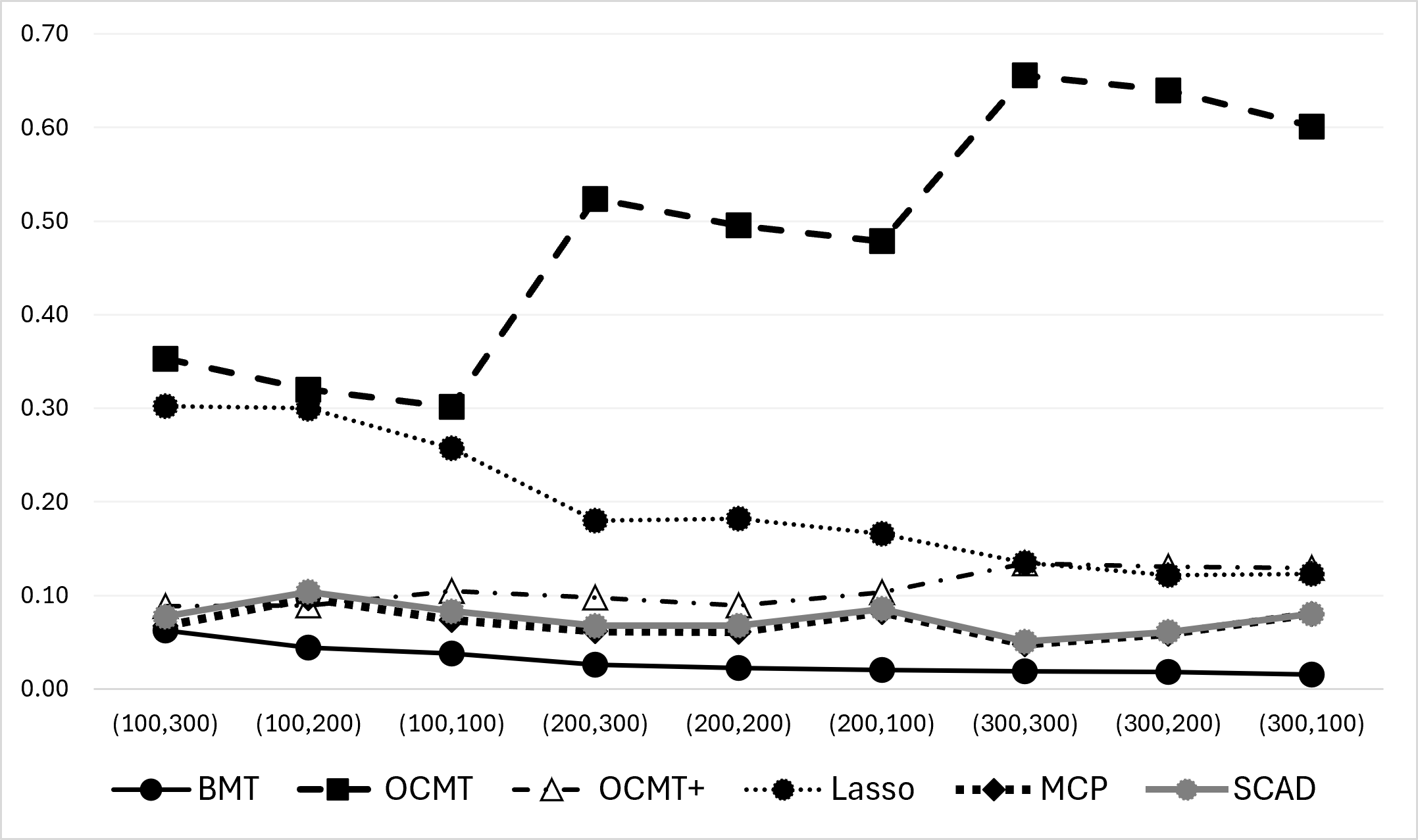}\\
\end{minipage}
\begin{minipage}{0.48\linewidth}
\centering
\vspace{0.2em}
    V = 4, $\pi=1$
\includegraphics[height=0.28\textheight]{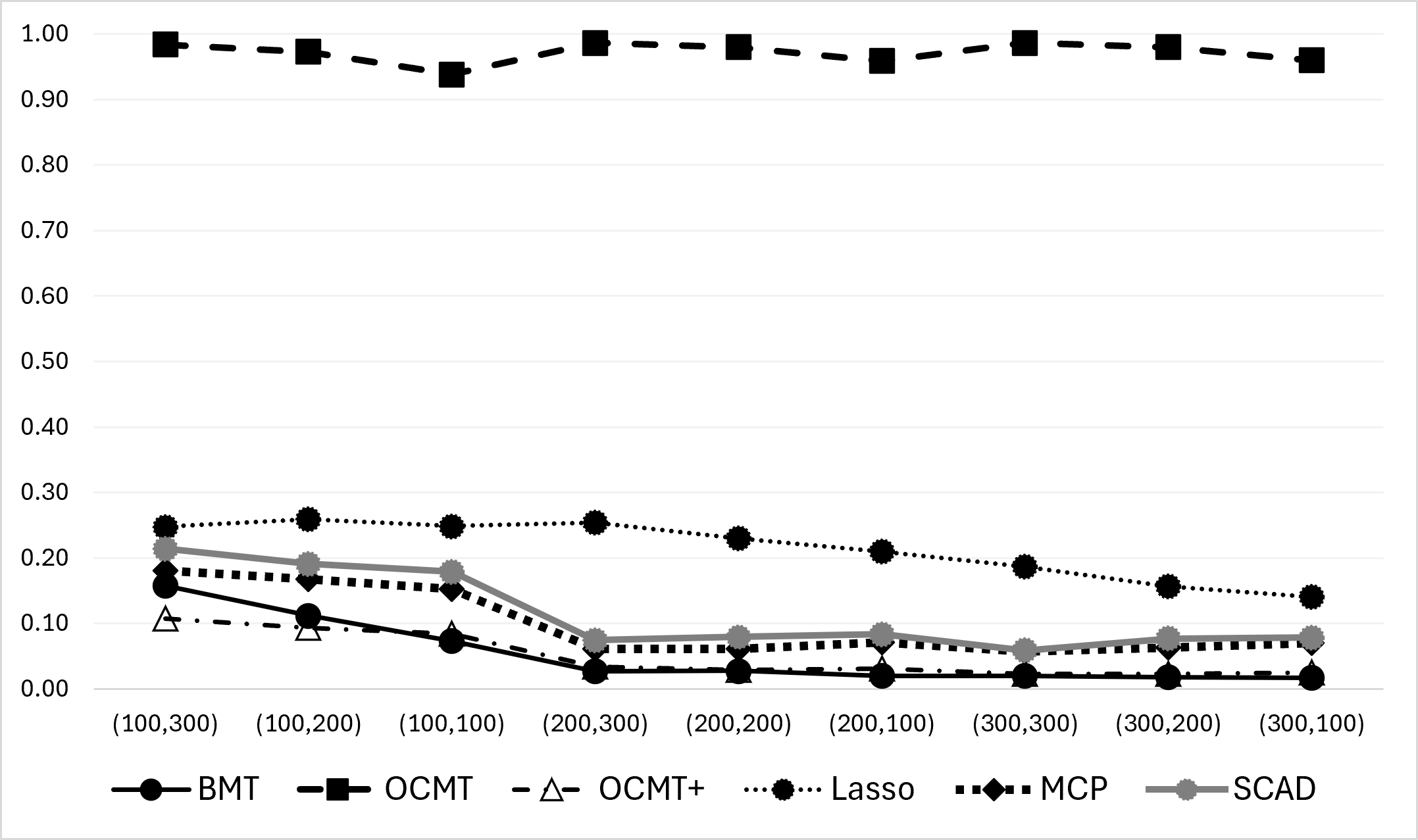}\\
\end{minipage}
\hfill
\begin{minipage}{0.48\linewidth}
\centering
\vspace{0.2em}
    V = 2, $\pi=1$
\includegraphics[height=0.28\textheight]{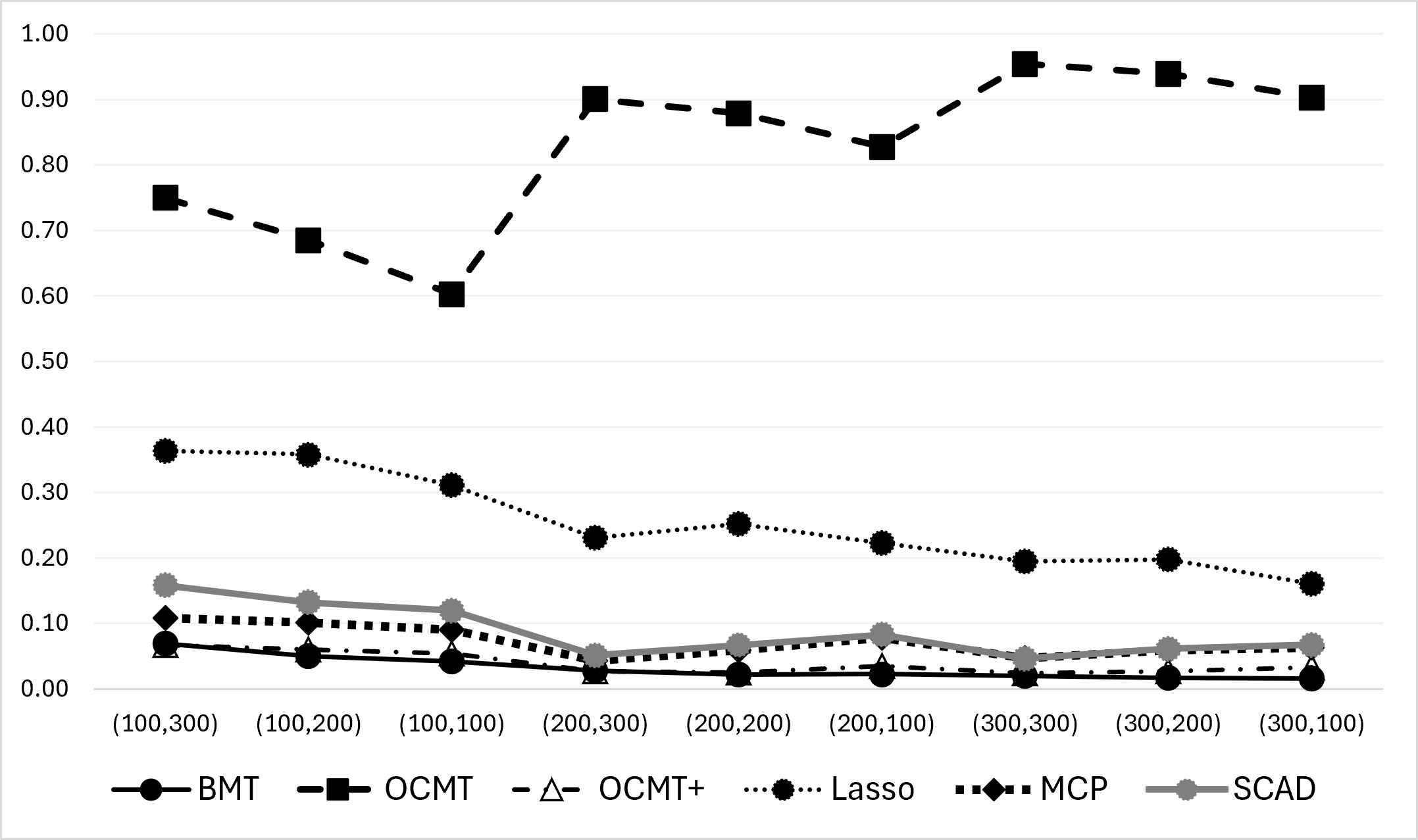}\\
\end{minipage}
\vspace{-0.4em}
\captionof{figure}{Visual Summary of FDR Performance for $\alpha=0.8$}
\end{minipage}
\end{landscape}

\begin{landscape}
\thispagestyle{empty} 
\centering
\begin{minipage}{0.95\linewidth}
\vspace{-0.5em}
\begin{minipage}{0.48\linewidth}
\centering
\vspace{-0.4em}
    V = 4, $\pi=0$
\includegraphics[width=\linewidth,height=0.28\textheight,keepaspectratio]{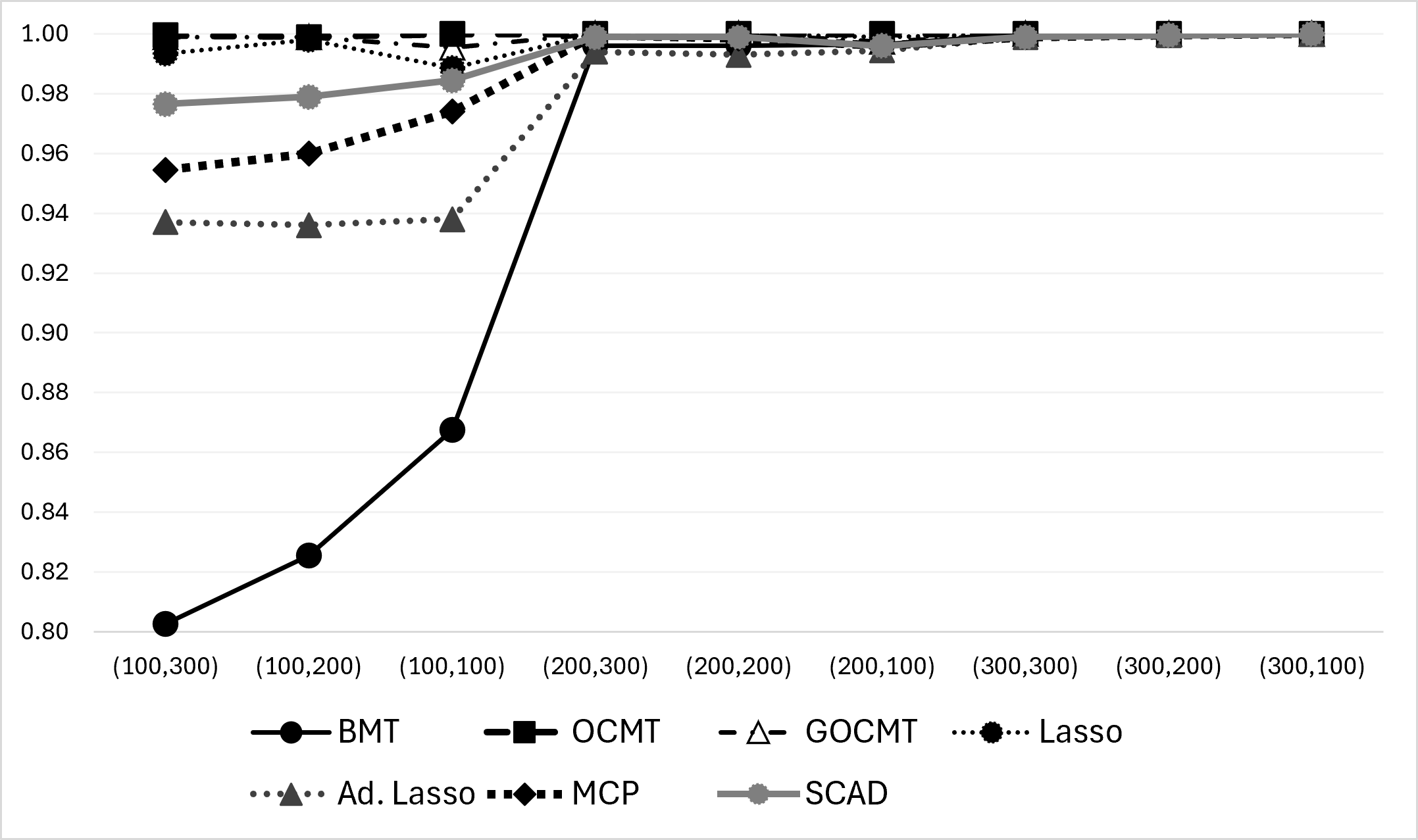}\\
\end{minipage}
\hfill
\begin{minipage}{0.48\linewidth}
\centering
\vspace{-0.4em}
    V = 2, $\pi=0$
\includegraphics[height=0.28\textheight]{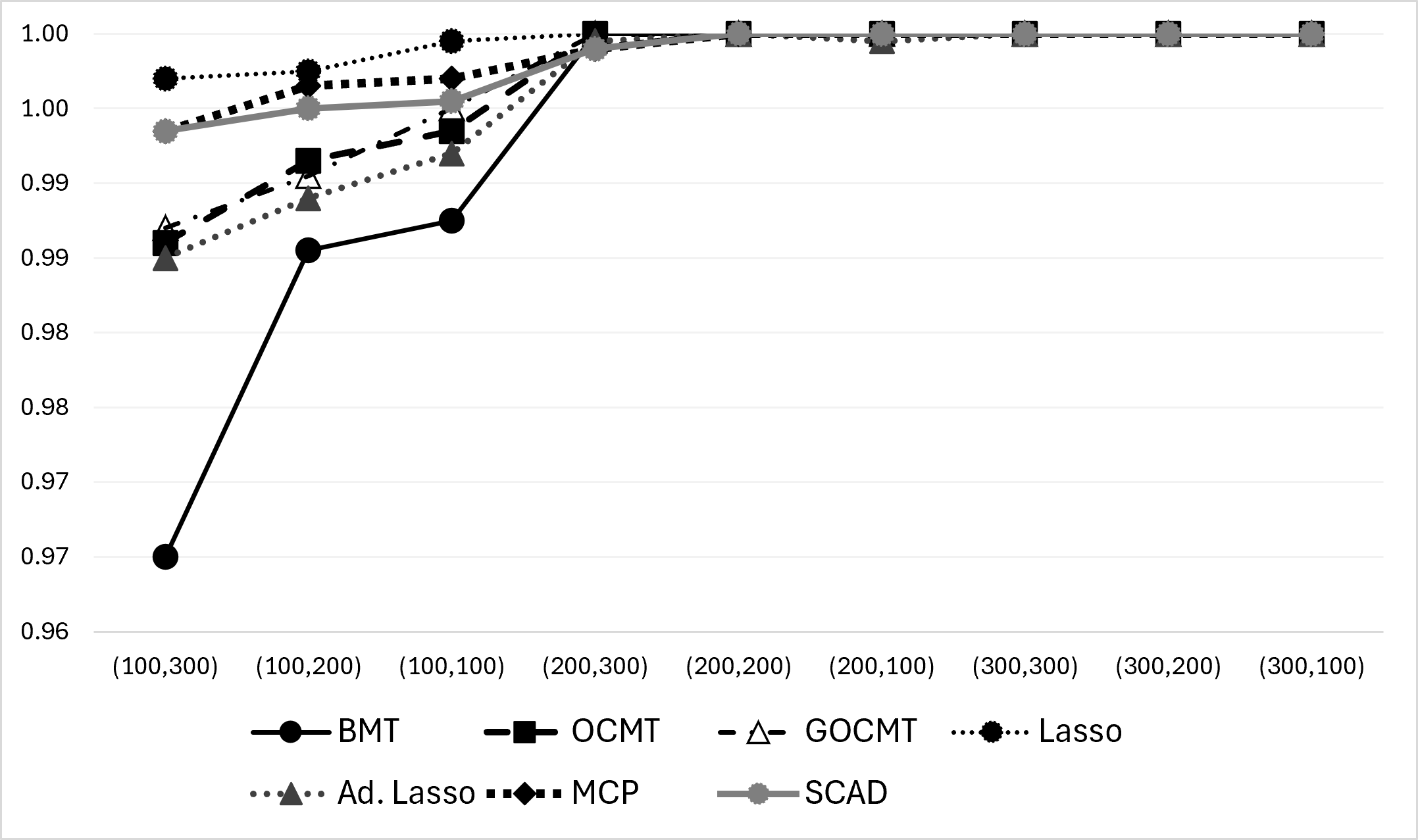}\\
\end{minipage}
\begin{minipage}{0.48\linewidth}
\centering
\vspace{0.4em}
    V = 4, $\pi=0.5$
\includegraphics[height=0.28\textheight]{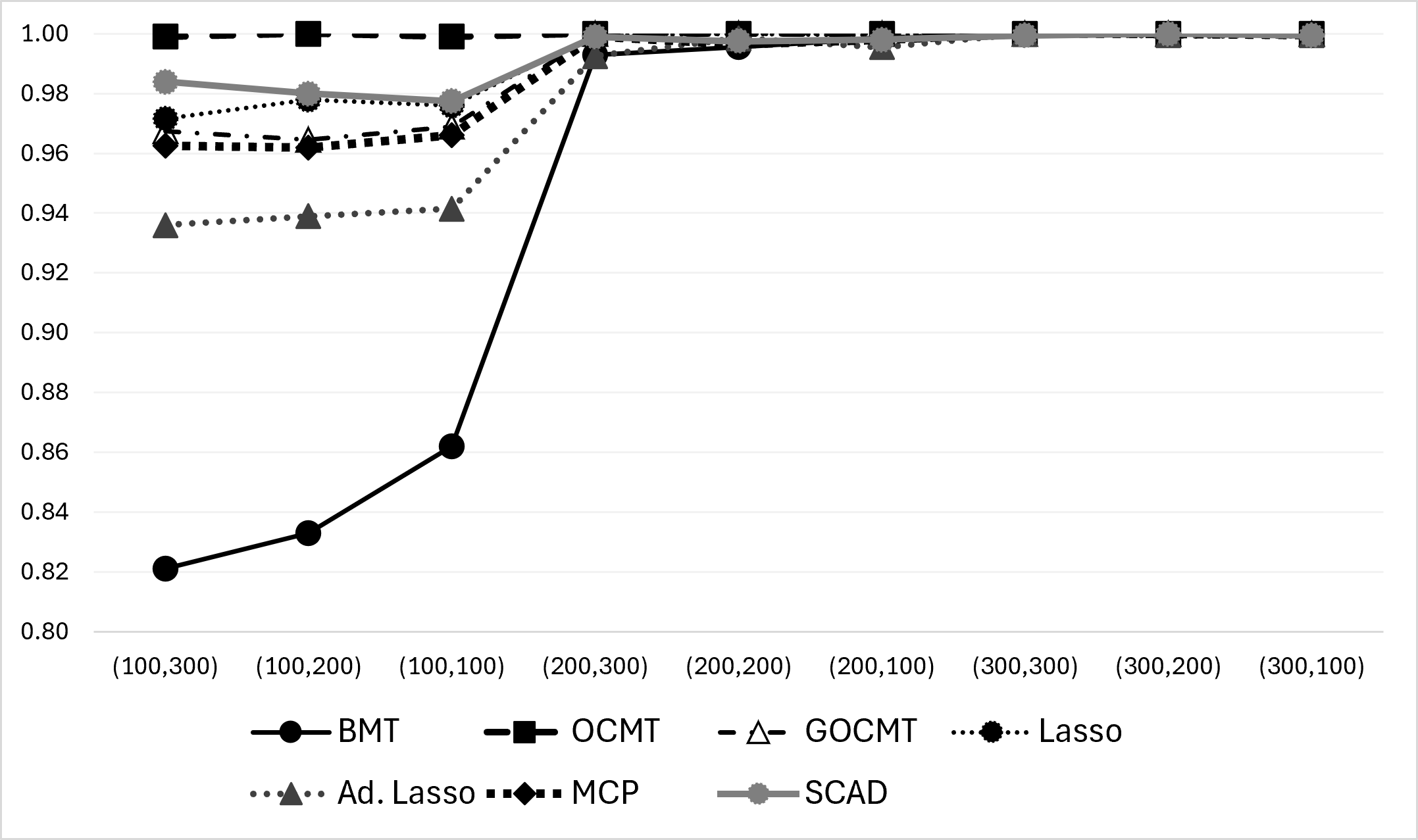}\\
\end{minipage}
\hfill
\begin{minipage}{0.48\linewidth}
\centering
\vspace{0.4em}
    V = 2, $\pi=0.5$
\includegraphics[height=0.28\textheight]{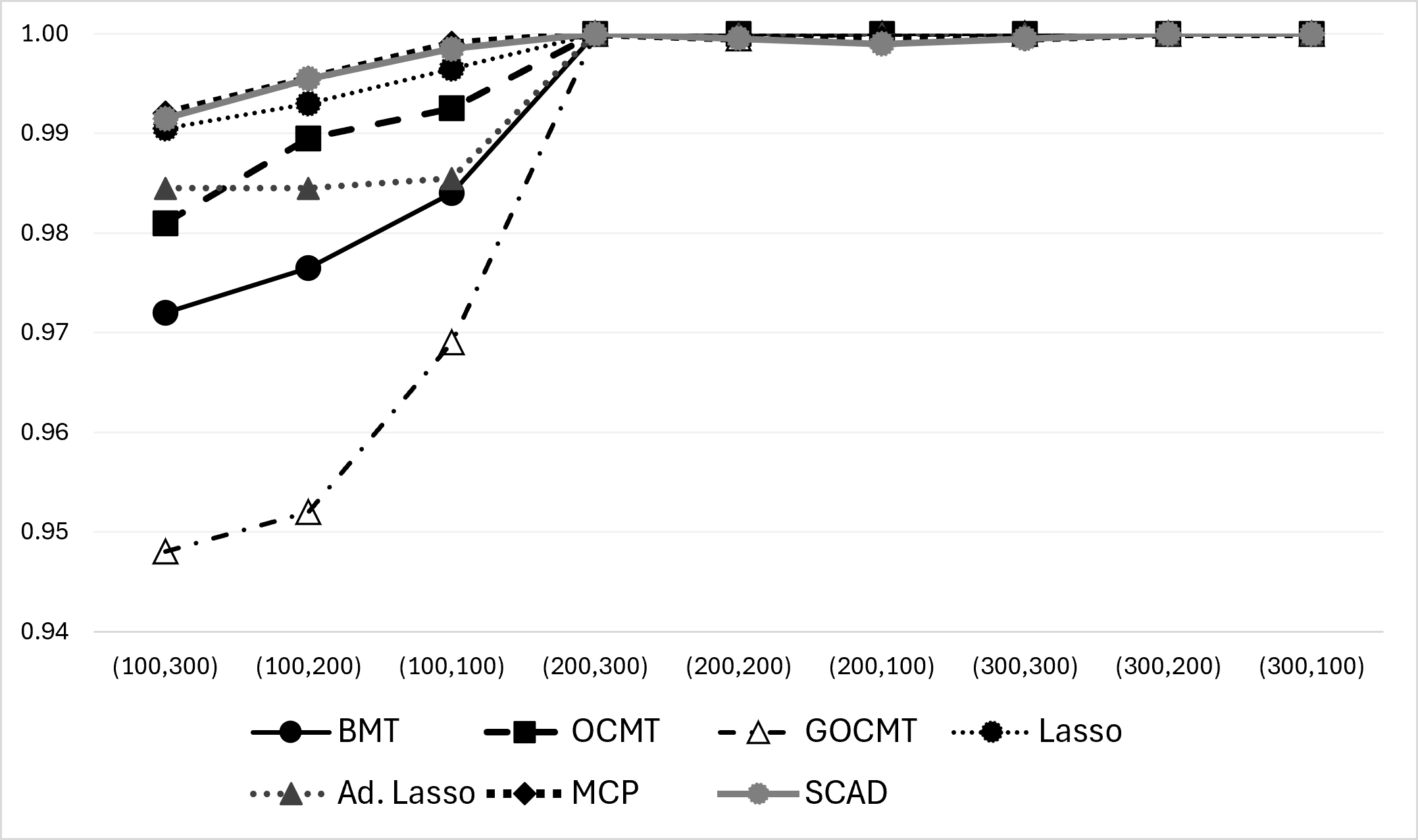}\\
\end{minipage}
\begin{minipage}{0.48\linewidth}
\centering
\vspace{0.2em}
    V = 4, $\pi=1$
\includegraphics[height=0.28\textheight]{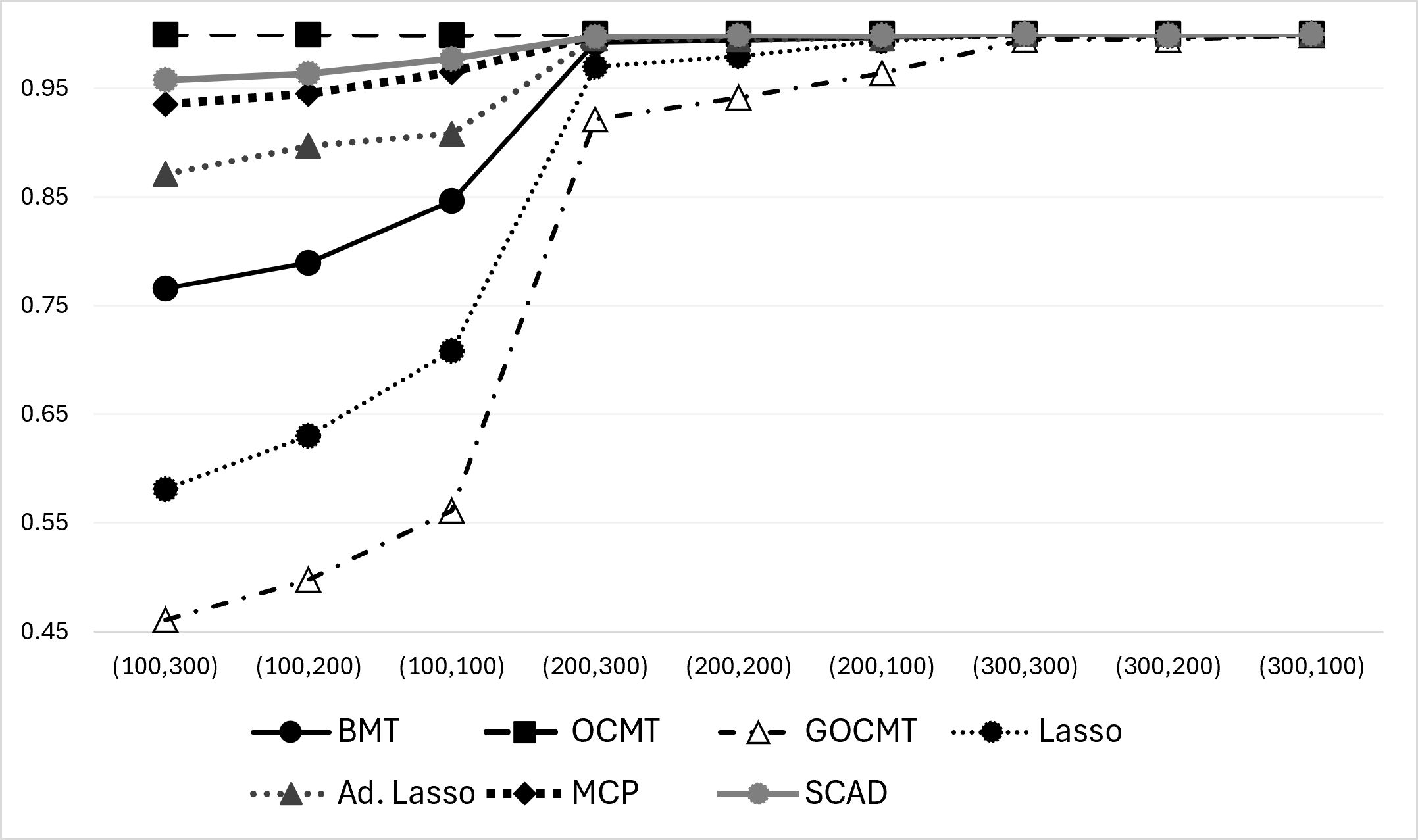}\\
\end{minipage}
\hfill
\begin{minipage}{0.48\linewidth}
\centering
\vspace{0.2em}
    V = 2, $\pi=1$
\includegraphics[height=0.28\textheight]{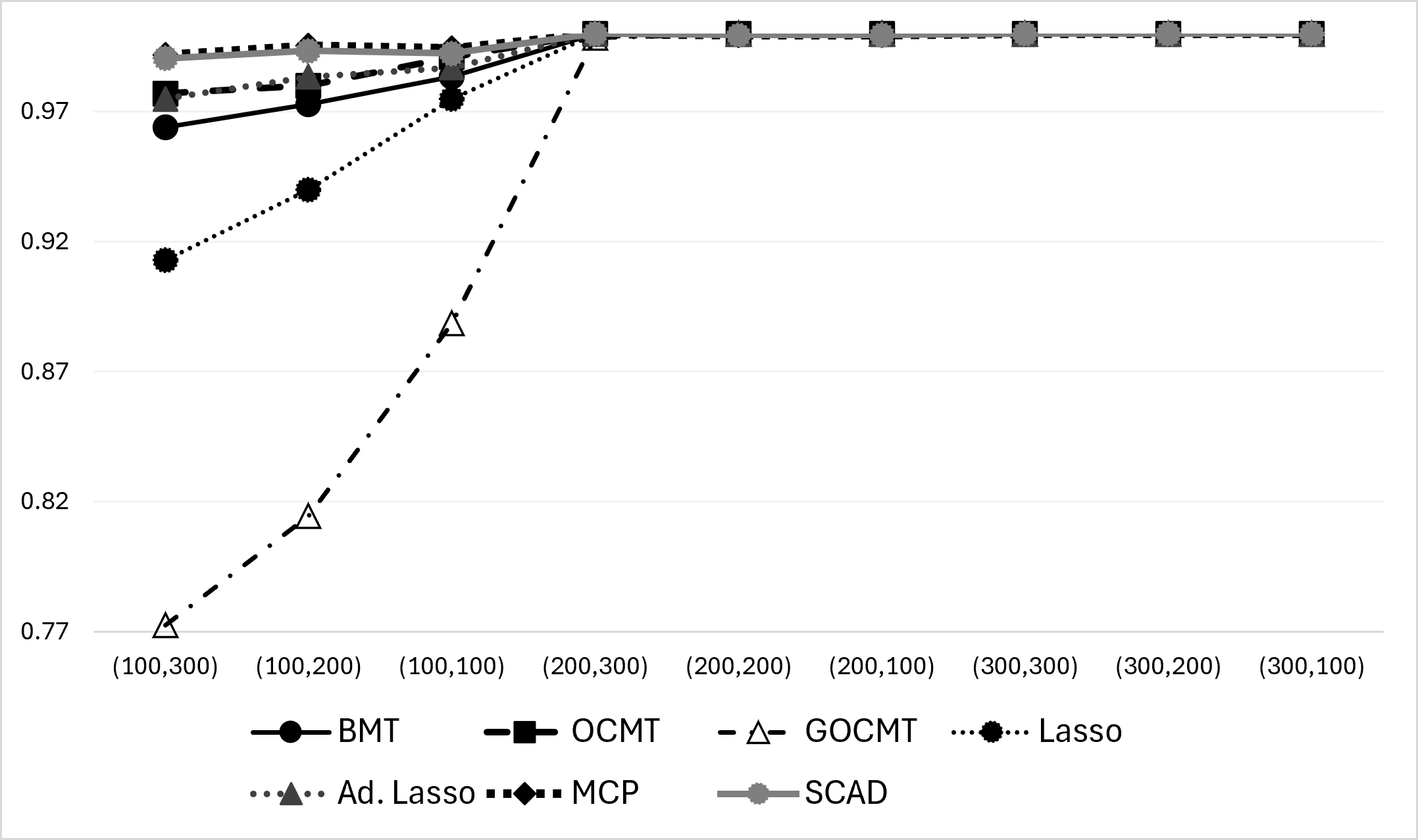}\\
\end{minipage}
\vspace{-0.4em}
\captionof{figure}{Visual Summary of TPR Performance for $\alpha=0.8$}
\end{minipage}
\end{landscape}

\begin{landscape}
\thispagestyle{empty} 
\centering
\begin{minipage}{0.95\linewidth}
\vspace{-0.5em}
\begin{minipage}{0.48\linewidth}
\centering
\vspace{-0.4em}
    V = 4, $\pi=0$
\includegraphics[width=\linewidth,height=0.28\textheight,keepaspectratio]{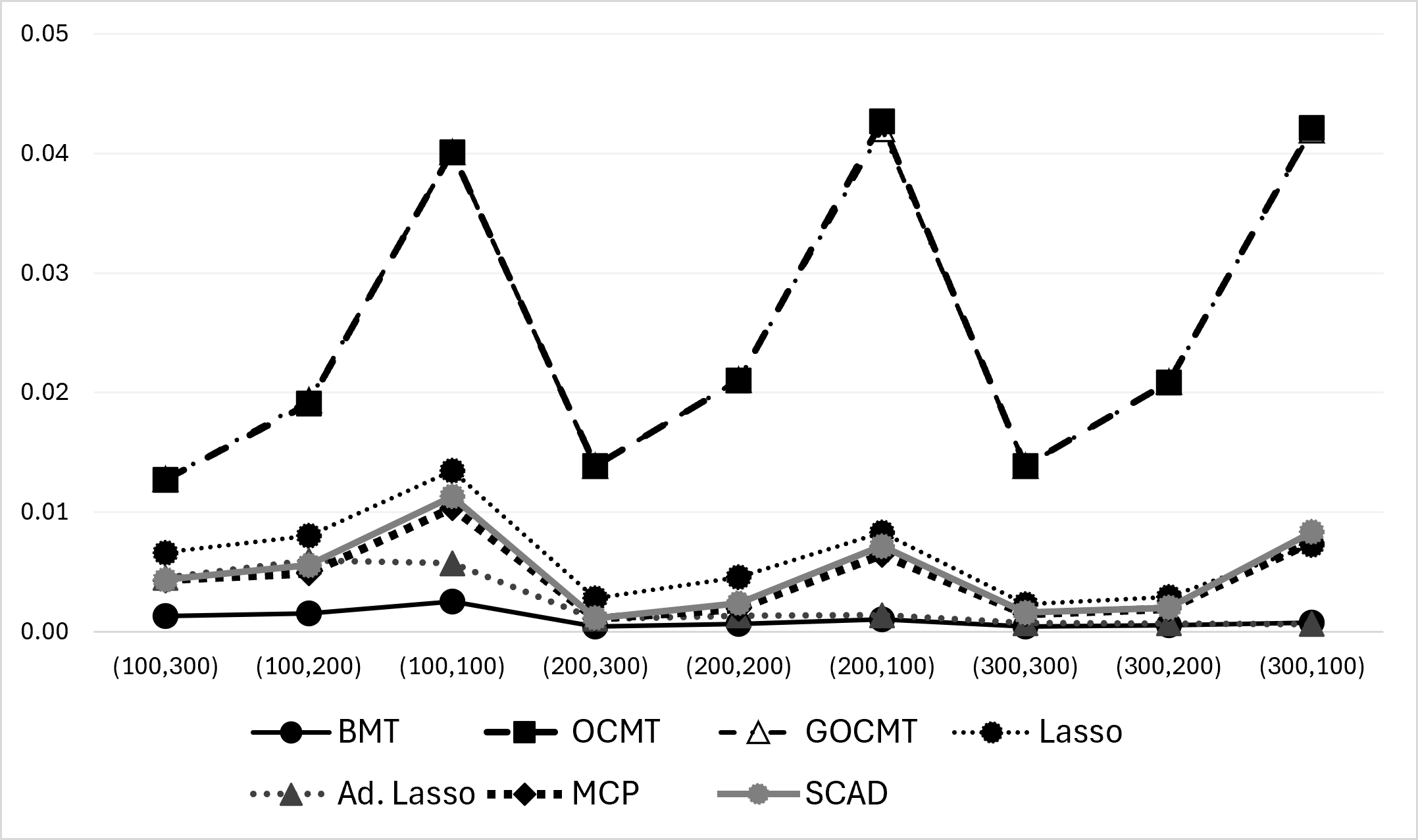}\\
\end{minipage}
\hfill
\begin{minipage}{0.48\linewidth}
\centering
\vspace{-0.4em}
    V = 2, $\pi=0$
\includegraphics[height=0.28\textheight]{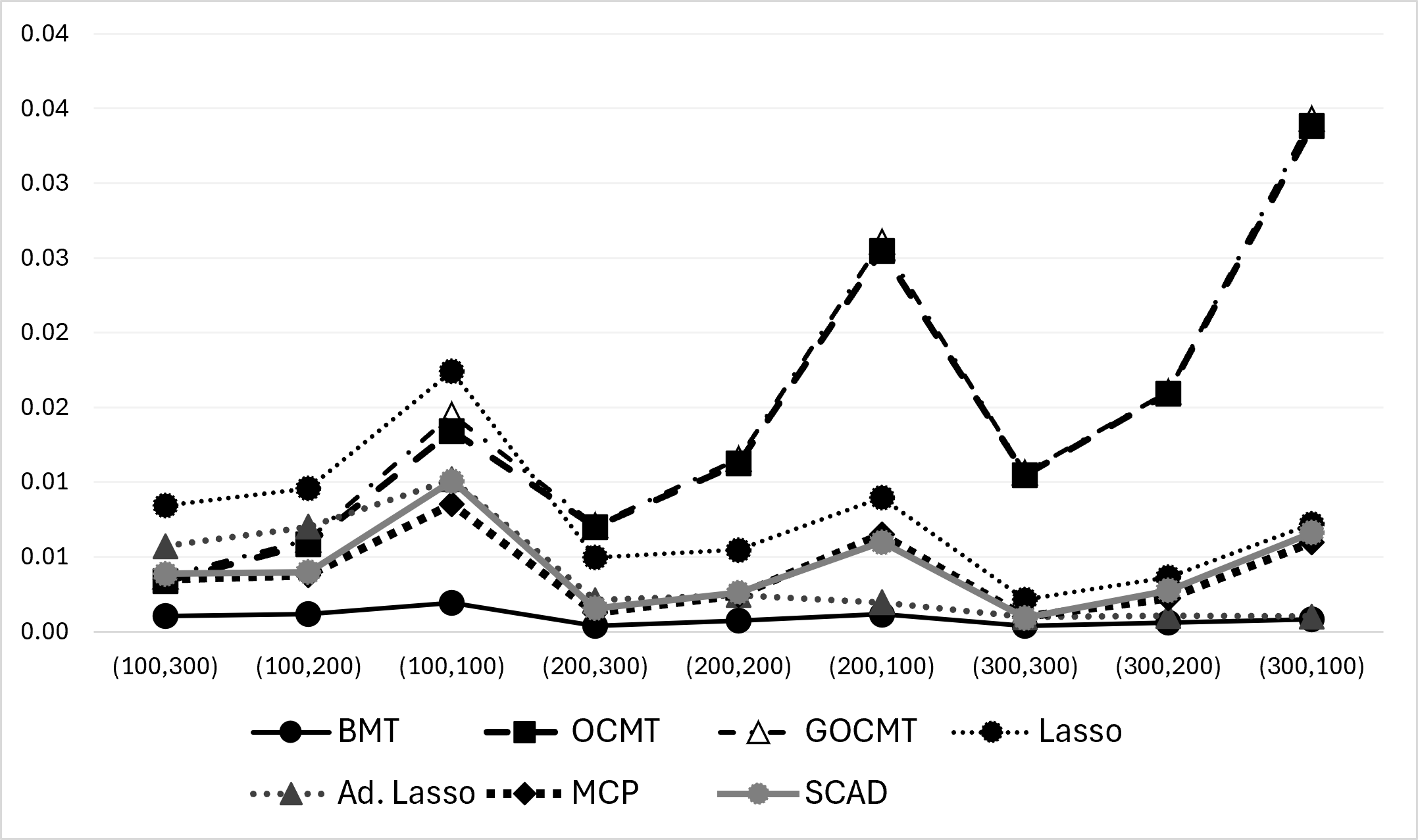}\\
\end{minipage}
\begin{minipage}{0.48\linewidth}
\centering
\vspace{0.4em}
    V = 4, $\pi=0.5$
\includegraphics[height=0.28\textheight]{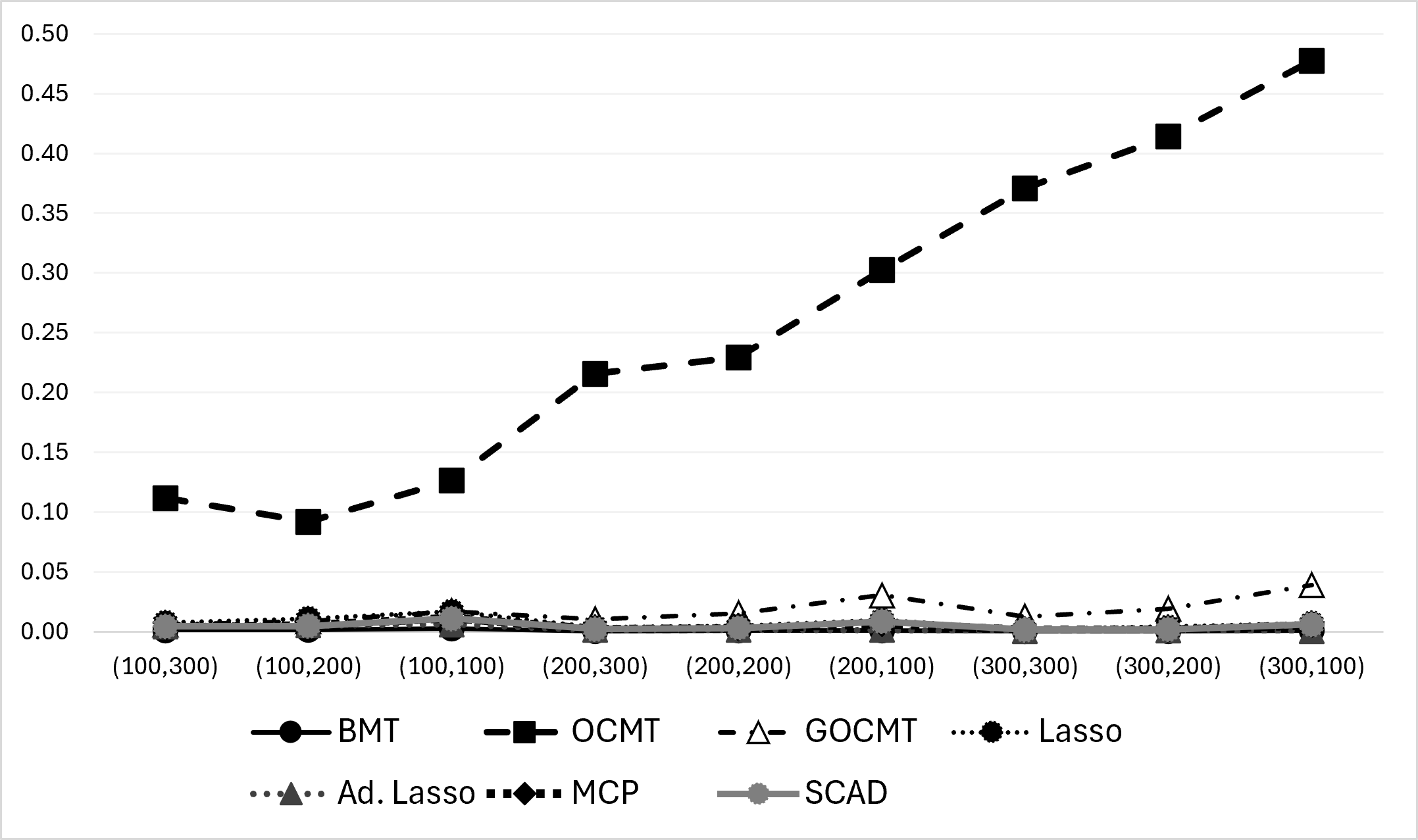}\\
\end{minipage}
\hfill
\begin{minipage}{0.48\linewidth}
\centering
\vspace{0.4em}
    V = 2, $\pi=0.5$
\includegraphics[height=0.28\textheight]{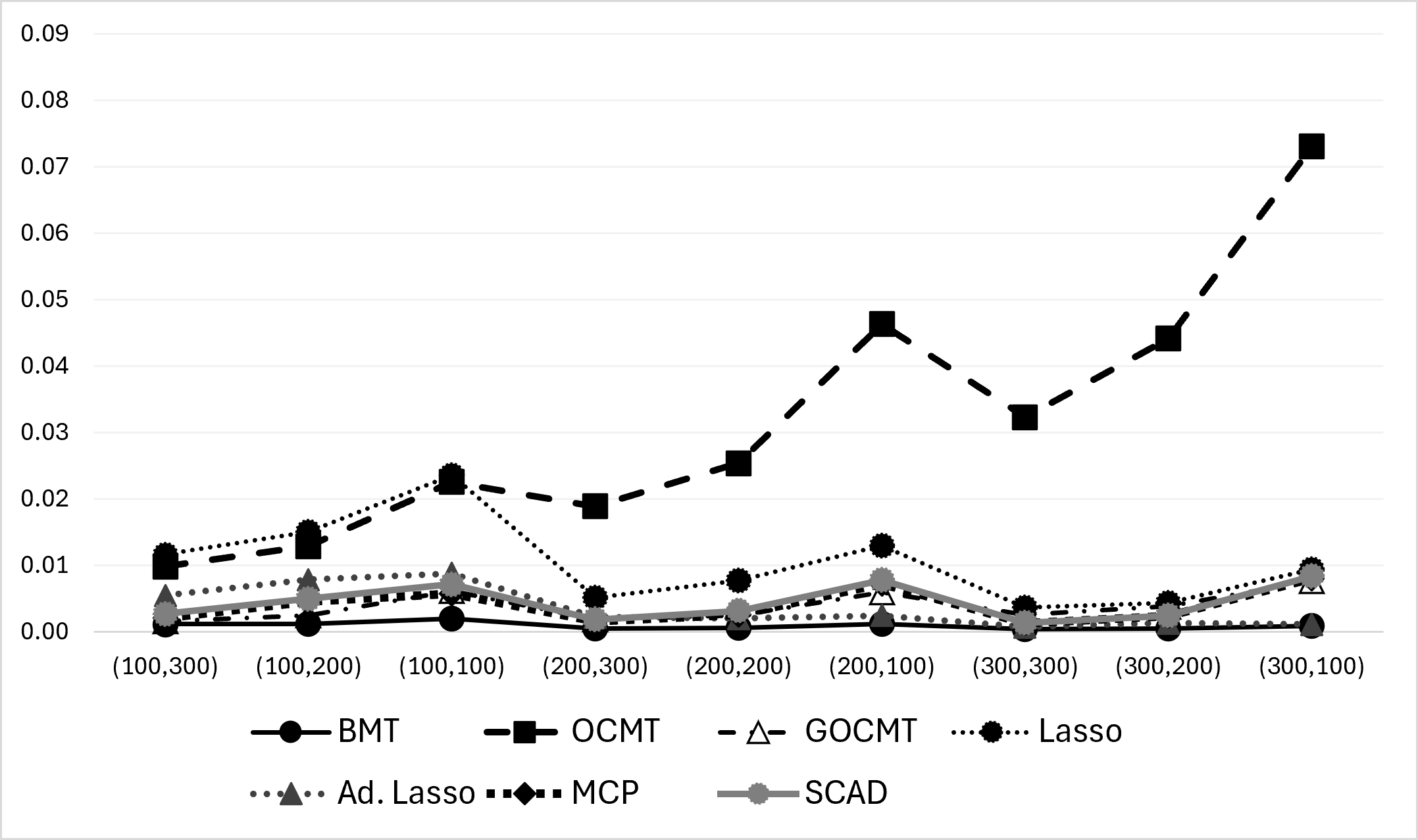}\\
\end{minipage}
\begin{minipage}{0.48\linewidth}
\centering
\vspace{0.2em}
    V = 4, $\pi=1$
\includegraphics[height=0.28\textheight]{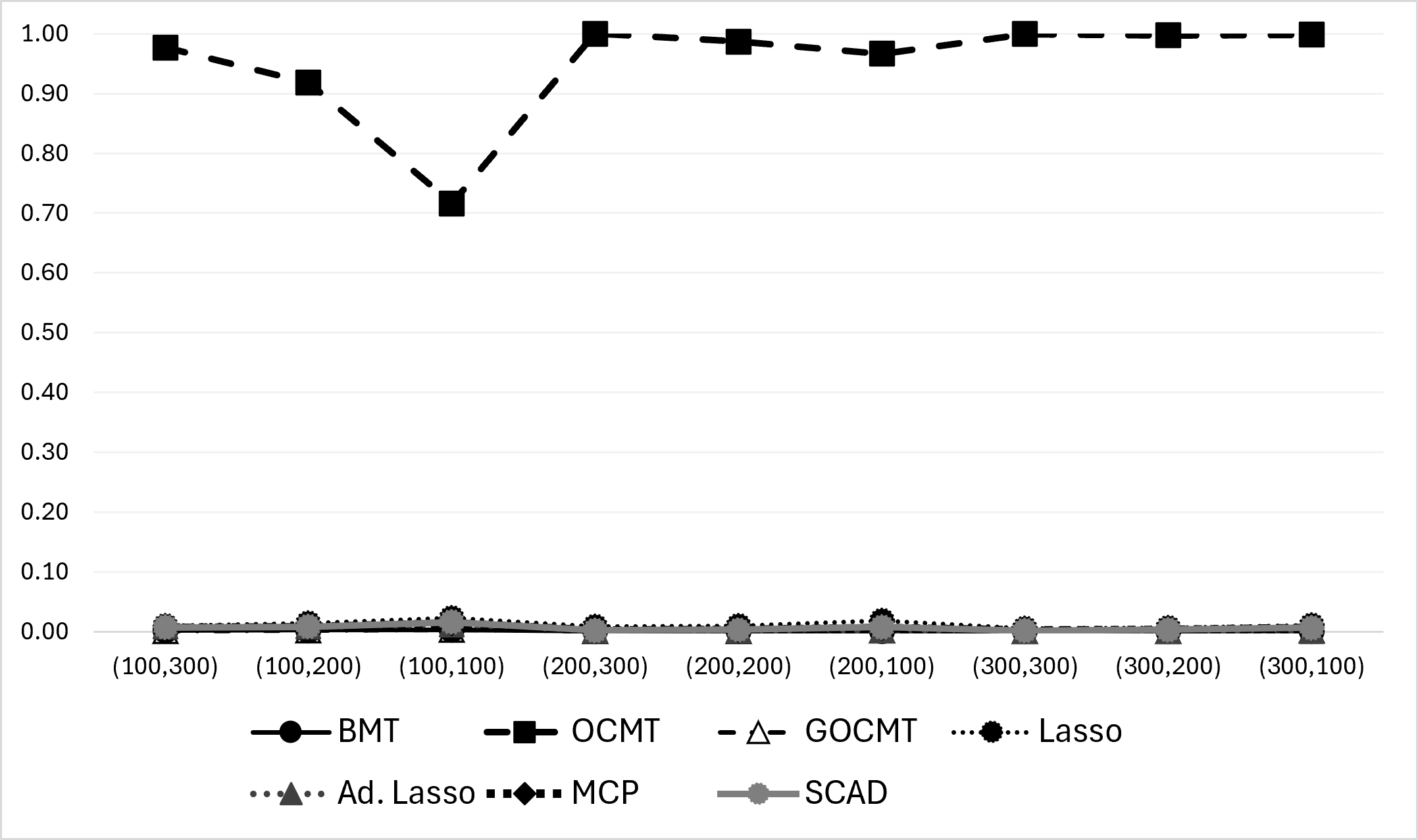}\\
\end{minipage}
\hfill
\begin{minipage}{0.48\linewidth}
\centering
\vspace{0.2em}
    V = 2, $\pi=1$
\includegraphics[height=0.28\textheight]{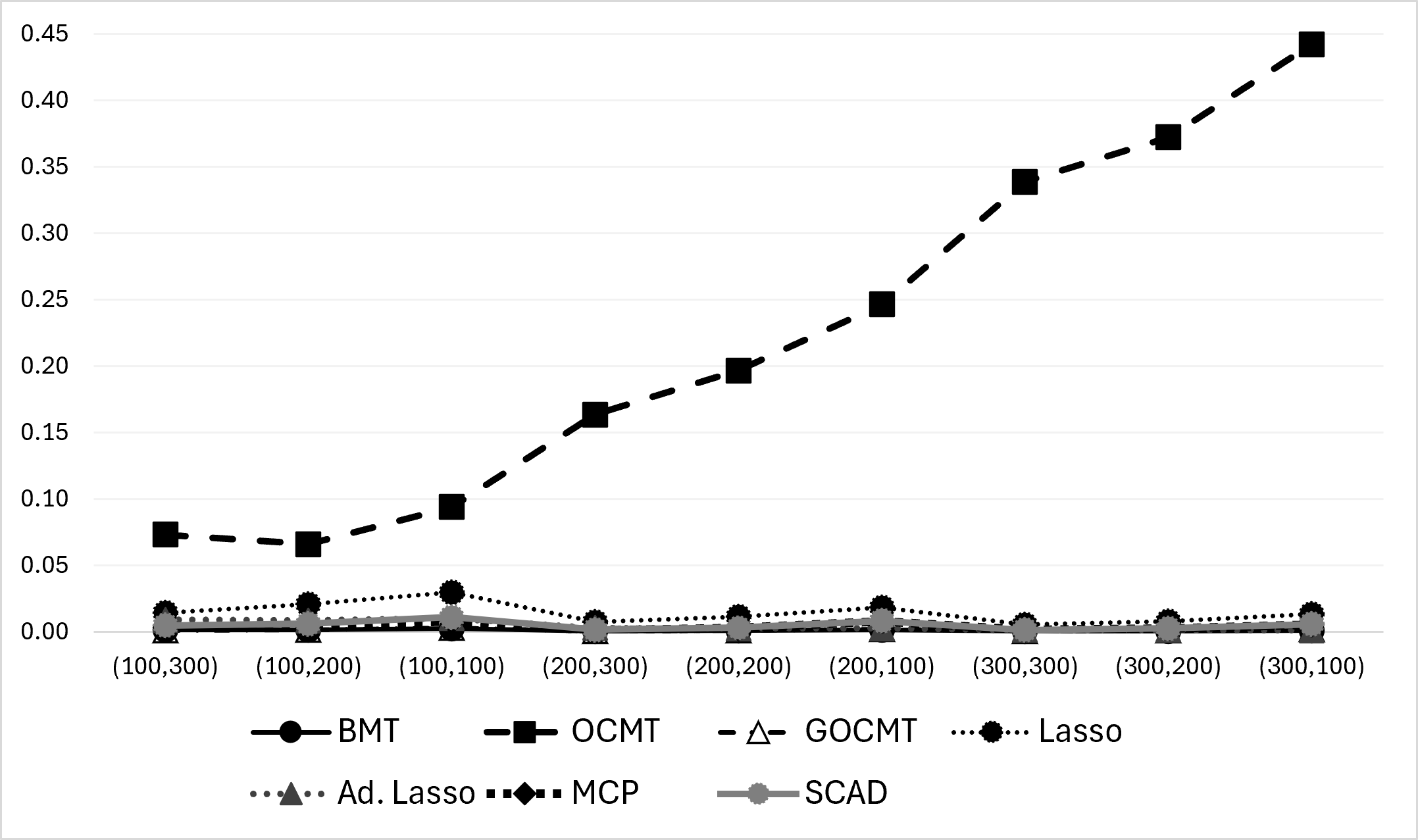}\\
\end{minipage}
\vspace{-0.4em}
\captionof{figure}{Visual Summary of FPR Performance for $\alpha=0.8$}
\end{minipage}
\end{landscape}

\begin{landscape}
\thispagestyle{empty} 
\centering
\begin{minipage}{0.95\linewidth}
\vspace{-0.5em}
\begin{minipage}{0.48\linewidth}
\centering
\vspace{-0.4em}
    V = 4, $\pi=0$
\includegraphics[width=\linewidth,height=0.28\textheight,keepaspectratio]{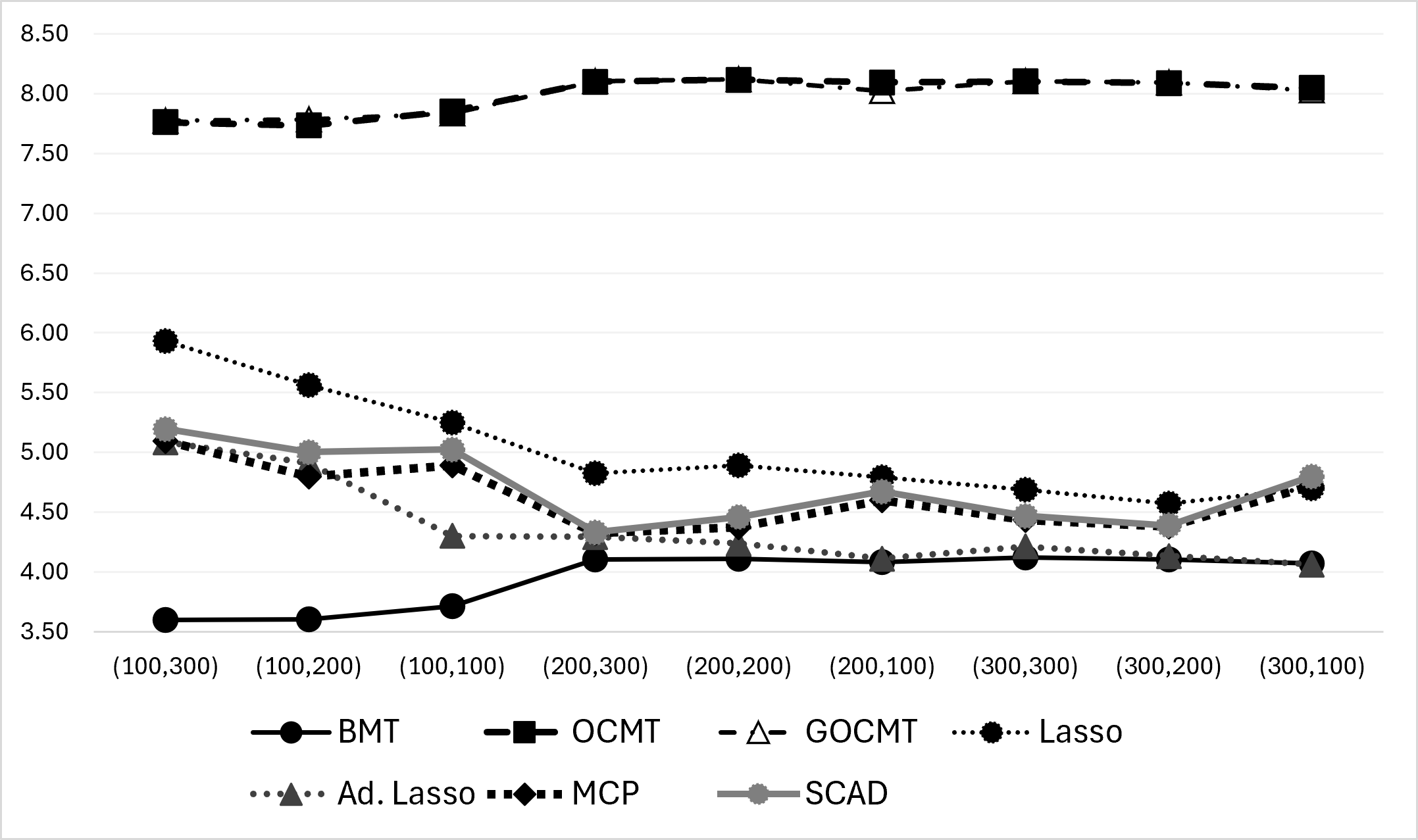}\\
\end{minipage}
\hfill
\begin{minipage}{0.48\linewidth}
\centering
\vspace{-0.4em}
    V = 2, $\pi=0$
\includegraphics[height=0.28\textheight]{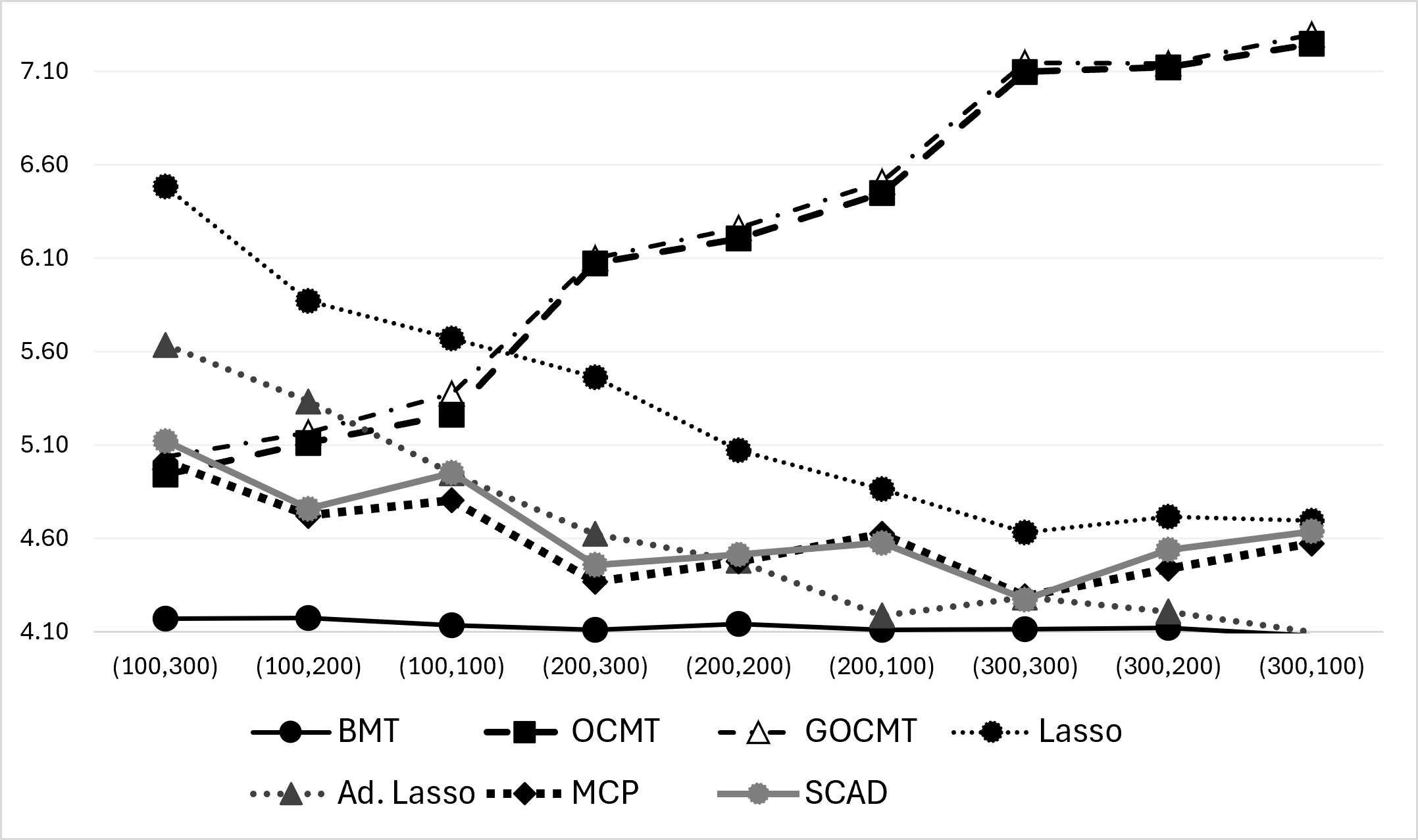}\\
\end{minipage}
\begin{minipage}{0.48\linewidth}
\centering
\vspace{0.4em}
    V = 4, $\pi=0.5$
\includegraphics[height=0.28\textheight]{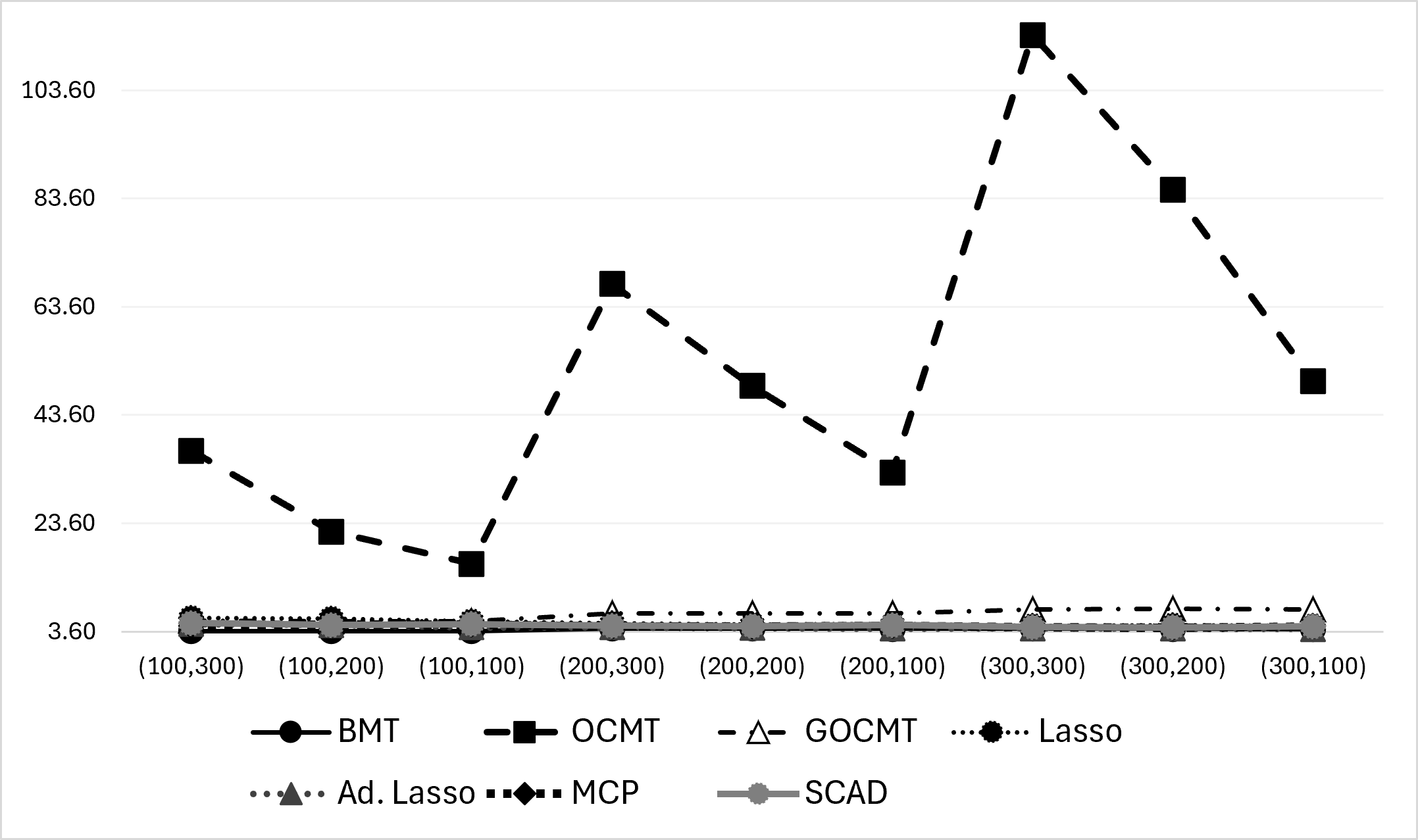}\\
\end{minipage}
\hfill
\begin{minipage}{0.48\linewidth}
\centering
\vspace{0.4em}
    V = 2, $\pi=0.5$
\includegraphics[height=0.28\textheight]{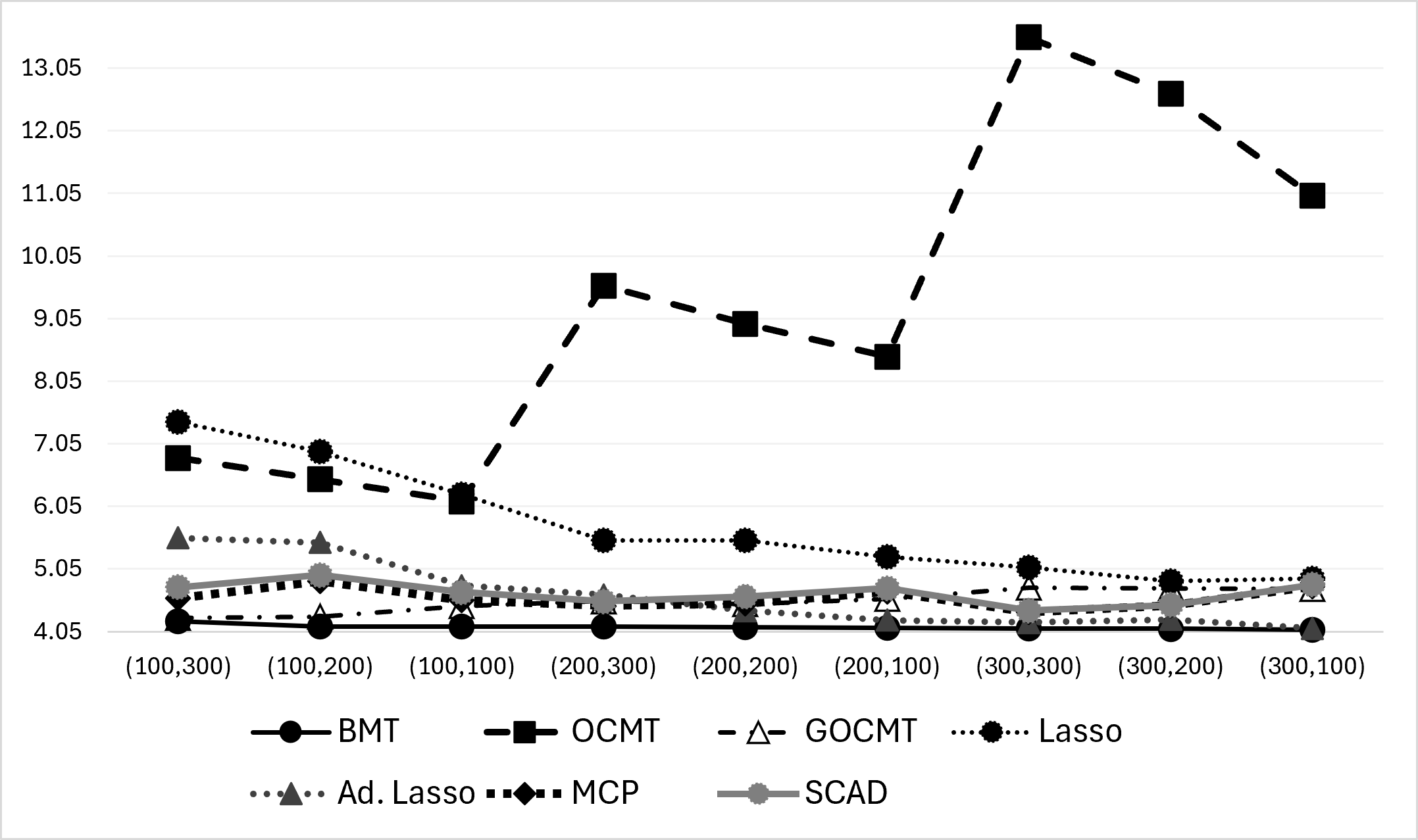}\\
\end{minipage}
\begin{minipage}{0.48\linewidth}
\centering
\vspace{0.2em}
    V = 4, $\pi=1$
\includegraphics[height=0.28\textheight]{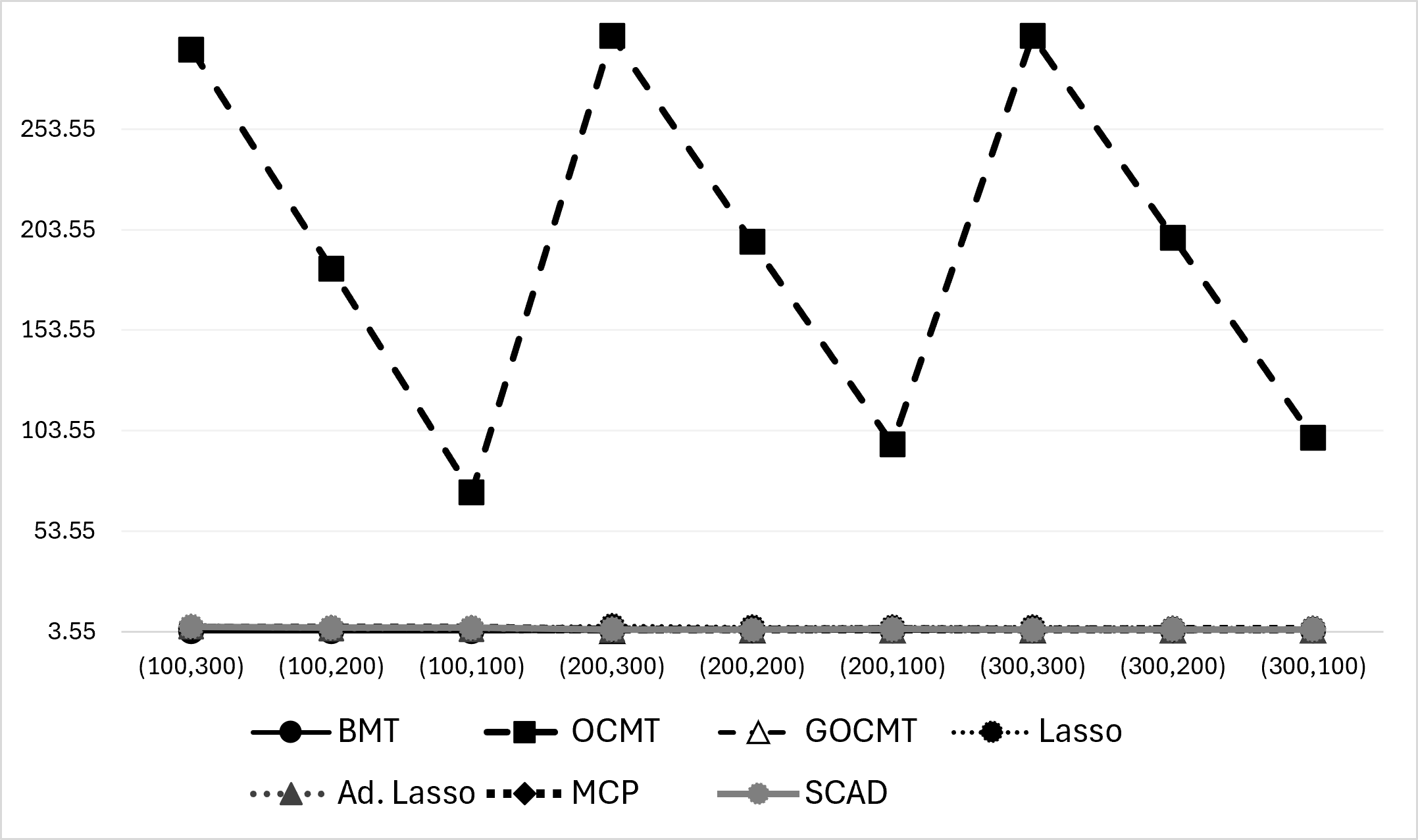}\\
\end{minipage}
\hfill
\begin{minipage}{0.48\linewidth}
\centering
\vspace{0.2em}
    V = 2, $\pi=1$
\includegraphics[height=0.28\textheight]{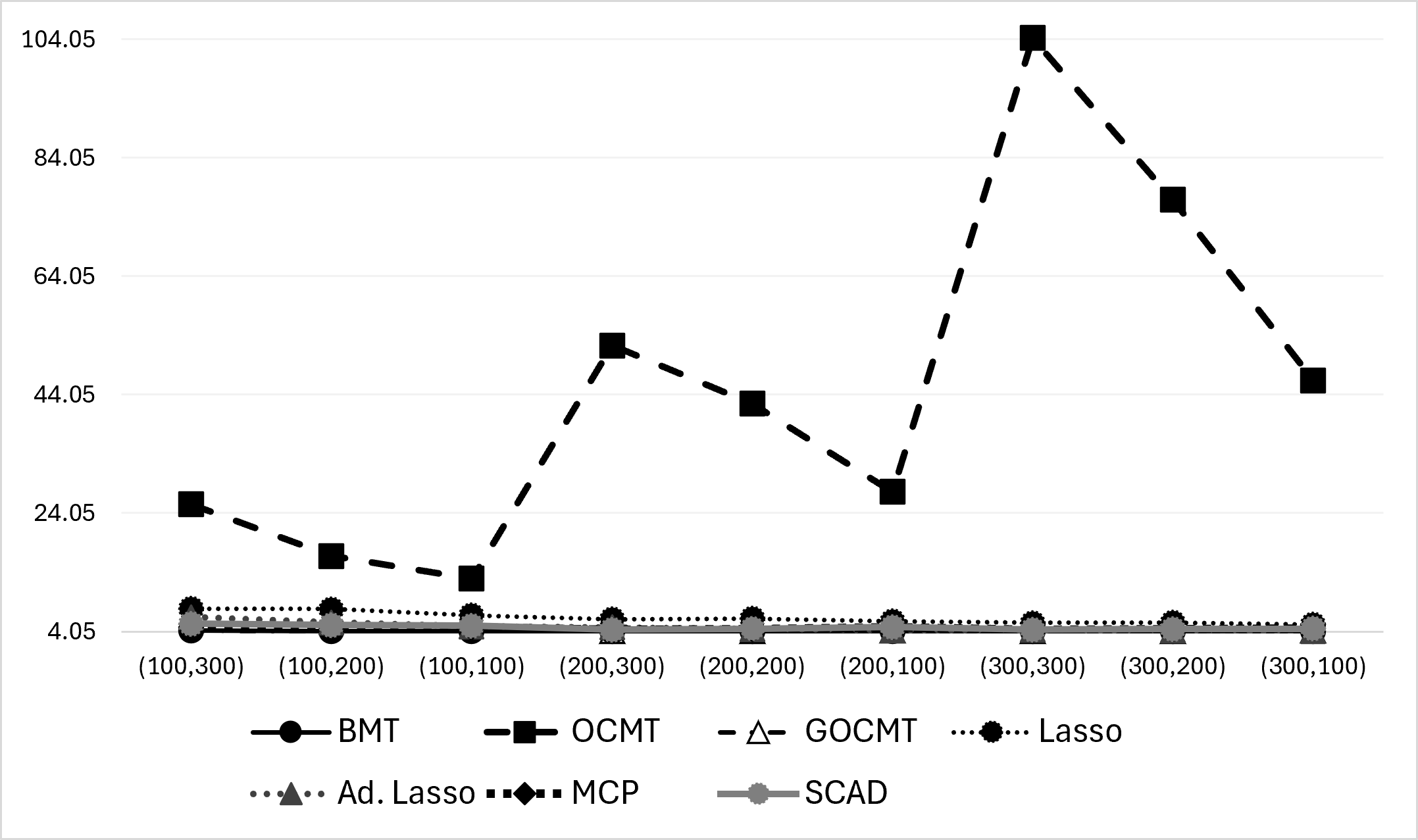}\\
\end{minipage}
\vspace{-0.4em}
\captionof{figure}{Visual Summary of Performance of Average Model Size for $\alpha=0.8$}
\end{minipage}
\end{landscape}

\begin{landscape}
\thispagestyle{empty} 
\centering
\begin{minipage}{0.95\linewidth}
\vspace{-0.5em}
\begin{minipage}{0.48\linewidth}
\centering
\vspace{-0.4em}
    V = 4, $\pi=0$
\includegraphics[width=\linewidth,height=0.28\textheight,keepaspectratio]{figures_n/RMSE_VIF_4_pi_0_alpha_08.png}\\
\end{minipage}
\hfill
\begin{minipage}{0.48\linewidth}
\centering
\vspace{-0.4em}
    V = 2, $\pi=0$
\includegraphics[height=0.28\textheight]{figures_n/RMSE_VIF_2_pi_0_alpha_08.png}\\
\end{minipage}
\begin{minipage}{0.48\linewidth}
\centering
\vspace{0.4em}
    V = 4, $\pi=0.5$
\includegraphics[height=0.28\textheight]{figures_n/RMSE_VIF_4_pi_0.5_alpha_08.png}\\
\end{minipage}
\hfill
\begin{minipage}{0.48\linewidth}
\centering
\vspace{0.4em}
    V = 2, $\pi=0.5$
\includegraphics[height=0.28\textheight]{figures_n/RMSE_VIF_2_pi_0.5_alpha_08.png}\\
\end{minipage}
\begin{minipage}{0.48\linewidth}
\centering
\vspace{0.2em}
    V = 4, $\pi=1$
\includegraphics[height=0.28\textheight]{figures_n/RMSE_VIF_4_pi_1_alpha_08.png}\\
\end{minipage}
\hfill
\begin{minipage}{0.48\linewidth}
\centering
\vspace{0.2em}
    V = 2, $\pi=1$
\includegraphics[height=0.28\textheight]{figures_n/RMSE_VIF_2_pi_1_alpha_08.png}\\
\end{minipage}
\vspace{-0.4em}
\captionof{figure}{Visual Summary of RMSE Performance for $\alpha=0.8$}
\end{minipage}
\end{landscape}

\begin{landscape}
\thispagestyle{empty} 
\centering
\begin{minipage}{0.95\linewidth}
\vspace{-0.5em}
\begin{minipage}{0.48\linewidth}
\centering
\vspace{-0.4em}
    V = 4, $\pi=0$
\includegraphics[width=\linewidth,height=0.28\textheight,keepaspectratio]{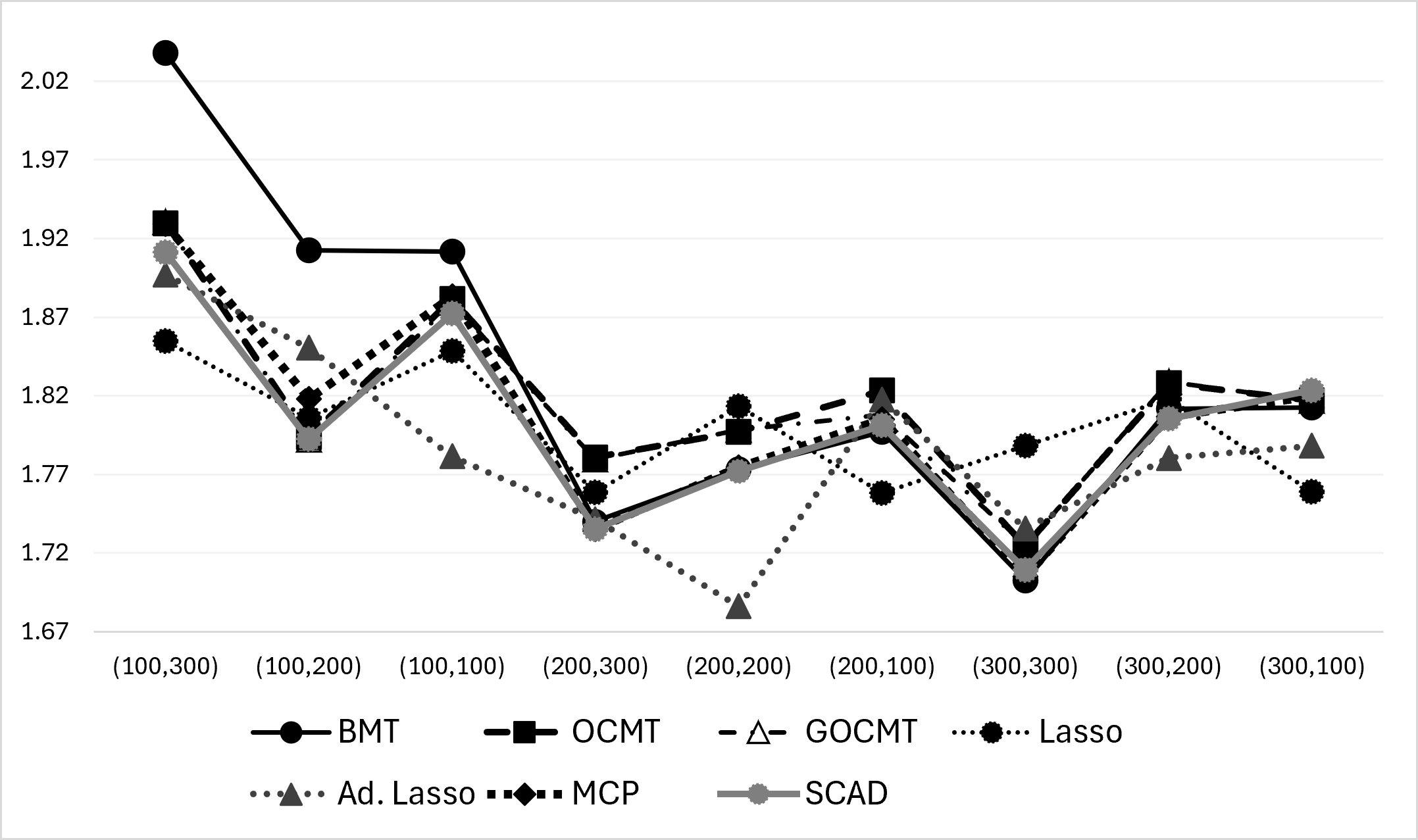}\\
\end{minipage}
\hfill
\begin{minipage}{0.48\linewidth}
\centering
\vspace{-0.4em}
    V = 2, $\pi=0$
\includegraphics[height=0.28\textheight]{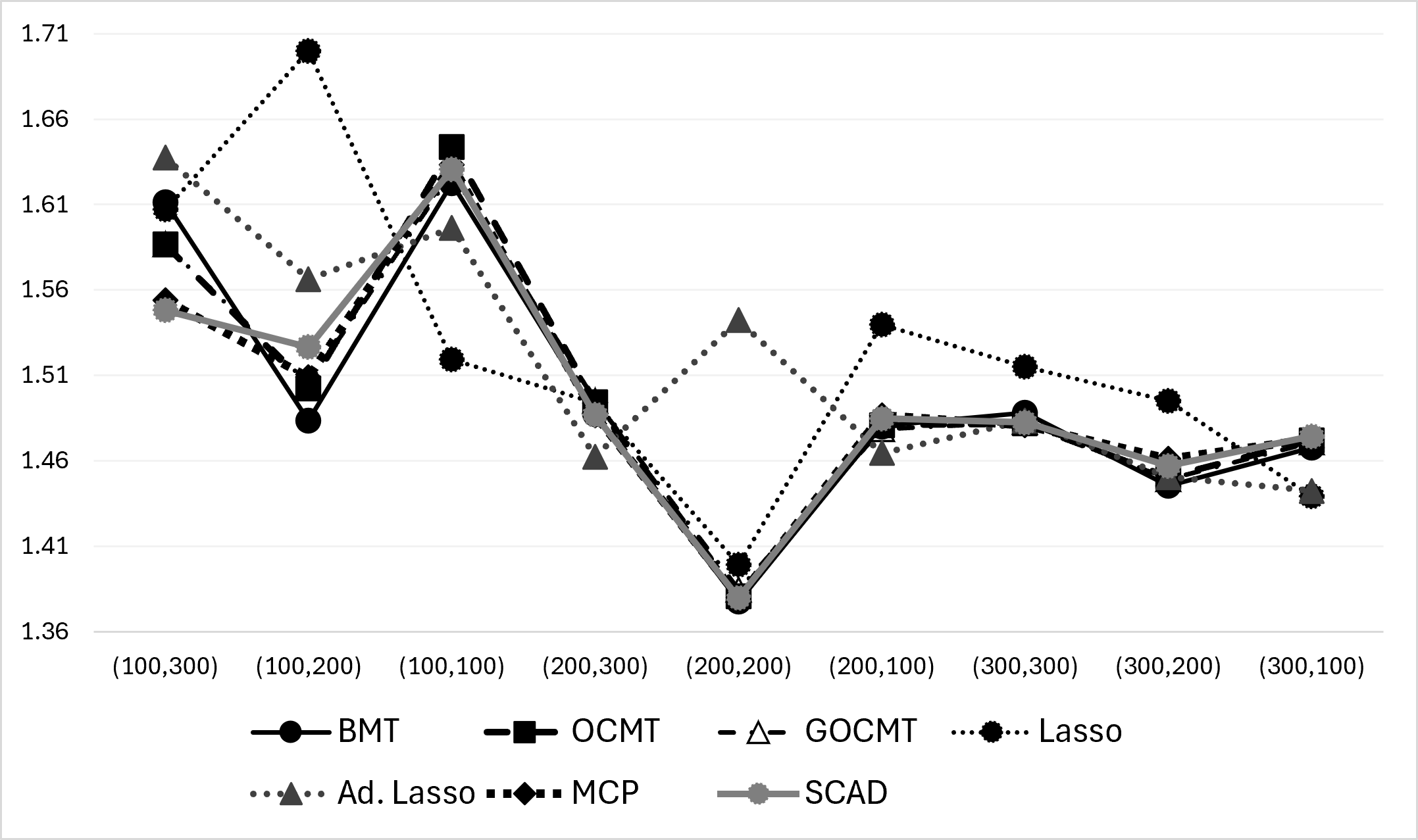}\\
\end{minipage}
\begin{minipage}{0.48\linewidth}
\centering
\vspace{0.4em}
    V = 4, $\pi=0.5$
\includegraphics[height=0.28\textheight]{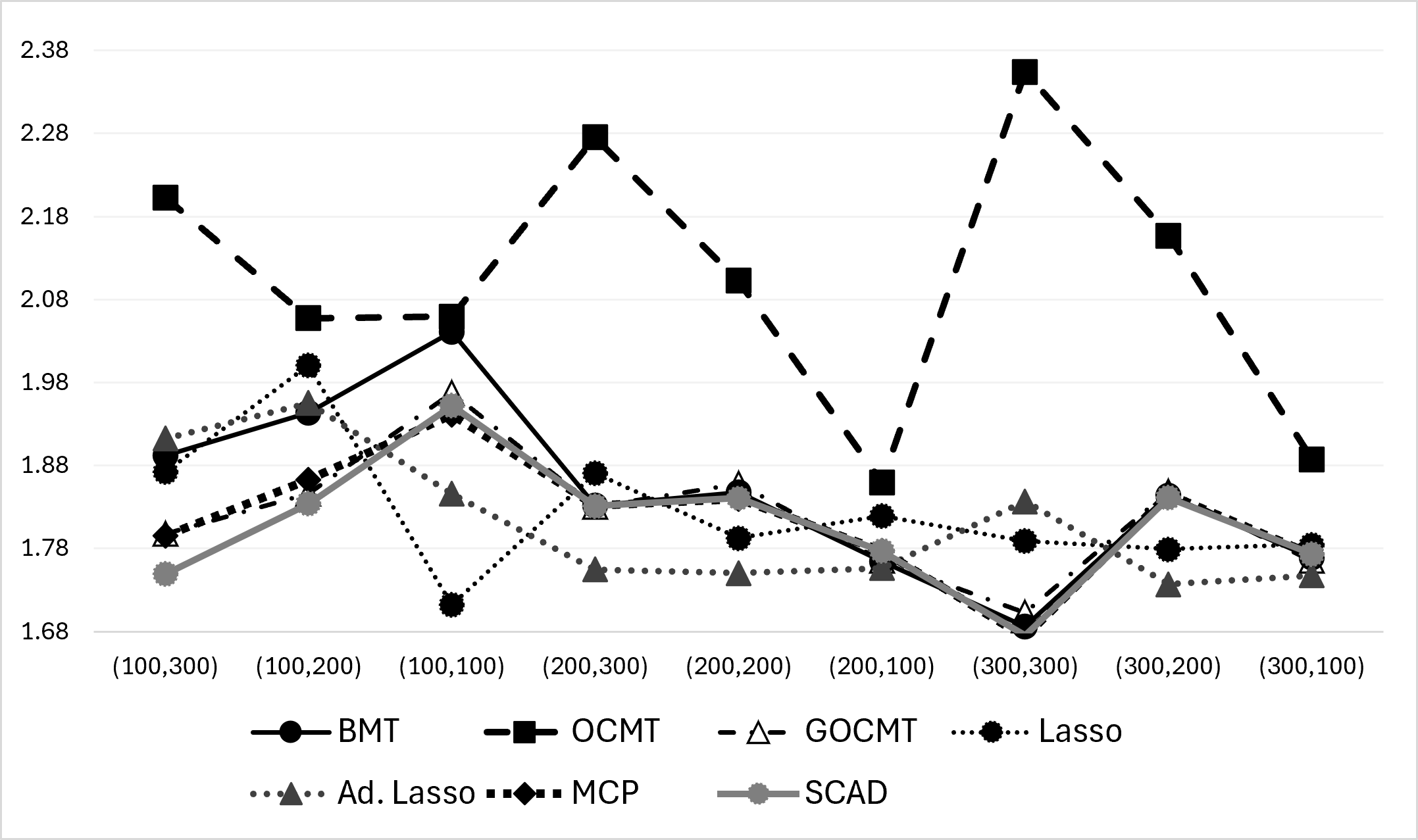}\\
\end{minipage}
\hfill
\begin{minipage}{0.48\linewidth}
\centering
\vspace{0.4em}
    V = 2, $\pi=0.5$
\includegraphics[height=0.28\textheight]{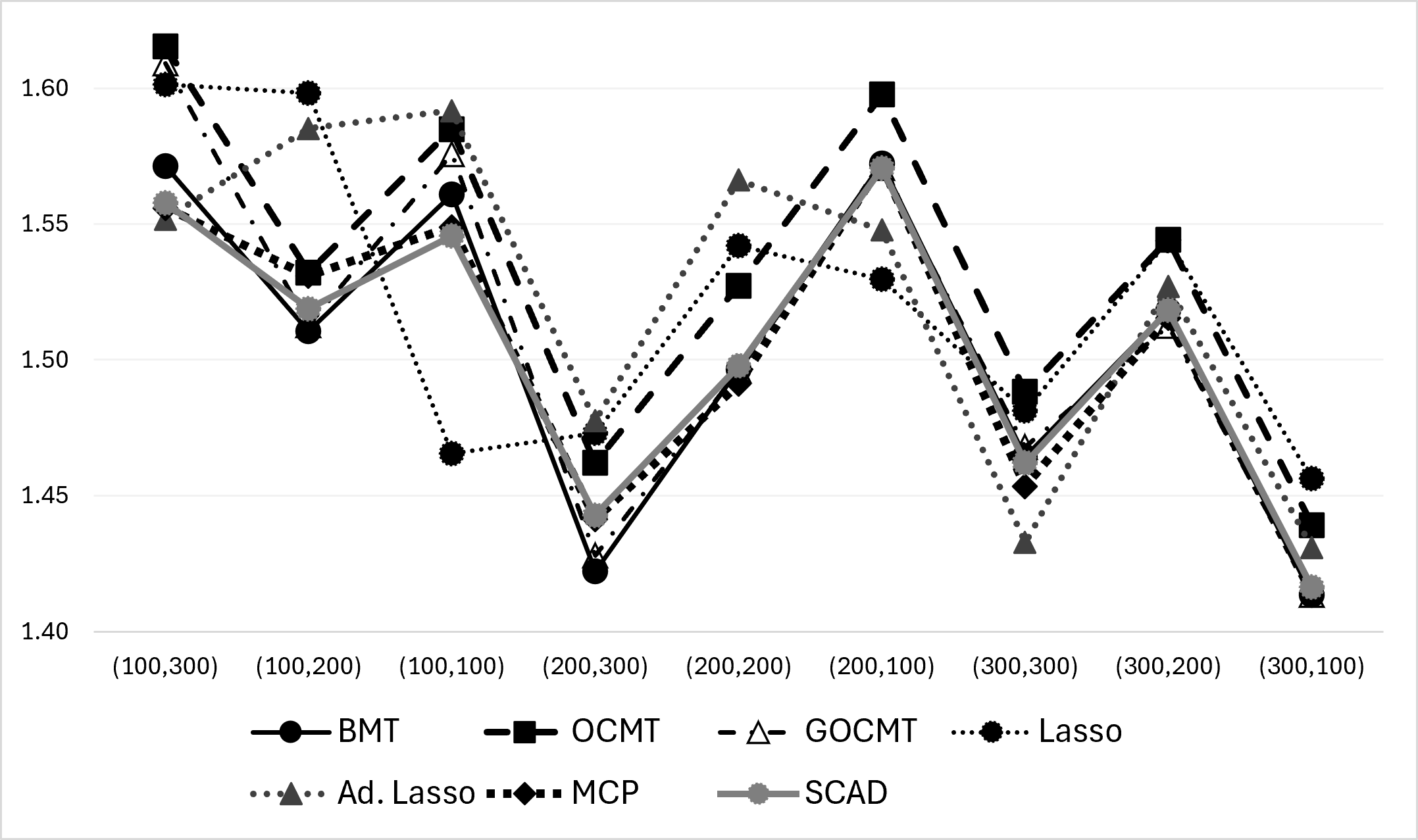}\\
\end{minipage}
\begin{minipage}{0.48\linewidth}
\centering
\vspace{0.2em}
    V = 4, $\pi=1$
\includegraphics[height=0.28\textheight]{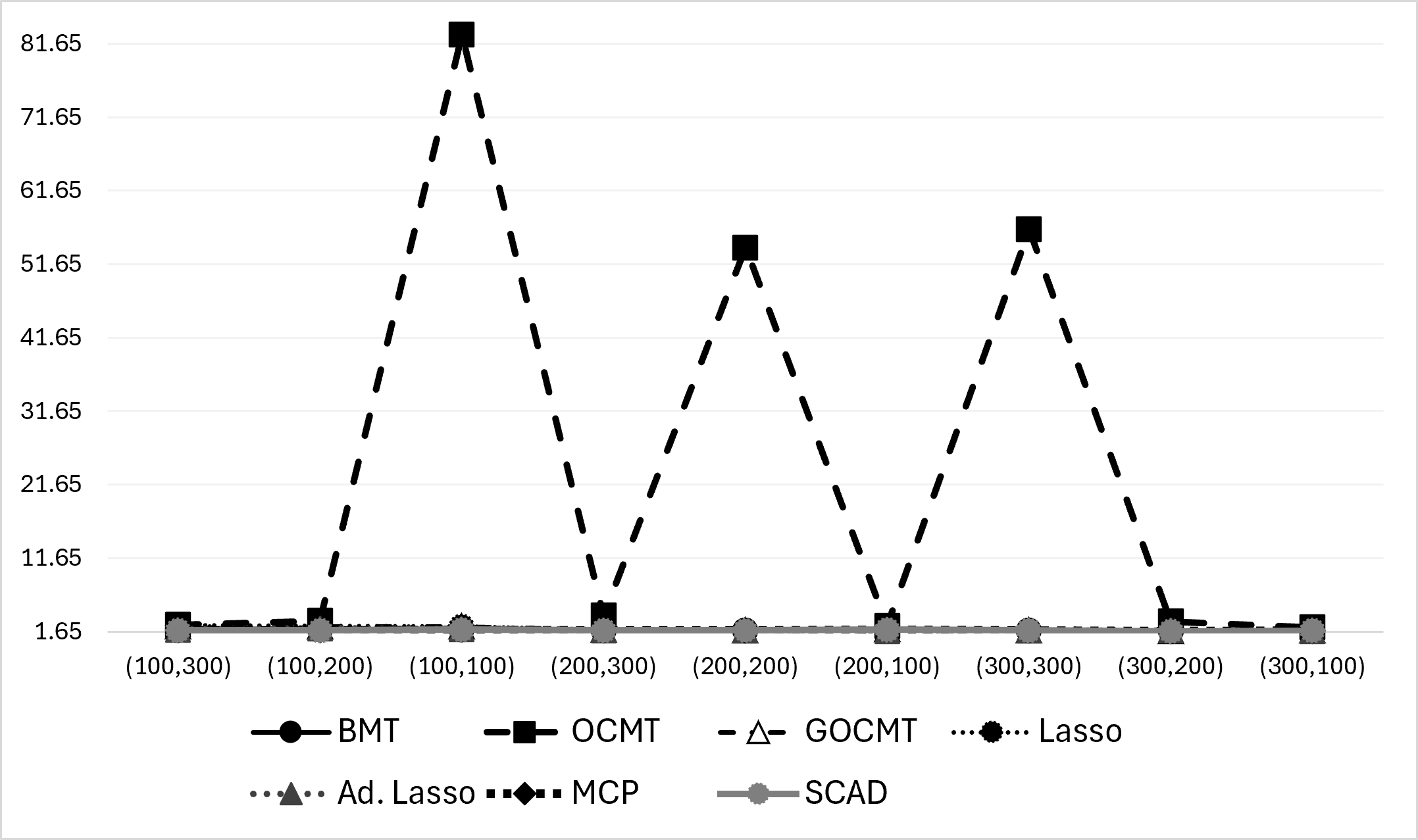}\\
\end{minipage}
\hfill
\begin{minipage}{0.48\linewidth}
\centering
\vspace{0.2em}
    V = 2, $\pi=1$
\includegraphics[height=0.28\textheight]{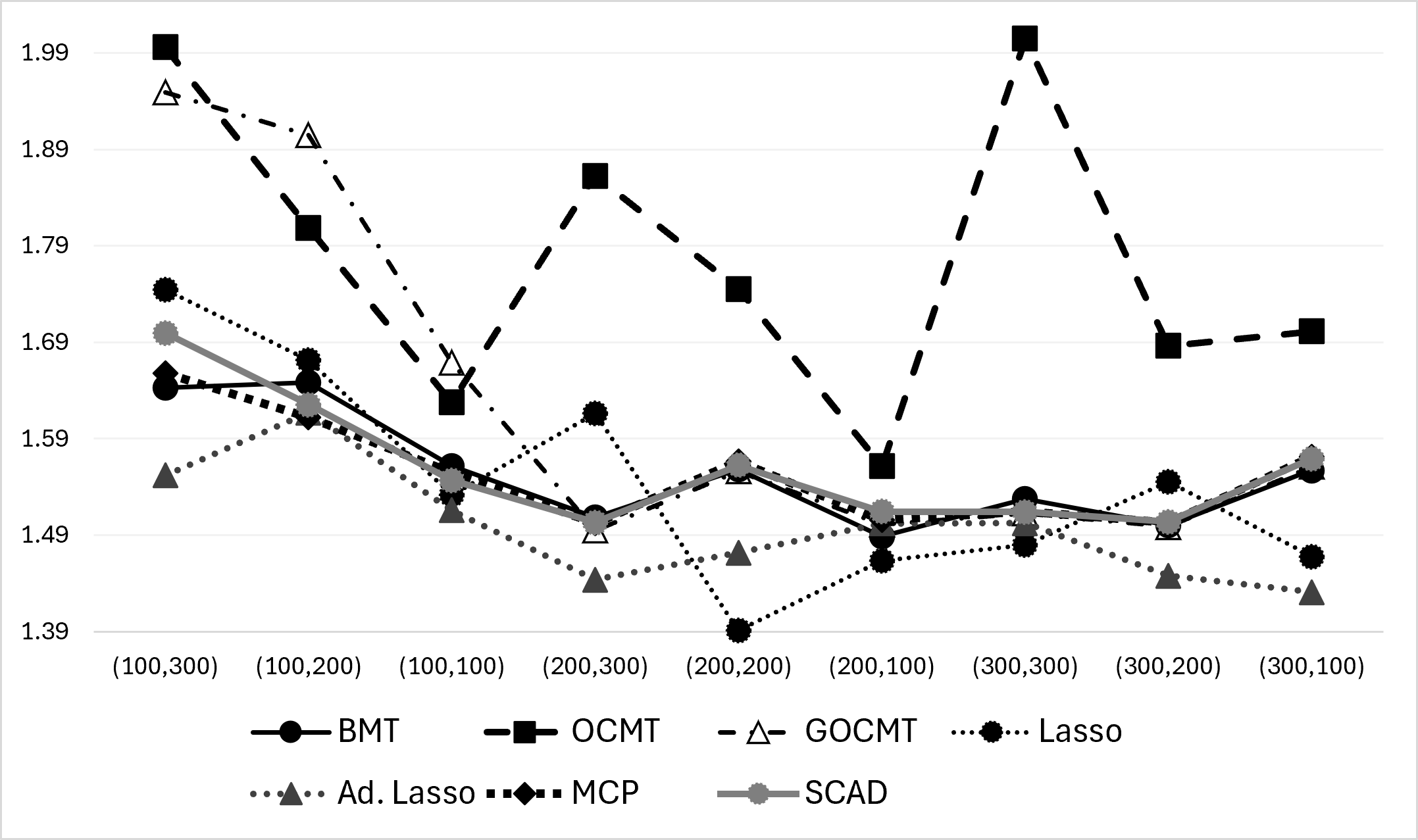}\\
\end{minipage}
\vspace{-0.4em}
\captionof{figure}{Visual Summary of RMSFE Performance for $\alpha=0.8$}
\end{minipage}
\end{landscape}


\begin{landscape}
\vspace*{\fill}
\begin{table}[htbp]
\thispagestyle{empty}
\centering
\small
\setlength{\tabcolsep}{3.0pt}
\caption{\textbf{Matthews Correlation Coefficient (MCC)}; $\alpha=0.8$, $k=4$}
\makebox[\linewidth][c]{%
\begin{tabular}{|c|c|ccc|ccc|ccc|ccc|ccc|ccc|ccc|}
\hline
&
& \multicolumn{3}{c|}{\textbf{BMT}}
& \multicolumn{3}{c|}{\textbf{OCMT}}
& \multicolumn{3}{c|}{\textbf{GOCMT}}
& \multicolumn{3}{c|}{\textbf{Lasso}}
& \multicolumn{3}{c|}{\textbf{Ad. Lasso}}
& \multicolumn{3}{c|}{\textbf{MCP}}
& \multicolumn{3}{c|}{\textbf{SCAD}} \\
\cline{3-23}

& $n/T$
& 100 & 200 & 300
& 100 & 200 & 300
& 100 & 200 & 300
& 100 & 200 & 300
& 100 & 200 & 300
& 100 & 200 & 300
& 100 & 200 & 300 \\
\hline

\multirow{3}{*}{\shortstack{V=4\\$\pi$=0}}
& 100
& 0.897 & 0.987 & 0.992
& 0.703 & 0.688 & 0.690
& 0.700 & 0.691 & 0.692
& 0.882 & 0.926 & 0.935
& 0.920 & 0.985 & 0.995
& 0.914 & 0.953 & 0.950
& 0.913 & 0.948 & 0.946 \\
& 200
& 0.868 & 0.985 & 0.989
& 0.716 & 0.695 & 0.696
& 0.714 & 0.695 & 0.696
& 0.889 & 0.923 & 0.948
& 0.887 & 0.977 & 0.989
& 0.918 & 0.969 & 0.969
& 0.919 & 0.965 & 0.968 \\
& 300
& 0.847 & 0.985 & 0.987
& 0.718 & 0.698 & 0.698
& 0.716 & 0.698 & 0.698
& 0.880 & 0.931 & 0.940
& 0.884 & 0.975 & 0.985
& 0.914 & 0.977 & 0.965
& 0.920 & 0.971 & 0.963 \\
\hline

\multirow{3}{*}{\shortstack{V=4\\$\pi$=0.5}}
& 100
& 0.892 & 0.989 & 0.992
& 0.501 & 0.315 & 0.218
& 0.833 & 0.758 & 0.712
& 0.862 & 0.919 & 0.941
& 0.917 & 0.986 & 0.996
& 0.907 & 0.952 & 0.959
& 0.909 & 0.939 & 0.955 \\
& 200
& 0.872 & 0.982 & 0.991
& 0.466 & 0.280 & 0.180
& 0.834 & 0.763 & 0.714
& 0.857 & 0.927 & 0.937
& 0.892 & 0.979 & 0.990
& 0.918 & 0.962 & 0.974
& 0.923 & 0.956 & 0.970 \\
& 300
& 0.864 & 0.981 & 0.990
& 0.421 & 0.254 & 0.163
& 0.830 & 0.764 & 0.719
& 0.850 & 0.917 & 0.948
& 0.870 & 0.968 & 0.988
& 0.919 & 0.965 & 0.968
& 0.910 & 0.956 & 0.966 \\
\hline

\multirow{3}{*}{\shortstack{V=4\\$\pi$=1}}
& 100
& 0.877 & 0.987 & 0.990
& 0.123 & 0.028 & 0.003
& 0.674 & 0.962 & 0.986
& 0.621 & 0.869 & 0.916
& 0.856 & 0.981 & 0.992
& 0.888 & 0.956 & 0.957
& 0.879 & 0.948 & 0.953 \\
& 200
& 0.830 & 0.982 & 0.990
& 0.028 & 0.006 & 0.003
& 0.635 & 0.951 & 0.985
& 0.547 & 0.850 & 0.907
& 0.818 & 0.966 & 0.987
& 0.873 & 0.964 & 0.962
& 0.868 & 0.954 & 0.954 \\
& 300
& 0.796 & 0.981 & 0.989
& 0.008 & 0.000 & 0.000
& 0.596 & 0.937 & 0.984
& 0.518 & 0.828 & 0.889
& 0.752 & 0.958 & 0.978
& 0.859 & 0.963 & 0.967
& 0.850 & 0.955 & 0.966 \\
\hline

\multirow{3}{*}{\shortstack{V=2\\$\pi$=0}}
& 100
& 0.973 & 0.988 & 0.991
& 0.873 & 0.789 & 0.735
& 0.865 & 0.785 & 0.732
& 0.872 & 0.926 & 0.941
& 0.922 & 0.984 & 0.991
& 0.946 & 0.956 & 0.955
& 0.936 & 0.958 & 0.953 \\
& 200
& 0.968 & 0.985 & 0.987
& 0.888 & 0.811 & 0.751
& 0.882 & 0.807 & 0.749
& 0.866 & 0.917 & 0.941
& 0.903 & 0.966 & 0.985
& 0.950 & 0.967 & 0.966
& 0.944 & 0.966 & 0.962 \\
& 300
& 0.949 & 0.988 & 0.988
& 0.898 & 0.821 & 0.753
& 0.891 & 0.819 & 0.750
& 0.847 & 0.902 & 0.949
& 0.888 & 0.958 & 0.979
& 0.941 & 0.969 & 0.976
& 0.938 & 0.965 & 0.977 \\
\hline

\multirow{3}{*}{\shortstack{V=2\\$\pi$=0.5}}
& 100
& 0.970 & 0.989 & 0.991
& 0.817 & 0.698 & 0.603
& 0.924 & 0.942 & 0.927
& 0.839 & 0.901 & 0.927
& 0.927 & 0.979 & 0.992
& 0.955 & 0.950 & 0.950
& 0.949 & 0.947 & 0.950 \\
& 200
& 0.963 & 0.988 & 0.990
& 0.805 & 0.692 & 0.579
& 0.925 & 0.951 & 0.928
& 0.813 & 0.892 & 0.930
& 0.888 & 0.971 & 0.982
& 0.939 & 0.963 & 0.965
& 0.935 & 0.959 & 0.964 \\
& 300
& 0.951 & 0.986 & 0.990
& 0.782 & 0.674 & 0.567
& 0.924 & 0.947 & 0.926
& 0.809 & 0.894 & 0.921
& 0.883 & 0.956 & 0.984
& 0.956 & 0.964 & 0.972
& 0.949 & 0.961 & 0.970 \\
\hline

\multirow{3}{*}{\shortstack{V=2\\$\pi$=1}}
& 100
& 0.968 & 0.987 & 0.991
& 0.583 & 0.356 & 0.232
& 0.903 & 0.980 & 0.981
& 0.790 & 0.865 & 0.903
& 0.919 & 0.981 & 0.994
& 0.942 & 0.951 & 0.962
& 0.923 & 0.948 & 0.959 \\
& 200
& 0.958 & 0.988 & 0.991
& 0.520 & 0.305 & 0.194
& 0.854 & 0.986 & 0.985
& 0.735 & 0.848 & 0.884
& 0.882 & 0.962 & 0.981
& 0.938 & 0.965 & 0.966
& 0.919 & 0.960 & 0.964 \\
& 300
& 0.944 & 0.985 & 0.989
& 0.457 & 0.282 & 0.172
& 0.824 & 0.985 & 0.987
& 0.717 & 0.862 & 0.887
& 0.819 & 0.948 & 0.976
& 0.933 & 0.976 & 0.973
& 0.902 & 0.971 & 0.973 \\
\hline

\end{tabular}
}%
\vspace{0.5em}

\parbox{\linewidth}{
\small
\textit{\textbf{Notes}:} Matthews Correlation Coefficient (MCC) is given by
$\left(\text{TP} \cdot \text{TN} - \text{FP} \times \text{FN}\right)/
\left(\sqrt{(\text{TP}+\text{FP})(\text{TP}+\text{FN})
(\text{TN}+\text{FP})(\text{TN}+\text{FN})}\right)$.
MCC ranges from $-1$ (perfect misclassification) to $+1$ (perfect selection), with $0$ indicating random guessing.
Lasso denotes the PCA prefiltered Lasso with the one standard error (1SE) penalty selection rule, while Ad. Lasso denotes adaptive Lasso with the 1SE rule applied at both penalisation stages (1SE--1SE).
}

\end{table}
\vspace*{\fill}
\end{landscape}

\begin{landscape}
\vspace*{\fill}
\begin{table}[htbp]
\thispagestyle{empty}
\centering
\small
\setlength{\tabcolsep}{3.0pt}
\caption{\textbf{Matthews Correlation Coefficient (MCC)}; $\alpha=0.4$, $k=4$}
\makebox[\linewidth][c]{%
\begin{tabular}{|c|c|ccc|ccc|ccc|ccc|ccc|ccc|ccc|}
\hline
&
& \multicolumn{3}{c|}{\textbf{BMT}}
& \multicolumn{3}{c|}{\textbf{OCMT}}
& \multicolumn{3}{c|}{\textbf{GOCMT}}
& \multicolumn{3}{c|}{\textbf{Lasso}}
& \multicolumn{3}{c|}{\textbf{Ad. Lasso}}
& \multicolumn{3}{c|}{\textbf{MCP}}
& \multicolumn{3}{c|}{\textbf{SCAD}} \\
\cline{3-23}

& $n/T$
& 100 & 200 & 300
& 100 & 200 & 300
& 100 & 200 & 300
& 100 & 200 & 300
& 100 & 200 & 300
& 100 & 200 & 300
& 100 & 200 & 300 \\
\hline

\multirow{3}{*}{\shortstack{V=4\\$\pi$=0}}
& 100
& 0.912 & 0.986 & 0.991
& 0.712 & 0.687 & 0.689
& 0.708 & 0.691 & 0.693
& 0.888 & 0.925 & 0.942
& 0.906 & 0.986 & 0.993
& 0.915 & 0.959 & 0.958
& 0.912 & 0.954 & 0.953 \\
& 200
& 0.873 & 0.987 & 0.991
& 0.727 & 0.696 & 0.697
& 0.726 & 0.695 & 0.696
& 0.876 & 0.916 & 0.938
& 0.881 & 0.975 & 0.991
& 0.934 & 0.967 & 0.972
& 0.933 & 0.960 & 0.968 \\
& 300
& 0.854 & 0.987 & 0.991
& 0.733 & 0.698 & 0.698
& 0.730 & 0.698 & 0.698
& 0.874 & 0.921 & 0.935
& 0.869 & 0.961 & 0.983
& 0.919 & 0.968 & 0.972
& 0.925 & 0.962 & 0.967 \\
\hline

\multirow{3}{*}{\shortstack{V=4\\$\pi$=0.5}}
& 100
& 0.884 & 0.987 & 0.989
& 0.537 & 0.370 & 0.263
& 0.843 & 0.784 & 0.717
& 0.850 & 0.921 & 0.946
& 0.912 & 0.986 & 0.996
& 0.905 & 0.952 & 0.960
& 0.904 & 0.946 & 0.958 \\
& 200
& 0.870 & 0.985 & 0.991
& 0.515 & 0.334 & 0.228
& 0.839 & 0.774 & 0.722
& 0.837 & 0.925 & 0.946
& 0.875 & 0.977 & 0.989
& 0.914 & 0.965 & 0.959
& 0.909 & 0.956 & 0.957 \\
& 300
& 0.859 & 0.982 & 0.991
& 0.488 & 0.308 & 0.205
& 0.851 & 0.785 & 0.726
& 0.847 & 0.920 & 0.939
& 0.866 & 0.961 & 0.983
& 0.916 & 0.968 & 0.974
& 0.918 & 0.958 & 0.968 \\
\hline

\multirow{3}{*}{\shortstack{V=4\\$\pi$=1}}
& 100
& 0.863 & 0.989 & 0.991
& 0.174 & 0.050 & 0.008
& 0.666 & 0.964 & 0.986
& 0.634 & 0.857 & 0.910
& 0.848 & 0.978 & 0.993
& 0.880 & 0.942 & 0.954
& 0.869 & 0.931 & 0.950 \\
& 200
& 0.829 & 0.981 & 0.990
& 0.052 & 0.021 & 0.009
& 0.606 & 0.941 & 0.983
& 0.533 & 0.855 & 0.886
& 0.781 & 0.962 & 0.988
& 0.867 & 0.960 & 0.967
& 0.851 & 0.952 & 0.961 \\
& 300
& 0.772 & 0.985 & 0.988
& 0.016 & 0.001 & 0.001
& 0.553 & 0.936 & 0.987
& 0.477 & 0.845 & 0.895
& 0.717 & 0.941 & 0.970
& 0.849 & 0.973 & 0.972
& 0.836 & 0.957 & 0.973 \\
\hline

\multirow{3}{*}{\shortstack{V=2\\$\pi$=0}}
& 100
& 0.973 & 0.991 & 0.993
& 0.904 & 0.823 & 0.754
& 0.903 & 0.821 & 0.752
& 0.850 & 0.916 & 0.940
& 0.923 & 0.973 & 0.989
& 0.941 & 0.947 & 0.963
& 0.937 & 0.944 & 0.960 \\
& 200
& 0.965 & 0.989 & 0.989
& 0.911 & 0.846 & 0.773
& 0.907 & 0.845 & 0.771
& 0.847 & 0.917 & 0.938
& 0.880 & 0.959 & 0.984
& 0.960 & 0.963 & 0.967
& 0.953 & 0.964 & 0.966 \\
& 300
& 0.945 & 0.988 & 0.989
& 0.910 & 0.851 & 0.788
& 0.909 & 0.846 & 0.787
& 0.823 & 0.914 & 0.930
& 0.860 & 0.945 & 0.974
& 0.948 & 0.967 & 0.978
& 0.942 & 0.965 & 0.975 \\
\hline

\multirow{3}{*}{\shortstack{V=2\\$\pi$=0.5}}
& 100
& 0.973 & 0.987 & 0.990
& 0.846 & 0.741 & 0.653
& 0.941 & 0.955 & 0.949
& 0.799 & 0.893 & 0.927
& 0.913 & 0.983 & 0.988
& 0.939 & 0.950 & 0.955
& 0.931 & 0.944 & 0.953 \\
& 200
& 0.960 & 0.989 & 0.991
& 0.843 & 0.739 & 0.639
& 0.930 & 0.964 & 0.947
& 0.799 & 0.889 & 0.923
& 0.876 & 0.955 & 0.983
& 0.946 & 0.966 & 0.966
& 0.929 & 0.965 & 0.969 \\
& 300
& 0.951 & 0.988 & 0.989
& 0.823 & 0.726 & 0.625
& 0.914 & 0.965 & 0.952
& 0.777 & 0.881 & 0.927
& 0.852 & 0.946 & 0.971
& 0.949 & 0.973 & 0.977
& 0.941 & 0.974 & 0.975 \\
\hline

\multirow{3}{*}{\shortstack{V=2\\$\pi$=1}}
& 100
& 0.966 & 0.991 & 0.991
& 0.632 & 0.424 & 0.292
& 0.888 & 0.980 & 0.986
& 0.758 & 0.866 & 0.902
& 0.909 & 0.977 & 0.988
& 0.934 & 0.949 & 0.960
& 0.915 & 0.949 & 0.956 \\
& 200
& 0.955 & 0.989 & 0.991
& 0.585 & 0.366 & 0.238
& 0.826 & 0.984 & 0.989
& 0.729 & 0.842 & 0.884
& 0.840 & 0.960 & 0.983
& 0.937 & 0.974 & 0.964
& 0.902 & 0.968 & 0.962 \\
& 300
& 0.946 & 0.989 & 0.990
& 0.536 & 0.342 & 0.218
& 0.811 & 0.985 & 0.989
& 0.679 & 0.832 & 0.873
& 0.820 & 0.943 & 0.969
& 0.930 & 0.970 & 0.973
& 0.884 & 0.967 & 0.971 \\
\hline

\end{tabular}
}%
\vspace{0.5em}

\parbox{\linewidth}{
\small
\textit{\textbf{Notes}:} Matthews Correlation Coefficient (MCC) is given by $\left(\text{TP} \cdot \text{TN} - \text{FP} \times \text{FN}\right)/
\left(\sqrt{(\text{TP}+\text{FP})(\text{TP}+\text{FN})
(\text{TN}+\text{FP})(\text{TN}+\text{FN})}\right)$.
MCC ranges from $-1$ (perfect misclassification) to $+1$ (perfect selection), with $0$ indicating random guessing.
Lasso denotes the PCA prefiltered Lasso with the one standard error (1SE) penalty selection rule, while Ad. Lasso denotes adaptive Lasso with the 1SE rule applied at both penalisation stages (1SE--1SE).
}

\end{table}
\vspace*{\fill}
\end{landscape}

\begin{landscape}
\vspace*{\fill}
\begin{table}[htbp]
\thispagestyle{empty}
\centering
\small
\setlength{\tabcolsep}{3.0pt}
\caption{\textbf{F1 Score}; $\alpha=0.8$, $k=4$}
\makebox[\linewidth][c]{%
\begin{tabular}{|c|c|ccc|ccc|ccc|ccc|ccc|ccc|ccc|}
\hline
&
& \multicolumn{3}{c|}{\textbf{BMT}}
& \multicolumn{3}{c|}{\textbf{OCMT}}
& \multicolumn{3}{c|}{\textbf{GOCMT}}
& \multicolumn{3}{c|}{\textbf{Lasso}}
& \multicolumn{3}{c|}{\textbf{Ad. Lasso}}
& \multicolumn{3}{c|}{\textbf{MCP}}
& \multicolumn{3}{c|}{\textbf{SCAD}} \\
\cline{3-23}

& $n/T$
& 100 & 200 & 300
& 100 & 200 & 300
& 100 & 200 & 300
& 100 & 200 & 300
& 100 & 200 & 300
& 100 & 200 & 300
& 100 & 200 & 300 \\
\hline

\multirow{3}{*}{\shortstack{V=4\\$\pi$=0}}
& 100
& 0.895 & 0.987 & 0.992
& 0.679 & 0.662 & 0.664
& 0.675 & 0.665 & 0.666
& 0.876 & 0.923 & 0.933
& 0.917 & 0.984 & 0.994
& 0.909 & 0.950 & 0.947
& 0.908 & 0.945 & 0.943 \\
& 200
& 0.864 & 0.984 & 0.988
& 0.686 & 0.661 & 0.662
& 0.683 & 0.661 & 0.662
& 0.880 & 0.917 & 0.944
& 0.877 & 0.975 & 0.988
& 0.911 & 0.967 & 0.967
& 0.912 & 0.962 & 0.966 \\
& 300
& 0.841 & 0.984 & 0.987
& 0.685 & 0.662 & 0.661
& 0.683 & 0.661 & 0.661
& 0.868 & 0.926 & 0.936
& 0.874 & 0.973 & 0.984
& 0.906 & 0.975 & 0.962
& 0.912 & 0.969 & 0.959 \\
\hline

\multirow{3}{*}{\shortstack{V=4\\$\pi$=0.5}}
& 100
& 0.892 & 0.988 & 0.992
& 0.444 & 0.247 & 0.165
& 0.827 & 0.742 & 0.690
& 0.854 & 0.915 & 0.938
& 0.914 & 0.986 & 0.996
& 0.902 & 0.949 & 0.956
& 0.904 & 0.935 & 0.952 \\
& 200
& 0.869 & 0.981 & 0.991
& 0.385 & 0.187 & 0.105
& 0.823 & 0.741 & 0.683
& 0.844 & 0.921 & 0.932
& 0.885 & 0.977 & 0.990
& 0.912 & 0.958 & 0.972
& 0.918 & 0.952 & 0.968 \\
& 300
& 0.859 & 0.981 & 0.989
& 0.334 & 0.155 & 0.082
& 0.818 & 0.739 & 0.687
& 0.834 & 0.910 & 0.944
& 0.860 & 0.965 & 0.987
& 0.912 & 0.961 & 0.965
& 0.902 & 0.952 & 0.963 \\
\hline

\multirow{3}{*}{\shortstack{V=4\\$\pi$=1}}
& 100
& 0.877 & 0.987 & 0.990
& 0.115 & 0.080 & 0.077
& 0.659 & 0.962 & 0.986
& 0.604 & 0.860 & 0.912
& 0.854 & 0.981 & 0.992
& 0.882 & 0.953 & 0.954
& 0.872 & 0.945 & 0.950 \\
& 200
& 0.827 & 0.982 & 0.990
& 0.052 & 0.040 & 0.039
& 0.611 & 0.947 & 0.984
& 0.527 & 0.836 & 0.899
& 0.810 & 0.964 & 0.986
& 0.864 & 0.961 & 0.959
& 0.858 & 0.950 & 0.950 \\
& 300
& 0.793 & 0.981 & 0.989
& 0.031 & 0.026 & 0.026
& 0.578 & 0.932 & 0.984
& 0.492 & 0.810 & 0.877
& 0.739 & 0.954 & 0.976
& 0.847 & 0.959 & 0.964
& 0.836 & 0.950 & 0.964 \\
\hline

\multirow{3}{*}{\shortstack{V=2\\$\pi$=0}}
& 100
& 0.973 & 0.988 & 0.991
& 0.868 & 0.776 & 0.716
& 0.860 & 0.772 & 0.713
& 0.864 & 0.922 & 0.938
& 0.917 & 0.983 & 0.991
& 0.942 & 0.953 & 0.952
& 0.932 & 0.955 & 0.950 \\
& 200
& 0.967 & 0.985 & 0.987
& 0.881 & 0.796 & 0.726
& 0.875 & 0.791 & 0.725
& 0.854 & 0.910 & 0.936
& 0.893 & 0.963 & 0.983
& 0.945 & 0.964 & 0.963
& 0.939 & 0.963 & 0.958 \\
& 300
& 0.947 & 0.988 & 0.988
& 0.892 & 0.806 & 0.727
& 0.884 & 0.803 & 0.724
& 0.830 & 0.892 & 0.945
& 0.876 & 0.953 & 0.977
& 0.934 & 0.967 & 0.974
& 0.931 & 0.962 & 0.975 \\
\hline

\multirow{3}{*}{\shortstack{V=2\\$\pi$=0.5}}
& 100
& 0.970 & 0.988 & 0.991
& 0.806 & 0.671 & 0.559
& 0.922 & 0.940 & 0.925
& 0.828 & 0.895 & 0.922
& 0.923 & 0.978 & 0.992
& 0.953 & 0.946 & 0.947
& 0.946 & 0.943 & 0.947 \\
& 200
& 0.962 & 0.988 & 0.990
& 0.787 & 0.653 & 0.517
& 0.921 & 0.949 & 0.924
& 0.793 & 0.883 & 0.924
& 0.878 & 0.969 & 0.980
& 0.933 & 0.960 & 0.962
& 0.930 & 0.955 & 0.961 \\
& 300
& 0.949 & 0.986 & 0.989
& 0.759 & 0.628 & 0.498
& 0.920 & 0.944 & 0.922
& 0.787 & 0.883 & 0.913
& 0.870 & 0.951 & 0.983
& 0.952 & 0.961 & 0.970
& 0.944 & 0.957 & 0.968 \\
\hline

\multirow{3}{*}{\shortstack{V=2\\$\pi$=1}}
& 100
& 0.968 & 0.987 & 0.991
& 0.539 & 0.286 & 0.175
& 0.899 & 0.980 & 0.981
& 0.775 & 0.857 & 0.897
& 0.914 & 0.980 & 0.993
& 0.938 & 0.948 & 0.960
& 0.917 & 0.945 & 0.957 \\
& 200
& 0.957 & 0.988 & 0.991
& 0.451 & 0.209 & 0.113
& 0.841 & 0.986 & 0.985
& 0.706 & 0.832 & 0.873
& 0.872 & 0.959 & 0.980
& 0.932 & 0.961 & 0.963
& 0.912 & 0.957 & 0.961 \\
& 300
& 0.942 & 0.984 & 0.989
& 0.375 & 0.177 & 0.087
& 0.808 & 0.984 & 0.987
& 0.686 & 0.846 & 0.876
& 0.801 & 0.943 & 0.973
& 0.927 & 0.974 & 0.970
& 0.892 & 0.968 & 0.971 \\
\hline

\end{tabular}
}%
\vspace{0.5em}

\parbox{\linewidth}{
\small
\textit{\textbf{Notes}:} The F1 Score is the harmonic mean of precision and recall, defined as
$F_1 = \frac{2 \cdot \text{TP}}
{2 \cdot \text{TP} + \text{FP} + \text{FN}}$.
Values closer to 1 indicate better performance.
Lasso denotes the PCA prefiltered Lasso with the one standard error (1SE) penalty selection rule, while Ad. Lasso denotes adaptive Lasso with the 1SE rule applied at both penalisation stages (1SE--1SE).
}

\end{table}
\vspace*{\fill}
\end{landscape}

\begin{landscape}
\vspace*{\fill}
\begin{table}[htbp]
\thispagestyle{empty}
\centering
\small
\setlength{\tabcolsep}{3.0pt}
\caption{\textbf{F1 Score}; $\alpha=0.4$, $k=4$}
\makebox[\linewidth][c]{%
\begin{tabular}{|c|c|ccc|ccc|ccc|ccc|ccc|ccc|ccc|}
\hline
&
& \multicolumn{3}{c|}{\textbf{BMT}}
& \multicolumn{3}{c|}{\textbf{OCMT}}
& \multicolumn{3}{c|}{\textbf{GOCMT}}
& \multicolumn{3}{c|}{\textbf{Lasso}}
& \multicolumn{3}{c|}{\textbf{Ad. Lasso}}
& \multicolumn{3}{c|}{\textbf{MCP}}
& \multicolumn{3}{c|}{\textbf{SCAD}} \\
\cline{3-23}

& $n/T$
& 100 & 200 & 300
& 100 & 200 & 300
& 100 & 200 & 300
& 100 & 200 & 300
& 100 & 200 & 300
& 100 & 200 & 300
& 100 & 200 & 300 \\
\hline

\multirow{3}{*}{\shortstack{V=4\\$\pi$=0}}
& 100
& 0.910 & 0.985 & 0.991
& 0.689 & 0.661 & 0.663
& 0.685 & 0.665 & 0.667
& 0.882 & 0.922 & 0.940
& 0.901 & 0.985 & 0.992
& 0.910 & 0.956 & 0.955
& 0.907 & 0.952 & 0.950 \\
& 200
& 0.868 & 0.987 & 0.991
& 0.699 & 0.661 & 0.663
& 0.698 & 0.661 & 0.662
& 0.865 & 0.909 & 0.934
& 0.872 & 0.973 & 0.990
& 0.929 & 0.964 & 0.970
& 0.928 & 0.957 & 0.966 \\
& 300
& 0.847 & 0.986 & 0.991
& 0.704 & 0.661 & 0.662
& 0.699 & 0.661 & 0.662
& 0.862 & 0.915 & 0.930
& 0.856 & 0.956 & 0.981
& 0.912 & 0.965 & 0.969
& 0.919 & 0.959 & 0.963 \\
\hline

\multirow{3}{*}{\shortstack{V=4\\$\pi$=0.5}}
& 100
& 0.882 & 0.987 & 0.989
& 0.484 & 0.300 & 0.200
& 0.839 & 0.770 & 0.695
& 0.841 & 0.917 & 0.943
& 0.908 & 0.985 & 0.996
& 0.900 & 0.949 & 0.958
& 0.898 & 0.942 & 0.955 \\
& 200
& 0.866 & 0.984 & 0.991
& 0.442 & 0.236 & 0.140
& 0.831 & 0.754 & 0.692
& 0.822 & 0.919 & 0.942
& 0.866 & 0.975 & 0.988
& 0.906 & 0.962 & 0.955
& 0.902 & 0.952 & 0.953 \\
& 300
& 0.853 & 0.981 & 0.991
& 0.403 & 0.202 & 0.110
& 0.843 & 0.764 & 0.695
& 0.832 & 0.913 & 0.934
& 0.855 & 0.957 & 0.981
& 0.909 & 0.965 & 0.972
& 0.911 & 0.955 & 0.965 \\
\hline

\multirow{3}{*}{\shortstack{V=4\\$\pi$=1}}
& 100
& 0.863 & 0.989 & 0.991
& 0.140 & 0.083 & 0.077
& 0.657 & 0.963 & 0.986
& 0.616 & 0.847 & 0.905
& 0.846 & 0.977 & 0.993
& 0.873 & 0.938 & 0.951
& 0.861 & 0.926 & 0.947 \\
& 200
& 0.827 & 0.981 & 0.990
& 0.062 & 0.042 & 0.040
& 0.581 & 0.937 & 0.982
& 0.510 & 0.841 & 0.875
& 0.771 & 0.960 & 0.987
& 0.857 & 0.956 & 0.965
& 0.838 & 0.948 & 0.957 \\
& 300
& 0.769 & 0.984 & 0.987
& 0.035 & 0.026 & 0.026
& 0.535 & 0.932 & 0.986
& 0.453 & 0.830 & 0.885
& 0.702 & 0.936 & 0.967
& 0.839 & 0.971 & 0.970
& 0.822 & 0.953 & 0.971 \\
\hline

\multirow{3}{*}{\shortstack{V=2\\$\pi$=0}}
& 100
& 0.973 & 0.990 & 0.993
& 0.902 & 0.814 & 0.737
& 0.900 & 0.811 & 0.735
& 0.840 & 0.912 & 0.937
& 0.918 & 0.971 & 0.988
& 0.937 & 0.944 & 0.961
& 0.932 & 0.941 & 0.957 \\
& 200
& 0.964 & 0.989 & 0.988
& 0.906 & 0.835 & 0.753
& 0.901 & 0.834 & 0.751
& 0.831 & 0.910 & 0.933
& 0.868 & 0.955 & 0.983
& 0.956 & 0.959 & 0.965
& 0.949 & 0.960 & 0.963 \\
& 300
& 0.943 & 0.988 & 0.988
& 0.905 & 0.839 & 0.767
& 0.904 & 0.833 & 0.767
& 0.802 & 0.906 & 0.924
& 0.844 & 0.938 & 0.971
& 0.942 & 0.964 & 0.977
& 0.936 & 0.962 & 0.973 \\
\hline

\multirow{3}{*}{\shortstack{V=2\\$\pi$=0.5}}
& 100
& 0.973 & 0.987 & 0.990
& 0.839 & 0.720 & 0.619
& 0.941 & 0.954 & 0.947
& 0.783 & 0.887 & 0.923
& 0.908 & 0.982 & 0.987
& 0.935 & 0.947 & 0.952
& 0.926 & 0.940 & 0.950 \\
& 200
& 0.958 & 0.988 & 0.991
& 0.830 & 0.710 & 0.591
& 0.926 & 0.962 & 0.945
& 0.778 & 0.879 & 0.916
& 0.864 & 0.951 & 0.981
& 0.941 & 0.963 & 0.963
& 0.922 & 0.962 & 0.967 \\
& 300
& 0.949 & 0.987 & 0.989
& 0.808 & 0.692 & 0.569
& 0.909 & 0.963 & 0.949
& 0.749 & 0.869 & 0.921
& 0.835 & 0.939 & 0.967
& 0.944 & 0.971 & 0.974
& 0.935 & 0.971 & 0.973 \\
\hline

\multirow{3}{*}{\shortstack{V=2\\$\pi$=1}}
& 100
& 0.965 & 0.991 & 0.991
& 0.595 & 0.357 & 0.224
& 0.883 & 0.980 & 0.986
& 0.738 & 0.857 & 0.896
& 0.903 & 0.976 & 0.988
& 0.930 & 0.945 & 0.957
& 0.910 & 0.947 & 0.953 \\
& 200
& 0.954 & 0.989 & 0.991
& 0.529 & 0.271 & 0.145
& 0.818 & 0.983 & 0.989
& 0.701 & 0.826 & 0.873
& 0.824 & 0.956 & 0.981
& 0.931 & 0.971 & 0.960
& 0.893 & 0.965 & 0.959 \\
& 300
& 0.944 & 0.988 & 0.990
& 0.466 & 0.237 & 0.119
& 0.800 & 0.985 & 0.989
& 0.645 & 0.813 & 0.859
& 0.800 & 0.937 & 0.967
& 0.923 & 0.967 & 0.971
& 0.871 & 0.964 & 0.969 \\
\hline

\end{tabular}
}%
\vspace{0.5em}

\parbox{\linewidth}{
\small
\textit{\textbf{Notes}:} The F1 Score is the harmonic mean of precision and recall, defined as
$F_1 = \frac{2 \cdot \text{TP}}
{2 \cdot \text{TP} + \text{FP} + \text{FN}}$.
Values closer to 1 indicate better performance.
Lasso denotes the PCA prefiltered Lasso with the one standard error (1SE) penalty selection rule, while Ad. Lasso denotes adaptive Lasso with the 1SE rule applied at both penalisation stages (1SE--1SE).
}

\end{table}
\vspace*{\fill}
\end{landscape}

\begin{landscape}
\vspace*{\fill}
\begin{table}[htbp]
\thispagestyle{empty}
\centering
\small
\setlength{\tabcolsep}{3.0pt}
\caption{\textbf{True Discovery Rate (TDR)}; $\alpha=0.8$, $k=4$}
\makebox[\linewidth][c]{%
\begin{tabular}{|c|c|ccc|ccc|ccc|ccc|ccc|ccc|ccc|}
\hline
&
& \multicolumn{3}{c|}{\textbf{BMT}}
& \multicolumn{3}{c|}{\textbf{OCMT}}
& \multicolumn{3}{c|}{\textbf{GOCMT}}
& \multicolumn{3}{c|}{\textbf{Lasso}}
& \multicolumn{3}{c|}{\textbf{Ad. Lasso}}
& \multicolumn{3}{c|}{\textbf{MCP}}
& \multicolumn{3}{c|}{\textbf{SCAD}} \\
\cline{3-23}

& $n/T$
& 100 & 200 & 300
& 100 & 200 & 300
& 100 & 200 & 300
& 100 & 200 & 300
& 100 & 200 & 300
& 100 & 200 & 300
& 100 & 200 & 300 \\
\hline

\multirow{3}{*}{\shortstack{V=4\\$\pi$=0}}
& 100
& 0.945 & 0.981 & 0.986
& 0.516 & 0.495 & 0.497
& 0.516 & 0.501 & 0.499
& 0.808 & 0.872 & 0.887
& 0.923 & 0.980 & 0.991
& 0.882 & 0.927 & 0.918
& 0.869 & 0.919 & 0.912 \\
& 200
& 0.930 & 0.976 & 0.979
& 0.526 & 0.494 & 0.495
& 0.522 & 0.493 & 0.495
& 0.814 & 0.864 & 0.907
& 0.869 & 0.967 & 0.982
& 0.899 & 0.948 & 0.947
& 0.883 & 0.940 & 0.945 \\
& 300
& 0.914 & 0.976 & 0.976
& 0.524 & 0.495 & 0.494
& 0.523 & 0.494 & 0.494
& 0.803 & 0.877 & 0.893
& 0.863 & 0.963 & 0.976
& 0.898 & 0.961 & 0.940
& 0.888 & 0.950 & 0.936 \\
\hline

\multirow{3}{*}{\shortstack{V=4\\$\pi$=0.5}}
& 100
& 0.939 & 0.981 & 0.986
& 0.297 & 0.144 & 0.092
& 0.741 & 0.597 & 0.530
& 0.794 & 0.862 & 0.898
& 0.912 & 0.980 & 0.993
& 0.879 & 0.922 & 0.932
& 0.868 & 0.901 & 0.925 \\
& 200
& 0.927 & 0.971 & 0.983
& 0.251 & 0.106 & 0.056
& 0.740 & 0.596 & 0.521
& 0.782 & 0.873 & 0.889
& 0.870 & 0.965 & 0.984
& 0.898 & 0.938 & 0.955
& 0.887 & 0.926 & 0.947 \\
& 300
& 0.923 & 0.972 & 0.981
& 0.214 & 0.087 & 0.043
& 0.728 & 0.594 & 0.526
& 0.775 & 0.857 & 0.909
& 0.833 & 0.953 & 0.980
& 0.895 & 0.941 & 0.944
& 0.862 & 0.926 & 0.942 \\
\hline

\multirow{3}{*}{\shortstack{V=4\\$\pi$=1}}
& 100
& 0.927 & 0.980 & 0.983
& 0.062 & 0.041 & 0.040
& 0.916 & 0.969 & 0.975
& 0.616 & 0.790 & 0.859
& 0.836 & 0.970 & 0.986
& 0.847 & 0.929 & 0.930
& 0.821 & 0.916 & 0.922 \\
& 200
& 0.888 & 0.972 & 0.982
& 0.027 & 0.020 & 0.020
& 0.906 & 0.970 & 0.977
& 0.539 & 0.766 & 0.843
& 0.772 & 0.946 & 0.977
& 0.833 & 0.939 & 0.937
& 0.809 & 0.920 & 0.923 \\
& 300
& 0.842 & 0.973 & 0.980
& 0.016 & 0.013 & 0.013
& 0.893 & 0.966 & 0.977
& 0.533 & 0.740 & 0.813
& 0.682 & 0.931 & 0.962
& 0.819 & 0.939 & 0.943
& 0.786 & 0.925 & 0.941 \\
\hline

\multirow{3}{*}{\shortstack{V=2\\$\pi$=0}}
& 100
& 0.964 & 0.979 & 0.985
& 0.784 & 0.644 & 0.561
& 0.771 & 0.638 & 0.557
& 0.789 & 0.874 & 0.899
& 0.879 & 0.973 & 0.985
& 0.916 & 0.928 & 0.925
& 0.902 & 0.930 & 0.922 \\
& 200
& 0.955 & 0.972 & 0.977
& 0.807 & 0.672 & 0.576
& 0.798 & 0.665 & 0.574
& 0.778 & 0.858 & 0.897
& 0.850 & 0.943 & 0.974
& 0.919 & 0.944 & 0.942
& 0.910 & 0.943 & 0.935 \\
& 300
& 0.939 & 0.979 & 0.978
& 0.828 & 0.686 & 0.576
& 0.815 & 0.683 & 0.572
& 0.750 & 0.833 & 0.912
& 0.829 & 0.929 & 0.964
& 0.909 & 0.948 & 0.957
& 0.904 & 0.941 & 0.959 \\
\hline

\multirow{3}{*}{\shortstack{V=2\\$\pi$=0.5}}
& 100
& 0.962 & 0.979 & 0.984
& 0.699 & 0.521 & 0.399
& 0.895 & 0.897 & 0.871
& 0.743 & 0.834 & 0.877
& 0.893 & 0.965 & 0.987
& 0.926 & 0.918 & 0.920
& 0.917 & 0.914 & 0.920 \\
& 200
& 0.956 & 0.978 & 0.982
& 0.680 & 0.505 & 0.360
& 0.911 & 0.911 & 0.869
& 0.700 & 0.818 & 0.878
& 0.827 & 0.951 & 0.969
& 0.903 & 0.939 & 0.941
& 0.896 & 0.932 & 0.939 \\
& 300
& 0.938 & 0.974 & 0.981
& 0.647 & 0.476 & 0.344
& 0.912 & 0.903 & 0.866
& 0.698 & 0.820 & 0.865
& 0.818 & 0.925 & 0.973
& 0.933 & 0.938 & 0.952
& 0.923 & 0.932 & 0.949 \\
\hline

\multirow{3}{*}{\shortstack{V=2\\$\pi$=1}}
& 100
& 0.958 & 0.977 & 0.984
& 0.397 & 0.173 & 0.097
& 0.946 & 0.965 & 0.966
& 0.682 & 0.777 & 0.839
& 0.876 & 0.968 & 0.989
& 0.909 & 0.922 & 0.937
& 0.880 & 0.917 & 0.932 \\
& 200
& 0.950 & 0.978 & 0.983
& 0.315 & 0.121 & 0.061
& 0.940 & 0.975 & 0.973
& 0.626 & 0.748 & 0.803
& 0.817 & 0.935 & 0.967
& 0.898 & 0.940 & 0.941
& 0.868 & 0.933 & 0.938 \\
& 300
& 0.931 & 0.972 & 0.980
& 0.250 & 0.100 & 0.046
& 0.932 & 0.973 & 0.976
& 0.615 & 0.769 & 0.805
& 0.719 & 0.913 & 0.958
& 0.892 & 0.957 & 0.953
& 0.842 & 0.949 & 0.953 \\
\hline

\end{tabular}
}%
\vspace{0.5em}

\parbox{\linewidth}{
\small
\textit{\textbf{Notes}:} True Discovery Rate (TDR) is the proportion of selected
variables that are truly relevant. Higher values indicate better precision in
variable selection.
Lasso denotes the PCA prefiltered Lasso with the one standard error (1SE) penalty selection rule, while Ad. Lasso denotes adaptive Lasso with the 1SE rule applied at both penalisation stages (1SE--1SE).
}

\end{table}
\vspace*{\fill}
\end{landscape}

\begin{landscape}
\vspace*{\fill}
\begin{table}[htbp]
\thispagestyle{empty}
\centering
\small
\setlength{\tabcolsep}{3.0pt}
\caption{\textbf{True Discovery Rate (TDR)}; $\alpha=0.4$, $k=4$}
\makebox[\linewidth][c]{%
\begin{tabular}{|c|c|ccc|ccc|ccc|ccc|ccc|ccc|ccc|}
\hline
&
& \multicolumn{3}{c|}{\textbf{BMT}}
& \multicolumn{3}{c|}{\textbf{OCMT}}
& \multicolumn{3}{c|}{\textbf{GOCMT}}
& \multicolumn{3}{c|}{\textbf{Lasso}}
& \multicolumn{3}{c|}{\textbf{Ad. Lasso}}
& \multicolumn{3}{c|}{\textbf{MCP}}
& \multicolumn{3}{c|}{\textbf{SCAD}} \\
\cline{3-23}

& $n/T$
& 100 & 200 & 300
& 100 & 200 & 300
& 100 & 200 & 300
& 100 & 200 & 300
& 100 & 200 & 300
& 100 & 200 & 300
& 100 & 200 & 300 \\
\hline

\multirow{3}{*}{\shortstack{V=4\\$\pi$=0}}
& 100
& 0.960 & 0.976 & 0.984
& 0.530 & 0.494 & 0.496
& 0.528 & 0.501 & 0.502
& 0.815 & 0.871 & 0.899
& 0.894 & 0.982 & 0.989
& 0.882 & 0.933 & 0.930
& 0.874 & 0.926 & 0.923 \\
& 200
& 0.938 & 0.977 & 0.984
& 0.543 & 0.494 & 0.496
& 0.541 & 0.494 & 0.495
& 0.793 & 0.854 & 0.891
& 0.856 & 0.960 & 0.986
& 0.913 & 0.944 & 0.951
& 0.899 & 0.932 & 0.944 \\
& 300
& 0.918 & 0.979 & 0.983
& 0.549 & 0.495 & 0.495
& 0.543 & 0.494 & 0.494
& 0.787 & 0.862 & 0.884
& 0.830 & 0.939 & 0.973
& 0.895 & 0.945 & 0.952
& 0.893 & 0.934 & 0.944 \\
\hline

\multirow{3}{*}{\shortstack{V=4\\$\pi$=0.5}}
& 100
& 0.942 & 0.979 & 0.980
& 0.331 & 0.181 & 0.113
& 0.768 & 0.638 & 0.536
& 0.782 & 0.867 & 0.907
& 0.906 & 0.979 & 0.994
& 0.874 & 0.923 & 0.934
& 0.858 & 0.912 & 0.930 \\
& 200
& 0.925 & 0.976 & 0.984
& 0.298 & 0.138 & 0.077
& 0.759 & 0.613 & 0.533
& 0.762 & 0.871 & 0.904
& 0.848 & 0.962 & 0.981
& 0.880 & 0.943 & 0.931
& 0.863 & 0.928 & 0.927 \\
& 300
& 0.927 & 0.974 & 0.983
& 0.265 & 0.116 & 0.059
& 0.777 & 0.628 & 0.535
& 0.774 & 0.863 & 0.893
& 0.833 & 0.937 & 0.970
& 0.885 & 0.944 & 0.956
& 0.873 & 0.928 & 0.946 \\
\hline

\multirow{3}{*}{\shortstack{V=4\\$\pi$=1}}
& 100
& 0.919 & 0.984 & 0.983
& 0.076 & 0.043 & 0.040
& 0.915 & 0.971 & 0.976
& 0.637 & 0.777 & 0.850
& 0.824 & 0.964 & 0.988
& 0.838 & 0.910 & 0.925
& 0.812 & 0.892 & 0.919 \\
& 200
& 0.881 & 0.973 & 0.981
& 0.033 & 0.021 & 0.020
& 0.919 & 0.974 & 0.971
& 0.542 & 0.771 & 0.808
& 0.727 & 0.940 & 0.979
& 0.830 & 0.936 & 0.943
& 0.790 & 0.920 & 0.931 \\
& 300
& 0.827 & 0.975 & 0.977
& 0.018 & 0.013 & 0.013
& 0.868 & 0.969 & 0.980
& 0.484 & 0.765 & 0.820
& 0.636 & 0.902 & 0.948
& 0.812 & 0.955 & 0.952
& 0.773 & 0.929 & 0.953 \\
\hline

\multirow{3}{*}{\shortstack{V=2\\$\pi$=0}}
& 100
& 0.965 & 0.983 & 0.987
& 0.841 & 0.698 & 0.590
& 0.837 & 0.695 & 0.587
& 0.755 & 0.858 & 0.897
& 0.883 & 0.956 & 0.981
& 0.908 & 0.914 & 0.939
& 0.900 & 0.910 & 0.934 \\
& 200
& 0.952 & 0.980 & 0.979
& 0.854 & 0.729 & 0.611
& 0.847 & 0.727 & 0.608
& 0.746 & 0.857 & 0.894
& 0.814 & 0.932 & 0.973
& 0.934 & 0.936 & 0.944
& 0.924 & 0.939 & 0.941 \\
& 300
& 0.934 & 0.978 & 0.979
& 0.860 & 0.737 & 0.631
& 0.857 & 0.727 & 0.631
& 0.709 & 0.853 & 0.878
& 0.784 & 0.909 & 0.957
& 0.919 & 0.943 & 0.961
& 0.906 & 0.941 & 0.956 \\
\hline

\multirow{3}{*}{\shortstack{V=2\\$\pi$=0.5}}
& 100
& 0.963 & 0.977 & 0.982
& 0.747 & 0.579 & 0.459
& 0.919 & 0.920 & 0.909
& 0.683 & 0.822 & 0.876
& 0.866 & 0.971 & 0.980
& 0.904 & 0.919 & 0.925
& 0.891 & 0.909 & 0.923 \\
& 200
& 0.952 & 0.979 & 0.984
& 0.742 & 0.570 & 0.431
& 0.927 & 0.934 & 0.904
& 0.683 & 0.810 & 0.867
& 0.806 & 0.924 & 0.971
& 0.916 & 0.942 & 0.943
& 0.888 & 0.941 & 0.948 \\
& 300
& 0.940 & 0.978 & 0.980
& 0.718 & 0.551 & 0.410
& 0.912 & 0.935 & 0.911
& 0.654 & 0.799 & 0.873
& 0.769 & 0.910 & 0.950
& 0.918 & 0.955 & 0.960
& 0.905 & 0.956 & 0.958 \\
\hline

\multirow{3}{*}{\shortstack{V=2\\$\pi$=1}}
& 100
& 0.954 & 0.983 & 0.985
& 0.450 & 0.226 & 0.128
& 0.945 & 0.965 & 0.975
& 0.644 & 0.778 & 0.835
& 0.862 & 0.962 & 0.980
& 0.895 & 0.916 & 0.933
& 0.864 & 0.916 & 0.927 \\
& 200
& 0.943 & 0.980 & 0.983
& 0.388 & 0.163 & 0.079
& 0.923 & 0.970 & 0.980
& 0.623 & 0.736 & 0.806
& 0.747 & 0.932 & 0.971
& 0.897 & 0.955 & 0.938
& 0.840 & 0.945 & 0.936 \\
& 300
& 0.931 & 0.979 & 0.982
& 0.331 & 0.140 & 0.064
& 0.919 & 0.973 & 0.980
& 0.580 & 0.718 & 0.786
& 0.719 & 0.905 & 0.947
& 0.884 & 0.948 & 0.952
& 0.814 & 0.943 & 0.948 \\
\hline

\end{tabular}
}%
\vspace{0.5em}

\parbox{\linewidth}{
\small
\textit{\textbf{Notes}:} True Discovery Rate (TDR) is the proportion of selected
variables that are truly relevant. Higher values indicate better precision in
variable selection. Lasso denotes the PCA prefiltered Lasso with the one standard error (1SE) penalty selection rule, while Ad. Lasso denotes adaptive Lasso with the 1SE rule applied at both penalisation stages (1SE--1SE).
}

\end{table}
\vspace*{\fill}
\end{landscape}

\begin{landscape}
\vspace*{\fill}
\begin{table}[htbp]
\thispagestyle{empty}
\centering
\small
\setlength{\tabcolsep}{3.0pt}
\caption{\textbf{False Discovery Rate (FDR)}; $\alpha=0.8$, $k=4$}
\makebox[\linewidth][c]{%
\begin{tabular}{|c|c|ccc|ccc|ccc|ccc|ccc|ccc|ccc|}
\hline
&
& \multicolumn{3}{c|}{\textbf{BMT}}
& \multicolumn{3}{c|}{\textbf{OCMT}}
& \multicolumn{3}{c|}{\textbf{GOCMT}}
& \multicolumn{3}{c|}{\textbf{Lasso}}
& \multicolumn{3}{c|}{\textbf{Ad. Lasso}}
& \multicolumn{3}{c|}{\textbf{MCP}}
& \multicolumn{3}{c|}{\textbf{SCAD}} \\
\cline{3-23}

& $n/T$
& 100 & 200 & 300
& 100 & 200 & 300
& 100 & 200 & 300
& 100 & 200 & 300
& 100 & 200 & 300
& 100 & 200 & 300
& 100 & 200 & 300 \\
\hline

\multirow{3}{*}{\shortstack{V=4\\$\pi$=0}}
& 100
& 0.055 & 0.019 & 0.014
& 0.484 & 0.505 & 0.503
& 0.484 & 0.499 & 0.501
& 0.186 & 0.128 & 0.113
& 0.077 & 0.020 & 0.009
& 0.118 & 0.073 & 0.082
& 0.131 & 0.081 & 0.088 \\
& 200
& 0.070 & 0.024 & 0.021
& 0.474 & 0.506 & 0.505
& 0.478 & 0.507 & 0.505
& 0.186 & 0.136 & 0.093
& 0.131 & 0.033 & 0.018
& 0.101 & 0.052 & 0.053
& 0.117 & 0.060 & 0.055 \\
& 300
& 0.086 & 0.024 & 0.024
& 0.476 & 0.505 & 0.506
& 0.477 & 0.506 & 0.506
& 0.197 & 0.123 & 0.107
& 0.137 & 0.037 & 0.024
& 0.102 & 0.039 & 0.060
& 0.112 & 0.050 & 0.064 \\
\hline

\multirow{3}{*}{\shortstack{V=4\\$\pi$=0.5}}
& 100
& 0.061 & 0.019 & 0.014
& 0.703 & 0.856 & 0.908
& 0.259 & 0.403 & 0.470
& 0.206 & 0.138 & 0.102
& 0.088 & 0.020 & 0.007
& 0.121 & 0.078 & 0.068
& 0.132 & 0.099 & 0.075 \\
& 200
& 0.073 & 0.029 & 0.017
& 0.749 & 0.894 & 0.944
& 0.260 & 0.404 & 0.479
& 0.218 & 0.127 & 0.111
& 0.130 & 0.035 & 0.016
& 0.102 & 0.062 & 0.045
& 0.113 & 0.074 & 0.053 \\
& 300
& 0.077 & 0.028 & 0.019
& 0.786 & 0.913 & 0.957
& 0.272 & 0.406 & 0.474
& 0.225 & 0.143 & 0.091
& 0.167 & 0.047 & 0.020
& 0.105 & 0.059 & 0.056
& 0.138 & 0.074 & 0.058 \\
\hline

\multirow{3}{*}{\shortstack{V=4\\$\pi$=1}}
& 100
& 0.073 & 0.020 & 0.017
& 0.938 & 0.959 & 0.960
& 0.084 & 0.031 & 0.025
& 0.248 & 0.210 & 0.141
& 0.164 & 0.030 & 0.014
& 0.153 & 0.071 & 0.070
& 0.179 & 0.084 & 0.078 \\
& 200
& 0.112 & 0.028 & 0.018
& 0.973 & 0.980 & 0.980
& 0.094 & 0.030 & 0.023
& 0.259 & 0.230 & 0.157
& 0.228 & 0.054 & 0.023
& 0.167 & 0.061 & 0.063
& 0.191 & 0.080 & 0.077 \\
& 300
& 0.158 & 0.027 & 0.020
& 0.984 & 0.987 & 0.987
& 0.107 & 0.034 & 0.023
& 0.247 & 0.254 & 0.187
& 0.318 & 0.069 & 0.038
& 0.181 & 0.061 & 0.057
& 0.214 & 0.075 & 0.059 \\
\hline

\multirow{3}{*}{\shortstack{V=2\\$\pi$=0}}
& 100
& 0.036 & 0.021 & 0.015
& 0.216 & 0.356 & 0.439
& 0.229 & 0.362 & 0.443
& 0.211 & 0.126 & 0.101
& 0.121 & 0.027 & 0.015
& 0.084 & 0.072 & 0.075
& 0.098 & 0.070 & 0.078 \\
& 200
& 0.045 & 0.028 & 0.023
& 0.193 & 0.328 & 0.424
& 0.202 & 0.335 & 0.426
& 0.222 & 0.142 & 0.103
& 0.150 & 0.057 & 0.026
& 0.081 & 0.056 & 0.058
& 0.090 & 0.057 & 0.065 \\
& 300
& 0.061 & 0.021 & 0.022
& 0.172 & 0.314 & 0.424
& 0.185 & 0.317 & 0.428
& 0.250 & 0.167 & 0.088
& 0.171 & 0.071 & 0.036
& 0.091 & 0.052 & 0.043
& 0.096 & 0.059 & 0.041 \\
\hline

\multirow{3}{*}{\shortstack{V=2\\$\pi$=0.5}}
& 100
& 0.038 & 0.021 & 0.016
& 0.301 & 0.479 & 0.601
& 0.105 & 0.103 & 0.129
& 0.257 & 0.166 & 0.123
& 0.107 & 0.035 & 0.013
& 0.074 & 0.082 & 0.080
& 0.083 & 0.086 & 0.080 \\
& 200
& 0.044 & 0.022 & 0.018
& 0.320 & 0.495 & 0.640
& 0.089 & 0.089 & 0.131
& 0.300 & 0.182 & 0.122
& 0.173 & 0.049 & 0.031
& 0.097 & 0.061 & 0.059
& 0.104 & 0.068 & 0.061 \\
& 300
& 0.062 & 0.026 & 0.019
& 0.353 & 0.524 & 0.656
& 0.088 & 0.097 & 0.134
& 0.302 & 0.180 & 0.135
& 0.182 & 0.075 & 0.027
& 0.067 & 0.062 & 0.048
& 0.077 & 0.068 & 0.051 \\
\hline

\multirow{3}{*}{\shortstack{V=2\\$\pi$=1}}
& 100
& 0.042 & 0.023 & 0.016
& 0.603 & 0.827 & 0.903
& 0.054 & 0.035 & 0.034
& 0.312 & 0.223 & 0.161
& 0.124 & 0.032 & 0.011
& 0.091 & 0.078 & 0.063
& 0.120 & 0.083 & 0.068 \\
& 200
& 0.050 & 0.022 & 0.017
& 0.685 & 0.879 & 0.939
& 0.060 & 0.025 & 0.027
& 0.358 & 0.252 & 0.197
& 0.183 & 0.065 & 0.033
& 0.102 & 0.060 & 0.059
& 0.132 & 0.067 & 0.062 \\
& 300
& 0.069 & 0.028 & 0.020
& 0.750 & 0.900 & 0.954
& 0.068 & 0.027 & 0.024
& 0.363 & 0.231 & 0.195
& 0.281 & 0.087 & 0.042
& 0.108 & 0.043 & 0.047
& 0.158 & 0.051 & 0.047 \\
\hline

\end{tabular}
}%
\vspace{0.5em}

\parbox{\linewidth}{
\small
\textit{\textbf{Notes}:} False Discovery Rate (FDR) is the proportion of
irrelevant variables among the selected ones. Lower values indicate more
accurate variable selection.
Lasso denotes the PCA prefiltered Lasso with the one standard error
(1SE) penalty selection rule, while Ad. Lasso denotes adaptive Lasso
with the 1SE rule applied at both penalisation stages (1SE--1SE).
}

\end{table}
\vspace*{\fill}

\end{landscape}

\begin{landscape}
\vspace*{\fill}
\begin{table}[htbp]
\thispagestyle{empty}
\centering
\small
\setlength{\tabcolsep}{3.0pt}
\caption{\textbf{False Discovery Rate (FDR)}; $\alpha=0.4$, $k=4$}
\makebox[\linewidth][c]{%
\begin{tabular}{|c|c|ccc|ccc|ccc|ccc|ccc|ccc|ccc|}
\hline
&
& \multicolumn{3}{c|}{\textbf{BMT}}
& \multicolumn{3}{c|}{\textbf{OCMT}}
& \multicolumn{3}{c|}{\textbf{GOCMT}}
& \multicolumn{3}{c|}{\textbf{Lasso}}
& \multicolumn{3}{c|}{\textbf{Ad. Lasso}}
& \multicolumn{3}{c|}{\textbf{MCP}}
& \multicolumn{3}{c|}{\textbf{SCAD}} \\
\cline{3-23}

& $n/T$
& 100 & 200 & 300
& 100 & 200 & 300
& 100 & 200 & 300
& 100 & 200 & 300
& 100 & 200 & 300
& 100 & 200 & 300
& 100 & 200 & 300 \\
\hline

\multirow{3}{*}{\shortstack{V=4\\$\pi$=0}}
& 100
& 0.040 & 0.024 & 0.016
& 0.470 & 0.506 & 0.504
& 0.472 & 0.499 & 0.498
& 0.183 & 0.129 & 0.101
& 0.106 & 0.018 & 0.011
& 0.118 & 0.067 & 0.070
& 0.126 & 0.074 & 0.077 \\
& 200
& 0.062 & 0.023 & 0.016
& 0.457 & 0.506 & 0.504
& 0.459 & 0.506 & 0.505
& 0.207 & 0.146 & 0.109
& 0.144 & 0.040 & 0.014
& 0.087 & 0.056 & 0.049
& 0.101 & 0.068 & 0.056 \\
& 300
& 0.082 & 0.021 & 0.017
& 0.451 & 0.505 & 0.505
& 0.457 & 0.506 & 0.506
& 0.213 & 0.138 & 0.116
& 0.170 & 0.061 & 0.027
& 0.105 & 0.055 & 0.048
& 0.107 & 0.066 & 0.056 \\
\hline

\multirow{3}{*}{\shortstack{V=4\\$\pi$=0.5}}
& 100
& 0.058 & 0.021 & 0.020
& 0.669 & 0.819 & 0.887
& 0.232 & 0.362 & 0.464
& 0.218 & 0.133 & 0.093
& 0.094 & 0.021 & 0.006
& 0.126 & 0.077 & 0.066
& 0.142 & 0.088 & 0.070 \\
& 200
& 0.075 & 0.024 & 0.016
& 0.702 & 0.862 & 0.923
& 0.241 & 0.387 & 0.467
& 0.236 & 0.129 & 0.096
& 0.152 & 0.038 & 0.019
& 0.120 & 0.057 & 0.069
& 0.137 & 0.072 & 0.073 \\
& 300
& 0.073 & 0.026 & 0.017
& 0.735 & 0.884 & 0.941
& 0.223 & 0.372 & 0.465
& 0.222 & 0.137 & 0.107
& 0.167 & 0.063 & 0.030
& 0.115 & 0.056 & 0.044
& 0.127 & 0.072 & 0.054 \\
\hline

\multirow{3}{*}{\shortstack{V=4\\$\pi$=1}}
& 100
& 0.081 & 0.016 & 0.017
& 0.924 & 0.957 & 0.960
& 0.085 & 0.029 & 0.024
& 0.253 & 0.223 & 0.150
& 0.176 & 0.036 & 0.012
& 0.162 & 0.090 & 0.075
& 0.188 & 0.108 & 0.081 \\
& 200
& 0.119 & 0.027 & 0.019
& 0.967 & 0.979 & 0.980
& 0.081 & 0.026 & 0.029
& 0.252 & 0.229 & 0.192
& 0.273 & 0.060 & 0.021
& 0.170 & 0.064 & 0.057
& 0.210 & 0.080 & 0.069 \\
& 300
& 0.173 & 0.025 & 0.023
& 0.982 & 0.987 & 0.987
& 0.132 & 0.031 & 0.020
& 0.252 & 0.231 & 0.180
& 0.364 & 0.098 & 0.052
& 0.188 & 0.045 & 0.048
& 0.227 & 0.071 & 0.047 \\
\hline

\multirow{3}{*}{\shortstack{V=2\\$\pi$=0}}
& 100
& 0.035 & 0.017 & 0.013
& 0.159 & 0.302 & 0.410
& 0.163 & 0.305 & 0.413
& 0.245 & 0.142 & 0.103
& 0.117 & 0.044 & 0.019
& 0.092 & 0.086 & 0.061
& 0.100 & 0.090 & 0.066 \\
& 200
& 0.048 & 0.020 & 0.021
& 0.146 & 0.271 & 0.389
& 0.153 & 0.273 & 0.392
& 0.254 & 0.143 & 0.106
& 0.186 & 0.068 & 0.027
& 0.066 & 0.064 & 0.056
& 0.076 & 0.061 & 0.059 \\
& 300
& 0.066 & 0.022 & 0.021
& 0.140 & 0.263 & 0.369
& 0.143 & 0.273 & 0.369
& 0.291 & 0.147 & 0.122
& 0.216 & 0.091 & 0.043
& 0.081 & 0.057 & 0.039
& 0.094 & 0.059 & 0.044 \\
\hline

\multirow{3}{*}{\shortstack{V=2\\$\pi$=0.5}}
& 100
& 0.037 & 0.023 & 0.018
& 0.253 & 0.421 & 0.541
& 0.081 & 0.080 & 0.091
& 0.317 & 0.178 & 0.124
& 0.134 & 0.029 & 0.020
& 0.096 & 0.081 & 0.075
& 0.109 & 0.091 & 0.077 \\
& 200
& 0.048 & 0.021 & 0.016
& 0.258 & 0.430 & 0.569
& 0.073 & 0.066 & 0.096
& 0.313 & 0.190 & 0.133
& 0.194 & 0.076 & 0.029
& 0.084 & 0.058 & 0.057
& 0.112 & 0.059 & 0.052 \\
& 300
& 0.060 & 0.022 & 0.020
& 0.282 & 0.449 & 0.590
& 0.088 & 0.065 & 0.089
& 0.344 & 0.201 & 0.127
& 0.231 & 0.090 & 0.050
& 0.082 & 0.045 & 0.040
& 0.095 & 0.044 & 0.042 \\
\hline

\multirow{3}{*}{\shortstack{V=2\\$\pi$=1}}
& 100
& 0.046 & 0.017 & 0.015
& 0.550 & 0.774 & 0.872
& 0.055 & 0.035 & 0.025
& 0.352 & 0.222 & 0.165
& 0.138 & 0.038 & 0.020
& 0.105 & 0.084 & 0.067
& 0.136 & 0.084 & 0.073 \\
& 200
& 0.057 & 0.020 & 0.017
& 0.612 & 0.837 & 0.921
& 0.077 & 0.030 & 0.020
& 0.361 & 0.264 & 0.194
& 0.253 & 0.068 & 0.029
& 0.103 & 0.045 & 0.062
& 0.160 & 0.055 & 0.064 \\
& 300
& 0.069 & 0.021 & 0.018
& 0.669 & 0.860 & 0.936
& 0.081 & 0.027 & 0.020
& 0.384 & 0.282 & 0.214
& 0.281 & 0.095 & 0.053
& 0.116 & 0.052 & 0.048
& 0.186 & 0.057 & 0.052 \\
\hline

\end{tabular}
}%
\vspace{0.5em}

\parbox{\linewidth}{
\small
\textit{\textbf{Notes}:} False Discovery Rate (FDR) is the proportion of irrelevant variables among the selected ones. Lower values indicate more accurate variable selection. Lasso denotes the PCA prefiltered Lasso with the one standard error (1SE) penalty selection rule, while Ad. Lasso denotes adaptive Lasso with the 1SE rule applied at both penalisation stages (1SE--1SE).
}

\end{table}
\vspace*{\fill}
\end{landscape}

\begin{landscape}
\vspace*{\fill}
\begin{table}[htbp]
\thispagestyle{empty}
\centering
\small
\setlength{\tabcolsep}{3.0pt}
\caption{\textbf{True Positive Rate (TPR)}; $\alpha=0.8$, $k=4$}
\makebox[\linewidth][c]{%
\begin{tabular}{|c|c|ccc|ccc|ccc|ccc|ccc|ccc|ccc|}
\hline
&
& \multicolumn{3}{c|}{\textbf{BMT}}
& \multicolumn{3}{c|}{\textbf{OCMT}}
& \multicolumn{3}{c|}{\textbf{GOCMT}}
& \multicolumn{3}{c|}{\textbf{Lasso}}
& \multicolumn{3}{c|}{\textbf{Ad. Lasso}}
& \multicolumn{3}{c|}{\textbf{MCP}}
& \multicolumn{3}{c|}{\textbf{SCAD}} \\
\cline{3-23}

& $n/T$
& 100 & 200 & 300
& 100 & 200 & 300
& 100 & 200 & 300
& 100 & 200 & 300
& 100 & 200 & 300
& 100 & 200 & 300
& 100 & 200 & 300 \\
\hline

\multirow{3}{*}{\shortstack{V=4\\$\pi$=0}}
& 100
& 0.868 & 0.997 & 1.000
& 1.000 & 1.000 & 1.000
& 0.996 & 0.998 & 1.000
& 0.989 & 0.999 & 1.000
& 0.938 & 0.995 & 1.000
& 0.974 & 0.996 & 1.000
& 0.985 & 0.996 & 1.000 \\
& 200
& 0.826 & 0.996 & 1.000
& 0.999 & 1.000 & 1.000
& 0.999 & 1.000 & 1.000
& 0.998 & 1.000 & 1.000
& 0.936 & 0.993 & 1.000
& 0.960 & 0.999 & 1.000
& 0.979 & 0.999 & 1.000 \\
& 300
& 0.803 & 0.996 & 1.000
& 1.000 & 1.000 & 1.000
& 0.999 & 1.000 & 1.000
& 0.994 & 1.000 & 1.000
& 0.937 & 0.994 & 0.999
& 0.955 & 1.000 & 0.999
& 0.977 & 0.999 & 0.999 \\
\hline

\multirow{3}{*}{\shortstack{V=4\\$\pi$=0.5}}
& 100
& 0.862 & 0.999 & 1.000
& 0.999 & 1.000 & 1.000
& 0.969 & 1.000 & 1.000
& 0.976 & 1.000 & 1.000
& 0.942 & 0.996 & 1.000
& 0.966 & 0.998 & 1.000
& 0.978 & 0.998 & 1.000 \\
& 200
& 0.833 & 0.996 & 1.000
& 1.000 & 1.000 & 1.000
& 0.965 & 1.000 & 1.000
& 0.978 & 1.000 & 1.000
& 0.939 & 0.999 & 1.000
& 0.962 & 0.997 & 1.000
& 0.980 & 0.998 & 1.000 \\
& 300
& 0.821 & 0.993 & 1.000
& 0.999 & 1.000 & 1.000
& 0.968 & 1.000 & 1.000
& 0.972 & 1.000 & 1.000
& 0.936 & 0.993 & 1.000
& 0.963 & 0.999 & 1.000
& 0.984 & 0.999 & 1.000 \\
\hline

\multirow{3}{*}{\shortstack{V=4\\$\pi$=1}}
& 100
& 0.846 & 0.996 & 1.000
& 0.999 & 1.000 & 1.000
& 0.561 & 0.964 & 0.999
& 0.708 & 0.993 & 1.000
& 0.908 & 0.996 & 1.000
& 0.965 & 0.998 & 1.000
& 0.977 & 0.999 & 1.000 \\
& 200
& 0.790 & 0.994 & 1.000
& 1.000 & 1.000 & 1.000
& 0.498 & 0.941 & 0.995
& 0.630 & 0.980 & 0.998
& 0.897 & 0.995 & 1.000
& 0.945 & 0.999 & 0.999
& 0.964 & 0.999 & 0.999 \\
& 300
& 0.766 & 0.992 & 1.000
& 1.000 & 1.000 & 1.000
& 0.461 & 0.921 & 0.995
& 0.581 & 0.970 & 1.000
& 0.871 & 0.996 & 1.000
& 0.935 & 0.998 & 1.000
& 0.958 & 0.998 & 1.000 \\
\hline

\multirow{3}{*}{\shortstack{V=2\\$\pi$=0}}
& 100
& 0.988 & 1.000 & 1.000
& 0.994 & 1.000 & 1.000
& 0.995 & 1.000 & 1.000
& 1.000 & 1.000 & 1.000
& 0.992 & 1.000 & 1.000
& 0.997 & 1.000 & 1.000
& 0.996 & 1.000 & 1.000 \\
& 200
& 0.986 & 1.000 & 1.000
& 0.992 & 1.000 & 1.000
& 0.991 & 1.000 & 1.000
& 0.998 & 1.000 & 1.000
& 0.989 & 1.000 & 1.000
& 0.997 & 1.000 & 1.000
& 0.995 & 1.000 & 1.000 \\
& 300
& 0.965 & 1.000 & 1.000
& 0.986 & 1.000 & 1.000
& 0.987 & 1.000 & 1.000
& 0.997 & 1.000 & 1.000
& 0.985 & 1.000 & 1.000
& 0.994 & 0.999 & 1.000
& 0.994 & 0.999 & 1.000 \\
\hline

\multirow{3}{*}{\shortstack{V=2\\$\pi$=0.5}}
& 100
& 0.984 & 1.000 & 1.000
& 0.993 & 1.000 & 1.000
& 0.969 & 1.000 & 1.000
& 0.997 & 1.000 & 1.000
& 0.986 & 1.000 & 1.000
& 0.999 & 1.000 & 1.000
& 0.999 & 0.999 & 1.000 \\
& 200
& 0.977 & 1.000 & 1.000
& 0.990 & 1.000 & 1.000
& 0.952 & 1.000 & 1.000
& 0.993 & 1.000 & 1.000
& 0.985 & 1.000 & 1.000
& 0.996 & 1.000 & 1.000
& 0.996 & 1.000 & 1.000 \\
& 300
& 0.972 & 1.000 & 1.000
& 0.981 & 1.000 & 1.000
& 0.948 & 1.000 & 1.000
& 0.991 & 1.000 & 1.000
& 0.985 & 1.000 & 1.000
& 0.992 & 1.000 & 1.000
& 0.992 & 1.000 & 1.000 \\
\hline

\multirow{3}{*}{\shortstack{V=2\\$\pi$=1}}
& 100
& 0.984 & 1.000 & 1.000
& 0.991 & 1.000 & 1.000
& 0.889 & 1.000 & 1.000
& 0.975 & 1.000 & 1.000
& 0.987 & 1.000 & 1.000
& 0.995 & 1.000 & 1.000
& 0.993 & 1.000 & 1.000 \\
& 200
& 0.973 & 1.000 & 1.000
& 0.980 & 1.000 & 1.000
& 0.815 & 1.000 & 1.000
& 0.940 & 1.000 & 1.000
& 0.984 & 1.000 & 1.000
& 0.996 & 1.000 & 1.000
& 0.994 & 1.000 & 1.000 \\
& 300
& 0.964 & 1.000 & 1.000
& 0.977 & 1.000 & 1.000
& 0.773 & 0.999 & 1.000
& 0.913 & 1.000 & 1.000
& 0.975 & 1.000 & 1.000
& 0.992 & 1.000 & 1.000
& 0.991 & 1.000 & 1.000 \\
\hline

\end{tabular}
}%
\vspace{0.5em}

\parbox{\linewidth}{
\small
\textit{\textbf{Notes}:} True Positive Rate (TPR), also known as sensitivity or recall, is the proportion of relevant variables correctly identified among all truly relevant variables. Higher values indicate better detection of true signals. Lasso denotes the PCA prefiltered Lasso with the one standard error (1SE) penalty selection rule, while Ad. Lasso denotes adaptive Lasso
with the 1SE rule applied at both penalisation stages (1SE--1SE).
}

\end{table}
\vspace*{\fill}
\end{landscape}

\begin{landscape}
\vspace*{\fill}
\begin{table}[htbp]
\thispagestyle{empty}
\centering
\small
\setlength{\tabcolsep}{3.0pt}
\caption{\textbf{True Positive Rate (TPR)}; $\alpha=0.4$, $k=4$}
\makebox[\linewidth][c]{%
\begin{tabular}{|c|c|ccc|ccc|ccc|ccc|ccc|ccc|ccc|}
\hline
&
& \multicolumn{3}{c|}{\textbf{BMT}}
& \multicolumn{3}{c|}{\textbf{OCMT}}
& \multicolumn{3}{c|}{\textbf{GOCMT}}
& \multicolumn{3}{c|}{\textbf{Lasso}}
& \multicolumn{3}{c|}{\textbf{Ad. Lasso}}
& \multicolumn{3}{c|}{\textbf{MCP}}
& \multicolumn{3}{c|}{\textbf{SCAD}} \\
\cline{3-23}

& $n/T$
& 100 & 200 & 300
& 100 & 200 & 300
& 100 & 200 & 300
& 100 & 200 & 300
& 100 & 200 & 300
& 100 & 200 & 300
& 100 & 200 & 300 \\
\hline

\multirow{3}{*}{\shortstack{V=4\\$\pi$=0}}
& 100
& 0.880 & 0.999 & 1.000
& 1.000 & 1.000 & 1.000
& 0.995 & 0.998 & 1.000
& 0.992 & 1.000 & 1.000
& 0.945 & 0.993 & 0.999
& 0.978 & 0.999 & 1.000
& 0.981 & 0.998 & 1.000 \\
& 200
& 0.827 & 0.999 & 1.000
& 0.999 & 1.000 & 1.000
& 0.999 & 1.000 & 1.000
& 0.996 & 1.000 & 1.000
& 0.938 & 0.998 & 1.000
& 0.973 & 1.000 & 1.000
& 0.985 & 1.000 & 1.000 \\
& 300
& 0.811 & 0.996 & 1.000
& 0.997 & 1.000 & 1.000
& 0.998 & 1.000 & 1.000
& 0.998 & 1.000 & 1.000
& 0.948 & 0.994 & 0.999
& 0.968 & 0.999 & 0.999
& 0.978 & 1.000 & 1.000 \\
\hline

\multirow{3}{*}{\shortstack{V=4\\$\pi$=0.5}}
& 100
& 0.847 & 0.998 & 1.000
& 1.000 & 1.000 & 1.000
& 0.953 & 1.000 & 1.000
& 0.970 & 1.000 & 1.000
& 0.941 & 0.997 & 1.000
& 0.966 & 0.999 & 1.000
& 0.981 & 1.000 & 1.000 \\
& 200
& 0.833 & 0.997 & 1.000
& 0.999 & 1.000 & 1.000
& 0.952 & 1.000 & 1.000
& 0.963 & 1.000 & 1.000
& 0.935 & 0.997 & 1.000
& 0.973 & 0.997 & 1.000
& 0.984 & 0.997 & 1.000 \\
& 300
& 0.810 & 0.992 & 1.000
& 0.998 & 1.000 & 1.000
& 0.952 & 1.000 & 1.000
& 0.966 & 1.000 & 1.000
& 0.933 & 0.997 & 1.000
& 0.970 & 1.000 & 0.999
& 0.986 & 0.999 & 0.999 \\
\hline

\multirow{3}{*}{\shortstack{V=4\\$\pi$=1}}
& 100
& 0.828 & 0.996 & 1.000
& 0.999 & 1.000 & 1.000
& 0.548 & 0.965 & 0.999
& 0.720 & 0.986 & 1.000
& 0.905 & 0.997 & 1.000
& 0.961 & 0.996 & 1.000
& 0.970 & 0.997 & 1.000 \\
& 200
& 0.795 & 0.993 & 1.000
& 0.999 & 1.000 & 1.000
& 0.475 & 0.924 & 0.998
& 0.605 & 0.984 & 1.000
& 0.880 & 0.994 & 1.000
& 0.940 & 0.997 & 1.000
& 0.958 & 0.999 & 1.000 \\
& 300
& 0.735 & 0.997 & 1.000
& 0.998 & 1.000 & 1.000
& 0.423 & 0.918 & 0.996
& 0.545 & 0.970 & 1.000
& 0.852 & 0.994 & 1.000
& 0.919 & 0.998 & 1.000
& 0.944 & 0.997 & 1.000 \\
\hline

\multirow{3}{*}{\shortstack{V=2\\$\pi$=0}}
& 100
& 0.987 & 1.000 & 1.000
& 0.990 & 1.000 & 1.000
& 0.991 & 1.000 & 1.000
& 1.000 & 1.000 & 1.000
& 0.990 & 1.000 & 1.000
& 0.996 & 1.000 & 1.000
& 0.998 & 1.000 & 1.000 \\
& 200
& 0.984 & 1.000 & 1.000
& 0.987 & 1.000 & 1.000
& 0.985 & 1.000 & 1.000
& 0.999 & 1.000 & 1.000
& 0.988 & 1.000 & 1.000
& 0.999 & 1.000 & 1.000
& 0.997 & 1.000 & 1.000 \\
& 300
& 0.963 & 1.000 & 1.000
& 0.977 & 1.000 & 1.000
& 0.978 & 1.000 & 1.000
& 0.997 & 1.000 & 1.000
& 0.984 & 1.000 & 1.000
& 0.994 & 1.000 & 1.000
& 0.996 & 1.000 & 1.000 \\
\hline

\multirow{3}{*}{\shortstack{V=2\\$\pi$=0.5}}
& 100
& 0.988 & 1.000 & 1.000
& 0.991 & 1.000 & 1.000
& 0.975 & 1.000 & 1.000
& 0.995 & 1.000 & 1.000
& 0.988 & 1.000 & 1.000
& 0.998 & 1.000 & 1.000
& 0.998 & 1.000 & 1.000 \\
& 200
& 0.975 & 1.000 & 1.000
& 0.986 & 1.000 & 1.000
& 0.947 & 1.000 & 1.000
& 0.985 & 1.000 & 1.000
& 0.988 & 1.000 & 1.000
& 0.994 & 1.000 & 1.000
& 0.992 & 1.000 & 1.000 \\
& 300
& 0.970 & 1.000 & 1.000
& 0.974 & 1.000 & 1.000
& 0.930 & 1.000 & 1.000
& 0.983 & 1.000 & 1.000
& 0.986 & 1.000 & 1.000
& 0.996 & 1.000 & 1.000
& 0.995 & 1.000 & 1.000 \\
\hline

\multirow{3}{*}{\shortstack{V=2\\$\pi$=1}}
& 100
& 0.985 & 1.000 & 1.000
& 0.988 & 1.000 & 1.000
& 0.866 & 1.000 & 1.000
& 0.968 & 1.000 & 1.000
& 0.988 & 1.000 & 1.000
& 0.997 & 0.999 & 1.000
& 0.996 & 0.999 & 1.000 \\
& 200
& 0.974 & 1.000 & 1.000
& 0.977 & 1.000 & 1.000
& 0.791 & 1.000 & 1.000
& 0.930 & 1.000 & 1.000
& 0.988 & 1.000 & 1.000
& 0.995 & 1.000 & 1.000
& 0.995 & 1.000 & 1.000 \\
& 300
& 0.968 & 1.000 & 1.000
& 0.972 & 1.000 & 1.000
& 0.759 & 1.000 & 1.000
& 0.884 & 1.000 & 1.000
& 0.977 & 1.000 & 1.000
& 0.995 & 0.999 & 1.000
& 0.991 & 0.999 & 1.000 \\
\hline

\end{tabular}
}%
\vspace{0.5em}

\parbox{\linewidth}{
\small
\textit{\textbf{Notes}:} True Positive Rate (TPR), also known as sensitivity or recall, is the proportion of relevant variables correctly identified among all truly relevant variables. Higher values indicate better detection of true signals. Lasso denotes the PCA prefiltered Lasso with the one standard error (1SE) penalty selection rule, while Ad. Lasso denotes adaptive Lasso
with the 1SE rule applied at both penalisation stages (1SE--1SE).
}

\end{table}
\vspace*{\fill}
\end{landscape}

\begin{landscape}
\vspace*{\fill}
\begin{table}[htbp]
\thispagestyle{empty}
\centering
\small
\setlength{\tabcolsep}{3.0pt}
\caption{\textbf{False Positive Rate (FPR)}; $\alpha=0.8$, $k=4$}
\makebox[\linewidth][c]{%
\begin{tabular}{|c|c|ccc|ccc|ccc|ccc|ccc|ccc|ccc|}
\hline
&
& \multicolumn{3}{c|}{\textbf{BMT}}
& \multicolumn{3}{c|}{\textbf{OCMT}}
& \multicolumn{3}{c|}{\textbf{GOCMT}}
& \multicolumn{3}{c|}{\textbf{Lasso}}
& \multicolumn{3}{c|}{\textbf{Ad. Lasso}}
& \multicolumn{3}{c|}{\textbf{MCP}}
& \multicolumn{3}{c|}{\textbf{SCAD}} \\
\cline{3-23}

& $n/T$
& 100 & 200 & 300
& 100 & 200 & 300
& 100 & 200 & 300
& 100 & 200 & 300
& 100 & 200 & 300
& 100 & 200 & 300
& 100 & 200 & 300 \\
\hline

\multirow{3}{*}{\shortstack{V=4\\$\pi$=0}}
& 100
& 0.003 & 0.001 & 0.001
& 0.040 & 0.043 & 0.042
& 0.040 & 0.042 & 0.042
& 0.014 & 0.008 & 0.007
& 0.006 & 0.001 & 0.001
& 0.010 & 0.006 & 0.007
& 0.011 & 0.007 & 0.008 \\
& 200
& 0.002 & 0.001 & 0.001
& 0.019 & 0.021 & 0.021
& 0.019 & 0.021 & 0.021
& 0.008 & 0.005 & 0.003
& 0.006 & 0.001 & 0.001
& 0.005 & 0.002 & 0.002
& 0.006 & 0.002 & 0.002 \\
& 300
& 0.001 & 0.000 & 0.000
& 0.013 & 0.014 & 0.014
& 0.013 & 0.014 & 0.014
& 0.007 & 0.003 & 0.002
& 0.005 & 0.001 & 0.001
& 0.004 & 0.001 & 0.001
& 0.004 & 0.001 & 0.002 \\
\hline

\multirow{3}{*}{\shortstack{V=4\\$\pi$=0.5}}
& 100
& 0.003 & 0.001 & 0.001
& 0.126 & 0.302 & 0.478
& 0.017 & 0.030 & 0.039
& 0.017 & 0.010 & 0.007
& 0.006 & 0.001 & 0.000
& 0.011 & 0.007 & 0.005
& 0.011 & 0.008 & 0.006 \\
& 200
& 0.002 & 0.001 & 0.000
& 0.092 & 0.229 & 0.414
& 0.008 & 0.015 & 0.019
& 0.011 & 0.005 & 0.004
& 0.005 & 0.001 & 0.001
& 0.005 & 0.003 & 0.002
& 0.005 & 0.003 & 0.002 \\
& 300
& 0.001 & 0.000 & 0.000
& 0.112 & 0.215 & 0.371
& 0.006 & 0.010 & 0.013
& 0.007 & 0.004 & 0.002
& 0.004 & 0.001 & 0.001
& 0.003 & 0.002 & 0.001
& 0.004 & 0.002 & 0.001 \\
\hline

\multirow{3}{*}{\shortstack{V=4\\$\pi$=1}}
& 100
& 0.003 & 0.001 & 0.001
& 0.716 & 0.967 & 0.998
& 0.002 & 0.002 & 0.001
& 0.023 & 0.018 & 0.011
& 0.012 & 0.002 & 0.001
& 0.014 & 0.006 & 0.006
& 0.016 & 0.007 & 0.007 \\
& 200
& 0.002 & 0.001 & 0.000
& 0.918 & 0.987 & 0.998
& 0.001 & 0.001 & 0.001
& 0.014 & 0.010 & 0.006
& 0.008 & 0.002 & 0.001
& 0.007 & 0.002 & 0.003
& 0.008 & 0.003 & 0.003 \\
& 300
& 0.002 & 0.000 & 0.000
& 0.977 & 1.000 & 1.000
& 0.001 & 0.001 & 0.000
& 0.009 & 0.008 & 0.006
& 0.009 & 0.002 & 0.001
& 0.006 & 0.002 & 0.002
& 0.007 & 0.002 & 0.001 \\
\hline

\multirow{3}{*}{\shortstack{V=2\\$\pi$=0}}
& 100
& 0.002 & 0.001 & 0.001
& 0.013 & 0.025 & 0.034
& 0.015 & 0.026 & 0.034
& 0.017 & 0.009 & 0.007
& 0.010 & 0.002 & 0.001
& 0.009 & 0.006 & 0.006
& 0.010 & 0.006 & 0.007 \\
& 200
& 0.001 & 0.001 & 0.001
& 0.006 & 0.011 & 0.016
& 0.006 & 0.012 & 0.016
& 0.010 & 0.005 & 0.004
& 0.007 & 0.002 & 0.001
& 0.004 & 0.002 & 0.002
& 0.004 & 0.003 & 0.003 \\
& 300
& 0.001 & 0.000 & 0.000
& 0.003 & 0.007 & 0.010
& 0.004 & 0.007 & 0.011
& 0.008 & 0.005 & 0.002
& 0.006 & 0.002 & 0.001
& 0.004 & 0.001 & 0.001
& 0.004 & 0.002 & 0.001 \\
\hline

\multirow{3}{*}{\shortstack{V=2\\$\pi$=0.5}}
& 100
& 0.002 & 0.001 & 0.001
& 0.023 & 0.046 & 0.073
& 0.006 & 0.006 & 0.008
& 0.024 & 0.013 & 0.009
& 0.009 & 0.002 & 0.001
& 0.006 & 0.007 & 0.008
& 0.007 & 0.008 & 0.008 \\
& 200
& 0.001 & 0.001 & 0.000
& 0.013 & 0.025 & 0.044
& 0.002 & 0.002 & 0.004
& 0.015 & 0.008 & 0.004
& 0.008 & 0.002 & 0.001
& 0.004 & 0.003 & 0.002
& 0.005 & 0.003 & 0.002 \\
& 300
& 0.001 & 0.000 & 0.000
& 0.010 & 0.019 & 0.032
& 0.002 & 0.002 & 0.003
& 0.012 & 0.005 & 0.004
& 0.005 & 0.002 & 0.001
& 0.002 & 0.002 & 0.001
& 0.003 & 0.002 & 0.001 \\
\hline

\multirow{3}{*}{\shortstack{V=2\\$\pi$=1}}
& 100
& 0.002 & 0.001 & 0.001
& 0.094 & 0.246 & 0.442
& 0.003 & 0.002 & 0.002
& 0.030 & 0.018 & 0.013
& 0.010 & 0.002 & 0.001
& 0.008 & 0.008 & 0.005
& 0.011 & 0.008 & 0.006 \\
& 200
& 0.001 & 0.001 & 0.000
& 0.066 & 0.196 & 0.372
& 0.001 & 0.001 & 0.001
& 0.021 & 0.011 & 0.008
& 0.008 & 0.002 & 0.001
& 0.004 & 0.003 & 0.002
& 0.006 & 0.003 & 0.002 \\
& 300
& 0.001 & 0.000 & 0.000
& 0.073 & 0.163 & 0.339
& 0.001 & 0.000 & 0.000
& 0.014 & 0.007 & 0.005
& 0.009 & 0.003 & 0.001
& 0.003 & 0.001 & 0.001
& 0.005 & 0.001 & 0.001 \\
\hline

\end{tabular}
}%
\vspace{0.5em}

\parbox{\linewidth}{
\small
\textit{\textbf{Notes}:} False Positive Rate (FPR) is the proportion of irrelevant variables incorrectly selected among all truly irrelevant variables. Lower values indicate better specificity in variable selection. Lasso denotes the PCA prefiltered Lasso with the one standard error (1SE) penalty selection rule, while Ad. Lasso denotes adaptive Lasso with the 1SE rule applied at both penalisation stages (1SE--1SE).
}

\end{table}
\vspace*{\fill}
\end{landscape}

\begin{landscape}
\vspace*{\fill}
\begin{table}[htbp]
\thispagestyle{empty}
\centering
\small
\setlength{\tabcolsep}{3.0pt}
\caption{\textbf{False Positive Rate (FPR)}; $\alpha=0.4$, $k=4$}
\makebox[\linewidth][c]{%
\begin{tabular}{|c|c|ccc|ccc|ccc|ccc|ccc|ccc|ccc|}
\hline
&
& \multicolumn{3}{c|}{\textbf{BMT}}
& \multicolumn{3}{c|}{\textbf{OCMT}}
& \multicolumn{3}{c|}{\textbf{GOCMT}}
& \multicolumn{3}{c|}{\textbf{Lasso}}
& \multicolumn{3}{c|}{\textbf{Ad. Lasso}}
& \multicolumn{3}{c|}{\textbf{MCP}}
& \multicolumn{3}{c|}{\textbf{SCAD}} \\
\cline{3-23}

& $n/T$
& 100 & 200 & 300
& 100 & 200 & 300
& 100 & 200 & 300
& 100 & 200 & 300
& 100 & 200 & 300
& 100 & 200 & 300
& 100 & 200 & 300 \\
\hline

\multirow{3}{*}{\shortstack{V=4\\$\pi$=0}}
& 100
& 0.002 & 0.001 & 0.001
& 0.038 & 0.043 & 0.042
& 0.039 & 0.042 & 0.042
& 0.013 & 0.009 & 0.006
& 0.008 & 0.001 & 0.001
& 0.012 & 0.006 & 0.006
& 0.012 & 0.006 & 0.007 \\
& 200
& 0.001 & 0.001 & 0.000
& 0.018 & 0.021 & 0.021
& 0.018 & 0.021 & 0.021
& 0.009 & 0.005 & 0.004
& 0.006 & 0.002 & 0.001
& 0.004 & 0.002 & 0.002
& 0.004 & 0.003 & 0.002 \\
& 300
& 0.001 & 0.000 & 0.000
& 0.012 & 0.014 & 0.014
& 0.012 & 0.014 & 0.014
& 0.006 & 0.003 & 0.003
& 0.005 & 0.002 & 0.001
& 0.004 & 0.001 & 0.001
& 0.003 & 0.002 & 0.002 \\
\hline

\multirow{3}{*}{\shortstack{V=4\\$\pi$=0.5}}
& 100
& 0.003 & 0.001 & 0.001
& 0.105 & 0.228 & 0.379
& 0.014 & 0.027 & 0.037
& 0.019 & 0.010 & 0.006
& 0.008 & 0.002 & 0.001
& 0.011 & 0.007 & 0.005
& 0.012 & 0.009 & 0.006 \\
& 200
& 0.002 & 0.001 & 0.000
& 0.067 & 0.162 & 0.304
& 0.007 & 0.014 & 0.018
& 0.011 & 0.005 & 0.003
& 0.007 & 0.001 & 0.001
& 0.006 & 0.002 & 0.003
& 0.006 & 0.003 & 0.003 \\
& 300
& 0.001 & 0.000 & 0.000
& 0.051 & 0.142 & 0.269
& 0.005 & 0.009 & 0.012
& 0.007 & 0.004 & 0.003
& 0.005 & 0.002 & 0.001
& 0.004 & 0.001 & 0.001
& 0.004 & 0.002 & 0.001 \\
\hline

\multirow{3}{*}{\shortstack{V=4\\$\pi$=1}}
& 100
& 0.003 & 0.001 & 0.001
& 0.586 & 0.928 & 0.994
& 0.002 & 0.002 & 0.001
& 0.025 & 0.018 & 0.012
& 0.012 & 0.002 & 0.001
& 0.014 & 0.009 & 0.006
& 0.016 & 0.010 & 0.007 \\
& 200
& 0.002 & 0.001 & 0.000
& 0.839 & 0.950 & 0.992
& 0.001 & 0.001 & 0.001
& 0.014 & 0.010 & 0.008
& 0.012 & 0.002 & 0.001
& 0.008 & 0.003 & 0.002
& 0.009 & 0.003 & 0.002 \\
& 300
& 0.002 & 0.000 & 0.000
& 0.947 & 0.996 & 0.999
& 0.001 & 0.000 & 0.000
& 0.010 & 0.007 & 0.005
& 0.010 & 0.002 & 0.001
& 0.006 & 0.001 & 0.001
& 0.007 & 0.002 & 0.001 \\
\hline

\multirow{3}{*}{\shortstack{V=2\\$\pi$=0}}
& 100
& 0.002 & 0.001 & 0.001
& 0.010 & 0.021 & 0.031
& 0.010 & 0.021 & 0.031
& 0.021 & 0.010 & 0.007
& 0.010 & 0.004 & 0.001
& 0.009 & 0.008 & 0.005
& 0.010 & 0.008 & 0.006 \\
& 200
& 0.001 & 0.001 & 0.001
& 0.004 & 0.009 & 0.014
& 0.005 & 0.009 & 0.014
& 0.012 & 0.005 & 0.004
& 0.009 & 0.003 & 0.001
& 0.003 & 0.003 & 0.002
& 0.003 & 0.003 & 0.002 \\
& 300
& 0.001 & 0.000 & 0.000
& 0.003 & 0.006 & 0.009
& 0.003 & 0.006 & 0.009
& 0.009 & 0.004 & 0.003
& 0.007 & 0.003 & 0.001
& 0.003 & 0.002 & 0.001
& 0.003 & 0.002 & 0.001 \\
\hline

\multirow{3}{*}{\shortstack{V=2\\$\pi$=0.5}}
& 100
& 0.002 & 0.001 & 0.001
& 0.018 & 0.037 & 0.057
& 0.004 & 0.005 & 0.005
& 0.032 & 0.014 & 0.009
& 0.011 & 0.002 & 0.002
& 0.009 & 0.008 & 0.007
& 0.011 & 0.009 & 0.007 \\
& 200
& 0.001 & 0.001 & 0.000
& 0.009 & 0.019 & 0.032
& 0.002 & 0.002 & 0.003
& 0.015 & 0.007 & 0.005
& 0.009 & 0.003 & 0.001
& 0.004 & 0.002 & 0.002
& 0.006 & 0.002 & 0.002 \\
& 300
& 0.001 & 0.000 & 0.000
& 0.007 & 0.014 & 0.024
& 0.001 & 0.001 & 0.002
& 0.013 & 0.006 & 0.003
& 0.008 & 0.003 & 0.001
& 0.003 & 0.001 & 0.001
& 0.003 & 0.001 & 0.001 \\
\hline

\multirow{3}{*}{\shortstack{V=2\\$\pi$=1}}
& 100
& 0.002 & 0.001 & 0.001
& 0.068 & 0.180 & 0.328
& 0.002 & 0.002 & 0.001
& 0.036 & 0.018 & 0.012
& 0.013 & 0.003 & 0.002
& 0.009 & 0.007 & 0.005
& 0.012 & 0.006 & 0.006 \\
& 200
& 0.001 & 0.001 & 0.000
& 0.045 & 0.138 & 0.276
& 0.002 & 0.001 & 0.001
& 0.020 & 0.011 & 0.008
& 0.012 & 0.003 & 0.001
& 0.004 & 0.002 & 0.002
& 0.007 & 0.002 & 0.003 \\
& 300
& 0.001 & 0.000 & 0.000
& 0.041 & 0.113 & 0.237
& 0.001 & 0.000 & 0.000
& 0.017 & 0.008 & 0.006
& 0.009 & 0.003 & 0.001
& 0.003 & 0.001 & 0.001
& 0.006 & 0.002 & 0.001 \\
\hline

\end{tabular}
}%
\vspace{0.5em}

\parbox{\linewidth}{
\small
\textit{\textbf{Notes}:} False Positive Rate (FPR) is the proportion of irrelevant variables incorrectly selected among all truly irrelevant variables. Lower values indicate better specificity in variable selection. Lasso denotes the PCA prefiltered Lasso with the one standard error (1SE) penalty selection rule, while Ad. Lasso denotes adaptive Lasso with the 1SE rule applied at both penalisation stages (1SE--1SE).
}

\end{table}
\vspace*{\fill}
\end{landscape}

\begin{landscape}
\vspace*{\fill}
\begin{table}[htbp]
\thispagestyle{empty}
\centering
\small
\setlength{\tabcolsep}{3.0pt}
\caption{\textbf{Exact Recovery}; $\alpha=0.8$, $k=4$}
\makebox[\linewidth][c]{%
\begin{tabular}{|c|c|ccc|ccc|ccc|ccc|ccc|ccc|ccc|}
\hline
&
& \multicolumn{3}{c|}{\textbf{BMT}}
& \multicolumn{3}{c|}{\textbf{OCMT}}
& \multicolumn{3}{c|}{\textbf{GOCMT}}
& \multicolumn{3}{c|}{\textbf{Lasso}}
& \multicolumn{3}{c|}{\textbf{Ad. Lasso}}
& \multicolumn{3}{c|}{\textbf{MCP}}
& \multicolumn{3}{c|}{\textbf{SCAD}} \\
\cline{3-23}

& $n/T$
& 100 & 200 & 300
& 100 & 200 & 300
& 100 & 200 & 300
& 100 & 200 & 300
& 100 & 200 & 300
& 100 & 200 & 300
& 100 & 200 & 300 \\
\hline

\multirow{3}{*}{\shortstack{V=4\\$\pi$=0}}
& 100
& 0.446 & 0.899 & 0.934
& 0.001 & 0.000 & 0.001
& 0.001 & 0.010 & 0.000
& 0.364 & 0.534 & 0.560
& 0.560 & 0.914 & 0.974
& 0.595 & 0.770 & 0.747
& 0.573 & 0.741 & 0.725 \\
& 200
& 0.315 & 0.875 & 0.898
& 0.005 & 0.000 & 0.000
& 0.002 & 0.000 & 0.001
& 0.427 & 0.512 & 0.644
& 0.469 & 0.879 & 0.949
& 0.588 & 0.835 & 0.826
& 0.581 & 0.795 & 0.818 \\
& 300
& 0.229 & 0.868 & 0.889
& 0.005 & 0.000 & 0.000
& 0.003 & 0.000 & 0.000
& 0.415 & 0.544 & 0.598
& 0.481 & 0.879 & 0.936
& 0.579 & 0.869 & 0.810
& 0.596 & 0.813 & 0.792 \\
\hline

\multirow{3}{*}{\shortstack{V=4\\$\pi$=0.5}}
& 100
& 0.435 & 0.908 & 0.933
& 0.001 & 0.000 & 0.000
& 0.183 & 0.012 & 0.006
& 0.385 & 0.544 & 0.634
& 0.542 & 0.904 & 0.973
& 0.557 & 0.764 & 0.791
& 0.527 & 0.694 & 0.762 \\
& 200
& 0.320 & 0.851 & 0.918
& 0.001 & 0.001 & 0.000
& 0.164 & 0.030 & 0.001
& 0.367 & 0.578 & 0.639
& 0.480 & 0.891 & 0.953
& 0.572 & 0.796 & 0.844
& 0.575 & 0.754 & 0.809 \\
& 300
& 0.321 & 0.848 & 0.904
& 0.003 & 0.000 & 0.001
& 0.112 & 0.028 & 0.002
& 0.358 & 0.553 & 0.676
& 0.385 & 0.841 & 0.943
& 0.603 & 0.820 & 0.835
& 0.548 & 0.769 & 0.826 \\
\hline

\multirow{3}{*}{\shortstack{V=4\\$\pi$=1}}
& 100
& 0.407 & 0.893 & 0.917
& 0.000 & 0.001 & 0.000
& 0.128 & 0.746 & 0.883
& 0.080 & 0.403 & 0.561
& 0.374 & 0.869 & 0.951
& 0.502 & 0.783 & 0.802
& 0.441 & 0.741 & 0.771 \\
& 200
& 0.289 & 0.864 & 0.915
& 0.000 & 0.001 & 0.000
& 0.067 & 0.707 & 0.881
& 0.064 & 0.350 & 0.515
& 0.280 & 0.798 & 0.925
& 0.449 & 0.819 & 0.810
& 0.433 & 0.737 & 0.765 \\
& 300
& 0.257 & 0.863 & 0.904
& 0.000 & 0.000 & 0.000
& 0.050 & 0.667 & 0.877
& 0.032 & 0.305 & 0.495
& 0.186 & 0.760 & 0.884
& 0.394 & 0.808 & 0.818
& 0.365 & 0.772 & 0.798 \\
\hline

\multirow{3}{*}{\shortstack{V=2\\$\pi$=0}}
& 100
& 0.801 & 0.902 & 0.925
& 0.245 & 0.056 & 0.001
& 0.221 & 0.059 & 0.003
& 0.391 & 0.599 & 0.660
& 0.646 & 0.917 & 0.947
& 0.755 & 0.797 & 0.768
& 0.698 & 0.793 & 0.766 \\
& 200
& 0.752 & 0.864 & 0.890
& 0.300 & 0.093 & 0.005
& 0.288 & 0.075 & 0.009
& 0.380 & 0.569 & 0.666
& 0.598 & 0.849 & 0.931
& 0.759 & 0.839 & 0.812
& 0.709 & 0.837 & 0.803 \\
& 300
& 0.680 & 0.899 & 0.890
& 0.348 & 0.094 & 0.009
& 0.333 & 0.095 & 0.009
& 0.373 & 0.521 & 0.721
& 0.553 & 0.814 & 0.899
& 0.748 & 0.837 & 0.850
& 0.708 & 0.816 & 0.848 \\
\hline

\multirow{3}{*}{\shortstack{V=2\\$\pi$=0.5}}
& 100
& 0.770 & 0.903 & 0.925
& 0.182 & 0.023 & 0.000
& 0.537 & 0.570 & 0.489
& 0.333 & 0.500 & 0.623
& 0.647 & 0.891 & 0.962
& 0.771 & 0.764 & 0.773
& 0.742 & 0.761 & 0.768 \\
& 200
& 0.734 & 0.895 & 0.913
& 0.158 & 0.033 & 0.003
& 0.526 & 0.628 & 0.488
& 0.278 & 0.472 & 0.618
& 0.526 & 0.855 & 0.914
& 0.707 & 0.815 & 0.821
& 0.665 & 0.805 & 0.815 \\
& 300
& 0.665 & 0.878 & 0.910
& 0.141 & 0.018 & 0.002
& 0.525 & 0.588 & 0.461
& 0.269 & 0.501 & 0.598
& 0.509 & 0.801 & 0.919
& 0.770 & 0.811 & 0.846
& 0.738 & 0.797 & 0.841 \\
\hline

\multirow{3}{*}{\shortstack{V=2\\$\pi$=1}}
& 100
& 0.785 & 0.895 & 0.924
& 0.025 & 0.002 & 0.001
& 0.570 & 0.840 & 0.841
& 0.215 & 0.382 & 0.534
& 0.603 & 0.907 & 0.966
& 0.717 & 0.791 & 0.805
& 0.635 & 0.777 & 0.786 \\
& 200
& 0.740 & 0.896 & 0.919
& 0.012 & 0.001 & 0.000
& 0.489 & 0.881 & 0.875
& 0.153 & 0.358 & 0.451
& 0.509 & 0.811 & 0.895
& 0.685 & 0.821 & 0.818
& 0.582 & 0.796 & 0.814 \\
& 300
& 0.684 & 0.866 & 0.902
& 0.007 & 0.001 & 0.000
& 0.381 & 0.878 & 0.883
& 0.110 & 0.403 & 0.429
& 0.320 & 0.771 & 0.873
& 0.640 & 0.852 & 0.852
& 0.528 & 0.830 & 0.853 \\
\hline

\end{tabular}
}%
\vspace{0.5em}

\parbox{\linewidth}{
\small
\textit{\textbf{Notes}:} Exact recovery is the proportion of replications
in which the exact model is recovered. Higher values indicate better
model selection performance.
Lasso denotes the PCA prefiltered Lasso with the one standard error
(1SE) penalty selection rule, while Ad. Lasso denotes adaptive Lasso
with the 1SE rule applied at both penalisation stages (1SE--1SE).
}

\end{table}
\vspace*{\fill}
\end{landscape}

\begin{landscape}
\vspace*{\fill}
\begin{table}[htbp]
\thispagestyle{empty}
\centering
\small
\setlength{\tabcolsep}{3.0pt}
\caption{\textbf{Exact Recovery}; $\alpha=0.4$, $k=4$}
\makebox[\linewidth][c]{%
\begin{tabular}{|c|c|ccc|ccc|ccc|ccc|ccc|ccc|ccc|}
\hline
&
& \multicolumn{3}{c|}{\textbf{BMT}}
& \multicolumn{3}{c|}{\textbf{OCMT}}
& \multicolumn{3}{c|}{\textbf{GOCMT}}
& \multicolumn{3}{c|}{\textbf{Lasso}}
& \multicolumn{3}{c|}{\textbf{Ad. Lasso}}
& \multicolumn{3}{c|}{\textbf{MCP}}
& \multicolumn{3}{c|}{\textbf{SCAD}} \\
\cline{3-23}

& $n/T$
& 100 & 200 & 300
& 100 & 200 & 300
& 100 & 200 & 300
& 100 & 200 & 300
& 100 & 200 & 300
& 100 & 200 & 300
& 100 & 200 & 300 \\
\hline

\multirow{3}{*}{\shortstack{V=4\\$\pi$=0}}
& 100
& 0.489 & 0.883 & 0.923
& 0.002 & 0.000 & 0.001
& 0.002 & 0.007 & 0.006
& 0.398 & 0.546 & 0.610
& 0.532 & 0.908 & 0.956
& 0.618 & 0.795 & 0.792
& 0.582 & 0.761 & 0.775 \\
& 200
& 0.303 & 0.890 & 0.920
& 0.006 & 0.000 & 0.000
& 0.004 & 0.001 & 0.001
& 0.379 & 0.518 & 0.599
& 0.461 & 0.873 & 0.956
& 0.647 & 0.825 & 0.824
& 0.620 & 0.787 & 0.799 \\
& 300
& 0.271 & 0.882 & 0.919
& 0.007 & 0.000 & 0.000
& 0.007 & 0.000 & 0.000
& 0.368 & 0.530 & 0.584
& 0.432 & 0.820 & 0.930
& 0.611 & 0.822 & 0.843
& 0.604 & 0.789 & 0.823 \\
\hline

\multirow{3}{*}{\shortstack{V=4\\$\pi$=0.5}}
& 100
& 0.389 & 0.893 & 0.904
& 0.005 & 0.000 & 0.001
& 0.207 & 0.057 & 0.001
& 0.335 & 0.560 & 0.679
& 0.553 & 0.918 & 0.980
& 0.556 & 0.781 & 0.805
& 0.527 & 0.742 & 0.797 \\
& 200
& 0.371 & 0.871 & 0.925
& 0.004 & 0.001 & 0.001
& 0.177 & 0.026 & 0.003
& 0.307 & 0.591 & 0.650
& 0.419 & 0.859 & 0.942
& 0.597 & 0.822 & 0.796
& 0.558 & 0.768 & 0.784 \\
& 300
& 0.311 & 0.859 & 0.919
& 0.003 & 0.001 & 0.001
& 0.220 & 0.048 & 0.003
& 0.353 & 0.584 & 0.646
& 0.389 & 0.803 & 0.910
& 0.590 & 0.810 & 0.850
& 0.560 & 0.757 & 0.820 \\
\hline

\multirow{3}{*}{\shortstack{V=4\\$\pi$=1}}
& 100
& 0.366 & 0.906 & 0.918
& 0.001 & 0.000 & 0.000
& 0.107 & 0.777 & 0.886
& 0.082 & 0.349 & 0.552
& 0.340 & 0.847 & 0.959
& 0.469 & 0.745 & 0.789
& 0.417 & 0.689 & 0.771 \\
& 200
& 0.301 & 0.854 & 0.914
& 0.001 & 0.001 & 0.001
& 0.071 & 0.684 & 0.858
& 0.046 & 0.349 & 0.461
& 0.215 & 0.771 & 0.927
& 0.415 & 0.804 & 0.817
& 0.347 & 0.765 & 0.781 \\
& 300
& 0.186 & 0.867 & 0.889
& 0.001 & 0.000 & 0.000
& 0.055 & 0.675 & 0.896
& 0.030 & 0.334 & 0.482
& 0.137 & 0.689 & 0.844
& 0.353 & 0.842 & 0.851
& 0.308 & 0.767 & 0.843 \\
\hline

\multirow{3}{*}{\shortstack{V=2\\$\pi$=0}}
& 100
& 0.806 & 0.917 & 0.937
& 0.397 & 0.117 & 0.020
& 0.370 & 0.116 & 0.021
& 0.329 & 0.548 & 0.639
& 0.643 & 0.880 & 0.942
& 0.739 & 0.761 & 0.827
& 0.718 & 0.749 & 0.812 \\
& 200
& 0.730 & 0.902 & 0.900
& 0.417 & 0.142 & 0.027
& 0.398 & 0.140 & 0.024
& 0.327 & 0.547 & 0.659
& 0.509 & 0.824 & 0.921
& 0.786 & 0.808 & 0.826
& 0.749 & 0.816 & 0.822 \\
& 300
& 0.657 & 0.893 & 0.898
& 0.427 & 0.176 & 0.042
& 0.424 & 0.161 & 0.041
& 0.297 & 0.561 & 0.602
& 0.472 & 0.774 & 0.893
& 0.749 & 0.828 & 0.858
& 0.719 & 0.826 & 0.848 \\
\hline

\multirow{3}{*}{\shortstack{V=2\\$\pi$=0.5}}
& 100
& 0.799 & 0.888 & 0.915
& 0.232 & 0.043 & 0.013
& 0.616 & 0.665 & 0.628
& 0.247 & 0.486 & 0.604
& 0.583 & 0.903 & 0.941
& 0.731 & 0.776 & 0.777
& 0.694 & 0.755 & 0.771 \\
& 200
& 0.727 & 0.905 & 0.923
& 0.253 & 0.050 & 0.005
& 0.565 & 0.712 & 0.590
& 0.231 & 0.445 & 0.592
& 0.495 & 0.791 & 0.916
& 0.733 & 0.825 & 0.831
& 0.642 & 0.824 & 0.839 \\
& 300
& 0.682 & 0.899 & 0.909
& 0.219 & 0.070 & 0.002
& 0.514 & 0.705 & 0.612
& 0.221 & 0.446 & 0.600
& 0.432 & 0.784 & 0.873
& 0.745 & 0.856 & 0.882
& 0.683 & 0.864 & 0.871 \\
\hline

\multirow{3}{*}{\shortstack{V=2\\$\pi$=1}}
& 100
& 0.749 & 0.920 & 0.925
& 0.018 & 0.000 & 0.000
& 0.557 & 0.838 & 0.883
& 0.158 & 0.388 & 0.491
& 0.587 & 0.887 & 0.936
& 0.676 & 0.749 & 0.808
& 0.577 & 0.730 & 0.796 \\
& 200
& 0.721 & 0.906 & 0.920
& 0.014 & 0.001 & 0.001
& 0.448 & 0.858 & 0.904
& 0.110 & 0.324 & 0.473
& 0.390 & 0.807 & 0.917
& 0.666 & 0.872 & 0.824
& 0.527 & 0.833 & 0.821 \\
& 300
& 0.686 & 0.903 & 0.912
& 0.015 & 0.000 & 0.001
& 0.390 & 0.871 & 0.904
& 0.084 & 0.287 & 0.434
& 0.303 & 0.750 & 0.841
& 0.644 & 0.830 & 0.836
& 0.473 & 0.810 & 0.824 \\
\hline

\end{tabular}
}%
\vspace{0.5em}

\parbox{\linewidth}{
\small
\textit{\textbf{Notes}:} Exact recovery is the proportion of replications
in which the exact model is recovered. Higher values indicate better
model selection performance.
Lasso denotes the PCA prefiltered Lasso with the one standard error
(1SE) penalty selection rule, while Ad. Lasso denotes adaptive Lasso
with the 1SE rule applied at both penalisation stages (1SE--1SE).
}

\end{table}
\vspace*{\fill}
\end{landscape}

\begin{landscape}
\vspace*{\fill}
\begin{table}[htbp]
\thispagestyle{empty}
\centering
\small
\setlength{\tabcolsep}{3.0pt}
\caption{\textbf{Average Model Size (selected variables)}; $\alpha=0.8$, $k=4$}
\makebox[\linewidth][c]{%
\begin{tabular}{|c|c|ccc|ccc|ccc|ccc|ccc|ccc|ccc|}
\hline
&
& \multicolumn{3}{c|}{\textbf{BMT}}
& \multicolumn{3}{c|}{\textbf{OCMT}}
& \multicolumn{3}{c|}{\textbf{GOCMT}}
& \multicolumn{3}{c|}{\textbf{Lasso}}
& \multicolumn{3}{c|}{\textbf{Ad. Lasso}}
& \multicolumn{3}{c|}{\textbf{MCP}}
& \multicolumn{3}{c|}{\textbf{SCAD}} \\
\cline{3-23}

& $n/T$
& 100 & 200 & 300
& 100 & 200 & 300
& 100 & 200 & 300
& 100 & 200 & 300
& 100 & 200 & 300
& 100 & 200 & 300
& 100 & 200 & 300 \\
\hline

\multirow{3}{*}{\shortstack{V=4\\$\pi$=0}}
& 100
& 3.712 & 4.084 & 4.073
& 7.853 & 8.096 & 8.049
& 7.835 & 8.024 & 8.027
& 5.251 & 4.792 & 4.699
& 4.303 & 4.111 & 4.063
& 4.891 & 4.602 & 4.718
& 5.024 & 4.672 & 4.800 \\
& 200
& 3.606 & 4.110 & 4.106
& 7.735 & 8.119 & 8.090
& 7.786 & 8.125 & 8.095
& 5.561 & 4.893 & 4.574
& 4.913 & 4.238 & 4.134
& 4.797 & 4.372 & 4.378
& 5.004 & 4.458 & 4.391 \\
& 300
& 3.597 & 4.106 & 4.123
& 7.762 & 8.102 & 8.105
& 7.774 & 8.105 & 8.106
& 5.934 & 4.824 & 4.688
& 5.085 & 4.293 & 4.213
& 5.095 & 4.323 & 4.435
& 5.196 & 4.333 & 4.471 \\
\hline

\multirow{3}{*}{\shortstack{V=4\\$\pi$=0.5}}
& 100
& 3.707 & 4.092 & 4.071
& 16.120 & 33.008 & 49.846
& 5.506 & 6.907 & 7.702
& 5.542 & 4.941 & 4.640
& 4.341 & 4.096 & 4.041
& 4.934 & 4.623 & 4.518
& 4.963 & 4.797 & 4.573 \\
& 200
& 3.647 & 4.129 & 4.087
& 22.009 & 48.947 & 85.187
& 5.501 & 6.931 & 7.766
& 5.975 & 4.905 & 4.758
& 4.758 & 4.266 & 4.118
& 4.821 & 4.514 & 4.316
& 4.822 & 4.551 & 4.370 \\
& 300
& 3.616 & 4.115 & 4.096
& 37.012 & 67.774 & 113.772
& 5.592 & 6.959 & 7.703
& 6.084 & 5.074 & 4.651
& 4.998 & 4.363 & 4.169
& 4.695 & 4.502 & 4.415
& 5.165 & 4.636 & 4.439 \\
\hline

\multirow{3}{*}{\shortstack{V=4\\$\pi$=1}}
& 100
& 3.682 & 4.086 & 4.089
& 72.757 & 96.796 & 99.834
& 2.481 & 4.009 & 4.124
& 5.053 & 5.681 & 5.015
& 4.765 & 4.150 & 4.087
& 5.218 & 4.601 & 4.575
& 5.440 & 4.708 & 4.633 \\
& 200
& 3.594 & 4.115 & 4.090
& 183.950 & 197.481 & 199.535
& 2.252 & 3.909 & 4.093
& 5.255 & 5.818 & 5.188
& 5.175 & 4.333 & 4.157
& 5.232 & 4.462 & 4.553
& 5.444 & 4.543 & 4.649 \\
& 300
& 3.715 & 4.107 & 4.101
& 293.180 & 300.000 & 299.949
& 2.079 & 3.848 & 4.096
& 5.069 & 6.236 & 5.633
& 6.083 & 4.507 & 4.262
& 5.410 & 4.531 & 4.452
& 5.899 & 4.655 & 4.427 \\
\hline

\multirow{3}{*}{\shortstack{V=2\\$\pi$=0}}
& 100
& 4.135 & 4.110 & 4.078
& 5.263 & 6.446 & 7.246
& 5.372 & 6.506 & 7.297
& 5.671 & 4.863 & 4.693
& 4.947 & 4.187 & 4.096
& 4.805 & 4.622 & 4.572
& 4.951 & 4.574 & 4.636 \\
& 200
& 4.173 & 4.141 & 4.118
& 5.111 & 6.204 & 7.120
& 5.166 & 6.263 & 7.142
& 5.870 & 5.070 & 4.714
& 5.332 & 4.478 & 4.207
& 4.720 & 4.476 & 4.437
& 4.759 & 4.512 & 4.539 \\
& 300
& 4.170 & 4.110 & 4.112
& 4.940 & 6.072 & 7.097
& 5.035 & 6.100 & 7.145
& 6.484 & 5.463 & 4.630
& 5.634 & 4.622 & 4.283
& 5.013 & 4.367 & 4.288
& 5.121 & 4.456 & 4.273 \\
\hline

\multirow{3}{*}{\shortstack{V=2\\$\pi$=0.5}}
& 100
& 4.126 & 4.106 & 4.081
& 6.131 & 8.442 & 11.010
& 4.458 & 4.568 & 4.733
& 6.245 & 5.242 & 4.899
& 4.785 & 4.233 & 4.100
& 4.547 & 4.689 & 4.770
& 4.678 & 4.741 & 4.801 \\
& 200
& 4.133 & 4.114 & 4.093
& 6.481 & 8.964 & 12.643
& 4.285 & 4.482 & 4.739
& 6.921 & 5.510 & 4.855
& 5.470 & 4.383 & 4.245
& 4.858 & 4.500 & 4.463
& 4.951 & 4.605 & 4.474 \\
& 300
& 4.212 & 4.133 & 4.096
& 6.824 & 9.576 & 13.545
& 4.263 & 4.535 & 4.751
& 7.395 & 5.510 & 5.077
& 5.546 & 4.639 & 4.194
& 4.579 & 4.460 & 4.364
& 4.759 & 4.533 & 4.382 \\
\hline

\multirow{3}{*}{\shortstack{V=2\\$\pi$=1}}
& 100
& 4.138 & 4.116 & 4.083
& 12.965 & 27.635 & 46.456
& 3.805 & 4.189 & 4.172
& 6.744 & 5.748 & 5.242
& 4.920 & 4.220 & 4.081
& 4.720 & 4.727 & 4.504
& 5.005 & 4.777 & 4.540 \\
& 200
& 4.141 & 4.111 & 4.087
& 16.793 & 42.500 & 76.947
& 3.512 & 4.130 & 4.142
& 7.829 & 6.233 & 5.546
& 5.572 & 4.481 & 4.217
& 4.814 & 4.489 & 4.427
& 5.135 & 4.535 & 4.441 \\
& 300
& 4.193 & 4.143 & 4.101
& 25.526 & 52.328 & 104.277
& 3.375 & 4.137 & 4.122
& 7.854 & 6.099 & 5.501
& 6.485 & 4.759 & 4.314
& 4.804 & 4.281 & 4.357
& 5.343 & 4.383 & 4.346 \\
\hline

\end{tabular}
}%
\vspace{0.5em}

\parbox{\linewidth}{
\small
\textit{\textbf{Notes}:} Average model size is the average number of variables selected across replications.
Lasso denotes the PCA prefiltered Lasso with the one standard error (1SE) penalty selection rule, while Ad. Lasso denotes adaptive Lasso with the 1SE rule applied at both penalisation stages (1SE--1SE).
}

\end{table}
\vspace*{\fill}
\end{landscape}

\begin{landscape}
\vspace*{\fill}
\begin{table}[htbp]
\thispagestyle{empty}
\centering
\small
\setlength{\tabcolsep}{3.0pt}
\caption{\textbf{Average Model Size (selected variables)}; $\alpha=0.4$, $k=4$}
\makebox[\linewidth][c]{%
\begin{tabular}{|c|c|ccc|ccc|ccc|ccc|ccc|ccc|ccc|}
\hline
&
& \multicolumn{3}{c|}{\textbf{BMT}}
& \multicolumn{3}{c|}{\textbf{OCMT}}
& \multicolumn{3}{c|}{\textbf{GOCMT}}
& \multicolumn{3}{c|}{\textbf{Lasso}}
& \multicolumn{3}{c|}{\textbf{Ad. Lasso}}
& \multicolumn{3}{c|}{\textbf{MCP}}
& \multicolumn{3}{c|}{\textbf{SCAD}} \\
\cline{3-23}

& $n/T$
& 100 & 200 & 300
& 100 & 200 & 300
& 100 & 200 & 300
& 100 & 200 & 300
& 100 & 200 & 300
& 100 & 200 & 300
& 100 & 200 & 300 \\
\hline

\multirow{3}{*}{\shortstack{V=4\\$\pi$=0}}
& 100
& 3.691 & 4.120 & 4.085
& 7.674 & 8.114 & 8.066
& 7.679 & 8.036 & 8.017
& 5.204 & 4.834 & 4.603
& 4.581 & 4.079 & 4.073
& 5.059 & 4.545 & 4.560
& 5.117 & 4.603 & 4.642 \\
& 200
& 3.586 & 4.111 & 4.080
& 7.540 & 8.109 & 8.074
& 7.555 & 8.114 & 8.084
& 5.687 & 5.070 & 4.699
& 4.941 & 4.292 & 4.113
& 4.658 & 4.410 & 4.335
& 4.810 & 4.516 & 4.387 \\
& 300
& 3.600 & 4.086 & 4.087
& 7.433 & 8.107 & 8.097
& 7.519 & 8.121 & 8.100
& 5.749 & 4.979 & 4.751
& 5.402 & 4.493 & 4.236
& 4.964 & 4.393 & 4.373
& 4.889 & 4.474 & 4.463 \\
\hline

\multirow{3}{*}{\shortstack{V=4\\$\pi$=0.5}}
& 100
& 3.644 & 4.097 & 4.100
& 14.042 & 25.848 & 40.337
& 5.189 & 6.545 & 7.591
& 5.704 & 4.926 & 4.615
& 4.485 & 4.144 & 4.055
& 4.899 & 4.708 & 4.525
& 5.050 & 4.826 & 4.586 \\
& 200
& 3.638 & 4.110 & 4.081
& 17.094 & 35.780 & 63.587
& 5.274 & 6.739 & 7.622
& 6.006 & 4.953 & 4.620
& 5.045 & 4.240 & 4.155
& 4.976 & 4.416 & 4.555
& 5.156 & 4.590 & 4.564 \\
& 300
& 3.539 & 4.096 & 4.087
& 19.097 & 46.105 & 83.609
& 5.144 & 6.616 & 7.591
& 5.920 & 5.097 & 4.752
& 5.155 & 4.515 & 4.208
& 5.043 & 4.387 & 4.328
& 5.213 & 4.500 & 4.423 \\
\hline

\multirow{3}{*}{\shortstack{V=4\\$\pi$=1}}
& 100
& 3.638 & 4.064 & 4.087
& 60.235 & 93.042 & 99.452
& 2.413 & 4.007 & 4.124
& 5.288 & 5.716 & 5.124
& 4.778 & 4.190 & 4.068
& 5.198 & 4.808 & 4.597
& 5.444 & 4.978 & 4.667 \\
& 200
& 3.656 & 4.111 & 4.096
& 168.530 & 190.113 & 198.365
& 2.100 & 3.826 & 4.139
& 5.110 & 5.852 & 5.632
& 5.797 & 4.361 & 4.134
& 5.294 & 4.512 & 4.392
& 5.690 & 4.642 & 4.476 \\
& 300
& 3.625 & 4.114 & 4.116
& 284.399 & 298.920 & 299.634
& 1.997 & 3.810 & 4.082
& 5.057 & 5.814 & 5.403
& 6.419 & 4.657 & 4.373
& 5.353 & 4.323 & 4.355
& 5.809 & 4.507 & 4.330 \\
\hline

\multirow{3}{*}{\shortstack{V=2\\$\pi$=0}}
& 100
& 4.127 & 4.086 & 4.066
& 4.874 & 5.976 & 6.950
& 4.898 & 6.007 & 6.988
& 6.050 & 4.982 & 4.715
& 4.934 & 4.368 & 4.143
& 4.817 & 4.726 & 4.482
& 4.903 & 4.814 & 4.533 \\
& 200
& 4.182 & 4.105 & 4.109
& 4.790 & 5.702 & 6.746
& 4.828 & 5.724 & 6.775
& 6.265 & 5.021 & 4.786
& 5.806 & 4.584 & 4.203
& 4.560 & 4.502 & 4.408
& 4.655 & 4.511 & 4.444 \\
& 300
& 4.176 & 4.112 & 4.104
& 4.705 & 5.676 & 6.550
& 4.734 & 5.748 & 6.554
& 6.788 & 5.131 & 4.869
& 6.108 & 4.858 & 4.373
& 4.824 & 4.474 & 4.250
& 4.893 & 4.482 & 4.314 \\
\hline

\multirow{3}{*}{\shortstack{V=2\\$\pi$=0.5}}
& 100
& 4.143 & 4.119 & 4.092
& 5.664 & 7.504 & 9.442
& 4.329 & 4.441 & 4.508
& 7.055 & 5.347 & 4.882
& 4.966 & 4.190 & 4.145
& 4.878 & 4.753 & 4.634
& 5.061 & 4.855 & 4.668 \\
& 200
& 4.149 & 4.107 & 4.083
& 5.786 & 7.767 & 10.209
& 4.166 & 4.352 & 4.523
& 6.943 & 5.427 & 4.991
& 5.742 & 4.579 & 4.226
& 4.795 & 4.485 & 4.455
& 5.100 & 4.489 & 4.395 \\
& 300
& 4.192 & 4.116 & 4.102
& 5.928 & 8.196 & 11.003
& 4.158 & 4.344 & 4.477
& 7.891 & 5.735 & 4.912
& 6.213 & 4.864 & 4.407
& 4.749 & 4.380 & 4.305
& 4.895 & 4.383 & 4.356 \\
\hline

\multirow{3}{*}{\shortstack{V=2\\$\pi$=1}}
& 100
& 4.173 & 4.086 & 4.078
& 10.483 & 21.252 & 35.531
& 3.696 & 4.184 & 4.126
& 7.321 & 5.738 & 5.192
& 5.179 & 4.290 & 4.149
& 4.861 & 4.677 & 4.505
& 5.128 & 4.619 & 4.607 \\
& 200
& 4.177 & 4.100 & 4.084
& 12.703 & 31.088 & 58.024
& 3.471 & 4.153 & 4.100
& 7.637 & 6.226 & 5.640
& 6.352 & 4.554 & 4.222
& 4.767 & 4.382 & 4.489
& 5.352 & 4.467 & 4.517 \\
& 300
& 4.222 & 4.107 & 4.093
& 16.038 & 37.341 & 74.245
& 3.342 & 4.141 & 4.101
& 8.428 & 6.431 & 5.864
& 6.589 & 4.826 & 4.390
& 4.953 & 4.378 & 4.313
& 5.764 & 4.448 & 4.332 \\
\hline

\end{tabular}
}%
\vspace{0.5em}

\parbox{\linewidth}{
\small
\textit{\textbf{Notes}:} Average model size is the average number of variables selected across replications.
Lasso denotes the PCA prefiltered Lasso with the one standard error (1SE) penalty selection rule, while Ad. Lasso denotes adaptive Lasso with the 1SE rule applied at both penalisation stages (1SE--1SE).
}

\end{table}
\vspace*{\fill}
\end{landscape}

\begin{landscape}
\vspace*{\fill}
\begin{table}[htbp]
\thispagestyle{empty}
\centering
\small
\setlength{\tabcolsep}{3.0pt}
\caption{\textbf{Root Mean Squared Error (RMSE)}; $\alpha=0.8$, $k=4$}
\makebox[\linewidth][c]{%
\begin{tabular}{|c|c|ccc|ccc|ccc|ccc|ccc|ccc|ccc|}
\hline
&
& \multicolumn{3}{c|}{\textbf{BMT}}
& \multicolumn{3}{c|}{\textbf{OCMT}}
& \multicolumn{3}{c|}{\textbf{GOCMT}}
& \multicolumn{3}{c|}{\textbf{Lasso}}
& \multicolumn{3}{c|}{\textbf{Ad. Lasso}}
& \multicolumn{3}{c|}{\textbf{MCP}}
& \multicolumn{3}{c|}{\textbf{SCAD}} \\
\cline{3-23}

& $n/T$
& 100 & 200 & 300
& 100 & 200 & 300
& 100 & 200 & 300
& 100 & 200 & 300
& 100 & 200 & 300
& 100 & 200 & 300
& 100 & 200 & 300 \\
\hline

\multirow{3}{*}{\shortstack{V=4\\$\pi$=0}}
& 100
& 0.939 & 0.350 & 0.264
& 0.672 & 0.469 & 0.373
& 0.689 & 0.483 & 0.372
& 0.690 & 0.412 & 0.315
& 0.764 & 0.360 & 0.262
& 0.670 & 0.393 & 0.308
& 0.635 & 0.398 & 0.310 \\
& 200
& 1.062 & 0.370 & 0.278
& 0.722 & 0.480 & 0.379
& 0.729 & 0.481 & 0.379
& 0.683 & 0.423 & 0.311
& 0.838 & 0.402 & 0.283
& 0.721 & 0.371 & 0.295
& 0.658 & 0.372 & 0.298 \\
& 300
& 1.146 & 0.358 & 0.284
& 0.725 & 0.468 & 0.387
& 0.725 & 0.468 & 0.387
& 0.708 & 0.421 & 0.330
& 0.848 & 0.396 & 0.298
& 0.759 & 0.348 & 0.309
& 0.693 & 0.356 & 0.313 \\
\hline

\multirow{3}{*}{\shortstack{V=4\\$\pi$=0.5}}
& 100
& 0.952 & 0.347 & 0.267
& 1.102 & 1.171 & 1.381
& 0.700 & 0.432 & 0.366
& 0.753 & 0.438 & 0.336
& 0.757 & 0.363 & 0.260
& 0.716 & 0.394 & 0.306
& 0.673 & 0.406 & 0.310 \\
& 200
& 1.032 & 0.381 & 0.275
& 1.394 & 1.600 & 1.995
& 0.711 & 0.429 & 0.369
& 0.803 & 0.444 & 0.348
& 0.817 & 0.372 & 0.275
& 0.713 & 0.396 & 0.291
& 0.637 & 0.394 & 0.297 \\
& 300
& 1.082 & 0.399 & 0.280
& 1.609 & 2.241 & 4.683
& 0.720 & 0.440 & 0.362
& 0.844 & 0.472 & 0.341
& 0.857 & 0.419 & 0.271
& 0.707 & 0.389 & 0.306
& 0.665 & 0.399 & 0.308 \\
\hline

\multirow{3}{*}{\shortstack{V=4\\$\pi$=1}}
& 100
& 1.008 & 0.366 & 0.271
& 327.545 & 12.342 & 11.188
& 1.600 & 0.521 & 0.271
& 1.377 & 0.560 & 0.386
& 0.923 & 0.366 & 0.278
& 0.769 & 0.394 & 0.315
& 0.738 & 0.399 & 0.318 \\
& 200
& 1.172 & 0.396 & 0.284
& 2.980 & 335.514 & 32.086
& 1.719 & 0.627 & 0.313
& 1.522 & 0.645 & 0.427
& 1.001 & 0.416 & 0.285
& 0.846 & 0.394 & 0.330
& 0.790 & 0.403 & 0.338 \\
& 300
& 1.265 & 0.413 & 0.285
& 2.306 & 4.148 & 592.136
& 1.764 & 0.711 & 0.317
& 1.597 & 0.719 & 0.466
& 1.147 & 0.445 & 0.306
& 0.898 & 0.411 & 0.317
& 0.851 & 0.423 & 0.314 \\
\hline

\multirow{3}{*}{\shortstack{V=2\\$\pi$=0}}
& 100
& 0.450 & 0.252 & 0.200
& 0.454 & 0.301 & 0.253
& 0.450 & 0.302 & 0.254
& 0.553 & 0.334 & 0.253
& 0.528 & 0.263 & 0.202
& 0.440 & 0.281 & 0.234
& 0.455 & 0.277 & 0.236 \\
& 200
& 0.476 & 0.270 & 0.216
& 0.448 & 0.301 & 0.257
& 0.455 & 0.302 & 0.259
& 0.583 & 0.356 & 0.272
& 0.578 & 0.307 & 0.216
& 0.444 & 0.281 & 0.234
& 0.452 & 0.281 & 0.239 \\
& 300
& 0.587 & 0.257 & 0.212
& 0.492 & 0.289 & 0.257
& 0.495 & 0.289 & 0.258
& 0.639 & 0.396 & 0.265
& 0.612 & 0.318 & 0.227
& 0.496 & 0.282 & 0.219
& 0.499 & 0.289 & 0.218 \\
\hline

\multirow{3}{*}{\shortstack{V=2\\$\pi$=0.5}}
& 100
& 0.464 & 0.261 & 0.204
& 0.502 & 0.385 & 0.340
& 0.529 & 0.266 & 0.216
& 0.590 & 0.357 & 0.272
& 0.550 & 0.272 & 0.194
& 0.417 & 0.300 & 0.240
& 0.420 & 0.303 & 0.238 \\
& 200
& 0.519 & 0.262 & 0.205
& 0.550 & 0.394 & 0.370
& 0.598 & 0.263 & 0.211
& 0.672 & 0.393 & 0.289
& 0.615 & 0.286 & 0.218
& 0.484 & 0.289 & 0.230
& 0.479 & 0.296 & 0.229 \\
& 300
& 0.563 & 0.265 & 0.213
& 0.611 & 0.416 & 0.394
& 0.619 & 0.267 & 0.218
& 0.684 & 0.406 & 0.304
& 0.616 & 0.316 & 0.216
& 0.469 & 0.292 & 0.233
& 0.476 & 0.296 & 0.236 \\
\hline

\multirow{3}{*}{\shortstack{V=2\\$\pi$=1}}
& 100
& 0.478 & 0.257 & 0.206
& 0.774 & 0.772 & 0.995
& 0.783 & 0.246 & 0.204
& 0.716 & 0.400 & 0.299
& 0.562 & 0.262 & 0.195
& 0.464 & 0.297 & 0.231
& 0.490 & 0.298 & 0.233 \\
& 200
& 0.533 & 0.265 & 0.207
& 1.022 & 1.085 & 1.334
& 0.965 & 0.247 & 0.199
& 0.865 & 0.453 & 0.338
& 0.642 & 0.312 & 0.220
& 0.485 & 0.293 & 0.229
& 0.509 & 0.297 & 0.230 \\
& 300
& 0.596 & 0.274 & 0.211
& 1.262 & 1.300 & 1.863
& 1.062 & 0.265 & 0.202
& 0.939 & 0.449 & 0.337
& 0.757 & 0.336 & 0.231
& 0.512 & 0.275 & 0.226
& 0.554 & 0.283 & 0.226 \\
\hline

\end{tabular}
}%
\vspace{0.5em}

\parbox{\linewidth}{
\small
\textit{\textbf{Notes}:} RMSE is defined as
$\sqrt{\frac{1}{r}\sum_{j=1}^{r}
\left\|\widetilde{\boldsymbol{\beta}}_{n}^{(j)}
-\boldsymbol{\beta}_{n}\right\|^2}$,
where $\widetilde{\boldsymbol{\beta}}_{n}^{(j)}$ is the
$n$ dimensional vector that coincides with the post selection (OLS) estimates on the selected coordinates and is zero elsewhere, and $r$ is the number of Monte Carlo replications. RMSE captures the average in sample estimation error across replications, reflecting how closely the estimated coefficients match their true values.
}

\end{table}
\vspace*{\fill}
\end{landscape}

\begin{landscape}
\vspace*{\fill}
\begin{table}[htbp]
\thispagestyle{empty}
\centering
\small
\setlength{\tabcolsep}{3.0pt}
\caption{\textbf{Root Mean Squared Error (RMSE)}; $\alpha=0.4$, $k=4$}
\makebox[\linewidth][c]{%
\begin{tabular}{|c|c|ccc|ccc|ccc|ccc|ccc|ccc|ccc|}
\hline
&
& \multicolumn{3}{c|}{\textbf{BMT}}
& \multicolumn{3}{c|}{\textbf{OCMT}}
& \multicolumn{3}{c|}{\textbf{GOCMT}}
& \multicolumn{3}{c|}{\textbf{Lasso}}
& \multicolumn{3}{c|}{\textbf{Ad. Lasso}}
& \multicolumn{3}{c|}{\textbf{MCP}}
& \multicolumn{3}{c|}{\textbf{SCAD}} \\
\cline{3-23}

& $n/T$
& 100 & 200 & 300
& 100 & 200 & 300
& 100 & 200 & 300
& 100 & 200 & 300
& 100 & 200 & 300
& 100 & 200 & 300
& 100 & 200 & 300 \\
\hline

\multirow{3}{*}{\shortstack{V=4\\$\pi$=0}}
& 100
& 0.893 & 0.358 & 0.270
& 0.684 & 0.486 & 0.374
& 0.710 & 0.495 & 0.373
& 0.676 & 0.407 & 0.311
& 0.780 & 0.375 & 0.278
& 0.666 & 0.380 & 0.302
& 0.652 & 0.387 & 0.307 \\
& 200
& 1.048 & 0.350 & 0.278
& 0.685 & 0.471 & 0.386
& 0.688 & 0.472 & 0.388
& 0.708 & 0.443 & 0.333
& 0.842 & 0.380 & 0.275
& 0.671 & 0.368 & 0.301
& 0.624 & 0.376 & 0.306 \\
& 300
& 1.107 & 0.360 & 0.274
& 0.697 & 0.472 & 0.382
& 0.703 & 0.472 & 0.383
& 0.719 & 0.447 & 0.346
& 0.845 & 0.429 & 0.301
& 0.710 & 0.368 & 0.305
& 0.651 & 0.373 & 0.309 \\
\hline

\multirow{3}{*}{\shortstack{V=4\\$\pi$=0.5}}
& 100
& 0.999 & 0.357 & 0.273
& 0.980 & 0.896 & 1.040
& 0.753 & 0.431 & 0.361
& 0.803 & 0.440 & 0.327
& 0.782 & 0.355 & 0.250
& 0.720 & 0.395 & 0.298
& 0.680 & 0.400 & 0.300 \\
& 200
& 1.050 & 0.373 & 0.269
& 1.184 & 1.122 & 1.393
& 0.732 & 0.424 & 0.362
& 0.865 & 0.449 & 0.342
& 0.851 & 0.379 & 0.281
& 0.709 & 0.387 & 0.316
& 0.669 & 0.396 & 0.317 \\
& 300
& 1.093 & 0.404 & 0.283
& 1.281 & 1.524 & 1.797
& 0.738 & 0.430 & 0.374
& 0.832 & 0.470 & 0.362
& 0.855 & 0.421 & 0.288
& 0.719 & 0.383 & 0.305
& 0.656 & 0.391 & 0.312 \\
\hline

\multirow{3}{*}{\shortstack{V=4\\$\pi$=1}}
& 100
& 1.054 & 0.360 & 0.280
& 26.375 & 9.594 & 10.851
& 1.593 & 0.515 & 0.279
& 1.376 & 0.596 & 0.393
& 0.932 & 0.387 & 0.266
& 0.785 & 0.427 & 0.319
& 0.756 & 0.434 & 0.323 \\
& 200
& 1.171 & 0.407 & 0.290
& 3.321 & 375.769 & 29.131
& 1.703 & 0.692 & 0.302
& 1.557 & 0.631 & 0.465
& 1.082 & 0.434 & 0.278
& 0.868 & 0.415 & 0.314
& 0.825 & 0.411 & 0.320 \\
& 300
& 1.329 & 0.382 & 0.301
& 2.357 & 4.211 & 898.602
& 1.789 & 0.714 & 0.316
& 1.639 & 0.694 & 0.446
& 1.206 & 0.479 & 0.327
& 0.943 & 0.380 & 0.318
& 0.889 & 0.406 & 0.309 \\
\hline

\multirow{3}{*}{\shortstack{V=2\\$\pi$=0}}
& 100
& 0.458 & 0.253 & 0.201
& 0.455 & 0.293 & 0.252
& 0.447 & 0.294 & 0.255
& 0.585 & 0.346 & 0.261
& 0.523 & 0.289 & 0.203
& 0.463 & 0.298 & 0.226
& 0.455 & 0.299 & 0.228 \\
& 200
& 0.493 & 0.260 & 0.208
& 0.476 & 0.293 & 0.251
& 0.488 & 0.294 & 0.251
& 0.614 & 0.361 & 0.277
& 0.620 & 0.315 & 0.212
& 0.422 & 0.289 & 0.225
& 0.435 & 0.286 & 0.226 \\
& 300
& 0.598 & 0.259 & 0.216
& 0.525 & 0.288 & 0.250
& 0.521 & 0.290 & 0.250
& 0.665 & 0.373 & 0.289
& 0.677 & 0.350 & 0.238
& 0.480 & 0.284 & 0.221
& 0.474 & 0.284 & 0.225 \\
\hline

\multirow{3}{*}{\shortstack{V=2\\$\pi$=0.5}}
& 100
& 0.453 & 0.262 & 0.209
& 0.492 & 0.353 & 0.310
& 0.504 & 0.264 & 0.211
& 0.662 & 0.372 & 0.275
& 0.554 & 0.262 & 0.209
& 0.452 & 0.294 & 0.238
& 0.462 & 0.298 & 0.237 \\
& 200
& 0.532 & 0.261 & 0.207
& 0.547 & 0.371 & 0.340
& 0.613 & 0.259 & 0.212
& 0.700 & 0.390 & 0.292
& 0.621 & 0.321 & 0.218
& 0.479 & 0.285 & 0.234
& 0.501 & 0.284 & 0.229 \\
& 300
& 0.564 & 0.269 & 0.215
& 0.602 & 0.399 & 0.358
& 0.673 & 0.262 & 0.213
& 0.743 & 0.422 & 0.294
& 0.681 & 0.341 & 0.237
& 0.468 & 0.280 & 0.225
& 0.479 & 0.279 & 0.227 \\
\hline

\multirow{3}{*}{\shortstack{V=2\\$\pi$=1}}
& 100
& 0.475 & 0.257 & 0.209
& 0.701 & 0.649 & 0.782
& 0.839 & 0.254 & 0.204
& 0.763 & 0.399 & 0.297
& 0.576 & 0.283 & 0.210
& 0.465 & 0.310 & 0.232
& 0.481 & 0.305 & 0.235 \\
& 200
& 0.544 & 0.259 & 0.208
& 0.830 & 0.814 & 0.974
& 1.015 & 0.254 & 0.196
& 0.892 & 0.454 & 0.339
& 0.692 & 0.313 & 0.220
& 0.486 & 0.276 & 0.237
& 0.519 & 0.282 & 0.239 \\
& 300
& 0.593 & 0.263 & 0.211
& 0.994 & 0.943 & 1.208
& 1.094 & 0.253 & 0.199
& 1.017 & 0.468 & 0.359
& 0.750 & 0.345 & 0.244
& 0.511 & 0.292 & 0.228
& 0.568 & 0.294 & 0.231 \\
\hline

\end{tabular}
}%
\vspace{0.5em}

\parbox{\linewidth}{
\small
\textit{\textbf{Notes}:} RMSE is defined as
$\sqrt{\frac{1}{r}\sum_{j=1}^{r}
\left\|\widetilde{\boldsymbol{\beta}}_{n}^{(j)}
-\boldsymbol{\beta}_{n}\right\|^2}$,
where $\widetilde{\boldsymbol{\beta}}_{n}^{(j)}$ is the
$n$ dimensional vector that coincides with the post selection (OLS) estimates on the selected coordinates and is zero elsewhere, and $r$ is the number of Monte Carlo replications. RMSE captures the average in sample estimation error across replications, reflecting how closely the estimated coefficients match their true values.
}

\end{table}
\vspace*{\fill}
\end{landscape}

\begin{landscape}
\vspace*{\fill}
\begin{table}[htbp]
\thispagestyle{empty}
\centering
\small
\setlength{\tabcolsep}{3.0pt}
\caption{\textbf{Root Mean Squared Forecast Error (RMSFE)}; $\alpha=0.8$, $k=4$}
\makebox[\linewidth][c]{%
\begin{tabular}{|c|c|ccc|ccc|ccc|ccc|ccc|ccc|ccc|}
\hline
&
& \multicolumn{3}{c|}{\textbf{BMT}}
& \multicolumn{3}{c|}{\textbf{OCMT}}
& \multicolumn{3}{c|}{\textbf{GOCMT}}
& \multicolumn{3}{c|}{\textbf{Lasso}}
& \multicolumn{3}{c|}{\textbf{Ad. Lasso}}
& \multicolumn{3}{c|}{\textbf{MCP}}
& \multicolumn{3}{c|}{\textbf{SCAD}} \\
\cline{3-23}

& $n/T$
& 100 & 200 & 300
& 100 & 200 & 300
& 100 & 200 & 300
& 100 & 200 & 300
& 100 & 200 & 300
& 100 & 200 & 300
& 100 & 200 & 300 \\
\hline

\multirow{3}{*}{\shortstack{V=4\\$\pi$=0}}
& 100
& 1.912 & 1.797 & 1.813
& 1.882 & 1.823 & 1.817
& 1.879 & 1.808 & 1.817
& 1.849 & 1.758 & 1.759
& 1.781 & 1.818 & 1.788
& 1.883 & 1.805 & 1.820
& 1.872 & 1.801 & 1.824 \\
& 200
& 1.912 & 1.773 & 1.812
& 1.794 & 1.797 & 1.828
& 1.792 & 1.799 & 1.829
& 1.806 & 1.814 & 1.818
& 1.851 & 1.686 & 1.780
& 1.818 & 1.774 & 1.807
& 1.792 & 1.772 & 1.805 \\
& 300
& 2.038 & 1.740 & 1.702
& 1.930 & 1.781 & 1.725
& 1.931 & 1.779 & 1.725
& 1.855 & 1.759 & 1.788
& 1.897 & 1.741 & 1.736
& 1.929 & 1.735 & 1.705
& 1.911 & 1.735 & 1.709 \\
\hline

\multirow{3}{*}{\shortstack{V=4\\$\pi$=0.5}}
& 100
& 2.041 & 1.764 & 1.769
& 2.059 & 1.860 & 1.887
& 1.967 & 1.765 & 1.765
& 1.713 & 1.819 & 1.785
& 1.846 & 1.756 & 1.748
& 1.941 & 1.777 & 1.773
& 1.952 & 1.777 & 1.773 \\
& 200
& 1.943 & 1.848 & 1.844
& 2.057 & 2.103 & 2.157
& 1.846 & 1.858 & 1.848
& 2.001 & 1.793 & 1.779
& 1.955 & 1.750 & 1.737
& 1.862 & 1.840 & 1.845
& 1.834 & 1.841 & 1.841 \\
& 300
& 1.892 & 1.832 & 1.686
& 2.203 & 2.275 & 2.354
& 1.798 & 1.831 & 1.702
& 1.872 & 1.871 & 1.789
& 1.912 & 1.755 & 1.837
& 1.796 & 1.831 & 1.673
& 1.749 & 1.831 & 1.675 \\
\hline

\multirow{3}{*}{\shortstack{V=4\\$\pi$=1}}
& 100
& 1.936 & 1.886 & 1.759
& 82.883 & 2.458 & 2.197
& 2.301 & 1.932 & 1.757
& 2.286 & 1.769 & 1.741
& 1.964 & 1.743 & 1.730
& 1.863 & 1.879 & 1.760
& 1.896 & 1.891 & 1.757 \\
& 200
& 1.931 & 1.819 & 1.696
& 3.126 & 53.933 & 3.032
& 2.293 & 1.831 & 1.702
& 2.375 & 1.834 & 1.792
& 2.067 & 1.809 & 1.754
& 1.801 & 1.818 & 1.699
& 1.797 & 1.818 & 1.697 \\
& 300
& 2.040 & 1.807 & 1.778
& 2.579 & 3.814 & 56.369
& 2.329 & 1.837 & 1.795
& 2.500 & 1.873 & 1.825
& 1.913 & 1.828 & 1.784
& 1.828 & 1.803 & 1.787
& 1.822 & 1.806 & 1.788 \\
\hline

\multirow{3}{*}{\shortstack{V=2\\$\pi$=0}}
& 100
& 1.623 & 1.480 & 1.468
& 1.644 & 1.481 & 1.472
& 1.629 & 1.479 & 1.471
& 1.519 & 1.540 & 1.439
& 1.596 & 1.464 & 1.442
& 1.633 & 1.487 & 1.473
& 1.631 & 1.485 & 1.474 \\
& 200
& 1.483 & 1.378 & 1.445
& 1.502 & 1.380 & 1.450
& 1.505 & 1.384 & 1.449
& 1.700 & 1.399 & 1.495
& 1.566 & 1.542 & 1.450
& 1.509 & 1.380 & 1.461
& 1.526 & 1.380 & 1.457 \\
& 300
& 1.611 & 1.487 & 1.488
& 1.587 & 1.494 & 1.482
& 1.587 & 1.495 & 1.483
& 1.607 & 1.494 & 1.515
& 1.638 & 1.462 & 1.484
& 1.554 & 1.487 & 1.481
& 1.548 & 1.487 & 1.482 \\
\hline

\multirow{3}{*}{\shortstack{V=2\\$\pi$=0.5}}
& 100
& 1.561 & 1.572 & 1.414
& 1.585 & 1.598 & 1.439
& 1.576 & 1.570 & 1.413
& 1.466 & 1.530 & 1.456
& 1.592 & 1.548 & 1.431
& 1.549 & 1.572 & 1.414
& 1.546 & 1.570 & 1.416 \\
& 200
& 1.511 & 1.497 & 1.518
& 1.532 & 1.527 & 1.544
& 1.513 & 1.496 & 1.512
& 1.598 & 1.542 & 1.544
& 1.585 & 1.566 & 1.527
& 1.531 & 1.491 & 1.515
& 1.519 & 1.498 & 1.518 \\
& 300
& 1.571 & 1.422 & 1.464
& 1.615 & 1.462 & 1.488
& 1.609 & 1.428 & 1.468
& 1.601 & 1.473 & 1.481
& 1.552 & 1.478 & 1.433
& 1.556 & 1.441 & 1.453
& 1.558 & 1.443 & 1.462 \\
\hline

\multirow{3}{*}{\shortstack{V=2\\$\pi$=1}}
& 100
& 1.562 & 1.489 & 1.557
& 1.628 & 1.562 & 1.701
& 1.669 & 1.501 & 1.561
& 1.531 & 1.463 & 1.468
& 1.516 & 1.502 & 1.431
& 1.555 & 1.506 & 1.571
& 1.547 & 1.514 & 1.569 \\
& 200
& 1.648 & 1.558 & 1.500
& 1.808 & 1.745 & 1.686
& 1.905 & 1.556 & 1.499
& 1.672 & 1.391 & 1.545
& 1.617 & 1.472 & 1.448
& 1.612 & 1.566 & 1.503
& 1.625 & 1.562 & 1.504 \\
& 300
& 1.643 & 1.508 & 1.527
& 1.996 & 1.863 & 2.005
& 1.949 & 1.495 & 1.512
& 1.745 & 1.616 & 1.480
& 1.552 & 1.444 & 1.502
& 1.658 & 1.503 & 1.514
& 1.700 & 1.503 & 1.514 \\
\hline

\end{tabular}
}%
\vspace{0.5em}

\parbox{\linewidth}{
\small
\textit{\textbf{Notes}:} RMSFE (Root Mean Squared Forecast Error) is
defined as
$\sqrt{\frac{1}{r}\sum_{j=1}^{r}
\left(\frac{1}{S}\sum_{t=T+1}^{T+S}
\left(y_t^{(j)}-\widehat{y}_t^{(j)}\right)^2\right)}$,
where $r$ is the number of Monte Carlo replications and $S$ is the
length of the forecast evaluation period. We set $S=1$.
Here, $y_t^{(j)}$ and $\widehat{y}_t^{(j)}$ denote the actual and
predicted values, respectively, in replication $j$.
RMSFE captures the average out of sample forecast error.
}

\end{table}
\vspace*{\fill}
\end{landscape}

\begin{landscape}
\vspace*{\fill}
\begin{table}[htbp]
\thispagestyle{empty}
\centering
\small
\setlength{\tabcolsep}{3.0pt}
\caption{\textbf{Root Mean Squared Forecast Error (RMSFE)}; $\alpha=0.4$, $k=4$}
\makebox[\linewidth][c]{%
\begin{tabular}{|c|c|ccc|ccc|ccc|ccc|ccc|ccc|ccc|}
\hline
&
& \multicolumn{3}{c|}{\textbf{BMT}}
& \multicolumn{3}{c|}{\textbf{OCMT}}
& \multicolumn{3}{c|}{\textbf{GOCMT}}
& \multicolumn{3}{c|}{\textbf{Lasso}}
& \multicolumn{3}{c|}{\textbf{Ad. Lasso}}
& \multicolumn{3}{c|}{\textbf{MCP}}
& \multicolumn{3}{c|}{\textbf{SCAD}} \\
\cline{3-23}

& $n/T$
& 100 & 200 & 300
& 100 & 200 & 300
& 100 & 200 & 300
& 100 & 200 & 300
& 100 & 200 & 300
& 100 & 200 & 300
& 100 & 200 & 300 \\
\hline

\multirow{3}{*}{\shortstack{V=4\\$\pi$=0}}
& 100
& 1.909 & 1.696 & 1.790
& 1.862 & 1.710 & 1.792
& 1.872 & 1.704 & 1.796
& 1.956 & 1.822 & 1.794
& 1.939 & 1.793 & 1.762
& 1.888 & 1.719 & 1.797
& 1.844 & 1.713 & 1.793 \\
& 200
& 1.893 & 1.771 & 1.698
& 1.721 & 1.792 & 1.719
& 1.720 & 1.791 & 1.723
& 1.860 & 1.818 & 1.829
& 1.862 & 1.756 & 1.756
& 1.678 & 1.758 & 1.694
& 1.675 & 1.755 & 1.689 \\
& 300
& 1.954 & 1.824 & 1.738
& 1.854 & 1.847 & 1.773
& 1.851 & 1.846 & 1.774
& 1.910 & 1.804 & 1.742
& 1.953 & 1.796 & 1.750
& 1.855 & 1.825 & 1.753
& 1.818 & 1.824 & 1.761 \\
\hline

\multirow{3}{*}{\shortstack{V=4\\$\pi$=0.5}}
& 100
& 1.919 & 1.728 & 1.848
& 1.994 & 1.870 & 1.900
& 1.819 & 1.741 & 1.860
& 1.813 & 1.762 & 1.812
& 1.874 & 1.798 & 1.815
& 1.852 & 1.731 & 1.851
& 1.838 & 1.732 & 1.850 \\
& 200
& 1.983 & 1.785 & 1.627
& 2.060 & 1.992 & 1.870
& 1.864 & 1.788 & 1.650
& 1.888 & 1.849 & 1.807
& 1.867 & 1.826 & 1.729
& 1.865 & 1.800 & 1.634
& 1.865 & 1.796 & 1.629 \\
& 300
& 1.995 & 1.757 & 1.769
& 2.139 & 2.094 & 2.144
& 1.846 & 1.762 & 1.776
& 1.861 & 1.828 & 1.714
& 1.857 & 1.822 & 1.743
& 1.831 & 1.752 & 1.777
& 1.828 & 1.750 & 1.771 \\
\hline

\multirow{3}{*}{\shortstack{V=4\\$\pi$=1}}
& 100
& 2.005 & 1.779 & 1.765
& 7.710 & 2.388 & 2.110
& 2.325 & 1.800 & 1.769
& 2.262 & 1.811 & 1.831
& 1.921 & 1.614 & 1.712
& 1.930 & 1.784 & 1.770
& 1.923 & 1.792 & 1.768 \\
& 200
& 2.063 & 1.718 & 1.758
& 3.185 & 47.363 & 2.954
& 2.525 & 1.795 & 1.769
& 2.416 & 1.853 & 1.749
& 1.909 & 1.786 & 1.734
& 1.945 & 1.716 & 1.773
& 1.901 & 1.723 & 1.766 \\
& 300
& 2.103 & 1.689 & 1.827
& 2.765 & 3.683 & 98.773
& 2.423 & 1.766 & 1.839
& 2.525 & 1.834 & 1.742
& 2.043 & 1.864 & 1.752
& 1.962 & 1.673 & 1.842
& 1.949 & 1.688 & 1.831 \\
\hline

\multirow{3}{*}{\shortstack{V=2\\$\pi$=0}}
& 100
& 1.464 & 1.505 & 1.545
& 1.466 & 1.517 & 1.551
& 1.465 & 1.518 & 1.552
& 1.640 & 1.495 & 1.523
& 1.563 & 1.551 & 1.567
& 1.472 & 1.523 & 1.541
& 1.484 & 1.520 & 1.543 \\
& 200
& 1.478 & 1.569 & 1.418
& 1.491 & 1.548 & 1.421
& 1.494 & 1.552 & 1.422
& 1.617 & 1.480 & 1.514
& 1.579 & 1.533 & 1.533
& 1.483 & 1.581 & 1.424
& 1.488 & 1.582 & 1.425 \\
& 300
& 1.598 & 1.478 & 1.529
& 1.598 & 1.480 & 1.528
& 1.590 & 1.485 & 1.529
& 1.617 & 1.616 & 1.457
& 1.510 & 1.483 & 1.511
& 1.566 & 1.510 & 1.522
& 1.565 & 1.502 & 1.524 \\
\hline

\multirow{3}{*}{\shortstack{V=2\\$\pi$=0.5}}
& 100
& 1.528 & 1.492 & 1.477
& 1.521 & 1.498 & 1.474
& 1.544 & 1.498 & 1.480
& 1.603 & 1.467 & 1.430
& 1.504 & 1.511 & 1.353
& 1.554 & 1.492 & 1.490
& 1.551 & 1.498 & 1.489 \\
& 200
& 1.531 & 1.513 & 1.465
& 1.534 & 1.505 & 1.480
& 1.547 & 1.507 & 1.465
& 1.613 & 1.586 & 1.430
& 1.554 & 1.460 & 1.436
& 1.539 & 1.522 & 1.475
& 1.542 & 1.517 & 1.468 \\
& 300
& 1.601 & 1.526 & 1.502
& 1.593 & 1.578 & 1.518
& 1.651 & 1.519 & 1.504
& 1.679 & 1.547 & 1.433
& 1.626 & 1.501 & 1.469
& 1.574 & 1.511 & 1.503
& 1.558 & 1.518 & 1.510 \\
\hline

\multirow{3}{*}{\shortstack{V=2\\$\pi$=1}}
& 100
& 1.538 & 1.436 & 1.341
& 1.641 & 1.497 & 1.423
& 1.692 & 1.429 & 1.352
& 1.543 & 1.423 & 1.522
& 1.514 & 1.474 & 1.508
& 1.536 & 1.461 & 1.348
& 1.551 & 1.455 & 1.346 \\
& 200
& 1.506 & 1.445 & 1.484
& 1.612 & 1.587 & 1.606
& 1.775 & 1.442 & 1.481
& 1.728 & 1.492 & 1.487
& 1.678 & 1.455 & 1.455
& 1.506 & 1.442 & 1.489
& 1.528 & 1.448 & 1.489 \\
& 300
& 1.603 & 1.470 & 1.492
& 1.738 & 1.672 & 1.748
& 1.804 & 1.478 & 1.484
& 1.731 & 1.465 & 1.506
& 1.589 & 1.517 & 1.468
& 1.579 & 1.473 & 1.487
& 1.591 & 1.476 & 1.496 \\
\hline

\end{tabular}
}%
\vspace{0.5em}

\parbox{\linewidth}{
\small
\textit{\textbf{Notes}:} RMSFE (Root Mean Squared Forecast Error) is
defined as
$\sqrt{\frac{1}{r}\sum_{j=1}^{r}
\left(\frac{1}{S}\sum_{t=T+1}^{T+S}
\left(y_t^{(j)}-\widehat{y}_t^{(j)}\right)^2\right)}$,
where $r$ is the number of Monte Carlo replications and $S$ is the length of the forecast evaluation period. We set $S=1$. Here, $y_t^{(j)}$ and $\widehat{y}_t^{(j)}$ denote the actual and predicted values, respectively, in replication $j$. RMSFE captures the average out of sample forecast error.
}

\end{table}
\vspace*{\fill}
\end{landscape}

\begin{landscape}
\vspace*{\fill}
\begin{table}[htbp]
\thispagestyle{empty}
\centering
\small
\setlength{\tabcolsep}{3.0pt}
\caption{\textbf{Computational Time (minutes)}; $\alpha=0.8$, $k=4$}
\makebox[\linewidth][c]{%
\begin{tabular}{|c|c|ccc|ccc|ccc|ccc|ccc|ccc|ccc|}
\hline
&
& \multicolumn{3}{c|}{\textbf{BMT}}
& \multicolumn{3}{c|}{\textbf{OCMT}}
& \multicolumn{3}{c|}{\textbf{GOCMT}}
& \multicolumn{3}{c|}{\textbf{Lasso}}
& \multicolumn{3}{c|}{\textbf{Ad. Lasso}}
& \multicolumn{3}{c|}{\textbf{MCP}}
& \multicolumn{3}{c|}{\textbf{SCAD}} \\
\cline{3-23}

& $n/T$
& 100 & 200 & 300
& 100 & 200 & 300
& 100 & 200 & 300
& 100 & 200 & 300
& 100 & 200 & 300
& 100 & 200 & 300
& 100 & 200 & 300 \\
\hline

\multirow{3}{*}{\shortstack{V=4\\$\pi$=0}}
& 100
& 1.3 & 0.9 & 1.6
& 1.5 & 2.1 & 3.5
& 1.5 & 2.3 & 3.7
& 26.0 & 14.4 & 10.3
& 57.6 & 33.8 & 23.9
& 88.6 & 84.9 & 96.1
& 82.2 & 79.8 & 90.4 \\
& 200
& 0.7 & 1.7 & 3.4
& 2.6 & 4.3 & 7.3
& 2.6 & 4.4 & 7.5
& 6.4 & 53.4 & 44.0
& 13.3 & 116.1 & 94.3
& 373.8 & 399.4 & 456.9
& 361.8 & 374.1 & 428.9 \\
& 300
& 0.9 & 2.5 & 5.0
& 3.2 & 6.6 & 11.1
& 3.2 & 6.7 & 11.4
& 5.1 & 18.3 & 96.1
& 10.4 & 40.3 & 212.1
& 863.8 & 1066.1 & 1220.4
& 848.9 & 1017.5 & 1155.8 \\
\hline

\multirow{3}{*}{\shortstack{V=4\\$\pi$=0.5}}
& 100
& 0.3 & 0.9 & 1.6
& 1.8 & 5.3 & 9.1
& 1.0 & 2.0 & 3.6
& 24.0 & 15.4 & 11.1
& 51.0 & 32.1 & 23.0
& 92.3 & 101.4 & 118.1
& 86.6 & 95.5 & 112.0 \\
& 200
& 0.6 & 1.7 & 3.2
& 5.6 & 21.1 & 44.9
& 1.9 & 3.9 & 7.2
& 6.3 & 57.1 & 48.4
& 14.6 & 121.7 & 102.4
& 420.9 & 603.6 & 819.5
& 406.8 & 579.4 & 788.8 \\
& 300
& 0.9 & 2.5 & 5.0
& 11.4 & 48.9 & 122.1
& 2.8 & 5.7 & 11.0
& 5.2 & 18.7 & 104.8
& 12.4 & 44.8 & 225.2
& 981.0 & 1864.3 & 2824.9
& 971.7 & 1803.7 & 2722.2 \\
\hline

\multirow{3}{*}{\shortstack{V=4\\$\pi$=1}}
& 100
& 0.3 & 0.8 & 1.6
& 4.0 & 1.6 & 0.5
& 0.8 & 1.6 & 2.8
& 24.9 & 17.4 & 13.0
& 53.3 & 35.0 & 25.1
& 130.5 & 164.9 & 202.1
& 125.4 & 159.6 & 196.5 \\
& 200
& 0.6 & 1.6 & 3.2
& 9.5 & 6.4 & 1.8
& 1.4 & 3.0 & 5.2
& 6.7 & 62.2 & 55.8
& 15.9 & 130.3 & 107.9
& 613.3 & 1257.0 & 1760.9
& 598.2 & 1227.0 & 1694.5 \\
& 300
& 0.9 & 2.6 & 5.0
& 18.3 & 11.0 & 3.5
& 2.1 & 4.6 & 8.2
& 6.0 & 19.3 & 110.7
& 14.2 & 52.1 & 245.0
& 1457.5 & 4064.3 & 6396.8
& 1455.6 & 3992.6 & 6169.3 \\
\hline

\multirow{3}{*}{\shortstack{V=2\\$\pi$=0}}
& 100
& 0.4 & 0.9 & 1.6
& 0.9 & 1.9 & 3.3
& 1.0 & 1.9 & 3.6
& 19.7 & 12.6 & 9.9
& 41.6 & 27.5 & 26.3
& 93.1 & 86.9 & 93.7
& 86.3 & 81.6 & 87.9 \\
& 200
& 0.7 & 1.7 & 3.4
& 1.8 & 3.6 & 6.8
& 1.8 & 3.7 & 7.0
& 5.5 & 48.1 & 40.4
& 11.0 & 103.0 & 89.7
& 403.1 & 429.0 & 478.8
& 384.8 & 403.4 & 457.8 \\
& 300
& 1.1 & 2.5 & 5.0
& 2.6 & 5.4 & 10.2
& 2.6 & 5.5 & 10.6
& 4.5 & 16.3 & 86.1
& 9.5 & 37.0 & 185.7
& 957.4 & 1198.9 & 1297.8
& 940.6 & 1146.0 & 1224.7 \\
\hline

\multirow{3}{*}{\shortstack{V=2\\$\pi$=0.5}}
& 100
& 0.4 & 0.8 & 1.6
& 1.0 & 2.0 & 3.7
& 0.9 & 1.7 & 2.9
& 21.2 & 13.8 & 9.8
& 43.1 & 28.9 & 20.4
& 97.4 & 88.7 & 97.2
& 90.6 & 83.2 & 90.1 \\
& 200
& 0.7 & 1.6 & 3.3
& 2.0 & 4.2 & 8.7
& 1.7 & 3.3 & 5.7
& 5.7 & 50.5 & 42.4
& 11.7 & 103.3 & 87.9
& 411.6 & 450.9 & 500.3
& 397.3 & 424.9 & 472.3 \\
& 300
& 1.1 & 2.5 & 4.8
& 3.0 & 6.9 & 14.4
& 2.5 & 4.9 & 8.7
& 4.7 & 16.5 & 89.8
& 9.8 & 36.0 & 189.5
& 975.7 & 1232.4 & 1380.9
& 961.5 & 1168.1 & 1310.5 \\
\hline

\multirow{3}{*}{\shortstack{V=2\\$\pi$=1}}
& 100
& 0.4 & 0.8 & 1.6
& 1.5 & 4.7 & 8.9
& 0.8 & 1.6 & 2.8
& 21.3 & 14.7 & 11.1
& 44.3 & 30.2 & 22.1
& 98.6 & 92.6 & 104.0
& 92.3 & 87.1 & 98.8 \\
& 200
& 0.7 & 1.6 & 3.2
& 4.2 & 18.0 & 42.6
& 1.5 & 3.1 & 5.3
& 5.7 & 53.3 & 45.7
& 12.1 & 107.9 & 89.3
& 439.7 & 516.7 & 626.5
& 424.8 & 490.4 & 596.2 \\
& 300
& 1.0 & 2.4 & 4.9
& 9.2 & 38.7 & 114.8
& 2.2 & 4.5 & 8.1
& 4.7 & 17.6 & 91.9
& 10.2 & 40.6 & 194.4
& 1001.0 & 1509.5 & 2048.3
& 987.8 & 1451.8 & 1971.9 \\
\hline

\end{tabular}
}%
\vspace{0.5em}

\parbox{\linewidth}{
\small
\textit{\textbf{Notes}:} Computational time is the total time, in minutes, required to complete 2,000 Monte Carlo replications. Lower values indicate greater computational efficiency. All methods were implemented under the same computing environment. Lasso denotes the PCA prefiltered Lasso with the one standard error (1SE) penalty selection rule, while Ad. Lasso denotes adaptive Lasso with the 1SE rule applied at both penalisation stages (1SE--1SE).
}

\end{table}
\vspace*{\fill}
\end{landscape}

\begin{landscape}
\vspace*{\fill}
\begin{table}[htbp]
\thispagestyle{empty}
\centering
\small
\setlength{\tabcolsep}{3.0pt}
\caption{\textbf{Computational Time (minutes)}; $\alpha=0.4$, $k=4$}
\makebox[\linewidth][c]{%
\begin{tabular}{|c|c|ccc|ccc|ccc|ccc|ccc|ccc|ccc|}
\hline
&
& \multicolumn{3}{c|}{\textbf{BMT}}
& \multicolumn{3}{c|}{\textbf{OCMT}}
& \multicolumn{3}{c|}{\textbf{GOCMT}}
& \multicolumn{3}{c|}{\textbf{Lasso}}
& \multicolumn{3}{c|}{\textbf{Ad. Lasso}}
& \multicolumn{3}{c|}{\textbf{MCP}}
& \multicolumn{3}{c|}{\textbf{SCAD}} \\
\cline{3-23}

& $n/T$
& 100 & 200 & 300
& 100 & 200 & 300
& 100 & 200 & 300
& 100 & 200 & 300
& 100 & 200 & 300
& 100 & 200 & 300
& 100 & 200 & 300 \\
\hline

\multirow{3}{*}{\shortstack{V=4\\$\pi$=0}}
& 100
& 0.3 & 0.8 & 1.6
& 1.1 & 2.1 & 3.5
& 1.1 & 2.2 & 3.7
& 22.2 & 14.0 & 10.3
& 47.2 & 30.8 & 22.3
& 95.4 & 89.9 & 101.2
& 88.6 & 83.3 & 95.7 \\
& 200
& 0.6 & 1.6 & 3.2
& 2.1 & 4.1 & 7.1
& 2.2 & 4.3 & 7.4
& 6.1 & 54.8 & 45.4
& 12.2 & 115.8 & 96.5
& 416.3 & 440.0 & 501.2
& 402.2 & 413.2 & 470.2 \\
& 300
& 0.9 & 2.4 & 4.9
& 3.2 & 6.5 & 11.3
& 3.3 & 6.6 & 11.7
& 5.1 & 18.4 & 99.5
& 10.4 & 39.7 & 211.9
& 968.5 & 1222.7 & 1346.2
& 954.3 & 1156.7 & 1272.2 \\
\hline

\multirow{3}{*}{\shortstack{V=4\\$\pi$=0.5}}
& 100
& 0.3 & 0.8 & 1.6
& 1.6 & 4.4 & 8.4
& 0.9 & 1.9 & 3.5
& 24.4 & 15.2 & 11.3
& 50.8 & 31.6 & 22.9
& 103.1 & 104.5 & 120.0
& 96.3 & 98.7 & 113.2 \\
& 200
& 0.6 & 1.6 & 3.2
& 4.3 & 15.1 & 36.4
& 1.8 & 3.8 & 7.1
& 6.2 & 57.8 & 49.5
& 13.7 & 123.5 & 103.1
& 446.3 & 627.7 & 835.1
& 429.1 & 605.2 & 789.2 \\
& 300
& 0.9 & 2.4 & 4.9
& 7.3 & 33.5 & 91.1
& 2.6 & 5.7 & 11.0
& 5.1 & 18.7 & 103.8
& 11.0 & 42.9 & 225.8
& 1071.3 & 1960.6 & 2919.2
& 1066.6 & 1880.8 & 2807.9 \\
\hline

\multirow{3}{*}{\shortstack{V=4\\$\pi$=1}}
& 100
& 0.3 & 0.8 & 1.6
& 4.5 & 2.8 & 0.7
& 0.7 & 1.6 & 2.7
& 24.7 & 17.4 & 13.1
& 52.6 & 34.9 & 25.2
& 131.8 & 159.3 & 192.4
& 126.7 & 154.6 & 185.3 \\
& 200
& 0.6 & 1.6 & 3.1
& 12.2 & 14.7 & 3.9
& 1.3 & 2.9 & 5.2
& 6.0 & 62.6 & 55.5
& 15.4 & 130.2 & 107.8
& 604.4 & 1260.5 & 1699.4
& 591.6 & 1224.5 & 1634.8 \\
& 300
& 0.9 & 2.5 & 4.9
& 20.7 & 19.9 & 8.2
& 2.0 & 4.5 & 8.0
& 5.1 & 19.1 & 111.5
& 13.1 & 50.1 & 243.4
& 1469.5 & 3999.9 & 6328.8
& 1474.4 & 3887.0 & 6090.0 \\
\hline

\multirow{3}{*}{\shortstack{V=2\\$\pi$=0}}
& 100
& 0.4 & 0.8 & 1.6
& 0.9 & 1.7 & 3.1
& 0.9 & 1.9 & 3.4
& 19.7 & 13.1 & 9.4
& 43.8 & 28.7 & 20.3
& 109.0 & 96.8 & 102.2
& 102.4 & 89.7 & 96.9 \\
& 200
& 0.7 & 1.6 & 3.2
& 1.7 & 3.4 & 6.2
& 1.7 & 3.6 & 6.5
& 5.4 & 49.2 & 42.1
& 11.3 & 103.7 & 88.3
& 454.1 & 509.1 & 544.2
& 444.4 & 473.8 & 508.7 \\
& 300
& 1.0 & 2.5 & 4.9
& 2.5 & 5.3 & 9.6
& 2.6 & 5.4 & 10.0
& 4.5 & 16.4 & 87.9
& 9.5 & 34.8 & 188.5
& 1052.4 & 1369.8 & 1499.1
& 1055.2 & 1297.1 & 1409.9 \\
\hline

\multirow{3}{*}{\shortstack{V=2\\$\pi$=0.5}}
& 100
& 0.4 & 0.8 & 1.6
& 0.9 & 1.9 & 3.5
& 0.9 & 1.7 & 2.8
& 20.7 & 13.5 & 10.2
& 44.1 & 28.7 & 20.4
& 111.1 & 98.2 & 106.3
& 102.6 & 92.0 & 100.2 \\
& 200
& 0.7 & 1.6 & 3.2
& 1.8 & 3.9 & 7.4
& 1.6 & 3.2 & 5.5
& 5.7 & 51.5 & 44.0
& 11.4 & 105.2 & 90.6
& 458.8 & 517.6 & 560.5
& 445.0 & 486.0 & 529.5 \\
& 300
& 1.0 & 2.4 & 4.9
& 2.8 & 6.2 & 12.5
& 2.4 & 4.8 & 8.6
& 4.7 & 16.2 & 92.3
& 9.2 & 36.3 & 193.8
& 1072.6 & 1416.7 & 1562.6
& 1078.4 & 1347.5 & 1472.1 \\
\hline

\multirow{3}{*}{\shortstack{V=2\\$\pi$=1}}
& 100
& 0.4 & 0.8 & 1.6
& 1.4 & 3.7 & 7.9
& 0.8 & 1.6 & 2.7
& 24.0 & 15.3 & 11.1
& 48.9 & 30.5 & 21.8
& 120.3 & 104.7 & 113.5
& 111.1 & 99.2 & 107.5 \\
& 200
& 0.7 & 1.6 & 3.2
& 3.2 & 12.8 & 34.0
& 1.5 & 3.1 & 5.3
& 8.4 & 53.1 & 45.5
& 15.4 & 106.9 & 93.0
& 501.2 & 572.2 & 668.1
& 484.3 & 537.6 & 635.9 \\
& 300
& 1.1 & 2.4 & 4.8
& 6.4 & 26.8 & 80.7
& 2.2 & 4.6 & 8.0
& 6.0 & 16.7 & 92.8
& 11.7 & 38.0 & 198.3
& 1111.4 & 1643.6 & 2131.3
& 1095.9 & 1575.1 & 2048.5 \\
\hline

\end{tabular}
}%
\vspace{0.5em}

\parbox{\linewidth}{
\small
\textit{\textbf{Notes}:} Computational time is the total time, in minutes, required to complete 2,000 Monte Carlo replications. Lower values indicate greater computational efficiency. All methods were implemented under the same computing environment. Lasso denotes the PCA prefiltered Lasso with the one standard error (1SE) penalty selection rule, while Ad. Lasso denotes adaptive Lasso with the 1SE rule applied at both penalisation stages (1SE--1SE).
}

\end{table}
\vspace*{\fill}
\end{landscape}

\newpage
\section{Additional Simulation Evidence}

\renewcommand{\thetable}
{C\arabic{table}} \setcounter{table}{0}

This section describes two Monte Carlo designs used to illustrate the role of retesting in the BMT procedure. The first design considers a setting in which an inactive composite proxy can enter the model early and suppress a relatively weak signal. The second considers a setting in which one of the true signals is conditionally weak and may therefore be incorrectly removed by retesting. 

\subsubsection{Design 1: Composite proxy}

We consider a model with three true signals and one inactive composite proxy. Let
\[
x_{1t},x_{2t},x_{3t},\nu_t,u_t
\stackrel{\mathrm{iid}}{\sim}N(0,1),
\]
with all variables mutually independent. The outcome is generated as
\[
y_t=x_{1t}+0.3x_{2t}+x_{3t}+u_t,
\]
where \(x_{1t}\) and \(x_{3t}\) are relatively strong signals and \(x_{2t}\) is a weaker signal.

In addition to these three signals, the candidate set contains an inactive composite proxy,
\[
p_t
 =0.5x_{1t}+0.5x_{2t}+0.5x_{3t}+0.7\nu_t.
\]
The coefficient on \(p_t\) in the DGP for \(y_t\) is zero. Hence, \(p_t\) is not a member of the true model, despite being correlated with each of the three relevant regressors. The remaining $n-4$ candidate regressors are independent standard normal noise variables.  We evaluate performance across all combinations of $T
\in\{100, 200, 300\}$ and $n$. This setup allows us to assess how performance varies with dimensionality, without confounding effects from changes in sparsity. Thus, the analysis covers a total of 9 configurations.

The construction of the proxy implies
\[
\operatorname{Var}(p_t)
 =3(0.5)^2+0.7^2
 =1.24,
\]
while its marginal covariance with the dependent variable is
\[
\begin{aligned}
\operatorname{Cov}(p_t,y_t)
 &=0.5\operatorname{Cov}(x_{1t},y_t)
   +0.5\operatorname{Cov}(x_{2t},y_t)
   +0.5\operatorname{Cov}(x_{3t},y_t) \\
 &=0.5+0.5(0.3)+0.5\\
 &=1.15.
\end{aligned}
\]
By comparison,
\[
\operatorname{Cov}(x_{1t},y_t)
 =\operatorname{Cov}(x_{3t},y_t)=1,
\qquad
\operatorname{Cov}(x_{2t},y_t)=0.3.
\]
Thus, despite being inactive, the composite proxy has a larger marginal covariance with the outcome than any of the individual true signals. This makes the proxy particularly likely to enter early in a forward selection procedure.

The early inclusion of the composite proxy can mask the weaker signal \(x_{2t}\). Conditional on the proxy and the two stronger signals, the incremental contribution of \(x_{2t}\) can be substantially
reduced. At the same time, once the relevant signals have entered the model, the proxy becomes redundant because it has no direct effect on \(y_t\). A procedure that permits selected variables to be retested can therefore remove the proxy and, by changing the conditioning set, potentially make the weak signal detectable at a subsequent forward selection stage.

\subsubsection{Design 2: Weak signal}

In the second design all three structured candidate regressors are genuine signals, but one has a relatively small conditional effect.

Let
\[
x_{2t},x_{3t},e_t,\eta_t
\stackrel{\mathrm{iid}}{\sim}N(0,1),
\]
with all innovations mutually independent, and construct
\[
x_{1t}
 =0.6x_{2t}+0.6x_{3t}+\sqrt{0.3}\,e_t.
\]
The outcome is generated according to
\[
y_t
 =0.2x_{1t}+x_{2t}+x_{3t}+\sigma_u\eta_t,
\]
where \(\sigma_u\) is calibrated so that the population coefficient of determination is
\[
R^2=0.7.
\]
The remaining $n-3$ candidate regressors are independent standard normal noise variables. Thus, as in Design~1, we evaluate performance across all combinations of $T
\in\{100, 200, 300\}$ and $n$.

The construction of \(x_{1t}\) gives
\[
\operatorname{Var}(x_{1t})
 =0.6^2+0.6^2+0.3
 =1.02,
\]
and
\[
\operatorname{Var}(x_{1t}\mid x_{2t},x_{3t})=0.3.
\]
Although \(x_{1t}\) has a relatively small coefficient in the true model, it has a strong marginal association with the outcome. In particular,
\begin{equation*}
\operatorname{Cov}(x_{1t},y_t)
 =0.2(1.02)+0.6+0.6=1.404,
\end{equation*}
whereas
\[
\operatorname{Cov}(x_{2t},y_t)
 =\operatorname{Cov}(x_{3t},y_t)
 =1+0.2(0.6)
 =1.12.
\]
Consequently, \(x_{1t}\) is the strongest of the three signals in terms of marginal covariance and is therefore likely to be selected early.

Its conditional contribution, however, is much weaker. The variance of the signal component is
\begin{equation*}
\operatorname{Var}
 \left(0.2x_{1t}+x_{2t}+x_{3t}\right)
 =0.2^2(1.02)+1+1
   +2(0.2)(0.6)+2(0.2)(0.6)=2.5208.
\end{equation*}
Since \(R^2=0.7\), the disturbance variance is calibrated as
\[
\sigma_u^2
 =2.5208\left(\frac{1-0.7}{0.7}\right)
 \simeq 1.08034.
\]
The approximate noncentrality associated with \(x_{1t}\) in the joint true model is therefore
\[
\frac{\sqrt{200}(0.2)\sqrt{0.3}}
     {\sqrt{1.08034}}
 \simeq 1.49.
\]
Thus, although \(x_{1t}\) is highly relevant marginally, its conditional signal is relatively weak. In particular, its joint test statistic can fall below the threshold used in the backward retesting step.

Standard BMT can retain a genuine variable that entered at an early stage, whereas a procedure that subsequently retests previously selected regressors may delete \(x_{1t}\) once the two strongly related signals \(x_{2t}\) and \(x_{3t}\) have entered the model.

\subsubsection{Results for Design 1}

Table~\ref{tab:design1} reports the results for Design~1. 

Standard BMT performs increasingly well as $T$ rises, but the composite proxy remains an important source of selection error. For example, when $n=100$, the probability of selecting the proxy ranges from 0.381 to
0.454 as $T$ increases from 100 to 300. At the same time, the probability of selecting the weak signal rises from 0.253 to 0.860. Thus, a larger sample makes the genuine signal progressively easier to detect, but does
not eliminate the tendency of BMT to retain a proxy that entered at an earlier stage. This is reflected in exact recovery, which increases from 0.152 to 0.463 for $n=100$, but remains well below one even at $T=300$.

BMT-R1 substantially reduces this source of over selection. BMT-R1 applies a final retest after the forward BMT search. The proxy selection frequency falls relative to standard BMT, to approximately 0.05--0.07 across the nine $(n,T)$ combinations. This translates into sizeable improvements in MCC, RMSE and exact recovery. For example, at $(n,T)=(100,300)$, MCC increases from 0.899 under BMT to 0.953 under BMT-R1, while exact recovery rises from 0.463 to 0.788. BMT-R1 leaves the probability
of selecting $x_2$ unchanged relative to BMT in every design. Because retesting is applied only after the forward search has terminated, BMT-R1 cannot recover a signal that BMT failed to select.

BMT-R2 addresses this limitation by interleaving forward selection and backward retesting. Once an early proxy is removed, the forward search resumes using the reduced conditioning set, allowing a previously masked signal to enter. For every $(n,T)$ combination, BMT-R2 selects $x_2$ more frequently than either BMT or BMT-R1. For example, when $(n,T)=(100,200)$, the probability of selecting $x_2$ increases from 0.633 under BMT and BMT-R1 to 0.706 under BMT-R2. When $(n,T)=(200,200)$, the corresponding increase is from 0.549 to 0.653. At the same time, BMT-R2 has a similarly low proxy selection frequency to BMT-R1: the probability of selecting the proxy is only around 0.05--0.07. As a result, BMT-R2 generally provides the strongest overall performance
among the three BMT procedures. Its MCC is uniformly higher and its RMSE uniformly lower than those of both BMT and BMT-R1 across the nine $(n,T)$ combinations. Exact recovery also increases substantially. At $(n,T)=(100,300)$, for instance, exact recovery rises from 0.463 for BMT to 0.788 for BMT-R1 and further to 0.822 for BMT-R2.

The remaining methods provide useful benchmarks. OCMT performs relatively poorly in this design because it selects the composite proxy with probability essentially equal to one and has zero exact recovery in all nine cases. Lasso also selects the proxy with high probability, ranging from 0.847 to 0.971, and consequently exhibits low exact recovery rates. MCP and SCAD are considerably more successful at  distinguishing the signal from the proxy, particularly as $T$ increases. Nevertheless, BMT-R2 compares favourably with both methods in terms of MCC and exact recovery. 

\subsubsection{Results for Design 2}

Table~\ref{tab:design2} considers the complementary case in which the variable vulnerable to deletion is not an inactive proxy but a genuine signal. As discussed above, $x_1$ has the largest marginal association with the outcome and is therefore likely to enter the model early.
Conditional on $x_2$ and $x_3$, however, its  contribution is relatively weak. The design is therefore deliberately unfavourable to backward retesting: a procedure that reassesses previously selected variables can mistake a conditionally weak signal for an irrelevant regressor.

As $T$ increases, BMT identifies the true model with high probability. For $n=100$, its MCC rises from 0.790 at $T=100$ to 0.987 at $T=300$, while exact recovery increases from 0.455 to 0.915. The probability of retaining the conditionally weak signal is already 0.962 at $T=100$ and approaches one as $T$ increases. Similar patterns are observed for $n=200$ and $n=300$. Hence, although the signal is weak conditional on the other relevant regressors, the irreversible nature of standard BMT protects a variable that entered when its marginal evidence was strong.

The two retesting procedures behave very differently to ordinary BMT. Both BMT-R1 and BMT-R2 frequently remove the genuine signal during the backward retesting step. For BMT-R1, the probability of retaining the signal is only 0.422 when $(n,T)=(100,100)$ and falls to 0.048 at $(n,T)=(100,200)$. The corresponding probabilities for BMT-R2 are 0.422 and 0.046. Across most of the remaining configurations the selection probability is similarly very low. Exact recovery therefore is small and it never exceeds 0.045 for either retesting procedure, compared with values as high as 0.915 under standard BMT.

Both retesting procedures can delete the genuine signal once its conditional statistic falls below the threshold. Their similar MCC, RMSE and exact recovery results reflect this shared deletion rule.

The benchmark procedures perform substantially better in this design. OCMT achieves the highest MCC in most configurations and selects the weak signal with probability one throughout. Its exact recovery
probability ranges from 0.783 to 0.912. Lasso also performs strongly, with signal selection probabilities above 0.98 and exact recovery rates that increase markedly with $T$. MCP and SCAD improve substantially as
the time dimension increases, although their signal selection and exact recovery probabilities generally remain below those of OCMT, Lasso and standard BMT.

In these two designs, retesting improves selection when it removes an inactive proxy, but can worsen selection when it deletes a genuine signal with a weak conditional contribution.

\begin{table}[htbp]
\centering
\footnotesize
\setlength{\tabcolsep}{2.3pt}
\caption{\textbf{Design 1: Selection and Estimation Performance}}
\label{tab:design1}
\begin{tabular}{|l|c|ccc|ccc|ccc|ccc|ccc|}
\hline
&
&
\multicolumn{3}{c|}{\textbf{MCC}}
& \multicolumn{3}{c|}{\textbf{RMSE}}
& \multicolumn{3}{c|}{\shortstack{\textbf{Exact}\\\textbf{Recovery}}}
& \multicolumn{3}{c|}{\shortstack{$\mathbf{Pr}(\textbf{signal}$\\$\textbf{selected})$}} 
& \multicolumn{3}{c|}{\shortstack{$\mathbf{Pr}(\textbf{proxy}$\\$\textbf{selected})$}} \\
\cline{3-17}

\textbf{Method}
& $n/T$
& 100 & 200 & 300
& 100 & 200 & 300
& 100 & 200 & 300
& 100 & 200 & 300
& 100 & 200 & 300 \\
\hline

\multirow{3}{*}{\textbf{BMT}}
& 100
& 0.775 & 0.857 & 0.899
& 0.432 & 0.287 & 0.205
& 0.152 & 0.371 & 0.463
& 0.253 & 0.633 & 0.860
& 0.381 & 0.410 & 0.454 \\

& 200
& 0.758 & 0.842 & 0.890
& 0.469 & 0.306 & 0.219
& 0.126 & 0.327 & 0.443
& 0.208 & 0.549 & 0.818
& 0.384 & 0.413 & 0.457 \\

& 300
& 0.748 & 0.834 & 0.888
& 0.488 & 0.315 & 0.229
& 0.090 & 0.305 & 0.458
& 0.164 & 0.519 & 0.787
& 0.365 & 0.421 & 0.440 \\
\hline

\multirow{3}{*}{\textbf{BMT-R1}}
& 100
& 0.821 & 0.904 & 0.953
& 0.396 & 0.266 & 0.192
& 0.212 & 0.565 & 0.788
& 0.253 & 0.633 & 0.860
& 0.062 & 0.070 & 0.059 \\

& 200
& 0.804 & 0.891 & 0.945
& 0.432 & 0.279 & 0.202
& 0.169 & 0.490 & 0.740
& 0.208 & 0.549 & 0.818
& 0.052 & 0.060 & 0.052 \\

& 300
& 0.791 & 0.885 & 0.939
& 0.454 & 0.286 & 0.211
& 0.123 & 0.459 & 0.719
& 0.164 & 0.519 & 0.787
& 0.050 & 0.058 & 0.059 \\
\hline

\multirow{3}{*}{\textbf{BMT-R2}}
& 100
& 0.838 & 0.918 & 0.960
& 0.384 & 0.253 & 0.183
& 0.283 & 0.628 & 0.822
& 0.334 & 0.706 & 0.896
& 0.058 & 0.069 & 0.059 \\

& 200
& 0.818 & 0.911 & 0.956
& 0.425 & 0.261 & 0.189
& 0.227 & 0.583 & 0.795
& 0.272 & 0.653 & 0.875
& 0.048 & 0.057 & 0.051 \\

& 300
& 0.803 & 0.908 & 0.953
& 0.448 & 0.264 & 0.195
& 0.173 & 0.565 & 0.786
& 0.230 & 0.639 & 0.859
& 0.046 & 0.054 & 0.058 \\
\hline

\multirow{3}{*}{\textbf{OCMT}}
& 100
& 0.667 & 0.726 & 0.784
& 0.459 & 0.358 & 0.276
& 0.000 & 0.000 & 0.000
& 0.136 & 0.386 & 0.664
& 1.000 & 1.000 & 1.000 \\

& 200
& 0.662 & 0.714 & 0.767
& 0.464 & 0.372 & 0.300
& 0.000 & 0.000 & 0.000
& 0.101 & 0.316 & 0.566
& 0.997 & 1.000 & 1.000 \\

& 300
& 0.656 & 0.704 & 0.759
& 0.467 & 0.386 & 0.309
& 0.000 & 0.000 & 0.000
& 0.071 & 0.254 & 0.513
& 0.997 & 1.000 & 1.000 \\
\hline

\multirow{3}{*}{\textbf{Lasso}}
& 100
& 0.642 & 0.701 & 0.752
& 0.508 & 0.379 & 0.302
& 0.009 & 0.020 & 0.029
& 0.189 & 0.426 & 0.632
& 0.862 & 0.934 & 0.943 \\

& 200
& 0.637 & 0.684 & 0.731
& 0.521 & 0.402 & 0.327
& 0.005 & 0.012 & 0.016
& 0.138 & 0.353 & 0.562
& 0.857 & 0.939 & 0.965 \\

& 300
& 0.632 & 0.680 & 0.726
& 0.533 & 0.411 & 0.336
& 0.004 & 0.012 & 0.011
& 0.102 & 0.309 & 0.525
& 0.847 & 0.947 & 0.971 \\
\hline

\multirow{3}{*}{\textbf{MCP}}
& 100
& 0.764 & 0.879 & 0.919
& 0.424 & 0.265 & 0.193
& 0.173 & 0.515 & 0.671
& 0.411 & 0.792 & 0.925
& 0.361 & 0.210 & 0.115 \\

& 200
& 0.752 & 0.873 & 0.918
& 0.438 & 0.277 & 0.207
& 0.130 & 0.489 & 0.658
& 0.320 & 0.752 & 0.894
& 0.388 & 0.232 & 0.139 \\

& 300
& 0.747 & 0.866 & 0.929
& 0.449 & 0.286 & 0.202
& 0.112 & 0.459 & 0.687
& 0.273 & 0.731 & 0.897
& 0.393 & 0.263 & 0.129 \\
\hline

\multirow{3}{*}{\textbf{SCAD}}
& 100
& 0.747 & 0.875 & 0.916
& 0.425 & 0.258 & 0.191
& 0.146 & 0.497 & 0.654
& 0.424 & 0.805 & 0.927
& 0.458 & 0.247 & 0.119 \\

& 200
& 0.730 & 0.871 & 0.921
& 0.444 & 0.270 & 0.199
& 0.106 & 0.481 & 0.662
& 0.328 & 0.776 & 0.908
& 0.531 & 0.272 & 0.144 \\

& 300
& 0.720 & 0.864 & 0.927
& 0.458 & 0.278 & 0.197
& 0.082 & 0.459 & 0.678
& 0.275 & 0.752 & 0.909
& 0.553 & 0.310 & 0.151 \\
\hline

\end{tabular}

\vspace{0.5em}

\parbox{\linewidth}{
\small
\textit{\textbf{Notes}:}
The table reports results for the nine combinations of
$n,T\in\{100,200,300\}$. Within each method, rows correspond to $n=100,200,300$, while columns within each performance measure correspond to $T=100,200,300$. MCC denotes the Matthews correlation coefficient and
RMSE the root mean squared error. Exact Recovery is the proportion of Monte Carlo replications in which the true model is selected exactly.
$\Pr(\text{signal selected})$ and $\Pr(\text{proxy selected})$ denote the
proportions of replications in which the true signal, $x_2$, and the proxy variable, respectively, are selected. Higher values are preferred for MCC,
Exact Recovery, and $\Pr(\text{signal selected})$, whereas lower values are preferred for RMSE and $\Pr(\text{proxy selected})$.
}
\end{table}

\begin{table}[htbp]
\centering
\footnotesize
\setlength{\tabcolsep}{2.3pt}
\caption{\textbf{Design 2: Selection and Estimation Performance}}
\label{tab:design2}
\begin{tabular}{|l|c|ccc|ccc|ccc|ccc|}
\hline
&
&
\multicolumn{3}{c|}{\textbf{MCC}}
& \multicolumn{3}{c|}{\textbf{RMSE}}
& \multicolumn{3}{c|}{\shortstack{\textbf{Exact}\\\textbf{Recovery}}}
& \multicolumn{3}{c|}{\shortstack{$\mathbf{Pr}(\textbf{signal}$\\$\textbf{selected})$}} \\
\cline{3-14}

\textbf{Method}
& $n/T$
& 100 & 200 & 300
& 100 & 200 & 300
& 100 & 200 & 300
& 100 & 200 & 300 \\
\hline

\multirow{3}{*}{\textbf{BMT}}
& 100
& 0.790 & 0.973 & 0.987
& 1.232 & 0.367 & 0.206
& 0.455 & 0.864 & 0.915
& 0.962 & 0.997 & 1.000 \\

& 200
& 0.744 & 0.966 & 0.986
& 1.353 & 0.432 & 0.207
& 0.349 & 0.850 & 0.906
& 0.958 & 0.998 & 1.000 \\

& 300
& 0.719 & 0.956 & 0.987
& 1.416 & 0.515 & 0.202
& 0.299 & 0.820 & 0.911
& 0.959 & 0.996 & 1.000 \\
\hline

\multirow{3}{*}{\textbf{BMT-R1}}
& 100
& 0.689 & 0.794 & 0.808
& 1.235 & 0.415 & 0.300
& 0.007 & 0.023 & 0.045
& 0.422 & 0.048 & 0.051 \\

& 200
& 0.661 & 0.792 & 0.808
& 1.354 & 0.473 & 0.301
& 0.003 & 0.022 & 0.040
& 0.510 & 0.062 & 0.046 \\

& 300
& 0.645 & 0.784 & 0.808
& 1.417 & 0.550 & 0.296
& 0.003 & 0.009 & 0.036
& 0.557 & 0.071 & 0.040 \\
\hline

\multirow{3}{*}{\textbf{BMT-R2}}
& 100
& 0.689 & 0.794 & 0.808
& 1.235 & 0.415 & 0.300
& 0.007 & 0.023 & 0.045
& 0.422 & 0.046 & 0.050 \\

& 200
& 0.662 & 0.792 & 0.808
& 1.354 & 0.473 & 0.301
& 0.003 & 0.022 & 0.040
& 0.509 & 0.061 & 0.045 \\

& 300
& 0.646 & 0.785 & 0.809
& 1.417 & 0.550 & 0.296
& 0.003 & 0.009 & 0.036
& 0.557 & 0.071 & 0.039 \\
\hline

\multirow{3}{*}{\textbf{OCMT}}
& 100
& 0.974 & 0.984 & 0.987
& 0.347 & 0.235 & 0.190
& 0.831 & 0.890 & 0.909
& 1.000 & 1.000 & 1.000 \\

& 200
& 0.971 & 0.984 & 0.987
& 0.345 & 0.235 & 0.192
& 0.802 & 0.887 & 0.909
& 1.000 & 1.000 & 1.000 \\

& 300
& 0.968 & 0.983 & 0.988
& 0.344 & 0.232 & 0.191
& 0.783 & 0.879 & 0.912
& 1.000 & 1.000 & 1.000 \\
\hline

\multirow{3}{*}{\textbf{Lasso}}
& 100
& 0.906 & 0.952 & 0.970
& 0.459 & 0.266 & 0.206
& 0.619 & 0.781 & 0.858
& 0.981 & 0.997 & 0.998 \\

& 200
& 0.889 & 0.947 & 0.965
& 0.512 & 0.275 & 0.216
& 0.567 & 0.769 & 0.840
& 0.983 & 0.998 & 1.000 \\

& 300
& 0.889 & 0.945 & 0.965
& 0.511 & 0.273 & 0.216
& 0.563 & 0.762 & 0.840
& 0.982 & 0.997 & 0.999 \\
\hline

\multirow{3}{*}{\textbf{MCP}}
& 100
& 0.761 & 0.913 & 0.930
& 0.727 & 0.283 & 0.224
& 0.314 & 0.681 & 0.779
& 0.822 & 0.866 & 0.884 \\

& 200
& 0.728 & 0.907 & 0.945
& 0.909 & 0.298 & 0.221
& 0.278 & 0.634 & 0.804
& 0.834 & 0.871 & 0.901 \\

& 300
& 0.710 & 0.893 & 0.955
& 1.022 & 0.329 & 0.223
& 0.261 & 0.607 & 0.825
& 0.852 & 0.878 & 0.920 \\
\hline

\multirow{3}{*}{\textbf{SCAD}}
& 100
& 0.750 & 0.892 & 0.921
& 0.786 & 0.296 & 0.229
& 0.343 & 0.613 & 0.755
& 0.796 & 0.842 & 0.865 \\

& 200
& 0.745 & 0.874 & 0.934
& 0.937 & 0.335 & 0.226
& 0.359 & 0.560 & 0.764
& 0.811 & 0.853 & 0.879 \\

& 300
& 0.747 & 0.857 & 0.936
& 1.021 & 0.380 & 0.236
& 0.361 & 0.524 & 0.748
& 0.836 & 0.856 & 0.898 \\
\hline

\end{tabular}

\vspace{0.5em}

\parbox{\linewidth}{
\small
\textit{\textbf{Notes}:}
The table reports results for the nine combinations of
$n,T\in\{100,200,300\}$. Within each method, rows correspond to
$n=100,200,300$, while columns within each performance measure correspond
to $T=100,200,300$. MCC denotes the Matthews correlation coefficient and
RMSE the root mean squared error. Exact Recovery is the proportion of
Monte Carlo replications in which the true model is selected exactly.
$\Pr(\text{signal selected})$ denotes the proportion of replications in
which the true signal, $x_1$, is selected. Higher values are preferred for
MCC, Exact Recovery, and $\Pr(\text{signal selected})$, whereas lower
values are preferred for RMSE.
}

\end{table}

\clearpage
\section{Inflation controls}\label{app:inflation_controls}
\renewcommand{\thetable}{D\arabic{table}}
\setcounter{table}{0}

Tables \ref{tab:phillips_controls_change} and \ref{tab:phillips_controls_average} report all eleven forecast selectors and a controls only forecast, using the matched samples and compulsory inflation blocks described in Section~\ref{empirics2}. Only the 200 macroeconomic candidates are subject to selection; candidate counts exclude the intercept and inflation controls. The two targets coincide at $h=1$ and are reported once in Table~\ref{tab:phillips_controls_change}.

\begin{table}[!htbp]
\centering
\caption{Inflation Controls: Change in Inflation}
\label{tab:phillips_controls_change}
\scriptsize
\setlength{\tabcolsep}{3pt}
\begin{tabular}{l*{5}{r}}
\toprule
& No controls & $\Delta\pi$: 1 & $\Delta\pi$: 2 & $\Delta\pi$: 4 & $\pi$: 4 \\
\midrule
\multicolumn{6}{l}{\textit{$h=1$: both targets}} \\
\addlinespace[2pt]
BMT & 38.03 (2.00) & 43.00 (2.00) & -23.31 (1.28) & -19.76 (1.28) & -7.95 (1.28) \\
BMT-R1 & 38.03 (2.00) & 43.00 (2.00) & -14.21 (0.58) & -9.16 (0.58) & -7.95 (1.28) \\
BMT-R2 & 38.03 (2.00) & 43.00 (2.00) & -23.31 (1.28) & -19.76 (1.28) & -7.95 (1.28) \\
OCMT & 28.90 (5.19) & 35.95 (3.22) & -17.78 (0.58) & -9.85 (0.92) & -3.19 (4.00) \\
GOCMT & 31.15 (2.22) & 40.20 (2.47) & -158.31 (0.61) & -91.52 (1.08) & -64.47 (4.00) \\
Lasso (1SE) & -0.54 (0.00) & 10.40 (0.00) & 12.73 (0.00) & 13.08 (0.00) & 17.44 (0.00) \\
Lasso (MinMSE) & 28.92 (1.00) & 34.50 (1.86) & 27.66 (3.06) & 17.99 (3.86) & 10.92 (5.03) \\
Lasso-PF (1SE) & -5.87 (0.75) & 5.82 (1.33) & 10.05 (1.94) & 2.72 (4.64) & -4.73 (5.56) \\
Lasso-PF (MinMSE) & -6.08 (12.92) & -70.31 (25.06) & 5.80 (22.28) & 3.20 (43.50) & -315.93 (30.50) \\
AdLasso (1SE,1SE) & -0.54 (0.00) & 10.40 (0.00) & 12.73 (0.00) & 13.08 (0.00) & 17.44 (0.00) \\
AdLasso (MinMSE,MinMSE) & 28.92 (1.00) & 30.55 (1.47) & 25.88 (2.72) & 22.16 (3.08) & 9.08 (4.14) \\
Controls only & -0.54 (0.00) & 10.40 (0.00) & 12.73 (0.00) & 13.08 (0.00) & 17.44 (0.00) \\
\addlinespace[4pt]
\multicolumn{6}{l}{\textit{Change in inflation: $h=2$}} \\
\addlinespace[2pt]
BMT & -2.07 (1.06) & -10.03 (1.00) & 0.59 (1.00) & 1.22 (1.00) & 3.45 (1.00) \\
BMT-R1 & -0.20 (0.06) & -3.71 (0.00) & -2.09 (0.00) & 1.22 (1.00) & 3.45 (1.00) \\
BMT-R2 & -2.19 (1.03) & -10.03 (1.00) & 0.59 (1.00) & 1.22 (1.00) & 3.45 (1.00) \\
OCMT & -0.17 (0.00) & -3.94 (0.22) & -2.09 (0.00) & 1.15 (1.06) & 3.75 (1.61) \\
GOCMT & -0.17 (0.00) & -3.71 (0.00) & -4.18 (0.31) & 1.22 (1.00) & 3.45 (1.00) \\
Lasso (1SE) & -0.17 (0.00) & -3.71 (0.00) & -2.26 (1.08) & -3.11 (1.72) & 0.27 (0.00) \\
Lasso (MinMSE) & -47.07 (1.42) & -37.82 (2.14) & -48.24 (3.33) & -34.07 (5.42) & -24.77 (3.42) \\
Lasso-PF (1SE) & -0.13 (0.25) & -4.40 (0.47) & -3.42 (1.67) & -5.60 (2.94) & -6.39 (2.78) \\
Lasso-PF (MinMSE) & -18.61 (3.11) & -4.14 (3.06) & -0.65 (10.28) & -12.61 (14.61) & -18.61 (15.31) \\
AdLasso (1SE,1SE) & -0.17 (0.00) & -3.71 (0.00) & -1.56 (0.39) & -3.65 (0.36) & 0.27 (0.00) \\
AdLasso (MinMSE,MinMSE) & -47.39 (1.36) & -36.71 (1.67) & -47.68 (2.83) & -34.26 (4.03) & -26.05 (2.36) \\
Controls only & -0.17 (0.00) & -3.71 (0.00) & -2.09 (0.00) & -1.63 (0.00) & 0.27 (0.00) \\
\addlinespace[4pt]
\multicolumn{6}{l}{\textit{Change in inflation: $h=4$}} \\
\addlinespace[2pt]
BMT & -0.57 (1.44) & -2.82 (1.17) & -3.07 (1.17) & -10.84 (1.03) & -9.62 (1.17) \\
BMT-R1 & 1.28 (0.97) & 2.31 (0.86) & 1.81 (0.83) & -3.55 (0.61) & -4.87 (0.64) \\
BMT-R2 & -1.98 (1.36) & -2.82 (1.17) & -3.07 (1.17) & -10.84 (1.03) & -9.62 (1.17) \\
OCMT & -0.19 (0.25) & -9.00 (0.92) & -14.92 (1.03) & -9.25 (0.72) & -2.54 (0.56) \\
GOCMT & -4.50 (0.22) & -1.70 (1.03) & -4.17 (1.33) & -2.23 (0.81) & -1.29 (0.11) \\
Lasso (1SE) & -0.42 (0.00) & 1.73 (0.00) & 0.80 (0.00) & -1.88 (0.00) & -0.27 (0.00) \\
Lasso (MinMSE) & -0.25 (0.19) & -0.10 (0.28) & 0.54 (0.39) & -7.71 (0.28) & -8.21 (0.58) \\
Lasso-PF (1SE) & -0.42 (0.00) & 1.73 (0.00) & 0.80 (0.00) & -1.88 (0.00) & -0.27 (0.00) \\
Lasso-PF (MinMSE) & -1.81 (0.11) & 1.09 (0.22) & 0.25 (0.19) & -2.18 (0.19) & -2.65 (0.33) \\
AdLasso (1SE,1SE) & -0.42 (0.00) & 1.73 (0.00) & 0.80 (0.00) & -1.88 (0.00) & -0.27 (0.00) \\
AdLasso (MinMSE,MinMSE) & -0.25 (0.19) & -0.10 (0.28) & 0.54 (0.39) & -7.71 (0.28) & -8.21 (0.58) \\
Controls only & -0.42 (0.00) & 1.73 (0.00) & 0.80 (0.00) & -1.88 (0.00) & -0.27 (0.00) \\
\bottomrule
\end{tabular}
\par\vspace{3pt}
\begin{minipage}{\textwidth}
\footnotesize\textit{Notes:} Entries are $R^2_{\mathrm{OOS}}$ in per cent, with the average number of selected macroeconomic candidates in parentheses. All columns use the same initial estimation sample and 36 forecast origins. The no control column includes the intercept. The other columns also include the indicated number of inflation change terms, starting at $\Delta\pi_t$, or inflation level terms, starting at $\pi_t$. These terms are compulsory and are excluded from the reported candidate counts. All selected models are refitted by OLS with the intercept and inflation controls retained. The benchmark is a zero forecast of the corresponding target.
Controls only uses the intercept and the stated inflation block, without macroeconomic candidates. PF denotes principal component prefiltering; MinMSE and 1SE denote the minimum MSE and one standard error rules.
\end{minipage}
\end{table}

\begin{table}[!htbp]
\centering
\caption{Inflation Controls: Average Inflation}
\label{tab:phillips_controls_average}
\scriptsize
\setlength{\tabcolsep}{3pt}
\begin{tabular}{l*{5}{r}}
\toprule
& No controls & $\Delta\pi$: 1 & $\Delta\pi$: 2 & $\Delta\pi$: 4 & $\pi$: 4 \\
\midrule
\multicolumn{6}{l}{\textit{Average inflation: $h=2$}} \\
\addlinespace[2pt]
BMT & 36.93 (1.08) & 45.86 (2.00) & 21.41 (1.64) & 22.77 (1.69) & 29.80 (1.25) \\
BMT-R1 & 36.10 (1.03) & 45.86 (2.00) & 21.80 (1.58) & 22.77 (1.69) & 29.80 (1.25) \\
BMT-R2 & 36.10 (1.03) & 45.86 (2.00) & 21.41 (1.64) & 22.77 (1.69) & 29.80 (1.25) \\
OCMT & 27.99 (6.25) & 45.03 (2.53) & 27.66 (3.42) & -55.16 (4.75) & 28.25 (2.81) \\
GOCMT & 27.36 (3.69) & 42.02 (2.00) & -79.11 (1.72) & -67.62 (2.58) & -25.79 (2.50) \\
Lasso (1SE) & -0.57 (0.00) & 16.05 (0.00) & 15.45 (0.00) & 15.70 (0.00) & 24.86 (0.00) \\
Lasso (MinMSE) & 17.85 (2.14) & 1.78 (2.39) & 3.82 (3.03) & -55.16 (5.47) & 4.34 (6.61) \\
Lasso-PF (1SE) & -0.57 (0.00) & 15.65 (2.92) & 15.45 (0.00) & 0.54 (5.19) & 17.92 (3.06) \\
Lasso-PF (MinMSE) & -367.27 (6.42) & -366.27 (21.86) & -256.24 (10.03) & -120.75 (34.00) & -304.37 (25.25) \\
AdLasso (1SE,1SE) & -0.57 (0.00) & 16.05 (0.00) & 15.45 (0.00) & 15.70 (0.00) & 24.86 (0.00) \\
AdLasso (MinMSE,MinMSE) & 17.94 (2.11) & 2.83 (2.17) & 3.96 (2.78) & -44.52 (4.44) & 2.36 (5.81) \\
Controls only & -0.57 (0.00) & 16.05 (0.00) & 15.45 (0.00) & 15.70 (0.00) & 24.86 (0.00) \\
\addlinespace[4pt]
\multicolumn{6}{l}{\textit{Average inflation: $h=4$}} \\
\addlinespace[2pt]
BMT & 39.55 (4.28) & 47.06 (3.56) & -67.46 (2.25) & -66.21 (2.31) & 20.47 (1.83) \\
BMT-R1 & 40.19 (4.17) & 52.40 (2.69) & -67.69 (2.14) & -67.41 (2.22) & 20.15 (1.78) \\
BMT-R2 & 38.67 (4.08) & 46.66 (2.28) & -67.69 (2.14) & -67.41 (2.22) & 20.15 (1.78) \\
OCMT & 44.04 (6.31) & 44.71 (6.92) & 37.75 (6.11) & 6.91 (10.03) & -65.16 (10.92) \\
GOCMT & 33.55 (7.39) & 61.32 (3.78) & 63.11 (5.78) & 60.49 (12.11) & 11.27 (10.11) \\
Lasso (1SE) & -1.23 (0.08) & 12.43 (0.19) & 12.06 (1.17) & -0.39 (2.08) & 11.07 (2.86) \\
Lasso (MinMSE) & -4.99 (3.56) & -9.69 (5.19) & -45.96 (6.67) & -43.44 (8.47) & -13.75 (11.31) \\
Lasso-PF (1SE) & -84.64 (14.97) & -109.36 (27.19) & -9.39 (24.11) & -86.18 (33.17) & -137.80 (17.28) \\
Lasso-PF (MinMSE) & -365.66 (32.19) & 0.43 (54.61) & -170.99 (49.17) & 11.29 (62.42) & -126.97 (32.39) \\
AdLasso (1SE,1SE) & -1.23 (0.08) & 12.43 (0.19) & 12.04 (1.14) & -0.39 (2.06) & 11.15 (2.33) \\
AdLasso (MinMSE,MinMSE) & -0.98 (3.22) & -7.26 (4.28) & -40.23 (5.03) & -43.73 (5.61) & -15.61 (10.64) \\
Controls only & -1.14 (0.00) & 12.45 (0.00) & 12.40 (0.00) & 13.42 (0.00) & 30.09 (0.00) \\
\bottomrule
\end{tabular}
\par\vspace{3pt}
\begin{minipage}{\textwidth}
\footnotesize\textit{Notes:} Entries are $R^2_{\mathrm{OOS}}$ in per cent, with the average number of selected macroeconomic candidates in parentheses. All columns use the same initial estimation sample and 36 forecast origins. The no control column includes the intercept. The other columns also include the indicated number of inflation change terms, starting at $\Delta\pi_t$, or inflation level terms, starting at $\pi_t$. These terms are compulsory and are excluded from the reported candidate counts. All selected models are refitted by OLS with the intercept and inflation controls retained. The benchmark is a zero forecast of the corresponding target.
Controls only uses the intercept and the stated inflation block, without macroeconomic candidates. PF denotes principal component prefiltering; MinMSE and 1SE denote the minimum MSE and one standard error rules.
\end{minipage}
\end{table}

\clearpage
\putbib[boosting]
\end{bibunit}

\end{document}